\documentclass[12pt]{article}
\usepackage[left=3cm, right=3cm, bottom=4cm]{geometry}

\usepackage{amsmath, amsthm, amssymb,amsbsy}
\usepackage{graphicx}
\usepackage[dvipsnames]{xcolor}
\usepackage[colorlinks,citecolor=blue]{hyperref}
\usepackage{tikz}
\usetikzlibrary{shapes,arrows,positioning,calc, decorations.pathmorphing, patterns}
\usepackage{bbm}
\usepackage{float}
\usepackage{enumitem}
\usepackage{comment}

\usepackage{glossaries}

\newif\ifmaybe
\maybefalse

\usepackage[caption=false,font=normalsize,labelfon
t=sf,textfont=sf]{subfig}
\usepackage{framed}
\usepackage{tcolorbox}
\tcbuselibrary{breakable,skins}

\newtcolorbox{mybox}{                       % this will create the main box
   enhanced,
   colframe=blue!50,
   boxrule=1pt,
   width=1.05\textwidth,
   arc=1mm,
   breakable,                               % this will make the box breaks when 
}

\title{A Notion of Entropy for Stochastic Processes on Marked Rooted Graphs}
\author{Payam Delgosha and Venkat Anantharam\\[2mm]
\small Department of Electrical Engineering and Computer Sciences\\
\small University of California, Berkeley\\
\small \{pdelgosha, ananth\} @ berkeley.edu
}

\newcommand{\ev}[1]{\mathbb{E} \left [ #1 \right ] }
\newcommand{\evwrt}[2]{\mathbb{E}_{#1} \left [ #2 \right ] }
\newcommand{\pr}[1]{\mathbb{P} \left ( #1 \right ) }
\newcommand{\prwrt}[2]{\mathbb{P}_{#1} \left ( #2 \right ) }

\newcommand{\snorm}[1]{\Vert #1 \Vert}
\newcommand{\one}[1]{\mathbbm{1} \left [ #1 \right ]}
\newtheorem{lem}{Lemma}
\newtheorem{thm}{Theorem}
\newtheorem{definition}{Definition}
\newtheorem{prop}{Proposition}

\newtheorem{rem}{Remark}
\newtheorem{cor}{Corollary}

\newcommand{\mG}{\mathcal{G}}
\newcommand{\mH}{\mathcal{H}}

\newcommand{\vm}{\vec{m}}
\newcommand{\vu}{\vec{u}}
\newcommand{\vd}{\vec{d}}

\newcommand{\vmn}{\vec{m}^{(n)}}
\newcommand{\vun}{\vec{u}^{(n)}}

\newcommand{\mn}{m^{(n)}}
\newcommand{\tmn}{\widetilde{m}^{(n)}}
\newcommand{\un}{u^{(n)}}

\newcommand{\tun}{\widetilde{u}^{(n)}}
\newcommand{\tG}{\widetilde{G}}

\newcommand{\mAdelta}{\mathcal{A}^{(\delta)}}
\newcommand{\mAr}{\mathcal{A}^{(r)}}
\newcommand{\mAk}{\mathcal{A}^{(k)}}
\newcommand{\mD}{\mathcal{D}}
\newcommand{\mP}{\mathcal{P}}

\newcommand{\mM}{\mathcal{M}}
\newcommand{\mS}{\mathcal{S}}
\newcommand{\mF}{\mathcal{F}}

\newcommand{\mT}{\mathcal{T}}

\newcommand{\mC}{\mathcal{C}}
\newcommand{\mCdelta}{\mathcal{C}^{(\delta)}}
\newcommand{\mCk}{\mathcal{C}^{(k)}}

\newcommand{\mX}{\mathcal{X}}

\newcommand{\mGb}{\bar{\mathcal{G}}}
\newcommand{\mTb}{\bar{\mathcal{T}}}
\newcommand{\mGn}{\mathcal{G}^{(n)}}

\newcommand{\tx}{\tilde{x}}
\newcommand{\txp}{\tilde{x}'}

\newcommand{\mGnmnun}{\mathcal{G}^{(n)}_{\vmn, \vun}}
\newcommand{\Anmnun}{A^{(n)}_{\vmn, \vun}}
\newcommand{\Bnmnun}{B^{(n)}_{\vmn, \vun}}
\newcommand{\Yn}{Y^{(n)}}
\newcommand{\Gn}{G^{(n)}}
\newcommand{\Gni}{G^{(n_i)}}
\newcommand{\tGn}{\widetilde{G}^{(n)}}
\newcommand{\gn}{g^{(n)}}
\newcommand{\Hn}{H^{(n)}}

\newcommand{\Pn}{P^{(n)}}
\newcommand{\Sn}{S^{(n)}}
\newcommand{\tSn}{\widetilde{S}^{(n)}}

\newcommand{\reals}{\mathbb{R}}
\newcommand{\integers}{\mathbb{Z}}
\newcommand{\nats}{\mathbb{N}}

\newcommand{\ugwt}{\mathsf{UGWT}}
\newcommand{\cugwt}{\mathsf{CUGWT}}

\newcommand{\LP}{L\'{e}vy--Prokhorov }

\newcommand{\dlp}{d_\text{LP}} % Levy Prohorov distance
\newcommand{\dtv}{d_\text{TV}} % Levy Prohorov distance
\newcommand{\bch}{ \Sigma} % Bordenave Caputo entropy
\newcommand{\bchover}{\overline{\Sigma}}
\newcommand{\bchunder}{\underbar{$\Sigma$}}
\newcommand{\condmnun}{|_{(\vmn, \vun)}} % conditioning BC entropy on mn and un

\newcommand{\edgemark}{\Xi} % edge mark set
\newcommand{\vermark}{\Theta} % vertex mark set
\newcommand{\etype}{\varphi} % edge type

\newcommand{\vtype}{\Pi} % vertex type
\newcommand{\vvtype}{\vec{\Pi}} %the vertex type vector
\newcommand{\vdeg}{\vec{\deg}} % the degree vector
\newcommand{\dist}{\text{dist}} %distance between two vertices in a graph
\newcommand{\hP}{\widehat{P}} % sized biased distribution
\newcommand{\hQ}{\widehat{Q}} % sized biased distribution
\newcommand{\hmG}{\widehat{\mathcal{G}}}

\newcommand{\CB}{\mathsf{CB}} % colorblind version of a directed colored graph
\newcommand{\MCB}{\mathsf{MCB}} % marked colorblind version of a directed colored graph
\newcommand{\colored}{\mathsf{C}} % colored version of a marked graph (which is a colored directed graph
\newcommand{\CM}{\mathsf{CM}} % configuration model
\newcommand{\vD}{\vec{D}} % the degree sequence of a colored graph
\newcommand{\vDn}{\vec{D}^{(n)}} % the degree sequence of a colored graph
\newcommand{\Dn}{D^{(n)}} % the degree sequence of a colored graph

\newcommand{\vDpn}{\vec{D}^{'(n)}}
\newcommand{\Dpn}{D^{'(n)}}

\newcommand{\tP}{\widetilde{P}}
\newcommand{\tQ}{\widetilde{Q}}
\newcommand{\vbeta}{\vec{\beta}} % vector of marks
\newcommand{\vbetan}{\vec{\beta}^{(n)}} % vector of marks
\newcommand{\betan}{\beta^{(n)}} % vector of marks
\newcommand{\muk}{\mu^{(k)}}
\newcommand{\vtau}{\vec{\tau}}
\allowdisplaybreaks % to allow equations spread among pages

\colorlet{pedit}{ForestGreen!80!black}
\colorlet{pedit2}{blue}

\newglossary[nlg]{notation}{not}{ntn}{Notation} 
\newglossary[tlg]{term}{trm}{tmm}{Terminology} 
\newglossaryentry{not:mPX}{%
  type = {notation},
  name = {$\mP(X)$},
  description = {The set of Borel probability measures on a metric space $X$}
}%

\newglossaryentry{not:LPdist}{%
  type = {notation},
  name = {$\dlp(\mu, \nu)$},
  description = {The \LP distance between probability measures $\mu$ and $\nu$}
}%

\newglossaryentry{not:edgemarkset}{%
  type = {notation},
  name = {$\edgemark$},
  sort = {edgemarkset},
  description = {Edge mark set}
}%

\newglossaryentry{not:vermarkset}{%
  type = {notation},
  name = {$\vermark$},
  sort = {vermarkset},
  description = {Vertex mark set}
}%

\newglossaryentry{not:VG}{%
  type = {notation},
  name = {$V(G)$},
  description = {Vertex set of graph $G$}
}%

\newglossaryentry{not:edgemark}{%
  type = {notation},
  name = {$\xi_G(v,w)$},
  description = {The mark of edge $(v,w)$ towards vertex $w$}
}%

\newglossaryentry{not:vermark}{%
  type = {notation},
  name = {$\tau_G(u)$},
  description = {The mark of vertex $u$}
}%

\newglossaryentry{not:vtauG}{%
  type = {notation},
  name = {$\vtau_G$},
  description = {The vertex mark vector of $G$, i.e. $\vtau_G = (\tau_G(v): v \in
  V(G))$}
}%

\newglossaryentry{not:deg-x-x'}{%
  type = {notation},
  name = {$\deg_G^{x,x'}(o)$},
  description = {the number of vertices $v$ connected to $o$ in $G$ such that $\xi_G(v,o) = x$ and $\xi_G(o,v) = x'$}
}%

\newglossaryentry{not:deg}{%
  type = {notation},
  name = {$\deg_G(o)$},
  description = {the total number of vertices connected to $o$ in $G$, irrespective of their marks}
}%

\newglossaryentry{not:edgemarkcountvector}{%
  type = {notation},
  name = {$\vm_G$},
  description = {The edge mark count vector of the marked graph $G$, i.e. $\vm_G = (m_G(x, x'): x, x' \in \edgemark)$ where $m_G(x, x')$ is the number of edges $(i,j)$ in $G$ such that $\xi_G(i, j) = x, \xi_G(j,i) = x'$ or $\xi_G(i,j) = x', \xi_G(j,i) = x$}
}%

\newglossaryentry{not:vertexmarkcountvector}{%
  type = {notation},
  name = {$\vu_G$},
  description = {The vertex mark count vector of the marked graph $G$, i.e. $\vu_G = (u_G(\theta): \theta \in \vermark)$ where $u_G(\theta)$ denotes the number of vertices in $G$ carrying mark $\theta$}
}%

\newglossaryentry{not:vertices-connected}{%
  type = {notation},
  name = {$v \sim_G w$},
  description = {vertices $v$ and $w$ are adjacent in $G$}
}%

\newglossaryentry{not:dist_G}{%
  type = {notation},
  name = {$\dist_G(v,w)$},
  description = {The distance between vertices $v$ and $w$ in $G$, which is the
    length of the shortest path connecting $v$ to $w$}
}%

\newglossaryentry{not:G-o}{%
  type = {notation},
  name = {$(G,o)$},
  description = {for a connected marked graph $G$ on a finite or countably
    infinite vertex set and a vertex $o \in V(G)$, we call the pair $(G,o)$ a
    rooted connected marked graph. Also, for a marked graph $G$ which is not
    necessarily connected and a vertex $o \in V(G)$, we use $(G,o)$ to denote $(G(o), o)$, where $G(o)$ is
  the connected component of $o$ in $G$}
}%

\newglossaryentry{not:G-o_h}{%
  type = {notation},
  name = {$(G,o)_h$},
  description = {the $h$--neighborhood of vertex $o$ in $G$}
}%

\newglossaryentry{not:G-o_iso}{%
  type = {notation},
  name = {$[G,o]$},
  description = {the isomorphism class of 
  $(G,o)$}
}%

\newglossaryentry{not:G-o_iso_h}{%
  type = {notation},
  name = {$[G,o]_h$},
  description = {the isomorphism class of the $h$--neighborhood $(G,o)_h$}
}%

\newglossaryentry{not:G_star}{%
  type = {notation},
  name = {$\mGb_*(\edgemark, \vermark)$, $\mGb_*$},
  description = {The space of the isomorphism classes $[G,o]$ of all simple
    marked rooted graphs on a finite of
countable vertex set, where edge and vertex marks come from the sets $\edgemark$
and $\vermark$, respectively. When the mark sets are clear from the context, we
use $\mGb_*$ as a shorthand for $\mGb_*(\edgemark, \vermark)$}
}%

\newglossaryentry{not:T_star}{%
  type = {notation},
  name = {$\mTb_*(\edgemark, \vermark)$, $\mTb_*$},
  description = {Subset of $\mGb_*(\edgemark, \vermark)$ consisting of
    isomorphism classes $[T,o]$ where $T$ is a tree. When the mark sets are
    clear from the context, we use $\mTb_*$ as a shorthand for $\mTb_*(\edgemark,
    \vermark)$}
}%

\newglossaryentry{not:mu_h}{%
  type = {notation},
  name = {$\mu_h$},
  description = {
For $h \ge 0$, 
$\mu_h \in \mP(\mGb_*^h)$ is the law of
$[G,o]_h$, where $[G,o]$ has law $\mu \in \mP(\mGb_*)$}
}%

\newglossaryentry{not:G_star_h}{%
  type = {notation},
  name = {$\mGb_*^h$},
  description = {The space of isomorphism classes of connected rooted marked
    graphs with depth at most $h$}
}%

\newglossaryentry{not:hmG_star}{%
  type = {notation},
  name = {$\hmG_*$},
  description = {The space of isomorphism classes of rooted multigraphs}
}%

\newglossaryentry{not:UG}{%
  type = {notation},
  name = {$U(G)$},
  description = {Neighborhood structure of a typical vertex in a finite graph
    $G$}
}%

\newglossaryentry{not:Ehgg'}{%
  type = {notation},
  name = {$E_h(g,g')(G,o)$},
  description = {The number of $v \sim_G o$ such that $\etype_G^h(o,v) = (g,g')$}
}%

\newglossaryentry{not:Ehgg'-iso}{%
  type = {notation},
  name = {$E_h(g,g')([G,o])$},
  description = {For $[G,o] \in \mGb_*$, integer $h \geq 1$, and $g,g' \in
    \edgemark \times \mGb_*^{h-1}$, we write $E_h(g,g')([G,o])$ for
    $E_h(g,g')(G,o)$, where $(G,o)$ is an arbitrary member of $[G,o]$}
}

\newglossaryentry{not:e-p-g-g'}{%
  type = {notation},
  name = {$e_p(g,g')$},
  description = {For  $h \geq 1$, $P \in \mP(\mGb_*^h)$, and $g, g' \in \Xi
    \times \mGb_*^{h-1}$, we define $e_P(g, g') := \evwrt{P}{E_h(g, g')(G, o)}$.
  Here, $(G, o)$ is a member of the isomorphism class $[G, o]$ that has law $P$}
}

\newglossaryentry{not:Guv}{%
  type = {notation},
  name = {$G(u,v)$},
  description = {The pair $(\xi_G(u,v), (G',v))$ where $G'$ is obtained from $G$ by
    removing the edge $u,v$}
}%

\newglossaryentry{not:Guvh}{%
  type = {notation},
  name = {$G(u,v)_h$},
  description = {The pair $(\xi_G(u,v), (G',v)_h)$ where $G'$ is obtained from $G$ by
    removing the edge $u,v$}
}%

\newglossaryentry{not:Guv_iso}{%
  type = {notation},
  name = {$G[u,v]$},
  description = {The pair $(\xi_G(u,v), [G',v])$ where $G'$ is obtained from $G$ by
    removing the edge $u,v$}
}%

\newglossaryentry{not:Guvh_iso}{%
  type = {notation},
  name = {$G[u,v]_h$},
  description = {The pair $(\xi_G(u,v), [G',v]_h)$ where $G'$ is obtained from $G$ by
    removing the edge $u,v$}
}%

\newglossaryentry{not:ugwhP}{%
  type = {notation},
  name = {$\ensuremath{\ugwt}_h(P)$},
  description = {Unimodular Galton--Watson Tree with depth $h$ neighborhood distribution $P$}
}%

\newglossaryentry{not:oplus}{%
  type = {notation},
  name = {$t \oplus t'$},
  description = {For $t \in \edgemark \times \mTb_*^k$ and $t' \in \edgemark
    \times \mTb_*^l$, returns an object in $\mTb_*^{\max\{k, l+1\}}$, see Figure~\ref{fig:t_oplus_tprime}}
}%

\newglossaryentry{not:unimodular-measure}{%
  type = {notation},
  name = {$\mP_u(\mGb_*)$, $\mP_u(\mTb_*)$},
  description = {The set of unimodular measures on $\mGb_*$ and $\mTb_*$, respectively}
}%

\newglossaryentry{not:degxx'mu}{%
  type = {notation},
  name = {$\deg_{x,x'}(\mu)$},
  description = {The expected number of edges connected to the root in $\mu$ with mark $x$ towards the root and mark $x'$ towards the offspring}
}%

\newglossaryentry{not:degmu}{%
  type = {notation},
  name = {$\deg(\mu)$},
  description = {The expected degree at the root in $\mu$}
}%

\newglossaryentry{not:vdegmu}{%
  type = {notation},
  name = {$\vdeg(\mu)$},
  description = {The average degree vector at the root, i.e. $\vdeg(\mu) = (\deg_{x,x'}(\mu): x, x' \in \edgemark$)}
}%

\newglossaryentry{not:vtypemu}{%
  type = {notation},
  name = {$\vtype_\theta(\mu)$},
  description = {The probability of the root in $\mu$ having mark $\theta \in \vermark$}
}%

\newglossaryentry{not:vvtypemu}{%
  type = {notation},
  name = {$\vvtype(\mu)$},
  description = {The root mark probability vector of $\mu$, i.e. $\vvtype(\mu) = (\vtype_\theta(\mu): \theta \in \vermark)$}
}%

\newglossaryentry{not:Gnmu}{%
  type = {notation},
  name = {$\mGn_{\vm, \vu}$},
  description = {The set of marked graphs $G$ on the vertex set $[n]$ such that $\vm_G = \vm$ and $\vu_G = \vu$}
}%

\newglossaryentry{not:sofd}{%
  type = {notation},
  name = {$s(.)$},
  description = {The function defined on $[0,\infty)$ such that $s(0) =0$ and $s(d) = \frac{d}{2} - \frac{d}{2} \log d$}
}%

\newglossaryentry{not:sofvecd}{%
  type = {notation},
  name = {$s(\vd)$},
  description = {For an average degree vector $\vd = (d_{x,x'}: x, x' \in \edgemark)$, $s(\vd)$ is defined to be
    $\sum_{x,x' \in \edgemark} s(d_{x,x'})$}
}

\newglossaryentry{not:mPh}{%
  type = {notation},
  name = {$\mP_h$},
  description = {The set of strongly admissible probability measures on
  $P \in \mP(\mTb_*^h)$, that is, $P$ is admissible, $H(P) < \infty$ and
  $\evwrt{P}{\deg_T(o) \log \deg_T(o)} < \infty$. In
  Corollary~\ref{cor:deg-log-deg-Ph}, we have shown that $P \in \mP_h$ is
  equivalent to $P$ being admissible and  $\evwrt{P}{\deg_T(o) \log \deg_T(o)} <
  \infty$}
}%

\newglossaryentry{not:etype}{%
  type = {notation},
  name = {$\etype^h_G(u,v)$},
  description = {for two adjacent vertices $u$ and $v$ in a simple marked graph
    $G$ and $h \geq 1$,
  is defined to be the type of the edge between $u$ and $v$, i.e. the pair
  $(G[v,u]_{h-1}, G[u,v]_{h-1}) \in (\edgemark \times \mGb_*^{h-1}) \times (\edgemark \times \mGb_*^{h-1})$}
}%

\newglossaryentry{not:colorset}{%
  type = {notation},
  name = {$\mC$},
  sort = {colorset},
  description = {The set of colors in a colored configuration model, consisting of pairs $(i,j)$ where $1 \leq i,j \leq L$. Here, $L$ is a fixed integer}
}%

\newglossaryentry{not:mhGmC}{%
  type = {notation},
  name = {$\hmG(\mC)$},
  description = {The set of directed multigraphs on a finite or countable vertex
    set with edges having colors coming from $\mC$. More precisely, each member
    of $\hmG(\mC)$ is of the form $G = (V, \omega)$ where $V$ is a finite or
    countable vertex set, $\omega = (\omega_c: c \in \mC)$ and for $c \in \mC$,
    $\omega_c: V^2 \rightarrow \integers_+$. Furthermore, we require that $(i)$ For
    $c \in \mC_=$, $\omega_c(v,v)$ is even for all $v \in V$, and $\omega_c(u,v)
    = \omega_c(v,u)$ for all $u, v \in V$, $(ii)$ For $c \in \mC_{\neq}$, we
    have $\omega_c(u,v) = \omega_{\bar{c}}(v,u)$ for all $u,v \in V$ and $(iii)$
  For all $v \in V$ and $c \in \mC$, $\sum_{u \in V} \omega_c(v,u) < \infty$}
}%

\newglossaryentry{not:mhGmCstar}{%
 type = {notation},
  name = {$\hmG_*(\mC)$},
  description = {Set of equivalence classes of
rooted 
directed colored multigraphs
in $\hmG(\mC)$.}
}%

\newglossaryentry{not:mhG-star-mC}{%
  type = {notation},
  name = {$\hmG(\mC)$},
  description = {The space of isomorphism classes of rooted directed colored multigraphs}
}%

\newglossaryentry{not:mGmC}{%
  type = {notation},
  name = {$\mG(\mC)$},
  description = {The set of simple directed colored graphs in $\hmG(\mC)$, i.e. the set of $G \in \hmG(\mC)$ where $\CB(G)$ is simple}
}%

\newglossaryentry{not:colorblind}{%
  type = {notation},
  name = {$\CB(G)$},
  description = {The colorblind version of a graph $G \in \hmG(\mC)$}
}%

\newglossaryentry{not:mML}{%
  type = {notation},
  name = {$\mM_L$},
  description = {The set of $L \times L$ matrices with nonnegative integer valued entries}
}%

\newglossaryentry{not:mMLdelta}{%
  type = {notation},
  name = {$\mM_L^{(\delta)}$},
  description = {The set of $L \times L$ matrices with nonnegative integer
    valued entries bounded by $\delta$}
}%

\newglossaryentry{not:JhP}{%
  type = {notation},
  name = {$J_h(P)$},
  description = {For $h \ge 1$ and admissible $P \in \mP(\mTb_*^h)$ with $H(P) < \infty$ and $\evwrt{P}{\deg_T(o)}
 > 0$, define
$J_h(P) := -s(d) + H(P) - \frac{d}{2} H(\pi_P) - \sum_{t, t' \in \edgemark
  \times \mTb_*^{h-1}} \evwrt{P}{\log E_h(t, t')!}$, where $d := \evwrt{P}{\deg_T(o)}$ is the average degree at the root and
    $s(d) = \frac{d}{2} - \frac{d}{2} \log d$}
}

\newglossaryentry{not:colored-degree}{%
  type = {notation},
  name = {$\vD^G$},
  description = {The colored degree sequence of a graph $G = (V, \omega) \in \hmG(\mC)$. Here $\vD^G = (D^G(v): v \in V)$, where $D^G(v) = (D^G_c(v): c \in \mC)$ and $D^G_c(v)$ is the number of edges going out of $v$ with color $c$}
}%

\newglossaryentry{not:DGc}{%
  type = {notation},
  name = {$D^G_c(v)$},
  description = {For $G = (V, \omega)\in \hmG(\mC)$, $v \in V$ and $c \in \mC$,
    $D^G_c(v)$ denotes the number of edges going out of $v$ with color $c$}
}

\newglossaryentry{not:DG}{%
  type = {notation},
  name = {$D^G(v)$},
  description = {For $G = (V, \omega)\in \hmG(\mC)$ and $v \in V$, $D^G(v) =
    (D^G_c(v): c \in \mC)$ where $D^G_c(v)$ is the number of edges going out of $v$ with color $c$}
}

\newglossaryentry{not:g-m}{%
  type = {notation},
  name = {$g[m]$},
  description = {For $g \in \edgemark \times \mGb_*$, we call the $\edgemark$ component of $g$
its {\em mark component} and denote it by $g[m]$}
}%

\newglossaryentry{not:g-s}{%
  type = {notation},
  name = {$g[s]$},
  description = {For $g \in \edgemark \times \mGb_*$, we call the $\vermark$ component of $g$
its {\em subgraph component} and denote it by $g[s]$}
}%

\newglossaryentry{not:g-k}{%
  type = {notation},
  name = {$g_k$},
  description = {For $g \in \edgemark \times \mGb_*^h$ and an integer $k \ge 0$, we define $g_k \in
\edgemark \times \mGb_*^{\min\{h, k\}}$ to have the same mark component as $g$, i.e.\ $g_k[m] := g[m]$,
and subgraph component the truncation
  of the subgraph component of $g$ up to depth $k$, i.e.\ $g_k[s] := (g[s])_k$}
}%

\newglossaryentry{not:mDn}{%
  type = {notation},
  name = {$\mD_n$},
  description = {The set of vectors $(D(1), \dots, D(n))$ such that $D(i) \in \mM_L$ for $1 \leq i \leq n$ and $S := \sum_{i=1}^n D(i)$ is a symmetric matrix in $\mM_L$ with even diagonal entries}
}%

\newglossaryentry{not:mhGmD}{%
  type = {notation},
  name = {$\hmG(\vD)$},
  description = {The set of directed colored multigraphs $G \in \hmG(\mC)$ whose colored degree sequence coincides with $\vD \in \mD_n$}
}%

\newglossaryentry{not:hmG-star-mC}{%
  type = {notation},
  name = {$\hmG_*(\mC)$},
  description = {The set of equivalence classes of rooted directed colored
    multigraphs. Each member 
of $\hmG_*(\mC)$ is of the form $[G, o]$ where $G \in \hmG(\mC)$ is
connected and $o$ is a distinguished vertex in $G$}
}

\newglossaryentry{not:mGmDh}{%
  type = {notation},
  name = {$\mG(\vD,h)$},
  description = {The set of directed colored multigraphs $G \in \hmG(\vD)$ where $\CB(G)$ has no cycles of length $l \leq h$}
}%

\newglossaryentry{not:CMD}{%
  type = {notation},
  name = {$\CM(\vD)$},
  description = {For $\vD \in \mD_n$, denotes the law of the configuration model
    given the colored degree sequence $\vD$}
}%

\newglossaryentry{not:cugwt}{%
  type = {notation},
  name = {$\cugwt(P)$},
  description = {Colored Unimodular Galton--Watson Tree, which is a probability
    measure on $\hmG_*(\mC)$. Here, $\mC$ is a color set of size $L \times L$
    and $P$ is a probability distribution over $\mM_L$}
}%

\newglossaryentry{not:coloredG}{%
  type = {notation},
  name = {$\colored(G)$},
  description = {for a simple marked graph $G$ on the vertex set $[n]$, is
    defined to be a simple directed colored graph where each edge is replaced by
  two directed edges with colors coming from the type of that edge}
}%

\newglossaryentry{not:MCB}{%
  type = {notation},
  name = {$\MCB_{\vbeta}(H)$},
  description = {given a simple directed colored graph $H \in \mG(\mC)$ where $\mC
    = \mF \times \mF$ and  $\mF \subset \edgemark \times \mGb_*^{h-1}$ is finite,
    and also a vector $\vbeta = (\beta(v): v \in V)$ with elements in
    $\vermark$, is defined to be a simple marked graph on $V$ where a pair of
    directed edges, one
    directed from $u$ towards $v$ with color $(g,g')$ and one directed from $v$
    towards $u$ with color $(g',g)$, are substituted with a marked edge with
    mark $g[m]$ towards $u$ and $g'[m]$ towards $v$. Furthermore, a vertex $v$ is
    given mark $\beta(v)$}
}%

\newglossaryentry{not:otimes}{%
  type = {notation},
  name = {$(\theta, x) \otimes t$},
  description = {given $h \in \nats$,  $t \in \edgemark \times \mTb_*^{h-1}$, $x \in \edgemark$
and $\theta \in \vermark$, define $(\theta, x) \otimes t$ to be the  element in
$\mTb_*^h$ where the root $o$ has mark $\theta$,  and attached to it is one offspring $v$ where the
subtree of $v$ is isomorphic to $t[s]$ and the edge connecting $o$ to $v$ has
mark $x$ towards $o$ and $t[m]$ towards $v$, see Figure~\ref{fig:otimes}}
}%

\newglossaryentry{not:times}{%
  type = {notation},
  name = {$x \times s$},
  description = {For $x \in \edgemark$ and $s \in \mTb_*$, is the element $t
    \in \edgemark \times \mTb_*$ such that $t[m] = x$ and $t[s] = s$}
}%

\newglossaryentry{not:odot}{%
  type = {notation},
  name = {$s \odot s'$},
  description = {For two rooted trees $s, s'
\in \mTb_*$ which have the same mark at the root, define $s \odot s'$ to be the
element in $\mTb_*$ obtained by jointing $s$ and $s'$ at a common root, see Figure~\ref{fig:odot}}
}%

\newglossaryentry{not:nDbeta}{%
  type = {notation},
  name = {$n(\vD, \vbeta)$},
  description = {For colored degree sequence $\vD = (D(v): v \in [n])$ and
    vertex mark sequence $\vbeta = (\beta(v): v \in [n])$, is defined to be the number
    of distinct pairs $(\vD^\pi, \vbeta^\pi)$ where $\pi$ ranges over the set of
    permutations $\pi: [n] \rightarrow [n]$. Here, for $1\leq i \leq n$,
    $D^\pi(i) = D(\pi(i))$ and $\beta^\pi(i) = \beta(\pi(i))$}
}%

\newglossaryentry{not:Gn}{%
  type = {notation},
  name = {$\mG_n$},
  description = {The set of graphs on the vertex set $[n]$}
}%

\newglossaryentry{not:Gbarn}{%
  type = {notation},
  name = {$\mGb_n$},
  description = {The set of marked graphs on the vertex set $[n]$}
}%

\newglossaryentry{not:mG-star-mC}{%
  type = {notation},
  name = {$\mG_*(\mC)$},
  description = {The subset of $\hmG_*(\mC)$ consisting of equivalence classes
    of rooted directed colored graphs, i.e.\ for which the associated colorblind multigraph $\CB(G)$ is a graph}
}

\newglossaryentry{not:eP-c}{%
  type = {notation},
  name = {$e_P(c)$},
  description = {For $c = (t, t')$ where $t, t' \in \edgemark \times \mTb_*^{h-1}$ for some $h \geq 1$, $e_P(c)$ is defined to be $e_P(t, t')$}
}

\newglossaryentry{not:T-o}{%
  type = {notation},
  name = {$(T,o)$},
  description = {See $(G,o)$},
  nonumberlist
}

\newglossaryentry{not:T-o-h}{%
  type = {notation},
  name = {$(T,o)_h$},
  description = {See $(G,o)_h$},
  nonumberlist
}

\newglossaryentry{not:Eh-tt'-To}{%
  type = {notation},
  name = {$E_h(t,t')(T,o)$},
  description = {See $E_h(g,g')(G,o)$},
  nonumberlist
}

\newglossaryentry{not:Eh-tt'-To-iso}{%
  type = {notation},
  name = {$E_h(t,t')([T,o])$},
  description = {See $E_h(g,g')([G,o])$},
  nonumberlist
}

\newglossaryentry{not:Ekl-tt'-To}{%
  type = {notation},
  name = {$E_{k,l}(t,t')(T,o)$},
  description = {the number of $v\sim_T o $ such that $T(v,o)_{k-1} \equiv t$
    and $T(o,v)_{l-1} \equiv t'$}
}

\newglossaryentry{not:T-u-v}{%
  type = {notation},
  name = {$T(u,v)$},
  description = {See $G(u,v)$},
  nonumberlist
}

\newglossaryentry{not:T-u-v-h}{%
  type = {notation},
  name = {$T(u,v)_h$},
  description = {See $G(u,v)_h$},
  nonumberlist
}

\newglossaryentry{not:T-u-v-iso}{%
  type = {notation},
  name = {$T[u,v]$},
  description = {See $G[u,v]$},
  nonumberlist
}

\newglossaryentry{not:T-u-v-h-iso}{%
  type = {notation},
  name = {$T[u,v]_h$},
  description = {See $G[u,v]_h$},
  nonumberlist
}

\newglossaryentry{not:V-T}{%
  type = {notation},
  name = {$V(T)$},
  description = {See $V(G)$},
  nonumberlist
}

\newglossaryentry{not:T-o-iso}{%
  type = {notation},
  name = {$[T,o]$},
  description = {See $[T,o]$},
  nonumberlist
}

\newglossaryentry{not:T-o-h-iso}{%
  type = {notation},
  name = {$[T,o]_h$},
  description = {See $[G,o]_h$},
  nonumberlist
}

\newglossaryentry{not:deg-T-o}{%
  type = {notation},
  name = {$\deg_T(o)$},
  description = {See $\deg_G(o)$},
  nonumberlist
}

\newglossaryentry{not:deg-T-xx'-o}{%
  type = {notation},
  name = {$\deg_T^{x,x'}(o)$},
  description = {See $\deg_G^{x,x'}(o)$},
  nonumberlist
}

\newglossaryentry{not:dist-T}{%
  type = {notation},
  name = {$\dist_T(v,w)$},
  description = {See $\dist_G(v,w)$},
  nonumberlist
}

\newglossaryentry{not:phi-T-u-v}{%
  type = {notation},
  name = {$\etype_T^h(u,v)$},
  description = {See $\etype_G^h(u,v)$},
  nonumberlist
}

\newglossaryentry{not:tau-T}{%
  type = {notation},
  name = {$\tau_T(u)$},
  description = {See $\tau_G(u)$},
  nonumberlist
}

\newglossaryentry{not:vD-T}{%
  type = {notation},
  name = {$\vD^T$},
  description = {See $\vD^G$},
  nonumberlist
}

\newglossaryentry{not:vm-T}{%
  type = {notation},
  name = {$\vm_T$},
  description = {See $\vm_G$},
  nonumberlist
}

\newglossaryentry{not:vtau-T}{%
  type = {notation},
  name = {$\vtau_T$},
  description = {See $\vtau_G$},
  nonumberlist
}

\newglossaryentry{not:vu-T}{%
  type = {notation},
  name = {$\vu_T$},
  description = {See $\vu_G$},
  nonumberlist
}

\newglossaryentry{not:xi-T-v-w}{%
  type = {notation},
  name = {$\xi_T(v,w)$},
  description = {See $\xi_G(v,w)$},
  nonumberlist
}

\newglossaryentry{not:e-P-tt'}{%
  type = {notation},
  name = {$e_P(t,t')$},
  description = {See $e_P(g,g')$},
  nonumberlist
}

\newglossaryentry{not:tm}{%
  type = {notation},
  name = {$t[m]$},
  description = {See $g[m]$},
  nonumberlist
}

\newglossaryentry{not:ts}{%
  type = {notation},
  name = {$t[s]$},
  description = {See $g[s]$},
  nonumberlist
}

\newglossaryentry{not:tk}{%
  type = {notation},
  name = {$t_k$},
  description = {See $g_k$},
  nonumberlist
}

\newglossaryentry{not:v-simT-w}{%
  type = {notation},
  name = {$v \sim_T w$},
  description = {See $v \sim_G w$},
  nonumberlist
}

\newglossaryentry{trm:admissible}{%
  type = {term},
  name = {admissible},
  description = {A probability distribution $P \in \mP(\mGb_*^h)$ is called admissible if $\evwrt{P}{\deg_G(o)} < \infty$ and for all $g, g' \in \edgemark \times \mGb_*^{h-1}$, we have $e_P(g, g') = e_P(g',g)$}
}

\newglossaryentry{trm:strongadmissible}{%
  type = {term},
  name = {strongly admissible},
  description = {A probability distribution $P \in \mP(\mTb_*^h)$ is called
    strongly admissible if $P$ is admissible, $H(P) < \infty$ and
$\evwrt{P}{\deg_T(o) \log \deg_T(o)} < \infty$. Also, $\mP_h$ denotes the set of
strongly admissible probability distributions $P \in \mP(\mTb_*^h)$}
}

\newglossaryentry{trm:graphical}{%
  type = {term},
  name = {graphical},
  description = {With $h \in \nats$, $\mF \subset \edgemark \times \mTb_*^{h-1}$
    such that $|\mF| = L$ and $\mC = \mF \times \mF$, a  pair $(\theta, D)$
    where $\theta \in \vermark$ and $D = (D_{t,t'}: t, t' \in \mF) \in \mM_L$ is called
    graphical if there exists $[T,o] \in \mTb_*^h$ such that $\tau_T(o) =
    \theta$ and $E_h(t,t')(T,o) = D_{t,t'}$ for all $t, t' \in \mF$. Moreover,
    $E_h(t,t')(T,o) = 0 $ when either $t \notin \mF$ or $t' \notin \mF$}
}

\newglossaryentry{trm:h-treelike}{%
  type = {term},
  name = {$h$ tree--like},
  description = {A marked or unmarked graph $G$ is said to be $h$ tree--like if for all vertices $v$ in
  $G$, the depth $h$ local neighborhood of $v$ in $G$, i.e.\ $(G, v)_h$, is a
  rooted tree. This condition is equivalent to requiring that there is no cycle
  of length $2h+1$ or less in $G$}
}

\newglossaryentry{trm:Delta-graphical-matrix}{%
  type = {term},
  name = {$\Delta$--graphical matrix},
  description = {A matrix $D \in \mM_L^{(\delta)}$ is said to be $\Delta$--graphical if there
exists $r \in \Delta$ such that $D = D(r)$. Here, $\Delta \subset \mTb_*^h$ is
a finite set. Furthermore, $\mF$ denotes the set of $T[o,v]_{h-1}$ and
$T[v,o]_{h-1}$ arising from $[T,o] \in \Delta$ and vertices $v \sim_T o$. Also,
for $r \in \Delta$, $D(r) \in \mM_L^{(\delta)}$ is the matrix such that, for $t,
t' \in \mF$, $D_{t,t'}(r) = E_h(t, t')(r)$} 
}

\newglossaryentry{trm:Delta-graphical-rooted-graph}{%
  type = {term},
  name = {$\Delta$--graphical rooted directed colored graph},
  description = {We say that  a rooted directed colored graph $[F, o] \in
    \mG_*(\mC)$ is $\Delta$--graphical if for each vertex $v$ in $F$, $D^F(v)$ is
$\Delta$--graphical}
}

\makeglossaries
\definecolor{bluenodecolor}{RGB}{62,126,176}
\definecolor{rednodecolor}{RGB}{173,61,58}

\definecolor{blueedgecolor}{RGB}{3,151,255}
\definecolor{orangeedgecolor}{RGB}{255,149,8}

\tikzstyle{nodeB} = [fill=bluenodecolor, circle, inner sep = 4pt]
\tikzstyle{nodeR} = [fill=rednodecolor, rectangle, inner sep = 5pt]
\tikzstyle{edgeB} = [very thick, blueedgecolor]
\tikzstyle{edgeO} = [very thick, orangeedgecolor, decoration = {zigzag,segment length = 0.2cm, amplitude = 0.5mm},decorate]

\newcommand{\drawedge}[4]{\draw[edge#3] (#1) -- ($(#1)!0.5!(#2)$); \draw[edge#4]
  (#2) -- ($(#2)!0.5!(#1)$);}
\newcommand{\nodelabel}[3]{\node at ($(#1)+(#2:5mm)$) {#3};}

\begin{document}

\colorlet{Cyan}{cyan}
\colorlet{Orange}{orange}
\tikzstyle{Node} = [circle,fill,inner sep=1.5pt]
\tikzstyle{Node2} = [rectangle,fill,inner sep=2pt]
\tikzstyle{Root} = [circle,fill=magenta,inner sep=1.7pt]
\tikzstyle{Root2} = [rectangle,fill=magenta,inner sep=2.5pt]
\tikzstyle{Cedge} = [Cyan, thick]
\tikzstyle{Oedge} = [Orange, densely dotted, very thick]

\maketitle

% \begin{tikzpicture}[remember picture, overlay]
%   \node at ($(current page.south east)+(-2,1)$) {\verb+VN041-V5+};
% \end{tikzpicture}

% bb

\begin{abstract}
In this document, we introduce a notion of entropy for stochastic processes on marked
rooted graphs.
%The local weak limit theory for sparse graphs, also known as the objective
%method, due to Benjamini, Schramm,  Aldous, Steele and Lyons, provides a
%framework which enables one to make sense of a stationary process for  graphs
%\cite{BenjaminiSchramm01rec, aldous2004objective, aldous2007processes}.
%Employing this framework, we introduce a notion of entropy for probability
%distributions on marked rooted graphs. 
%\color{red}
For this, we employ the framework of local weak limit theory for sparse marked graphs, also known as the objective
method, due to Benjamini, Schramm,  Aldous, Steele and Lyons
%provides a
%framework which enables one to make sense of a stationary process for  graphs
\cite{BenjaminiSchramm01rec, aldous2004objective, aldous2007processes}.
%\color{black}
%This is a generalization 
%\color{red}
Our contribution is a generalization
%\color{black}
of the notion of
entropy introduced by Bordenave and Caputo 
\cite{bordenave2015large} to graphs which carry marks on their
vertices and  edges.

The theory of time series is the
engine driving an enormous range of applications in areas such as control
theory, communications, information theory and signal processing. It is to be
expected that a theory of stationary stochastic processes 
%for combinatorial structures, 
%\color{red}
indexed by combinatorial structures,
%\color{black}
in particular graphs, would eventually have a
similarly wide-ranging  impact.
% \color{red}
% vacomment: Note, new para for second part
% of the abstract.
% \color{black}
\end{abstract}

%aa
\section{Preliminaries and Notation}
\label{sec:notations}
% bb
 
$\nats$ denotes the set of natural numbers, $\integers_+$ the set of nonnegative integers, $\integers$ the set of integers and $\reals$  the set
real numbers. 
%$\nats$ and $\reals$ denote the set of natural numbers and the set of real numbers,
%respectively. Moreover, $\integers$ and $\integers_+$ denote the sets of
%integers and the set of nonnegative integers, respectively.
For $n \in \nats$, $[n]$
denotes the set $\{1, \dots, n\}$. For a probability distribution $Q$ defined on
a finite set, $H(Q)$ denotes the Shannon entropy of $Q$. For a metric space $X$,
we denote the set of Borel probability measures on $X$ by $\mP(X)$.
\glsadd{not:mPX}
For two Borel probability measures $\mu, \nu \in \mP(X)$, $\dlp(\mu, \nu)$ denotes the
\LP distance between $\mu$ and $\nu$ \cite{billingsley2013convergence}. \glsadd{not:LPdist}
All logarithms in this document are to the natural base. Equality by definition is denoted by $:=$.

\subsection{Marked Graphs}
\label{sec:marked-graphs}

All graphs in this
  document are defined on a finite or countably infinite vertex set, and are assumed to
  be locally finite, i.e. the degree of each vertex is finite.
  Given a graph $G$, we denote its vertex set by $V(G)$. 
  A simple graph is a graph without self-loops or multiple edges between pairs of vertices.
  A simple marked graph is a simple graph where each 
  edge carries two marks coming from a finite
  edge mark set, one towards each of its endpoints, and each vertex carries a mark from a
  finite vertex mark set. We denote the edge and
  vertex mark sets by $\edgemark$ and $\vermark$ respectively.
  %and assume that
  %they are fixed and finite sets unless otherwise stated.
  \glsadd{not:edgemarkset} \glsadd{not:vermarkset} 
  %{\color{pedit} Since} $\edgemark$ and $\vermark$ are finite, we may treat them as
 % totally ordered sets. 
  For an edge
  between vertices $v, w \in V(G)$, we denote its mark towards the vertex $v$ by
  $\xi_G(w, v)$, and its mark towards the vertex $w$ by $\xi_G(v, w)$.
  \glsadd{not:edgemark}
  Also,
  $\tau_G(v)$ denotes the mark of a vertex $v \in V(G)$. 
  Let $\mG_n$ denote the set of graphs and
   $\mGb_n$ the set of marked graphs on the vertex set $[n]$.
  \glsadd{not:VG}%
  \glsadd{not:vermark}%
  \glsadd{not:Gn}%
   \glsadd{not:Gbarn}%
   
   %\color{red}
   %Note to Payam:
   %The notation $\mG_n$ should also be included in the glossary.
   %\color{black}
%\pres{added to glossary}
  
  All graphs and marked graphs appearing in this document are also assumed to be simple, unless otherwise stated. Therefore we will use the terms ``graph" and ``marked graph" as synonymous with ``simple locally finite graph" and ``simple locally finite marked graph" respectively.
Further, since a graph can be considered to be a marked graph with the edge and vertex mark sets being of cardinality $1$, all definitions that are made for marked graphs will be considered to have been simultaneously made for graphs.
  
  Let $G$ be a finite marked graph. We define the \emph{edge mark count vector} of $G$ by $\vm_G := (m_G(x,x'): x,
  x' \in \edgemark)$ where $m_G(x,x')$ is the number of edges $(v,w)$ in $G$
  where $\xi_G(v,w) = x$ and $\xi_G(w,v) = x'$, or $\xi_G(v,w) = x'$ and
  $\xi_G(w,v) = x$. 
\glsadd{not:edgemarkcountvector}
Likewise, 
%for a finite marked graph $G$, 
we define the \emph{vertex mark count vector} of $G$ by $\vu_G :=
(u_G(\theta): \theta \in \vermark)$ where $u_G(\theta)$ is the number of
vertices $v \in V(G)$ with $\tau_G(v) = \theta$. 
\glsadd{not:vertexmarkcountvector}
 %It is also convenient to denote
 % the vertex mark vector of a marked graph $G$, with finite or countably infinite vertex set, by $\vtau_G$, i.e.\ $\vtau_G := (\tau_G(v): v \in V(G))$. \glsadd{not:vtauG}
 %\begin{marginpar}
    %{\color{pedit} Regarding your comment about redefining these when used in a
      %more general framework, I use these notations only for simple marked
      %graphs.} 
  %\end{marginpar}
  
For a marked graph $G$ and vertices $v, w \in V(G)$, we write $v \sim_G
w$ to denote that $v$ and $w$ are adjacent in $G$.
\glsadd{not:vertices-connected}
Moreover, for a vertex $o \in V(G)$, $\deg_G^{x,x'}(o)$ denotes the number of
vertices $v$ connected to $o$ in $G$ such that $\xi_G(v,o) = x$ and $\xi_G(o,v)
= x'$, \glsadd{not:deg-x-x'} and $\deg_G(o)$ denotes the degree of $o$,
i.e.\ the total number of vertices connected to $o$ in $G$, which is precisely
$\sum_{x, x' \in \edgemark} \deg_G^{x,x'}(o)$. \glsadd{not:deg}
Additionally, for vertices $v, w \in V(G)$, $\dist_G(v, w)$ denotes the distance
between $v$ and $w$, which is the length of the shortest path connecting $v$ to
$w$. \glsadd{not:dist_G}

A marked forest is a marked graph with no cycles. A marked tree is a connected marked forest.

\subsection{The Framework of Local Weak Convergence}
\label{sec:local-weak-convergence-framework}

%A rooted marked graph is a 
Given a connected marked graph $G$ on a finite
or countably infinite vertex set and a vertex 
$o \in V(G)$, we call the pair $(G,o)$ a \emph{rooted connected marked graph}.
%where $G$ is connected and has a distinguished
%vertex $o \in V(G)$. We denote such a rooted marked graph by .
We extend this notation to a marked graph $G$ that is not necessarily connected and 
a vertex $o \in V(G)$ by 
defining $(G,o)$ to be 
%denotes the connected component of the graph
%$G$ containing $o$, rooted at $o$, i.e.\ 
%$(G,o) = (G(o), o)$
$(G(o), o)$,
where $G(o)$ denotes
the connected component of $o$ in $G$. 
In general, we call $(G,o)$ a \emph{rooted marked graph}.
%\color{red}
%   Note to Payam:
%   The notation $(G,o)$ should also be included in the glossary.
%   \color{black}
\glsadd{not:G-o}   
%\pres{added to glossary.}

 \begin{definition}
    \label{def:isomorphism}
    %{\em [Defining $\mGb_*$] }
    
%We say that two rooted marked graphs
Let $G$ and $G'$ be marked graphs. Let $o \in V(G)$
and $o' \in V(G')$. We say that
$(G, o)$ and $(G',o')$ are isomorphic, and write $(G, o) \equiv
(G',o')$, if there exists a bijection between 
the sets of vertices 
%in the connected components $(G(o), o)$ and 
%$(G'(o'), o')$ of the respective roots
of $G(o)$ and $G'(o')$
which maps $o$ to $o'$ while preserving vertex marks, the adjacency structure of these connected components, and the
edge marks. 
\end{definition}

Isomorphism defines an equivalence relation on rooted connected marked graphs.
The isomorphism class of
a rooted marked graph 
$(G, o)$ is denoted by $[G,o]$, \glsadd{not:G-o_iso}
and is determined by $(G(o),o)$.
The set comprised of the
isomorphism classes $[G,o]$ of all rooted marked graphs on any finite or
countably infinite vertex set, 
%together with any distinguished vertex, 
where the edge and vertex marks come from the sets $\edgemark$
and $\vermark$ respectively, is denoted by 
$\mGb_*(\edgemark, \vermark)$.
When the mark sets are clear from the context, we
use $\mGb_*$ as a shorthand for $\mGb_*(\edgemark, \vermark)$. 
Likewise, let
$\mTb_*(\edgemark, \vermark)$ denote the subset of $\mGb_*(\edgemark, \vermark)$
consisting of all isomorphism classes $[T,o]$ where $(T,o)$ is a rooted
marked forest. As for general graphs, the isomorphism class of $(T,o)$ is determined by 
the marked tree
$(T(o),o)$, where $T(o)$ is the connected component of the vertex $o \in T$.
When the mark sets are clear from the context, we use
$\mTb_*$ as a shorthand for $\mTb_*(\edgemark, \vermark)$. \glsadd{not:T_star}

%i.e.\ an adjacency preserving map $\sigma: V(G) \rightarrow V(G')$
%where $\sigma(o) = o'$, for all $v \in V(G)$, we have $\tau_G(v) = \tau_{G'}(\sigma(o))$,
%and for all adjacent vertices $v, w \in V(G)$, we have $\xi_G(v,w) =
%\xi_{G'}(\sigma(v), \sigma(w))$.

For an integer $h \geq 0$, we
denote by $(G,o)_h$ the $h$--neighborhood of
the vertex $o \in V(G)$, rooted at $o$. This is defined by considering the subgraph of $G$ consisting of
all the vertices $v \in V(G)$ such that $\dist_G(o,v) \leq h$ and then making this
subgraph rooted at $o$. \glsadd{not:G-o_h}
The
isomorphism class of the $h$--neighborhood $(G, o)_h$ is denoted by $[G,o]_h$.
\glsadd{not:G-o_iso_h} 
It is straightforward to check that $[G,o]_h$ is determined by $[G,o]$.

%We turn $\mGb_*$ into a metric space using the  local metric $d_*$. 
For $[G,o],
[G',o'] \in \mGb_*$, we define $d_*([G,o], [G',o'])$ to be $1/(1+h_*)$, where
$h_*$ is the maximum over integers $h \geq 0$ such that $(G, o)_h \equiv
(G',o')_h$. If $(G, o)_h \equiv (G',o')_h$ for all $h \ge 0$, it is easy to see
that $(G,o) \equiv (G',o')$, i.e. $[G,o] = [G',o']$. In this case, $d_*([G,o], [G',o'])$ is defined to be
zero. It can be easily checked that $\mGb_*$, 
%(i.e. $\mGb_*(\edgemark, \vermark)$), 
equipped with $d_*$, is a metric space. In
particular, it satisfies the triangle equality. 
In fact, it can be shown, for any finite sets $\edgemark$ and
$\vermark$, that $\mGb_*(\edgemark, \vermark)$ and $\mTb_*(\edgemark, \vermark)$ are complete and separable metric spaces, i.e.
Polish spaces \cite{aldous2007processes}.\footnote{In fact, a more general
  statement without requiring that $\edgemark$ and $\vermark$ be finite sets holds,
but we refer the reader to \cite{aldous2007processes} for more details about this, as we do
not need that more general statement here.}

For an integer $h \geq 0$, let $\mGb_*^h \subset \mGb_*$ consist of isomorphism
classes of rooted marked graphs 
where all the vertices of the connected component of the root are at distance at most
$h$ from the root.
%with depth at most $h$.
\glsadd{not:G_star_h}
For instance,
for $[G, o] \in \mGb_*$, we have $[G,o]_h \in \mGb_*^h$. We define $\mTb_*^h
\subset \mTb_*$ similarly. Note that, by definition, we have $\mGb_*^0 \subset
\mGb_*^1 \subset \dots \subset \mGb_*$. Consequently, for $[G, o] \in \mGb_*^h$
and $0 \leq k \leq h$, we have $[G,o]_k \in \mGb_*^k$.

\ifmaybe
\color{red}
To simplify the notation, we sometimes denote a member of $\mGb_*$ by $g$
instead of the longer notation $[G,o]$. It can be seen that for any rooted
isomorphism class in $\mGb_*$, there is a canonical connected marked rooted
graph with the vertex set $\nats$ rooted at $0$ \cite{aldous2007processes}.
\color{blue}
Is this canonical connected marked rooted graph unique?
It would seem not, since many labellings might be possible for it.
\color{black}
%{\color{red} make sure all subsequent uses of $o$ in this case is changed to
%  $0$}
Based on this, by an abuse of notation, for $g \in \mGb_*$, we use
$\deg_g(0)$ for $\deg_G(0)$ where $(G, 0)$ is the canonical member of the
isomorphism class $g$.  Also, for $g \in \mGb_*$, and $h \in \nats$, $g_h$
denotes $[G,0]_h$ where $(G,0)$ is the canonical member of $g$. 
\color{black}
\color{blue}
This use of the notation $g$ seems inconsistent with its use in Section 2. It seems to me (at the moment) that it would be best to delete the red paragraph.
\color{black}
\fi

% {\color{pedit} It can be seen that one can assign to any rooted
% isomorphism class $[G,o]$ in $\mGb_*$ a  unique canonical connected marked rooted
% graph $(\tilde{G}, 0) \in [G,o]$ with the vertex set $\integers_+$ rooted at $0$ \cite{aldous2007processes}.}

% \begin{marginpar}
%   {\color{pedit}
%     I brought back the discussion for the canonical member emphasizing on its
%     uniqueness, since this concept is referred to in Section 2. Everything can
%     be stated without referring to this concept, but I thought it might be good
%     to mention it briefly.
%   }
% \end{marginpar}

For a Polish space 
%(complete and separable metric space) 
$X$, we say that a
sequence of Borel probability measures $(\mu_n \in \mP(X): n \in \nats)$ converges
weakly to $\mu \in \mP(X)$, and write $\mu_n \Rightarrow \mu$, if, for any bounded
continuous function $f: X \rightarrow \reals$, we have $\int f d \mu_n
\rightarrow \int f d \mu$.
If $X$ is Polish, then weak convergence is equivalent to
  convergence with respect to the \LP distance, i.e.\ $\mu_n \Rightarrow \mu$ is
  equivalent to $\dlp(\mu_n, \mu) \rightarrow 0$.
% see Theorem 6.8, page 73, of Billingsley, Patrick. Convergence of probability
% measures. John Wiley & Sons, 2013, second edition
See \cite{billingsley1971weak,
  billingsley2013convergence} for more details on weak convergence of
Borel probability measures.

For a marked graph $G$ on a finite vertex set, we define $U(G) \in
\mP(\mGb_*)$ as
\begin{equation}
  \label{eq:UG-def}
  U(G) := \frac{1}{|V(G)|} \sum_{v \in V(G)} \delta_{[G,v]},
\end{equation}
where
%, recalling from the above, 
$[G,v]$ denotes the isomorphism class of the
connected component of $v$ in $G$ rooted at $v$. In words, $U(G)$ is the
neighborhood structure of the graph $G$ from the point of view of a vertex
chosen uniformly at random. Moreover, for  $h \geq 0$, let
\begin{equation}
  \label{eq:UGh-def}
  U(G)_h := \frac{1}{|V(G)|} \sum_{v \in V(G)} \delta_{[G,v]_h},
\end{equation}
be the depth $h$ neighborhood structure of a vertex in $G$ chosen uniformly at
random. Note that $U(G)_h \in \mP(\mGb_*^h)$. 

Given a sequence $(G_n: n \in \nats)$ of marked graphs, if $U(G_n)
\Rightarrow \mu$  for some $\mu \in \mP(\mGb_*)$, then we say that the sequence
$G_n$ converges {\em in the local weak sense} to $\mu$, and say that $\mu$ is the
{\em local weak limit} of the sequence.  A Borel probability measure
$\mu \in \mP(\mGb_*)$ is called {\em sofic} if it is the local weak limit of a
sequence of finite marked graphs. Not all Borel probability measures on $\mGb_*$ are sofic. A
necessary condition for a measure to be sofic exists, called {\em unimodularity}
\cite{aldous2007processes}. To define this, let $\mGb_{**}$ be the set of
isomorphism classes $[G, o, v]$ of marked connected graphs with two distinguished
vertices $o, v \in V(G)$ (which are ordered, but need not be distinct). Here, isomorphism is naturally defined as a bijection
preserving marks and adjacency structure which maps the two distinguished
vertices of one object to the respective ones of the other. A measure $\mu \in \mP(\mGb_*)$ is
called unimodular if, for all measurable non--negative functions $f: \mGb_{**}
\rightarrow \reals_+$, we have
\begin{equation}
  \label{eq:unimodularity}
  \int \sum_{v \in V(G)}f([G,o,v]) d \mu([G,o]) = \int \sum_{v \in V(G)} f([G,v,o]) d \mu([G,o]),
\end{equation}
where in each expression the summation is over $v \in V(G)$ that are in the same connected component of $G$ as $o$, since otherwise the expression 
$[G,o,v]$ is not defined.
It can be seen that, in order to check unimodularity, it suffices to check the
above condition for functions $f$ such that $f([G,o,v]) = 0$ unless $v$ is
adjacent to $o$. This is called {\em involution invariance}
\cite{aldous2007processes}. We denote the set of unimodular probability measures
on $\mGb_*$ by $\mP_u(\mGb_*)$. Similarly, as $\mTb_* \subset \mGb_*$, we can
define the set of unimodular probability measures on $\mTb_*$, which we denote by
$\mP_u(\mTb_*)$. \glsadd{not:unimodular-measure}
%\color{red}
%Need to clarify whether $o$ and $v$ need to be distinct vertices. Need to clarify
%whether $f([G,o,v])$ should be considered as defined when $v$ and $o$ are not in the
%same connected component. Need to clarify how to make sense of the statements for
%the situation where the probability distribution on $\mGb_*$ gives all its mass to
%a marked graph comprised of a single vertex (i.e. the root). In particular, what
%is the claim about involution invariance saying in this case?
%\color{black}

For $\mu \in \mP(\mGb_*)$, and $\theta \in \vermark$, we denote by
$\vtype_\theta(\mu)$ the probability under $\mu$ of the root having mark
$\theta$, i.e.\ $\pr{\tau_G(o) = \theta}$ where $[G,o]$ has law $\mu$.\footnote{Here we observe that $\tau_G(o)$ is the same for all $(G,o)$ in the equivalence class $[G,o]$, so we can unambiguously write $\tau_G(o)$ given only 
  the equivalence class $[G,o]$.}
\glsadd{not:vtypemu}%
With this, let $\vvtype(\mu) := (\vtype_\theta(\mu): \theta \in \vermark)$ be
the \emph{probability vector of the root mark}. \glsadd{not:vvtypemu}
Also, for $x,x' \in \edgemark$, we define
$\deg_{x,x'}(\mu) := \ev{\deg^{x,x'}_G(o)}$ where $[G,o]$ has law $\mu$.\footnote{Here we observe that $\deg^{x,x'}_G(o)$ is the same for all $(G,o)$
in the equivalence class $[G,o]$.}
In fact,
$\deg_{x,x'}(\mu)$ denotes the expected number of edges connected to the root
with mark $x$ towards the root and mark $x'$ towards the other endpoint.
\glsadd{not:degxx'mu}
Moreover, let $\deg(\mu)$ be the expected degree at the root. Note that, by
definition, we have $\deg(\mu) = \sum_{x,x' \in \edgemark} \deg_{x,x'}(\mu)$.
\glsadd{not:degmu} Furthermore, let $\vdeg(\mu) := (\deg_{x,x'}(\mu): x,x' \in
\edgemark)$.
% be the degree vector associated to $\mu$.
%\begin{marginpar}
%{\color{blue}
%I deleted the phrase that called this a degree vector,
%because the term ``degree vector'' has been used in a restricted sense later.}
%\end{marginpar}
\glsadd{not:vdegmu}

Every $\mu \in \mP(\mGb_*)$ that appears in this document will be assumed to satisfy $\deg(\mu) < \infty$. However, for clarity, we will explicitly repeat this condition wherever necessary.

\subsection{Local Weak Convergence for Multigraphs}
\label{sec:multigraph-lwc}

The framework above, which was defined for (locally finite, simple) graphs, can be extended to
multigraphs, as defined in \cite[Section~2]{bordenave2015large}.
Here we
give a brief introduction, and refer the reader to \cite{bordenave2015large}, and also to \cite{aldous2004objective}, \cite{aldous2007processes}, for further reading.

A multigraph on a finite or countably infinite vertex set $V$ is a pair $G = (V, \omega)$
where $\omega: V^2 \rightarrow \integers_+$ is such that, for $u, v \in V$,
$\omega(u,u)$ is even and $\omega(u,v) = \omega(v,u)$. We interpret $\omega(u,u)
/ 2$ as the number of self-loops at vertex $u$, and $\omega(u,v)$ as the number
of edges between vertices $u$ and $v$. The degree of a vertex $u$ is defined to
be $\deg(u) := \sum_{v \in V} \omega(u,v)$. 
The notions of path, distance
and connectivity are naturally defined for multigraphs.
A multigraph $G$ is called locally
finite if $\deg(v) < \infty$ for all $v \in V$. 

All multigraphs encountered in this document will be locally
finite, so the term ``multigraph" will be considered synonymous with
``locally finite multigraph".
Further, we assume that all
multigraphs are unmarked.% unless otherwise stated. 

It can be checked that a multigraph is a graph
(i.e. a locally finite multigraph is a simple locally finite 
graph) precisely when $\omega(u,v) \in \{0,1\}$ for all 
pairs of vertices $u$ and $v$ (in particular, $\omega(u,u) = 0$
for all vertices $u$).

A rooted multigraph $(G, o)$ is a multigraph on a finite or countably infinite
vertex set $V$ together with a distinguished vertex $o \in V$. 

\begin{definition}
    \label{def:multigraph-isomorphism}
    %{\em [Isomorphism of rooted multigraphs]  }
    
Two rooted multigraphs
$(G_1, o_1) = ((V_1, \omega_1), o_1)$ and $(G_2, o_2) = ((V_2, \omega_2), o_2)$ are said to be isomorphic if there is a bijection $\sigma$
%$\sigma: V_1 \rightarrow V_2$  
between the sets of vertices of the respective connected components of the roots
which preserves the roots and
connectivity. Namely, \ $\sigma(o_1) = o_2$ and we have $\omega_2(\sigma(v), \sigma(u)) =
\omega_1(v,u)$ for all
$u$ and $v$ in the connected component of $o_1$.
%$u, v \in V_1$. 
We denote this by writing $(G_1, o_1)
\equiv (G_2, o_2)$. 
\end{definition}

This notion of isomorphism defines an equivalence relation on rooted connected multigraphs,
where the equivalence class to which a rooted multigraph belongs is determined by the connected component of the root.
Let $\hmG_*$ be the set of all equivalence classes $[G,o]$ of rooted
multigraphs corresponding to this isomorphism relation. 
\glsadd{not:hmG_star}
%\glsadd{not:mu_h}

For $h
\geq 0$, let $(G,o)_h$ denote the induced multigraph defined by the vertices in $G$
with distance no more than $h$ from $o$, rooted at $o$. 
%We equip $\hmG_*$ with the following local metric similar to the one defined in the previous section.
Let $[G_1, o_1], [G_2, o_2]
\in \hmG_*$ and $(G_1, o_1)$ and $(G_2, o_2)$ be arbitrary members of $[G_1,
o_1]$ and $[G_2, o_2]$, respectively.
The distance between $[G_1, o_1], [G_2, o_2]
\in \hmG_*$ is defined to be $1/(1+h_*)$, where $h_*$ is the maximum $h$ such that $(G_1, o_1)_h \equiv (G_2,
o_2)_h$. If $(G_1, o_1)_h \equiv (G_2,
o_2)_h$ for all $h \ge 0$, then we define
the distance to be zero, because this occurs precisely when 
$(G_1, o_1) \equiv (G_2,
o_2)$.
It can be checked that this distance defined on $\hmG_*$ is indeed a metric, and
$\hmG_*$ equipped with this metric is a Polish space \cite{aldous2007processes}.

%%% Local Variables: 
%%% mode: latex
%%% TeX-master: "Note-41_BC-ent-arxiv.tex"
%%% End: 

%aa

\section{Marked Unimodular Galton--Watson Trees}
\label{sec:markovian-ugwt}

% bb

In this section, we introduce an important class of unimodular probability distributions on
$\mTb_*$, called {\em marked unimodular Galton--Watson trees}. These probability distributions can be thought of as the counterpart of finite memory Markov processes in the local weak
convergence language.
The construction here is a generalization of the one in
Section 1.2 of \cite{bordenave2015large}.
Before giving the definition, we need to set up some
notation. 

Given $\mu \in \mP(\mGb_*)$, let $\mu_h \in \mP(\mGb_*^h)$ denote the law of
$[G,o]_h$, where $[G,o]$ has law $\mu$. \glsadd{not:mu_h}
We similarly define $\mu_h \in \mP(\mTb_*^h)$ for $\mu \in
\mP(\mTb_*)$, recalling that $\mTb_* \subset \mGb_*$. For a marked graph
$G$, on a finite or countably infinite vertex set, and adjacent vertices $u$ and $v$ in $G$, we define $G(u,v)$ to be the pair
$(\xi_G(u,v), (G',v))$ where  $G'$ is the connected component of $v$ in the graph
obtained from $G$ by removing the edge between $u$ and $v$. Similarly, for
 $h \geq 0$, $G(u,v)_h$ is defined as $(\xi_G(u,v), (G',v)_h)$.
\glsadd{not:Guv} \glsadd{not:Guvh}
% \begin{marginpar}
% {\color{pedit} The sentence in blue about $(G',v)_h$ does not seem correct to me. Note that the
%   distance in $G'$ is defined after we remove the edge between $u$ and $v$ so
%   that a path is not allowed to go over this edge, whereas in $(G,v)_h$ such a
%   thing is allowed. In other words, leaving the edge $(u,v)$ can cause certain
%   vertices to have distance less than $h$ from $v$ which in fact should have distance larger
%   than $h$ in $(G',v)_h$.}
% \end{marginpar}
% \color{blue}
% Note that $(G',v)_h$ can also be arrived at by first creating $(G,v)_h$ and then removing the edge between $u$ and $v$ if $h > 0$.
% \color{black}
See Figure~\ref{fig:Guv} for an example.
Let $G[u,v]$ denote the pair $(\xi_G(u,v), [G',v])$, so $G[u,v] \in \edgemark \times \mGb_*$. Likewise, for $h \geq 0$, let $G[u,v]_h$ denote
$(\xi_G(u,v), [G',v]_h)$, so $G[u,v]_h \in \edgemark \times \mGb_*^h$.
\glsadd{not:Guv_iso} \glsadd{not:Guvh_iso}

\begin{figure}
  \centering
  \begin{tikzpicture}
  \begin{scope}[xshift=-5cm]
  \node[nodeR,label={[label distance=1mm]90:1}] (n1) at (0,0) {};
  \node[nodeB,label={[label distance=1mm]180:2}] (n2) at (-1.5,-1.5) {};
  \node[nodeB,label={[label distance=1mm]0:3}] (n3) at (1.5,-1.5) {};
  \node[nodeR,label={[label distance=1mm]180:4}] (n4) at (-1.5,-3.7) {};
  \node[nodeR,label={[label distance=1mm]0:5}] (n5) at (1.5,-3.7) {};

  \drawedge{n1}{n3}{B}{O}
  
  \drawedge{n1}{n2}{B}{O}

  \drawedge{n2}{n4}{B}{O}

  \drawedge{n4}{n5}{B}{B}

  \drawedge{n3}{n5}{B}{B}

  \drawedge{n2}{n5}{O}{B}

  \node at (0,-4.7) {$(a)$};
\end{scope}

\begin{scope}
  \node[nodeB] (n3) at (0,0) {};
  \nodelabel{n3}{0}{3};
  \node[nodeR] (n5) at (0,-1.5) {};
  \nodelabel{n5}{0}{5};
  \node[nodeB] (n2) at (-0.9,-2.6) {};
  \nodelabel{n2}{180}{2};
  \node[nodeR] (n4) at (0.9,-2.6) {};
  \nodelabel{n4}{0}{4};
  \node[nodeR] (n1) at (0,-3.7) {};
  \nodelabel{n1}{0}{1};

  \drawedge{n3}{n5}{B}{B}
  \drawedge{n5}{n2}{B}{O}
  \drawedge{n5}{n4}{B}{B}
  \drawedge{n2}{n1}{O}{B}
  \drawedge{n2}{n4}{B}{O}

  \draw[edgeO] ($(n3)+(0,0.6)$) -- (n3);
  \node at (0,-4.7) {$(b)$};
\end{scope}

\begin{scope}[xshift=5cm]
  \node[nodeB] (n3) at (0,0) {};
  \nodelabel{n3}{0}{3};
  \node[nodeR] (n5) at (0,-1.5) {};
  \nodelabel{n5}{0}{5};
  \node[nodeB] (n2) at (-0.9,-2.6) {};
  \nodelabel{n2}{180}{2};
  \node[nodeR] (n4) at (0.9,-2.6) {};
  \nodelabel{n4}{0}{4};

  \drawedge{n3}{n5}{B}{B}
  \drawedge{n5}{n2}{B}{O}
  \drawedge{n5}{n4}{B}{B}
  \drawedge{n2}{n4}{B}{O}

  \draw[edgeO] ($(n3)+(0,0.6)$) -- (n3);
  \node at (0,-4.7) {$(c)$};
\end{scope}
  
\end{tikzpicture}

  \caption[Marked Graph]{$(a)$ A marked graph $G$ on the vertex set $\{1, \dots, 5\}$
    with vertex mark set $\vermark = \{\tikz{\node[nodeB,scale=0.5] at (0,0) {};}, \tikz{\node[nodeR,scale=0.5] at (0,0)
  {};}\}$ and
    edge mark set $\edgemark =
  \{\text{\color{blueedgecolor} Blue (solid)}, \text{\color{orangeedgecolor}
    Orange (wavy)} \}$. In $(b)$, $G(1,3)$ is
    illustrated where the first component $\xi_G(1,3)$ is depicted as a half
    edge with the corresponding mark going towards the root $3$,  and
  $(c)$ illustrates $G(1,3)_2$. Note that $G(1,3)$ can be interpreted as cutting
the edge between $1$ and $3$ and leaving the half edge connected to $3$ in
place. Moreover, note that, in constructing $G(u,v)$, although we are removing
the edge between $u$ and $v$, it might be the case that $u$ is still reachable
from $v$ through another path, as is the case in the above example.}
  \label{fig:Guv}
\end{figure}

For $g \in \edgemark \times \mGb_*$, we call the $\edgemark$ component of $g$
its {\em mark component} and denote it by $g[m]$. Moreover, we call the $\mGb_*$
component of $g$ its {\em subgraph component} and denote it by $g[s]$.
%\color{red}
%Note for Payam:
%The notation $g[m]$ and the notation $g[s]$ need to be included in the glossary.
%\color{black}
%\pres{added to glossary}
\glsadd{not:g-m}
\glsadd{not:g-s}
Given a marked graph $G$ and adjacent vertices $u$ and $v$ in $G$, and
for $g \in \edgemark \times \mGb_*$,  we write $G(u,v) \equiv g$ to denote that
$\xi_G(u,v) = {g[m]}$ and also $(G',v)$ falls in the isomorphism class {$g[s]$}.
We define the expression $G(u,v)_h \equiv g$ for $g \in \edgemark \times
\mGb_*^h$ in a similar fashion. 
For $g \in \edgemark \times \mGb_*^h$ and an integer $k \ge 0$, we define $g_k \in
\edgemark \times \mGb_*^{\min\{h, k\}}$ to have the same mark component as $g$, i.e.\ $g_k[m] := g[m]$,
and subgraph component the truncation
of the subgraph component of $g$ up to depth $k$, i.e.\ $g_k[s] := (g[s])_k$.
\glsadd{not:g-k}
%\color{red}
%Note for Payam:
%The notation $g_k$ needs to be included in the glossary.
%\color{black}
%\pres{added to glossary}
%$(g_s)_k$, i.e.\ the truncation
%of the subgraph component of $g$.
For a marked graph $G$,  two adjacent vertices $u, v$ in $G$, and  $h \geq 1$, we define the \emph{depth $h$ type} of the edge $(u,v)$ as
  \begin{equation}
    \label{eq:depth-h-type}
    \etype^h_G(u,v) := (G[v,u]_{h-1}, G[u,v]_{h-1}) \in (\edgemark \times \mGb_*^{h-1}) \times (\edgemark \times \mGb_*^{h-1}).
  \end{equation}
  Note that we have employed the convention that the first component on the
  right hand side (i.e.\ $G[v,u]_{h-1}$) is the neighborhood of the first vertex
  appearing on the left hand side (i.e.\ $u$). See Figure~\ref{fig:etype_example}
  for an example.
\glsadd{not:etype}
  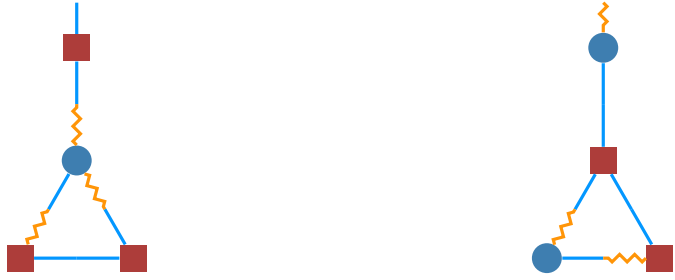
\begin{figure}
    \centering
    \begin{tikzpicture}
  \begin{scope}[xshift=-3.5cm]
    \node[nodeR] (n1) at (0,0) {};
    \node[nodeB] (n2) at (0,-1.5) {};
    \node[nodeR] (n3) at ($(0,-1.5) + (-120:1.5)$) {};
    \node[nodeR] (n4) at ($(0,-1.5) + (-60:1.5)$) {};

    \drawedge{n1}{n2}{B}{O}
    \drawedge{n2}{n3}{B}{O}
    \drawedge{n2}{n4}{O}{B}
    \drawedge{n3}{n4}{B}{B}

    \draw[edgeB] ($(n1)+(0,0.6)$) -- (n1);
  \end{scope}

  \begin{scope}[xshift=3.5cm]
    \node[nodeB] (n1) at (0,0) {};
    \node[nodeR] (n2) at (0,-1.5) {};
    \node[nodeB] (n3) at ($(0,-1.5) + (-120:1.5)$) {};
    \node[nodeR] (n4) at ($(0,-1.5) + (-60:1.5)$) {};

    \drawedge{n1}{n2}{B}{B}
    \drawedge{n2}{n3}{B}{O}
    \drawedge{n2}{n4}{B}{B}
    \drawedge{n3}{n4}{B}{O}

    \draw[edgeO] ($(n1)+(0,0.6)$) -- (n1);
  \end{scope}

\end{tikzpicture}

    \caption{$\etype_G^3(1,3)$ for the graph in Figure~\ref{fig:Guv}, with the
      first component on the left and the second component on the right. Note
      that the order in setting the notation $\etype_G^h(u,v)$ is chosen so that the first
      component ($G[3,1]_2$ here) is the neighborhood of the first vertex mentioned
      in the notation ($1$ in this example),  and the second component is the neighborhood of the vertex mentioned
      second ($3$ in this example).
      %\color{blue}
      %Also note that even though labels are shown on the vertices for clarity in this figure, the subgraph part of each of the two components of $\etype_G^3(1,3)$ in this example is an equivalence class, which should be thought of as unlabeled or labeled canonically.
      %\color{black}}
      Also note that  the subgraph part of each of the two
        components of $\etype_G^3(1,3)$ in this example is an equivalence class,
        which is the reason why there are no vertex labels.}
    \label{fig:etype_example}
  \end{figure}

  % \begin{marginpar}
  %   {\color{pedit}  I think it is better to remove vertex labels in
  %     Figure~\ref{fig:etype_example} to avoid confusion. With this, I have
  %     changed the last sentence in the caption.}
  % \end{marginpar}
  For a rooted marked graph $(G, o)$, integer $h \ge 1$, and $g, g' \in \Xi \times \mGb_*^{h-1}$, we define
  \begin{equation}
    \label{eq:Eh-g-g'}
    E_h(g, g')(G, o) := |\{v \sim_G o: \etype_G^h(o,v) = (g, g') \}|.
  \end{equation}
  \glsadd{not:Ehgg'}
  %When it is clear from the context, we may drop $(G, o)$ from the notation and
  %simply write $E_h(g, g')$.
Also, for $[G,o] \in \mGb_*$, we can write $E_h(g,g')([G,o])$ for
$E_h(g,g')(G,o)$, where $(G,o)$ is an arbitrary member of $[G,o]$.
\glsadd{not:Ehgg'-iso}
%\color{red}
%Note for Payam:
%The notation $E_h(g,g')([G,o])$ needs to be included in the glossary.
%\color{black}
%\pres{added to glossary}
This notation is well-defined, since $E_h(g, g')(G,o)$,
thought of as a function of $(G,o)$ for fixed integer $h \ge 1$
and $g, g' \in \Xi \times \mGb_*^{h-1}$, is invariant under rooted isomorphism.
%Note that since all such members are rooted isomorphic by definition, they will all result in
%  the same value for $E_h(g,g')(G,o)$.
% \begin{marginpar}
%   {
%   {\color{blue}
% The canonical member seems to be unique as an equivalence class but not as a graph, so perhaps $(G,0)$ should be written $[G,0]$?
% Also, the use of $g$ for members of 
% $\Xi \times \mGb_*$ as well as for members of 
% $\mGb_*$ could be confusing.
% }
% {\color{pedit}
% \textbf{Response:} The canonical member is not an equivalence class, for each
% equivalence class, there is a unique marked graph with labels in $\integers_+$
% with root being $0$. For uniqueness, see \cite{aldous2007processes} page 1461.
% }
% }
% \end{marginpar}

  For  $h \geq 1$, $P \in \mP(\mGb_*^h)$, and $g, g' \in \Xi \times \mGb_*^{h-1}$, define
  \begin{equation*}
    e_P(g, g') := \evwrt{P}{E_h(g, g')(G, o)}.
  \end{equation*}
  % \reversemarginpar%
  % \begin{marginpar}
  %   {
  %     {\color{blue}
  %       What is your convention for when you use $P$ or
  %       when you use $\mu$ for an element of
  %       $\mP(\mGb_*)$?
  %     }
  %     {\color{pedit}
  %       \textbf{Response:} I use $P$ for probability distributions on $\mGb_*^h$
  %       for integer $h \geq 0$, while $\mu$ is used for probability
  %       distributions on $\mGb_*$. 
  %     }
  %   }%
  % \end{marginpar}%
  Here, $(G, o)$ is a member of the isomorphism class $[G, o]$ that has law $P$.
  \glsadd{not:e-p-g-g'}
%\color{red}
%Note for Payam:
%The notation $e_P(g, g')$ needs to be included in the glossary.
%\color{black}
%\pres{added to glossary}
This notation is well-defined
for the same reason as above.
%since $E_h(g, g')(G,o)$,
%thought of as a function of $(G,o)$ for fixed integer $h \ge 1$
%and $g, g' \in \Xi \times \mGb_*^{h-1}$, is invariant under rooted isomorphism.

  \begin{definition}
    \label{def:admissible}
    %{\em [Admissible probability distributions on $\mGb_*^h$
%isomorphism classes of rooted marked graphs of depth bounded by a fixed bound ] }
    
      Let $h \ge 1$. A probability distribution $P \in \mP(\mGb_*^h)$ is called
  {\em admissible} if $\evwrt{P}{\deg_G(o)} < \infty$ and $e_P(g, g') = e_P(g', g)$
  for all $g, g' \in \edgemark \times \mGb_*^{h-1}$. \glsadd{trm:admissible}
\end{definition}
% \begin{marginpar}
%   {
% {  \color{blue}
%   Since the notion of admissibility seems to be important in the sequel, it should probably be highlighted as an explicit definition.
% }
% {\color{pedit} \textbf{Response:} Changed.}
% }
% \end{marginpar}

The following simple lemma indicates the importance of the concept of admissibility.
\begin{lem}
  \label{lem:unimodular-is-admissible}
Let $h \ge 1$,
%$h \in \nats$ 
and let
$\mu \in \mP_u(\mGb_*)$ be a unimodular probability measure with $\deg(\mu) < \infty$. Let
 $P := \mu_h$. Then $P$ is admissible.
\end{lem}

\begin{proof}
Using
the definition of unimodularity, for $g,g' \in \edgemark \times \mGb_*^{h-1}$, we have
\begin{align*}
  e_P(g,g') &= \evwrt{\mu}{\sum_{v \sim_G o} \one{\etype_G^h(o,v) = (g,g')}} \\
  &= \evwrt{\mu}{\sum_{v \sim_G o} \one{\etype_G^h(v,o) = (g,g')}} = \evwrt{\mu}{\sum_{v \sim_G o} \one{\etype_G^h(o,v) = (g',g)}} = e_P(g',g).
\end{align*}
\end{proof}

%\color{red}
%\begin{rem}
%\label{rem:zero-degree-admissible}
%Note that, for each integer $h \ge 0$, the probability distributions 
%$P \in \mP(\mGb_*^h)$ with $\evwrt{P}{\deg_G(o)} = 0$ are in one to one correspondence with the probability distributions on the vertex mark set $\vermark$. For $h \ge 1$ it can be checked that each $P \in \mP(\mGb_*^h)$ with $\evwrt{P}{\deg_G(o)} = 0$ is admissible.
%\end{rem}

%{\bf Think of deleting this remark.}
%\color{black}

  For the case of rooted marked trees,
 all the above notation can be defined similarly by substituting for
  $\mGb_*$ with $\mTb_*$, since $\mTb_* \subset \mGb_*$.
%\color{red}
%Note for Payam: 
%Introduce tree versions of all the relevant graph notation in the glossary. There are many items of notation of this sort that show up all the time. In the glossary these can be identified by 
%telling the reader to look at the graph version. For instance,
%put $T(u,v)$ as a glossary item, with the text saying
%``See $G(u,v)$'', etc. Some of the items that I noticed need to put in the glossary are $T(u,v)$, $T[u,v]$, $T(u,v)_h$, $T[u,v]_h$,
%$\xi_T(u,v)$, $t[m]$, $t[s]$, $\etype^h_T(u,v)$,
%$E_h(t,t')((T,o))$, $E_h(t,t')([T,o])$,
%$e_P(t,t')$. There may be others. Please check.
%\color{black}
%\pres{added these together with a number of other notations to glossary}
  While we have defined the notion of an admissible probability distribution
  $P$ for $P \in \mP(\mGb_*^h)$, $h \ge 1$, we will soon see that it suffices to be focused on the case $P \in \mP(\mTb_*^h)$, $h \ge 1$.%
  %isomorphism classes of marked rooted trees of depth bounded by a fixed bound. 
  \glsadd{not:T-o}%
  \glsadd{not:T-o-h} %
  \glsadd{not:Eh-tt'-To} %
  \glsadd{not:Eh-tt'-To-iso} %
  \glsadd{not:T-u-v} %
  \glsadd{not:T-u-v-h} %
  \glsadd{not:T-u-v-iso} %
  \glsadd{not:T-u-v-h-iso} %
  \glsadd{not:V-T} %
  \glsadd{not:T-o-iso} %

    \glsadd{not:T-o-h-iso} %
  \glsadd{not:deg-T-o} %
  \glsadd{not:deg-T-xx'-o} %
   \glsadd{not:dist-T} %
  \glsadd{not:phi-T-u-v} %
  \glsadd{not:tau-T} %
  \glsadd{not:vD-T} %
  \glsadd{not:vm-T} %
  \glsadd{not:vtau-T} %
  \glsadd{not:vu-T} %

  \glsadd{not:xi-T-v-w} %
  \glsadd{not:e-P-tt'} %
  \glsadd{not:tm} % 
  \glsadd{not:ts} %
  \glsadd{not:tk} %
  \glsadd{not:v-simT-w} %

For $t, t' \in \edgemark \times \mTb_*$, define
$t \oplus t' \in \mTb_*$ as the isomorphism class of the rooted tree $(T, o)$
where $o$ has a subtree isomorphic to $t[s]$, and $o$ has an extra offspring $v$
where the subtree rooted at $v$ is isomorphic to $t'[s]$. Furthermore,
$\xi_T(v,o) = t[m]$ and $\xi_T(o, v) = t'[m]$. See Figure~\ref{fig:t_oplus_tprime} for an example. Note
that, in general, $t \oplus t'$ is different from 
% \color{blue}
% and in a different equivalence class from
% \color{black}
% \begin{marginpar}%
%   {\color{pedit} No need to say ``and in a different equivalence class from''
%     since $t \oplus t'$ itself in a member in $\mTb_*$ and hence is an
%     equivalence class}
% \end{marginpar}%
$t' \oplus t$. Also, note that if $t\in \edgemark \times \mTb_*^{k}$ and $t' \in \edgemark \times \mTb_*^{l}$, then we have $t \oplus t' \in \mTb_*^{\max\{k, l+1 \}}$.
  \glsadd{not:oplus}

\begin{figure}
  \centering
      \scalebox{0.8}{
\begin{tikzpicture}

  \begin{scope}[xshift=-7.5cm]
    \node[nodeB] (n1) at (0,0) {};
    \node[nodeB] (n2) at (0,-1.5) {};
    \node[nodeR] (n3) at ($(0,-1.5) + (-135:1.5)$) {};
    \node[nodeB] (n4) at ($(0,-1.5) + (-45:1.5)$) {};

    %\drawedge{n1}{n2}{O}{O}
    \draw[edgeO] (n1) -- (n2);
    \drawedge{n2}{n3}{B}{B}
    \drawedge{n2}{n4}{O}{B}
    \draw[edgeO] ($(n1)+(0,0.6)$) -- (n1);

    \node[scale=1.5] at (0,1) {$t$};

  \end{scope}

    \begin{scope}[xshift=-2.5cm]
    \node[nodeB] (n1) at (0,0) {};
    \node[nodeB] (n2) at ($(-135:1.5)$) {};
    \node[nodeR] (n3) at ($(-45:1.5)$) {};

    \drawedge{n1}{n2}{B}{B}
    \drawedge{n1}{n3}{O}{B}
    
    \draw[edgeB] ($(n1)+(0,0.6)$) -- (n1);

    \node[scale=1.5] at (0,1) {$t'$};

  \end{scope}

  \begin{scope}[xshift=2.5cm]
    \node[nodeB] (n1) at (0,0) {};
    \node[nodeB] (n2) at (-135:1.9) {};
    \node[nodeB] (n3) at (-45:1.9) {};
    \node[nodeR] (n21) at ($(n2)+(-120:1.5)$) {};
    \node[nodeB] (n22) at ($(n2)+(-60:1.5)$) {};
    \node[nodeB] (n31) at ($(n3)+(-120:1.5)$) {};
    \node[nodeR] (n32) at ($(n3)+(-60:1.5)$) {};

    \draw[edgeO] (n1) -- (n2);
    \drawedge{n1}{n3}{O}{B}
    \drawedge{n2}{n21}{B}{B}
    \drawedge{n2}{n22}{O}{B}
    \drawedge{n3}{n31}{B}{B}
    \drawedge{n3}{n32}{O}{B}

    \node[scale=1.5] at (0,1) {$t \oplus t'$};
  \end{scope}

  \begin{scope}[xshift=7.5cm]
    \node[nodeB] (n0) at (0,0) {};
    \node[nodeB] (n1) at (-150:1.5) {};
    \node[nodeR] (n2) at (-90:1.5) {};
    \node[nodeB] (n3) at (-30:1.5) {};

    \drawedge{n0}{n1}{B}{B}
    \drawedge{n0}{n2}{O}{B}
    \drawedge{n0}{n3}{B}{O}

    \node[nodeB] (n33) at ($(n3)+(0,-1.2)$) {};
    \node[nodeR] (n331) at ($(n33) + (-120:1.2)$) {};
    \node[nodeB] (n332) at ($(n33) + (-60:1.2)$) {};

    \draw[edgeO] (n3) -- (n33);
    \drawedge{n33}{n331}{B}{B}
    \drawedge{n33}{n332}{O}{B}

    \node[scale=1.5] at (0,1) {$t' \oplus t$};
  \end{scope}

\end{tikzpicture}
}

  \caption{$t \oplus t'$ and $t' \oplus t$ for $t, t' \in \edgemark \times
    \mTb_*$ depicted on the left. We have employed our general convention in
    drawing objects in $\edgemark \times \mTb_*$, which is to draw the mark
    component as a half edge towards the root.  
    In the figures for $t \oplus t'$ and $t' \oplus t$ the root is vertex at the top of the figure. Note that in this example $t \oplus t'$ is
    different from $t' \oplus t$.
    %\color{blue}
    %Also, they are in different equivalence classes because the root in each of $t \oplus t'$ and
   %$t' \oplus t$ is the top vertex in the corresponding figure. 
    %\color{black}
    }
  \label{fig:t_oplus_tprime}
\end{figure}
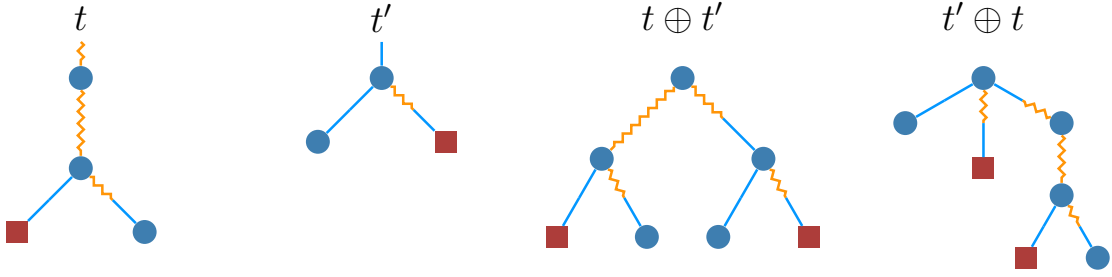
% \begin{marginpar}
%   {\color{pedit} I think the last sentence starting with ``Also, they are in
%     different equivalence classes \dots'' in the caption of
%     Figure~\ref{fig:t_oplus_tprime} is not necessary, since $t \oplus t'$ by
%     definition is a member in $\mTb_*$ and thus is an equivalence class itself.}
% \end{marginpar}

%    \normalmarginpar

The operation $\oplus$ described above helps to elucidate
the structure of marked rooted trees of fixed depth, i.e.
members of $\mTb_*^h$, $h \ge 0$.
Some of their properties are gathered in
Appendix~\ref{sec:marked-rooted-trees-some-props}.

Now, for $h \ge 1$, given an admissible $P \in \mP(\mTb_*^h)$, we define a Borel
probability measure $\ugwt_h(P) \in \mP(\mTb_*)$, which is called the
{\emph {marked unimodular Galton--Watson tree with  depth $h$ neighborhood distribution $P$}},
as follows. \glsadd{not:ugwhP}
For $t, t' \in \edgemark \times \mTb_*^{h-1}$ such that $e_P(t, t') > 0$, define $\hP_{t,t'} \in \mP(\edgemark \times \mTb_*^h)$ via:
  \begin{equation}
    \label{eq:size-biased-def}
    % \hP_{t, t'} (\tilde{t}) = P(\tilde{t} \oplus t') \left ( 1 + | \{ v \sim_{\tilde{t}_s} o : \tilde{t}_s(o,v) \equiv t', \xi_{\tilde{t}_s}(v,o) = t_m \} | \right ) \frac{\one{\tilde{t}_{h-1} = t}}{e_P(t, t')}.
    \hP_{t, t'} (\tilde{t}) := \one{\tilde{t}_{h-1} = t} \frac{P(\tilde{t} \oplus t')  E_h(t,t')(\tilde{t}\oplus t') }{e_P(t,t')}, \qquad \qquad \text{for }\tilde{t} \in \edgemark \times \mTb_*^h.
  \end{equation}
  Moreover, in case $e_P(t,t') = 0$, we define
    $\hP_{t,t'}(\tilde{t}) = \one{\tilde{t} = t}$.
    
    We first check that $\hP_{t, t'} (\tilde{t})$ defines a probability distribution over $\tilde{t}$. This is clear when
    $e_P(t,t') = 0$, so assume that $e_P(t,t') > 0$.
 By definition, we have
  \begin{equation*}
    e_P(t, t') = \sum_{t'' \in \mTb_*^h} P(t'') E_h(t,t')(t'').
  \end{equation*}
 Note that $E_h(t,t')(t'') > 0$ iff for some
  $\tilde{t} \in \edgemark \times \mTb_*^h$ with $\tilde{t}_{h-1} = t$, we have
  $t'' = \tilde{t} \oplus t'$.
Also, it is easy to see that two different $\tilde{t}^{(1)}$ and $\tilde{t}^{(2)}$ in
$\edgemark \times \mTb_*^h$ with $\tilde{t}^{(1)}_{h-1} = \tilde{t}^{(2)}_{h-1}
= t$ give rise to different objects $\tilde{t}^{(1)} \oplus t'$ and
$\tilde{t}^{(2)} \oplus t'$.
  This readily implies that summing
  $\hP_{t,t'}(\tilde{t})$ over all $\tilde{t} \in \edgemark \times \mTb_*^h$
  such that $\tilde{t}_{h-1} = t$ gives $1$, and hence $\hat{P}_{t,t'}(\tilde{t})$
  is a probability distribution over $\tilde{t}$.
  
    %As we will discuss below,
    %motivated by Corollary~\ref{cor:ugwt-gamma>0-ep>0} in Apendix~\ref{sec:UGWT-some-props},
    %this choice is arbitrary.

With this, we define $\ugwt_h(P)$ to be the law of $[T,o]$ where $(T,o)$ is the
 random rooted marked tree constructed as follows. First, we sample the $h$
 neighborhood of the root, $(T,o)_h$, according to $P$. Then,  for
 each offspring $v \sim_T o$ of the root, we sample $\tilde{t} \in \edgemark
 \times \mTb_*^h$ according to the law $\hP_{t,t'}(.)$ where $t = T[o,v]_{h-1}$ and
 $t' = T[v,o]_{h-1}$. Note that, by definition, we have $\tilde{t}_{h-1} = t$.
 This means that the subtree component of $\tilde{t}$ agrees
 with the subtree component of $t$ up to depth $h-1$. This allows us to 
add at most one layer to $T(o,v)_{h-1}$ so that $T(o,v)_h \equiv \tilde{t}$.
%extend the
%depth of the  subtree below $v$ by {\color{pedit} at most } one according to
%$\tilde{t}$. 
We carry out the same
procedure independently for each offspring of the root.
% to extend the depth of
%the tree rooted at $o$. 
At this step, the rooted tree has depth {
  at most }$h+1$. Subsequently, we follow the same procedure for
vertices at depth 2, 3, and so on inductively to construct $(T,o)$. More specifically,  for a vertex $v$ at depth
$k$ of $(T,o)$ with parent $w$, we sample $\tilde{t}$ from $\hP_{t,t'}(.)$ with
$t = T[w,v]_{h-1}$ and $t' = T[v,w]_{h-1}$.
%, and extend the subtree below $w$
%using $\tilde{t}$. 
Since by definition, we have $\tilde{t}_{h-1} = t$, we can add at
most one layer to $T(w,v)_{h-1}$ so that $T(w,v)_h \equiv \tilde{t}$. 
We do this independently for all vertices at depth $k$.
If, at the time we do the above procedure for vertices at depth
  $k$, there is no vertex at that depth, we stop the
  procedure.
%to make the depth of the graph equal to $h+k$ at the end of step $k$.
Finally, we define $\ugwt_h(P)$ to be the law of $[T,o]$.

As shown in Corollary~\ref{cor:ugwt-gamma>0-ep>0} in
Appendix~\ref{sec:UGWT-some-props}, if $[T,o]$ is outside a measure zero set with respect to
$\ugwt_h(P)$, for all vertices $v \in V(T) \setminus \{o\}$ we have $e_P(t,t')
> 0$ where $t = T[w,v]_{h-1}$ and $t' = T[v,w]_{h-1}$, with $w$ being the parent
of $v$. This means that the need to refer to the definition of $\hP_{t,t'}$ when $e_P(t,t') = 0$ 
will not arise, with probability $1$.
%is arbitrary in the definition of $\ugwt_h(P)$. 

%\color{red} {\bf I have not yet read the proof in Appendix~\ref{sec:UGWT-some-props}} \color{black}

% {\color{pedit}
%   Lemma~\ref{lem:UGWT-gamma-pos-ep-pos} in Appendix~\ref{sec:UGWT-some-props}
%   guarantees that, almost surely, throughout the construction, we have $e_P(t,t')
%   > 0$, where $t, t' \in \edgemark \times \mTb_*^{h-1}$ are those which appear in
%   $\hP_{t,t'}$ as in the above discussion. 
% }
% \begin{marginpar}
% {
% \color{blue}
% This discussion appears to be incomplete, because
% in order to define $\hP_{t,t'}(.)$ it is necessary
% to have $e_P(t, t') > 0$. How does one ensure that
% the $(t,t')$ pairs that show up in deeper stages of
% this construction always have the property
% that $e_P(t, t') > 0$?
% \color{pedit}
% \textbf{Response:} Please see Lemma~\ref{lem:UGWT-gamma-pos-ep-pos} in Appendix~\ref{sec:UGWT-some-props}.
% }
% \end{marginpar}

For each integer $h \ge 1$, the probability distribution 
$\ugwt_h(P) \in \mP(\mTb_*)$ satisfies a useful continuity property in its defining admissible probability distribution $P \in \mP(\mTb_*^h)$.
This is stated in the following Lemma~\ref{lem:Pn-conv-P--ugwt-Pn-conv-ugwt-P},
whose proof is in Appendix~\ref{sec:UGWT-convg}.

\begin{lem}
  \label{lem:Pn-conv-P--ugwt-Pn-conv-ugwt-P}
  Let $h \ge 1$.
  Assume that an admissible probability distribution $P \in \mP(\mTb_*^h)$
  together with a sequence of admissible probability distributions $\Pn \in
  \mP(\mTb_*^{h})$ are given such that $\Pn \Rightarrow P$ and, for all $t, t'
  \in \edgemark \times \mTb_*^{h-1}$, we have $e_{\Pn}(t,t') \rightarrow
  e_P(t,t')$. Then we have $\ugwt_h(\Pn) \Rightarrow \ugwt_h(P)$.
\end{lem}

The following Lemma~\ref{lem:ugwt-P-is-unimodular} justifies the terminology
used for the probability distribution $\ugwt_h(P) \in \mP(\mTb_*)$ constructed
from an admissible probability distribution 
$P \in \mP(\mTb_*^h)$, by establishing that $\ugwt_h(P)$ is unimodular.
The proof is given in Appendix~\ref{sec:UGWT-unimod}.

%\color{red}
%{\bf I have not yet cleaned up the proof of this claim in Appendix~\ref{sec:UGWT-unimod}.}
%\color{black}

\begin{lem}
  \label{lem:ugwt-P-is-unimodular}
  Let $h \ge 1$.
For an admissible probability distribution 
$P \in \mP(\mTb_*^h)$, let $\ugwt_h(P) \in \mP(\mTb_*)$ denote the 
marked unimodular Galton--Watson tree with  depth $h$ neighborhood distribution $P$.
Then $\ugwt_h(P)$ is a unimodular distribution.
\end{lem}

The following 
%summarizes the properties of the measure $\ugwt_h(P)$.
proposition states a key property of the probability distribution $\ugwt_h(P)$, which should be reminiscent of a finite order Markov property.
This is an important result for understanding the structure of $\ugwt_h(P)$.
The proof, which is provided in Appendix~\ref{sec:prop-markov-galt},
is very similar to the proof of the second part of
Proposition~1.1 in \cite{bordenave2015large}.

%\color{red}
%Note to self:
%I have not yet completely checked the proof of this claim in Appendix~\ref{sec:prop-markov-galt}.
%\color{black}

\begin{prop}
  \label{prop:ugwthP-markov}
  Let $h \ge 1$
  and let $P \in
  \mP(\mTb_*^h)$, i.e. $P$ is an admissible probability distribution.
  Then, for all $k \ge h$, we have
  \begin{equation}
    \label{eq:prop-consistency}
    \ugwt_k((\ugwt_h(P))_k) = \ugwt_h(P).
  \end{equation}
\end{prop}

%please note that 
The following proposition is not used in any way in the subsequent discussion. 
%\color{red} 
%{\bf This may be false, because the unimodularity of $\ugwt_h(P)$,
%which is used later, seems to have been addressed in Proposition 1.}
%\color{black}
The proof 
depends on several results to be developed during the course of this document,
and is provided in the last of the appendices, namely
Appendix~\ref{sec:prop-unim-galt}.

%aa
\begin{prop}
  \label{prop:ugwthP-properties}
  Given 
  %$h \in \nats$ 
  an integer $h \ge 1$
  and an admissible probability distribution $P \in
  \mP(\mTb_*^h)$, the probability distribution $\ugwt_h(P)$ is sofic. 
  %Moreover, for all $k > h$, we have
  %\begin{equation}
   % \label{eq:prop-consistency}
   % \ugwt_k((\ugwt_h(P))_k) = \ugwt_h(P).
  %\end{equation}
\end{prop}

%\color{red}
%Note to self:
%I have not yet completely checked the proof of this claim in Appendix~\ref{sec:prop-unim-galt}.
%I will do this after the issue of whether the proof actually establishes that $\ugwt_h(P)$ is sofic is resolved.
%\color{black}

%\color{red}
%Note to Payam:
%The proof in Appendix~\ref{sec:prop-unim-galt}
%does not seem to show that $\ugwt_h(P)$ is sofic.
%See the Appendix.
%\color{black}
%\pres{sorry, there was a short step missing at the end of the Appendix, which is
%added now.}

For $h \ge 1$ and admissible $P \in \mP(\mTb_*^h)$ such that $d:= \evwrt{P}{\deg_T(o)} > 0$, let $\pi_P$ denote the probability
distribution on $(\edgemark \times \mTb_*^{h-1}) \times (\edgemark \times
\mTb_*^{h-1})$ defined as
\begin{equation*}
  \pi_P(t, t') := \frac{e_P(t, t')}{d}.
\end{equation*}
%Note that
%the assumption that $P$ is admissible implies that $e_P(t, t') \leq d < \infty$ for
%all $t, t' \in \edgemark \times \mTb_*^{h-1}$. Moreover, 
Since for each $[T,
o] \in \mTb_*$ we have 
\begin{equation*}
\deg_T(o) = \sum_{t, t' \in \edgemark \times
  \mTb_*^{h-1}} E_h(t, t')(T, o), 
\end{equation*}
we have $d = \sum_{t, t' \in \edgemark
  \times \mTb_*^{h-1}} e_P(t, t')$. Consequently, $\pi_P$ is indeed a
probability distribution. 
%\marginpar{In \cite{bordenave2014large}, $\pi_P$ is only defined for  $P \in
%  \mP_h$. However, the assumption of $P$ being admissible is enough to give the
%  definition. }

 For $h \ge 1$ and admissible $P \in \mP(\mTb_*^h)$ with $H(P) < \infty$ and $\evwrt{P}{\deg_T(o)}
 > 0$, define
    \begin{equation}
      \label{eq:Jh-def}
      J_h(P) := -s(d) + H(P) - \frac{d}{2} H(\pi_P) - \sum_{t, t' \in \edgemark \times \mTb_*^{h-1}} \evwrt{P}{\log E_h(t, t')!},
    \end{equation}
    where $d := \evwrt{P}{\deg_T(o)}$ is the average degree at the root and
    $s(d) = \frac{d}{2} - \frac{d}{2} \log d$. \glsadd{not:JhP}
% \marginpar{In the BC paper, $J_h(P)$ is defined only if $P \in \mP_h$,
%    but I have defined it here for the case that $H(P) < \infty$. It turns out
%    to be well defined as is discussed below, since we do not have a $\infty -
%    \infty$ case. The reason for defining it this way is that it makes further
%    statements more general. Specifically, in the proof of upper bound, I can
%    treat that upper bound for a wider range of $\mu$'s and use it later on to
%    show that if $\mu_h \notin \mP_h$, then the entropy is $-\infty$.}
  Note that $s(d)$ is finite, since $d < \infty$. Also, $H(P) < \infty$,
   $H(\pi_P) \geq 0$, and for each $t, t' \in \edgemark \times
  \mTb_*^{h-1}$, $\evwrt{P}{\log E_h(t, t')!} \geq 0$. Thereby, $J_h(P)$ is
  well-defined and is in the range $[-\infty, \infty)$.
  %\color{red}
  %We extend the definition of $J_h(P)$ for $P \in \mP(\mTb_*^h)$, $h \ge 1$, with $\evwrt{P}{\deg_T(o)} = 0$ by defining it to be $H(P)$ for such $P$.
  %{\bf Think of deleting this sentence.}
  %\color{black}
  
  %\color{red}  Should include a lemma here about the continuity in $P$  of $J_h(P)$. 
  %{\bf This is false, so delete this.}
  %\color{black}
%{\color{red} I
%        should decide if I want to be more informative about its upper bound. I
%        can say it as a footnote and refer to the following theorem for that} 

    % \item For $P \in \mP(\mTb_*^h)$, define $\bar{J}_h(P)$ to be $J_h(P)$ if $P$
    %   is admissible and $H(P) < \infty$; otherwise, define $\bar{J}_h(P)$ to be
    %   $-\infty$. 

%{\color{red} In the following, I have distinguished between the cases
%  $\evwrt{\mu}{\deg_T(o) \log \deg_T(o)} = \infty$ and $\evwrt{\mu}{\deg_T(o)
%    \log \deg_T(o)} < \infty$, and hence I do not need to define $\bar{J}_h$}

%{\color{red} It might be better in the statement of the following theorem not to
%mention $\vmn$ and $\vun$ because in the previous theorem we have said that for
%any sequences it works and is well defined.}

%aa
\glsadd{trm:strongadmissible}
% to change: adding the notion ''strongly admissible''
% {
% For integer $h \geq 0$, define $\mP_h$ to be the set of probability measures on
% $P \in \mP(\mTb_*^h)$  such that $P$ is admissible, $H(P) < \infty$ and
% $\evwrt{P}{\deg_T(o) \log \deg_T(o)} < \infty$. Later, in
% Corollary~\ref{cor:deg-log-deg-Ph}, we will show that $P \in \mP_h$ is
% equivalent to $P$ being admissible and  $\evwrt{P}{\deg_T(o) \log \deg_T(o)} <
% \infty$. \glsadd{not:mPh}
% For admissible $P \in \mP(\mTb_*^h)$ such that $d:= \evwrt{P}{\deg_T(o)} > 0$, define $\pi_P$ to be the probability
% distribution on $(\edgemark \times \mTb_*^{h-1}) \times (\edgemark \times
% \mTb_*^{h-1})$ defined as
% }
\begin{definition}
    \label{def:stronglyadmissible}
    %{\em [Strongly admissible probability distributions on $\mTb_*^h$
 %isomorphism classes of rooted marked trees of depth bounded by a fixed bound
    %] }
    
For integer $h \geq 1$, we say that  a  probability distribution
$P \in \mP(\mTb_*^h)$ is \emph{strongly admissible} if  $P$ is admissible, $H(P) < \infty$, and
$\evwrt{P}{\deg_T(o) \log \deg_T(o)} < \infty$. Let $\mP_h$ denote the set of
strongly admissible probability distributions $P \in \mP(\mTb_*^h)$.
\end{definition}

%In Corollary~\ref{cor:deg-log-deg-Ph}, we will show that $P \in \mP_h$ is
%equivalent to $P$ being admissible, $P \in \mP(\mTb_*^h)$ (rather than just $P \in \mP(\mGb_*^h)$), and  $\evwrt{P}{\deg_T(o) \log \deg_T(o)} <
%\infty$. \glsadd{not:mPh}

In part 2 of Corollary~\ref{cor:deg-log-deg-Ph} of Lemma~\ref{lem:PPh--PtildePh+1} below, we show that, for $h \ge 1$ and $P \in \mP(\mTb_*^h)$,
the admissibility of $P$, together with 
the condition $\evwrt{P}{\deg_T(o) \log \deg_T(o)} < \infty$ 
is necessary and sufficient for $P$ to be strongly admissible, i.e. $P \in \mP_h$. 
Namely, the requirement that $H(P) < \infty$ in the definition of 
strong admissibility of $P$ is automatic given the other requirements, and need not be explicitly imposed.

In particular, this means that, for a unimodular $\mu \in \mP_u(\mTb_*)$,
if
$\evwrt{\mu}{\deg_T(o) \log \deg_T(o)} < \infty$, then for all 
%$h \in \nats$, 
$h \ge 1$
we have $\mu_h \in \mP_h$.
This is because $\mu$ being unimodular with $\deg(\mu) <
  \infty$ implies that $\mu_h$ is admissible for
all 
%$h \in \nats${\color{pedit}, 
$h \ge 1$,
as we show in Lemma~\ref{lem:unimodular-is-admissible}.
%\color{pedit} 

The proof of Lemma~\ref{lem:PPh--PtildePh+1} below is given in Appendix~\ref{sec:proof-lem-Ptilde-Ph+1}.

%\color{red}
%Note to self:
%I have not yet completely checked the proof in Appendix~\ref{sec:proof-lem-Ptilde-Ph+1}.
%I find the way the proof is written difficult to decipher
%(see the comment in the Appendix). It would be best to discuss this in the next phase.
%\color{black}

%aa
\begin{lem}
  \label{lem:PPh--PtildePh+1}
  Given a unimodular $\mu \in \mP_u(\mTb_*)$ and 
  %$h \in \nats$,
  an integer $h \ge 1$,
  assume that with $P
  := \mu_h$, we have $P$ is strongly admissible, i.e.\ $P \in \mP_h$. Then, with $\tP := \mu_{h+1}$, we have $\tP \in \mP_{h+1}$. 
  %is such that $\sum_{t,t' \in \edgemark \times
  %  \mTb_*^{h-1}} |e_P(t,t') \log e_P(t,t')| < \infty$ and
%  let $\mu := \ugwt_h(P)$ and  $\tP :=
 % \mu_{h+1} \in \mP(\mTb_*^{h+1})$. Then,  we have $\tP \in \mP_{h+1}$.
\end{lem}

%bb

%{\color{pedit}
%aa
\begin{cor}
  \label{cor:deg-log-deg-Ph}
  The following hold:
  \begin{enumerate}
  \item Assume that for a unimodular measure $\mu \in \mP_u(\mTb_*)$, we have
  $\evwrt{\mu}{\deg_T(o) \log \deg_T(o)} < \infty$. Then, for all
  %$h \in \nats$,
  integers $h \ge 1$, we have $\mu_h \in \mP_h$,
  i.e.\
  $\mu_h$ is strongly admissible.
  \item Let $h \ge 1$ and $P \in \mP(\mTb_*^h)$.
  Then $P \in \mP_h$, i.e.\ $P$ being strongly admissible,
  is equivalent to $P$ admissible and $\evwrt{P}{\deg_T(o) \log \deg_T(o)} <
  \infty$.
  \end{enumerate}
\end{cor}

% bb

\begin{proof}
  To prove part 1, let $\mu \in \mP_u(\mTb_*)$ with 
  $\evwrt{\mu}{\deg_T(o) \log \deg_T(o)} < \infty$.
  By Lemma~\ref{lem:PPh--PtildePh+1}, to show that $\mu_h \in \mP_h$
  for all $h \ge 1$, it suffices to show that $\mu_1 \in \mP_1$.
  Let $P:= \mu_1$. From $\evwrt{\mu}{\deg_T(o) \log \deg_T(o)} < \infty$
  we have $\evwrt{\mu}{\deg_T(o)} < \infty$. Since
  $\deg(\mu) = \evwrt{\mu}{\deg_T(o)}$, 
  from Lemma~\ref{lem:unimodular-is-admissible} we see that $P$ is admissible.
  We also have 
  $\evwrt{P}{\deg_T(o) \log \deg_T(o)} = \evwrt{\mu}{\deg_T(o) \log \deg_T(o)} < \infty$. By the definition of strong admissibility in 
  Definition~\ref{def:stronglyadmissible}, all that remains to show is that
  $H(P) < \infty$.
  %Admissibility of $P$ is
  %directly implied by the unimodularity of $\mu$. Indeed, for $t, t' \in \edgemark
  %\times \mTb_*^0$, we have
  %\begin{align*}
  %  e_P(t,t') &= \evwrt{\mu}{E_1(t,t')(T,o)} = \evwrt{\mu}{\sum_{v \sim_T o} \one{\etype_T^1(o,v) = (t,t')}} \\
  %            &\stackrel{(a)}{=} \evwrt{\mu}{\sum_{v \sim_T o } \one{\etype_T^1(o,v) = (t',t)}} = e_P(t',t),
  %\end{align*}
  %where in $(a)$, we have used the unimodularity of $\mu$. Also, we have
  %$\evwrt{P}{\deg_T(o)} = \evwrt{\mu}{\deg_T(o)} < \infty$. Hence, $P$ is
  %admissible. On the other hand, we have
  %\begin{equation*}
  % \evwrt{P}{\deg_T(o) \log \deg_T(o)} =
  %\evwrt{\mu}{\deg_T(o) \log \deg_T(o)} < \infty. 
  %\end{equation*}
 %Thereby, in order to show $P
  %\in \mP_1$, it remains to show that $H(P) < \infty$. To see this, 
  For this, observe that a rooted tree $[T,o] \in \mTb_*^1$
  is uniquely determined by knowing the integers 
  \begin{equation*}
    N_{x,x'}^{\theta, \theta'}(T,o) := |\{v \sim_T o: \xi_T(v,o) = x, \tau_T(o) = \theta, \xi_T(o,v) = x', \tau_T(v) = \theta'\}|,
  \end{equation*}
  for all $x, x' \in \edgemark$ and $\theta, \theta' \in \vermark$. On the other
  hand, for $x,
  x' \in \edgemark$ and $\theta, \theta' \in \vermark$,
  $\evwrt{P}{N_{x,x'}^{\theta, \theta'}(T,o)} \leq \evwrt{P}{\deg_T(o)} <
  \infty$. Consequently, when  $[T, o] \sim P$, the entropy of the random
  variable $N_{x,x'}^{\theta, \theta'}(T,o)$ is
  finite.
  %{\color{pedit}
    To see this, for $k \geq 0$, let $p_k$ denote the probability under $P$ that
    $N_{x,x'}^{\theta, \theta'}(T,o) = k$. Furthermore, let $q_k :=
    \frac{1}{2^{k+1}}$. Then we have
    \begin{equation*}
      H(N_{x,x'}^{\theta, \theta'}(T,o)) = \sum_{k=0}^\infty p_k \log \frac{1}{p_k} \stackrel{(a)}{\leq} \sum_{k=0}^\infty p_k \log \frac{1}{q_k} = \left(1 + \evwrt{P}{N_{x,x'}^{\theta,\theta'}(T,o)}\right) \log 2 < \infty,
    \end{equation*}
    where step (a) comes from Gibbs' inequality, i.e. the 
    nonnegativity of relative entropy, $\sum_{k=0}^\infty p_k \log \frac{p_k}{q_k} \ge 0$.
  %}%
  Since $\edgemark$ and $\vermark$ are finite sets, we have $H(P) < \infty$,
  which completes the proof of part 1.
  %and
  %thus $P \in \mP_1$. As a result, from
  %Lemma~\ref{lem:PPh--PtildePh+1} and induction, $\mu_h \in \mP_h$ for all $h
  %\in \nats$. 

To see part 2, first note that, by definition, if $P \in \mP_h$ then $P$ is admissible and
$\evwrt{P}{\deg_T(o)  \log \deg_T(o)} < \infty$. To show the other direction,
define $\mu := \ugwt_h(P)$. By Lemma~\ref{lem:ugwt-P-is-unimodular} we have $\mu \in \mP_u(\mT_*)$.
%\color{red}
%{\bf Need to figure out where this comes from. Currently it seems to involve 
%Proposition~\ref{prop:ugwthP-properties} but that is too deep in the chain of arguments and may even be using this Corollary in its proof }
%\color{black}
Further, we have
\begin{equation*}
 \evwrt{\mu}{\deg_T(o) \log \deg_T(o)} = \evwrt{P}{\deg_T(o) \log \deg_T(o)} <
\infty.
\end{equation*}
 Consequently, the first part of this corollary implies that $P = \mu_h \in
\mP_h$, and this completes the proof. 
\end{proof}

% bb

%%% Local Variables: 
%%% mode: latex
%%% TeX-master: "Note-41_BC-ent-arxiv.tex"
%%% End: 

%aa
\section{Towards the Definition of the Marked BC Entropy}
\label{sec:definition-entropy}

% bb

In this section, we make the initial steps towards defining a generalization of the notion of entropy defined in
\cite{bordenave2015large} for the marked regime discussed above. Our entropy
function is going to be defined 
for probability distributions $\mu \in \mP(\mGb_*)$ with 
$0 < \deg(\mu) < \infty$.
%on the set of Borel probability measures on $\mGb_*$.
We call this notion of entropy the {\em BC entropy} after Bordenave and Caputo.

Let the finite edge and vertex mark sets $\edgemark$ and $\vermark$
respectively be given. 
%Before defining the entropy, we need to set
%up some notation. 
An edge mark
count vector is defined to be a vector of nonnegative integers $\vm := (m(x, x'): x, x' \in \edgemark)$ such that
$m(x,x') = m(x',x)$ for all $x, x' \in \edgemark$. 
A vertex mark count
vector is defined to be a vector of nonnegative integers $\vu := (u(\theta): \theta \in
\Theta)$. Since $\edgemark$ is finite, we may assume it is an ordered set.
%With this, for an edge mark count vector $\vm$, 
We define $\snorm{\vm}_1 := 
\sum_{x \leq x' \in \edgemark} 
%\sum_{x \leq x'} 
m(x, x')$
%Furthermore, for a vertex mark count vector $\vu$, we define 
and $\snorm{\vu}_1 :=
\sum_{\theta \in \vermark} u(\theta)$.

For an integer $n \in \nats$ and edge mark and vertex mark count vectors $\vm$
and $\vu$, define $\mGn_{\vm, \vu}$ to be the set of 
%\color{blue}
marked
%\color{black}
graphs on the vertex set
$[n]$ such that $\vm_G = \vm$ and $\vu_G = \vu$. \glsadd{not:Gnmu} Note that
$\mGn_{\vm, \vu}$ is empty unless $\snorm{\vm}_1 \leq \binom{n}{2}$ and
$\snorm{\vu}_1 = n$. Furthermore, if these two conditions are satisfied, it is
easy to see that 
  \begin{equation}
    \label{eq:sizegnmnun-exact}
    |\mGn_{\vm,\vu} | = \frac{n!}{\prod_{\theta \in \vermark} u(\theta)!} \times \frac{\frac{n(n-1)}{2}!}{\prod_{x\leq x' \in \edgemark} m(x,x')! \times \left ( \frac{n(n-1)}{2} - \snorm{\vm}_1 \right)!} \times 2^{\sum_{x < x' \in \edgemark} m(x,x')}.
  \end{equation}

An {\em average degree vector} is defined to be a vector of nonnegative reals $\vd = (d_{x,x'} : x, x' \in \edgemark)$ such
  that for all $x, x' \in \edgemark$,  we have  $d_{x,x'} =
  d_{x',x}$. Moreover, we require that $\sum_{x,x' \in \edgemark} d_{x,x'} > 0$.
% \marginpar{\color{red} I
%     am convinced at this point that there is nowhere that we actually use the
%     assumption that $\sum d_{x,x'} > 0$, therefore I might want to drop it later, but I should
%     be careful in reducing the $-\infty$ part for not $\mTb_*$ supported}

%aa
\begin{definition}
\label{def:deg-seq-adapt}  
%{\em [Adapted sequence of edge and vertex mark count vectors] }

  Given an average degree vector $\vd$ and a probability distribution $Q = (q_\theta: \theta \in \vermark)$, we say that a sequence 
$(\vmn,\vun)$ of edge mark count vectors and vertex mark count vectors $\vmn$ and $\vun$ is adapted to $(\vd, Q)$, if the following conditions hold:
  \begin{enumerate}
  \item For each $n$, we have $\snorm{\vmn}_1 \leq \binom{n}{2}$ and $\snorm{\vun}_1 = n$.
  \item For $x \in \edgemark$, we have $\mn(x,x) / n \rightarrow d_{x,x}/2$.
  \item For $x \neq x' \in \edgemark$, we have $\mn(x, x') / n \rightarrow d_{x,x'} = d_{x',x}$.
  \item For $\theta \in \vermark$, we have $\un(\theta) / n \rightarrow q_\theta$.
  \item For $x, x' \in \edgemark$, $d_{x,x'} = 0$ implies $\mn(x, x') = 0$ for all $n$.
  \item For $\theta \in \vermark$, $q_\theta  = 0$ implies $\un(\theta) = 0$ for all $n$.
  \end{enumerate}
\end{definition}

%bb

If $\vmn$ and $\vun$ are sequences such that $(\vmn,\vun)$ is adapted to $(\vd, Q)$ then, using Stirling's approximation in \eqref{eq:sizegnmnun-exact}, we have
  \begin{equation}
    \label{eq:log-mGnmnun-Stirling}
    \log | \mGn_{\vmn, \vun} | =  \snorm{\vmn}_1 \log n + n H(Q)  + n \sum_{x,x' \in \edgemark} s(d_{x,x'}) + o(n),
  \end{equation}
  where \glsadd{not:sofd}
  \begin{equation*}
    s(d) :=
    \begin{cases}
      \frac{d}{2} - \frac{d}{2} \log d & d > 0, \\
      0 & d = 0.
    \end{cases}
  \end{equation*}
  % \begin{marginpar}
  %   {
  % \color{blue}
  % It would be good to include a proof of this in an appendix.
  % \color{pedit} \textbf{Response:} I am working on it. 
  %   }
  %   \end{marginpar}%
  See Appendix~\ref{sec:app-Stirling} for the details on how to
    derive~\eqref{eq:log-mGnmnun-Stirling}.
To simplify the notation, we may write $s(\vd)$ for $\sum_{x, x' \in \edgemark} s(d_{x,x'})$.
%\color{red}
%Note to Payam:
%The notation $s(\vd)$ needs to be included in the glossary.
%\color{black}
\glsadd{not:sofvecd}
%\pres{added to glossary}
% for the Stirling simplification in the unmarked regime, look at the note
% 2018-08-16_log-Gnm-unmarked-simplification.pdf
% and for the above Stirling approximation details in the marked regime, look at
% 2018-08-16_log-Gnmnun-marked-simplification.pdf

%\color{red}
%{\bf I have not yet checked the proof in Appendix~\ref{sec:app-Stirling}.}
%\color{black}

To lead up to the definition of the BC entropy in 
Definition~\ref{def:BC-entropy-new}, we now give
the definitions of upper and lower BC entropy.
%A full understanding of this definition depends on results to be proved in Theorems~\ref{thm:badcases} and 
%\ref{thm:bch-properties} in the next section.

\begin{definition}
  \label{def:BC-entropy}
  %{\em [Upper and lower BC entropy] }
  
Assume $\mu \in \mP(\mGb_*)$ is given, with $0 < \deg(\mu) < \infty$. For $\epsilon>0$, and edge and vertex mark count vectors
$\vm$ and $\vu$, define
  \begin{equation*}
    \mGn_{\vm, \vu} (\mu, \epsilon) := \{ G \in \mGn_{\vm, \vu}: \dlp(U(G), \mu) < \epsilon \}.
  \end{equation*}
  Fix an average degree vector $\vd$ and a probability distribution $Q = (q_\theta:
  \theta \in \vermark)$, and also fix sequences of edge and vertex mark
  count vectors $\vmn$ and $\vun$ such that $(\vmn,\vun)$ is adapted to $(\vd, Q)$. With these, define
  \begin{equation*}
    \bchover_{\vd, Q}(\mu, \epsilon)\condmnun := \limsup_{n \rightarrow \infty} \frac{\log |\mGn_{\vmn, \vun}(\mu, \epsilon)| - \snorm{\vmn}_1 \log n}{n},
  \end{equation*}
which we call the $\epsilon$--upper BC entropy. Since this is increasing
in $\epsilon$, we can define the {\em upper BC entropy} as 
  \begin{equation*}
    \bchover_{\vd, Q}(\mu)\condmnun := \lim_{\epsilon \downarrow 0} \bchover_{\vd, Q}(\mu, \epsilon)\condmnun.
  \end{equation*}
We may define the $\epsilon$--lower BC entropy $\bchunder_{\vd, Q}(\mu,
\epsilon)\condmnun$ similarly as
\begin{equation*}
    \bchunder_{\vd, Q}(\mu, \epsilon)\condmnun := \liminf_{n \rightarrow \infty} \frac{\log |\mGn_{\vmn, \vun}(\mu, \epsilon)| - \snorm{\vmn}_1 \log n}{n}.
  \end{equation*}
Since this is increasing
in $\epsilon$, we can define the {\em lower BC entropy} $\bchunder_{\vd, Q}(\mu)\condmnun$ as
\begin{equation*}
    \bchunder_{\vd, Q}(\mu)\condmnun := \lim_{\epsilon \downarrow 0} \bchunder_{\vd, Q}(\mu, \epsilon)\condmnun.
  \end{equation*}
 
    \end{definition}

To close this section, we prove an upper semicontinuity result that 
will be superseded later by the upper semicontinuity result of 
Theorem~\ref{thm:bch-semicont}.

%aa
\begin{lem}
  \label{lem:BC-ent-upper-semicontinouous}
  Assume that a sequence $\mu_k \in \mP(\mGb_*)$ together with  $\mu \in
  \mP(\mGb_*)$ are given such that $\mu_k \Rightarrow \mu$. Let 
  $\vd = (d_{x,x'}:x,x' \in \edgemark)$
  be an average degree vector and 
  $Q = (q_\theta: \theta \in \vermark)$
  a probability
  distribution. Let
  $\vmn, \vun$ be sequences such that $(\vmn,\vun)$ is adapted to $(\vd, Q)$. Then, we have
  \begin{equation*}
    \bchunder_{\vd, Q}(\mu)\condmnun \geq \limsup_{k \rightarrow \infty} \bchunder_{\vd,Q}(\mu_k)\condmnun.
  \end{equation*}
%  where in all terms, $\bchunder_{\vd, Q}$ is defined via the same sequences
%  $\vmn$ and $\vun$.
\end{lem}

% bb

\begin{proof}
  For $\epsilon>0$, let $B(\mu, \epsilon)$ denote the ball around
  $\mu$ of radius $\epsilon$ with respect to the \LP distance. Since $\mGb_*$ is
Polish, weak convergence in $\mP(\mGb_*)$ is equivalent to convergence with
respect to the \LP metric. Hence, for $\epsilon>0$, 
  $\mu_k \Rightarrow \mu$ implies that for $k$ large enough, we have $B(\mu, \epsilon) \supseteq
  B(\mu_k, \epsilon/2)$. Therefore, we have $|\mGn_{\vmn, \vun}(\mu_k,
  \epsilon/2)| \leq |\mGn_{\vmn, \vun}(\mu, \epsilon)|$. Consequently,
  \begin{equation*}
    \bchunder_{\vd, Q}(\mu, \epsilon)\condmnun \geq \bchunder_{\vd, Q}(\mu_k, \epsilon/2)\condmnun \geq \bchunder_{\vd, Q}(\mu_k)\condmnun.
  \end{equation*}
  Taking the limsup on the right hand side and then sending $\epsilon$ to zero on
  the left hand side, we get the desired result.
\end{proof}

%Note that although it seems at  first sight that the values $\bchover_{\vd,
%  Q}(\mu)\condmnun$ and $\bchunder_{\vd, Q}(\mu)\condmnun$ depend on the choice of sequences
%$\vmn$ and $\vun$, in fact, as we will see later {\color{pedit}in Theorem~\ref{thm:bch-properties}}, that value is invariant under
%the choice of these sequences as long as $(\vmn,\vun)$ is adapted to $(\vd, Q)$.
%\begin{marginpar}%
%  {
%\color{blue}
%Later where? Give a precise reference.
%\color{pedit}
%\textbf{Response:} okay.
%}
%\end{marginpar}%
%In the
%following sections we study the properties of BC entropy. 

%The proofs of parts \ref{thm:BC-invariant} and
%\ref{thm:BC-well} of
%Theorem~\ref{thm:bch-properties}, which were relied on in
%Definition \ref{def:BC-entropy} to define the BC entropy,
%depend on a sequence of propositions, which are stated
%in the next section and proved in Section \ref{sec:ent-properties} after the development of a supporting structure called the {\em colored configuration model} in 
%Section \ref{sec:color-conf-model}. In the process of developing this proof, we will also establish several properties of the BC entropy.

%%% Local Variables: 
%%% mode: latex
%%% TeX-master: "Note-41_BC-ent-arxiv.tex"
%%% End: 

%aa
\section{Definition of the Marked BC Entropy and Main Results}
\label{sec:main-results}

In this section, we state the main theorems proved in this document.
These theorems establish properties of the upper and lower marked BC entropy,
which enable us to 
define the marked BC entropy and establish some of its properties. 
%Theorem \ref{thm:bch-properties} states that the BC entropy is well-defined and presents some of its main properties. 
%Theorem~\ref{thm:Jh} can be thoughts of as presenting a method to compute the BC entropy. 
The main propositions that are used to prove these theorems are also 
stated in this section and we give the proofs of these theorems,
  assuming that the propositions are proved. The proofs of the propositions themselves are given later in 
the document.

% bb

The following Theorem~\ref{thm:badcases} shows that certain conditions must 
be met for the marked BC entropy to be of interest.

\begin{thm}
  \label{thm:badcases}
  Let an average degree vector $\vd = (d_{x,x'} : x,x' \in \edgemark)$ and a
  probability distribution $Q = (q_\theta: \theta \in \vermark)$ be given. Suppose $\mu \in \mP(\mGb_*)$ with
 $0 < \deg(\mu) < \infty$ satisfies any one of the following conditions:
 \begin{enumerate}
    \item $\mu$ is not unimodular.
    \item $\mu$ is not supported on $\mTb_*$.
    \item $\deg_{x,x'}(\mu) \neq d_{x,x'}$ for some $x,x' \in \edgemark$, or $\vtype_\theta(\mu) \neq q_\theta$ for some $\theta \in \vermark$.
    \end{enumerate}
    Then, for any choice of the
    sequences $\vmn$ and $\vun$ such that $(\vmn,\vun)$ is adapted to $(\vd, Q)$, we have $\bchover_{\vd, Q}(\mu)\condmnun = -\infty$. 
  \end{thm}
  
%Theorem~\ref{thm:badcases} shows that unless 
%$\vd = \vdeg(\mu)$, $Q = \vvtype(\mu)$, and $\mu$
%is a unimodular measure on $\mTb_*$, we have 
%$\bchover_{\vd, Q}(\mu)\condmnun = -\infty$
%(and hence $\bchunder_{\vd, Q}(\mu)\condmnun = -\infty$).
%Thus, unless $\mu \in \mP_u(\mTb_*)$, 
%$\vd = \vdeg(\mu)$, and $Q = \vvtype(\mu)$, 
%we may 
%unambiguously write $\bchover_{\vd, Q}(\mu)$ and 
%$\bchunder_{\vd, Q}(\mu)$ for 
%$\bchover_{\vd, Q}(\mu)\condmnun$ and $\bchunder_{\vd, Q}(\mu)\condmnun$ respectively for any $\vmn$ and $\vun$ such that $(\vmn,\vun)$ is adapted to $(\vd, Q)$ and these are both equal and equal to $-\infty$. 
%We may therefore also write $\bch_{\vd, Q}(\mu)$ for the common 
%value of $\bchover_{\vd, Q}(\mu)$ and 
%$\bchunder_{\vd, Q}(\mu)$, i.e. $-\infty$, in this case.

%We first show that under any of the  conditions mentioned in part~\ref{thm:BC--infty} of
%Theorem~\ref{thm:bch-properties}, the entropy is $-\infty$. This is
%done in Propositions~\ref{prop:BC-not-unim_dQnotmatch--infty} and
%\ref{prop:BC-no-tree--infty} below. The proofs of these statements are given in Section~\ref{sec:BC-conditions--infty}.
Theorem~\ref{thm:badcases} is proved by means of 
Propositions~\ref{prop:BC-not-unim_dQnotmatch--infty} and
\ref{prop:BC-no-tree--infty} below.

%aa
\begin{prop}
  \label{prop:BC-not-unim_dQnotmatch--infty}
  Assume that $\mu \in \mP(\mGb_*)$ with $0 < \deg(\mu) < \infty$ is given. Also,
assume that a degree vector
$\vd = (d_{x,x'}: x,x' \in \edgemark)$ and a probability distribution $Q =
(q_\theta: \theta \in \vermark)$ are given. Let $\vmn$ and $\vun$ be sequences such that $(\vmn,\vun)$ is adapted to
$(\vd, Q)$. 
  If $\mu$ is not unimodular, or $\vd \neq \vdeg(\mu)$, or $Q \neq
  \vvtype(\mu)$, we have $\bchover_{\vd, Q}(\mu)\condmnun = -\infty$. 
\end{prop}

% bb

%aa
\begin{prop}
  \label{prop:BC-no-tree--infty}
  Assume $\mu \in \mP(\mGb_*)$ with $0 < \deg(\mu) < \infty$ is given such that $\mu(\mTb_*) < 1$. Then, if
  $\vmn$ and $\vun$ are any sequences such that $(\vmn,\vun)$ is adapted to $(\vdeg(\mu), \vvtype(\mu))$,
  we have $\bchover_{\vdeg(\mu), \vvtype(\mu)}(\mu)\condmnun = -\infty$. 
\end{prop}

% bb

%Propositions~\ref{prop:BC-not-unim_dQnotmatch--infty}
%and~\ref{prop:BC-no-tree--infty} prove part \ref{thm:BC--infty} of
%Theorem~\ref{thm:bch-properties}, since $\bchover_{\vd, Q}(\mu)\condmnun = -\infty$
%implies $\bchunder_{\vd, Q}(\mu)\condmnun = -\infty$ and hence $\bch_{\vd, Q}(\mu)\condmnun = -\infty$.

The proofs of these statements are given in Section~\ref{sec:BC-conditions--infty} and it is immediate to see that they 
prove Theorem~\ref{thm:badcases}.
A consequence of Theorem~\ref{thm:badcases}
is that the only case of interest in the discussion of marked
BC entropy is when $\mu \in \mP_u(\mTb_*)$,
$\vd = \vdeg(\mu)$, $Q = \vvtype(\mu)$,
and the
sequences $\vmn$ and $\vun$ are such that $(\vmn,\vun)$ is adapted to
$(\vdeg(\mu), \vvtype(\mu))$.
Namely, the only upper and lower marked BC entropies of interest are 
$\bchover_{\vdeg(\mu), \vvtype(\mu)}(\mu)\condmnun$ and $\bchunder_{\vdeg(\mu), \vvtype(\mu)}(\mu)\condmnun$ respectively.

The following Theorem~\ref{thm:bch-properties} establishes that the upper and lower
marked BC entropies do not depend on the 
choice of the defining pair of sequences 
$(\vmn,\vun)$. Further, 
this theorem establishes that
the upper marked BC entropy 
is always equal to the lower marked BC entropy,

%aa
\begin{thm}
  \label{thm:bch-properties}
  Assume that an average degree vector $\vd = (d_{x,x'} : x,x' \in \edgemark)$ together with a
  probability distribution $Q = (q_\theta: \theta \in \vermark)$ are given. For
  any  $\mu \in \mP(\mGb_*)$ such that
 $0 < \deg(\mu) < \infty$, we have 
  \begin{enumerate}
  \item \label{thm:BC-invariant} The values of $\bchover_{\vd, Q}(\mu)\condmnun$ and
    $\bchunder_{\vd, Q}(\mu)\condmnun$ are invariant under the specific choice of the
    sequences $\vmn$ and $\vun$ such that $(\vmn,\vun)$ is adapted to $(\vd, Q)$. With this,
    we may simplify the notation and unambiguously write $\bchover_{\vd, Q}(\mu)$ and
    $\bchunder_{\vd, Q}(\mu)$. 
  \item \label{thm:BC-well} 
  %$\bch_{\vd, Q}(\mu)$ is well-defined, i.e.\
  $\bchover_{\vd, Q}(\mu) = \bchunder_{\vd, Q}(\mu)$. 
  We may therefore unambiguously write $\bch_{\vd, Q}(\mu)$ 
for this common value,
and call it the {\em marked BC entropy} of 
$\mu \in \mP(\mGb_*)$ for the 
average degree vector $\vd$ and a probability distribution $Q = (q_\theta:
  \theta \in \vermark)$.
  Moreover, $\bch_{\vd, Q}(\mu) \in [-\infty, s(\vd) + H(Q)]$.

  \end{enumerate}
\end{thm}

% bb

%aa
\iffalse
\begin{rem}
  In view of part~\ref{thm:BC--infty} of  Theorem~\ref{thm:bch-properties}, the
  entropy is $-\infty$ unless $\vd = \vdeg(\mu)$, $Q = \vvtype(\mu)$, and $\mu$
  is a unimodular measure on $\mTb_*$. Therefore, for $\mu \in \mP_u(\mTb_*)$
  with $\deg(\mu) \in (0,\infty)$, we simplify the notation and use $\bch(\mu)$
  for  $\bch_{\vdeg(\mu), \vvtype(\mu)}(\mu)$. Likewise, we may write
  $\bchunder(\mu)$ and $\bchover(\mu)$ for $\bchunder_{\vdeg(\mu),
    \vvtype(\mu)}(\mu)$ and $\bchover_{\vdeg(\mu),
    \vvtype(\mu)}(\mu)$, respectively. 
\end{rem}
\fi

From Theorem~\ref{thm:badcases} we conclude that unless 
$\vd = \vdeg(\mu)$, $Q = \vvtype(\mu)$, and $\mu$
  is a unimodular measure on $\mTb_*$, we have 
  $\bch_{\vd, Q}(\mu) = -\infty$. 
  In view of this, for $\mu \in \mP(\mGb_*)$
  with $\deg(\mu) < \infty$, we write $\bch(\mu)$
  for  $\bch_{\vdeg(\mu), \vvtype(\mu)}(\mu)$. Likewise, we may write
  $\bchunder(\mu)$ and $\bchover(\mu)$ for $\bchunder_{\vdeg(\mu),
    \vvtype(\mu)}(\mu)$ and $\bchover_{\vdeg(\mu),
    \vvtype(\mu)}(\mu)$, respectively. 
    Note that, unless $\mu \in \mP_u(\mTb_*)$, 
    we have $\bchover(\mu) = \bchunder(\mu) = \bch(\mu) = -\infty$.
    
We are now in a position to define the marked BC entropy.
%A full understanding of this definition depends on 
%properties to be proved in Theorems~\ref{thm:badcases} and 
%\ref{thm:bch-properties} in the next section.

\begin{definition}
  \label{def:BC-entropy-new}
  %{\em [BC entropy] }
  
  For $\mu \in \mP(\mGb_*)$
  with $0 < \deg(\mu) < \infty$, the marked BC entropy of $\mu$ is defined to be $\bch(\mu)$.
 
%Assume $\mu \in \mP(\mGb_*)$ is given, with $\deg(\mu) < \infty$. Fix an average degree vector $\vd$ and a probability distribution $Q = (q_\theta: \theta \in \vermark)$.
%Part \ref{thm:BC-well} of
%Theorem~\ref{thm:bch-properties} shows that
%$\bchover_{\vd, Q}(\mu) = \bchunder_{\vd, Q}(\mu)$.
%We may therefore unambiguously write $\bch_{\vd, Q}(\mu)$ 
%for this common value,
%and call it the {\em marked BC entropy} of 
%$\mu \in \mP(\mGb_*)$ for the 
%average degree vector $\vd$ and a probability distribution $Q = (q_\theta: \theta \in \vermark)$.

\end{definition}
  
Next, we give a recipe to compute the marked BC entropy for the marked unimodular
Galton--Watson trees defined in
Section~\ref{sec:markovian-ugwt}. 
%We also show that the entropy of a certain
%class of measures $\mu \in \mP_u(\mTb_*)$ can be approximated with the marked
%unimodular Galton--Watson trees with neighborhood distribution given by the
%truncation of $\mu$ up to any depth, and the approximation becomes
%asymptotically exact as we send this depth to infinity. 
We also characterize the marked BC entropy of any
$\mu \in \mP_u(\mTb_*)$ 
in terms of the marked BC entropies of
the marked
unimodular Galton--Watson trees with neighborhood distribution given by the
truncation of $\mu$ up to any depth.
%This will be shown in
%Theorem~\ref{thm:Jh} below. But before that, we need some notation. 

\glsadd{trm:strongadmissible}
% to change: adding the notion ''strongly admissible''
% {
% For integer $h \geq 0$, define $\mP_h$ to be the set of probability measures on
% $P \in \mP(\mTb_*^h)$  such that $P$ is admissible, $H(P) < \infty$ and
% $\evwrt{P}{\deg_T(o) \log \deg_T(o)} < \infty$. Later, in
% Corollary~\ref{cor:deg-log-deg-Ph}, we will show that $P \in \mP_h$ is
% equivalent to $P$ being admissible and  $\evwrt{P}{\deg_T(o) \log \deg_T(o)} <
% \infty$. \glsadd{not:mPh}
% For admissible $P \in \mP(\mTb_*^h)$ such that $d:= \evwrt{P}{\deg_T(o)} > 0$, define $\pi_P$ to be the probability
% distribution on $(\edgemark \times \mTb_*^{h-1}) \times (\edgemark \times
% \mTb_*^{h-1})$ defined as
% }

  \begin{thm}
  \label{thm:Jh}
  Let $\mu \in \mP_u(\mTb_*)$ be a unimodular probability measure with 
  $0 < \deg(\mu) < \infty$. %Moreover, assume that $\vmn$ and $\vun$ are any
%  sequences adapted to $(\vdeg(\mu), \vvtype(\mu))$.
  Then, 
  \begin{enumerate}
  \item \label{thm:Jh--infty} If $\evwrt{\mu}{\deg_T(o) \log \deg_T(o)} = \infty$, then $\bchunder(\mu)
    = \bchover(\mu) = \bch(\mu) = -\infty$.
  \item  \label{thm:Jh-lim} If $\evwrt{\mu}{\deg_T(o) \log \deg_T(o)} < \infty$, then, for each $h \ge 1$, the probability measure $\mu_h$ is admissible, and $H(\mu_h) < \infty$. Furthermore, the sequence $(J_h(\mu_h): h \geq 1)$ is
    nonincreasing, and
      \begin{equation*}
    \bchunder(\mu) = \bchover(\mu) = \bch(\mu) = \lim_{h \rightarrow \infty} J_h(\mu_h). 
  \end{equation*}
  \end{enumerate}
\end{thm}

%bb

Now, we proceed to state the propositions needed to prove Theorems~\ref{thm:bch-properties} and
\ref{thm:Jh} and explain how they  
prove the two theorems. We give the proofs of these propositions in Section~\ref{sec:ent-properties}.
Our proof techniques are similar to those given in
\cite{bordenave2015large}.

In view of Propositions~\ref{prop:BC-not-unim_dQnotmatch--infty}
and~\ref{prop:BC-no-tree--infty}, in order to address
parts~\ref{thm:BC-invariant} and \ref{thm:BC-well} of
Theorem~\ref{thm:bch-properties}, we may assume that $\mu \in \mP_u(\mTb_*)$, $\vd = \vdeg(\mu)$, and $Q = \vvtype(\mu)$, since otherwise
$\bchunder_{\vd, Q}(\mu) = \bchover_{\vd, Q}(\mu) = -\infty$. 
%Hence, we may simplify
%the notation and write $\bchunder(\mu)\condmnun$ and $\bchover(\mu)\condmnun$ for
%$\bchunder_{\vdeg(\mu), \vvtype(\mu)}(\mu)\condmnun$ and $\bchover_{\vdeg(\mu), \vvtype(\mu)}(\mu)\condmnun$, respectively{\color{pedit}, where $\vmn$ and
%$\vun$ are sequences such that $(\vmn,\vun)$ is adapted to $(\vdeg(\mu), \vvtype(\mu))$.} 
%Next, we show that  $\bchunder(\mu)\condmnun =
%\bchover(\mu)\condmnun$ and the common value is invariant under the choice of $\vmn$ and
%$\vun$. 
To prove part~\ref{thm:BC-invariant} of
Theorem~\ref{thm:bch-properties},
the strategy is to find a lower bound for $\bchunder_{\vdeg(\mu), \vvtype(\mu)}(\mu)\condmnun$ and
 an upper bound for $\bchover_{\vdeg(\mu), \vvtype(\mu)}(\mu)\condmnun$, and then to show that they match. 
 %However, it
%turns out that it is more convenient to first consider the case that $\mu_h \in
%\mP_h$ {\color{pedit} is strongly admissible} for all $h \in \nats$. Recall that $\mP_h$ is the set of probability
%distributions $P \in \mP(\mTb_*^h)$ where $P$ is admissible, $H(P) < \infty$ and
%$\evwrt{P}{\deg_T(o) \log \deg_T(o)} < \infty$. 
We first prove a lower bound for $\bchunder_{\vdeg(\mu), \vvtype(\mu)}(\mu)\condmnun$ when $\mu$ is of the
form $\ugwt_h(P)$ for $P \in \mP_h$ 
%{\color{pedit} being strongly admissible}.
being strongly admissible.

%aa
\begin{prop}
  \label{prop:lower-bound}
  Let $h \ge 1$. Let $P \in \mP_h$, i.e.
  $P$ is strongly admissible. Assume that with  $\mu :=
  \ugwt_h(P)$ we have  $0 < \deg(\mu) < \infty$. 
  Then, if $\vmn$ and $\vun$ are any sequences such that $(\vmn,\vun)$ is adapted to $(\vdeg(\mu),
  \vvtype(\mu))$, we have 
  \begin{equation*}
    \bchunder_{\vdeg(\mu), \vvtype(\mu)}(\mu)\condmnun \geq J_h(P).
  \end{equation*}
%  Here, $\bchunder(\mu)$ is obtained by using the sequences $\vmn$ and $\vun$. 
\end{prop}

% bb

%{\color{pedit} Note that, by definition, $\deg(\mu)$ is the expected degree at
 % the root in $\mu$, which is the same as the expected degree at the root in
 % $P$.}
The proof of Proposition~\ref{prop:lower-bound} is given in
Section~\ref{sec:lowerbound}. 
Now, for a unimodular probability measure $\mu \in \mP_u(\mTb_*)$ such that $0 < \deg(\mu) < \infty$ and $\evwrt{\mu}{\deg_T(o)
  \log \deg_T(o)} < \infty$,
  %Now, we validate part~\ref{thm:Jh-lim} of
%Theorem~\ref{thm:Jh} for such $\mu$. Note that
Corollary~\ref{cor:deg-log-deg-Ph} implies that,  %for all $ h \in \nats$, 
for all $h \ge 1$,
$\mu_h$ 
is strongly admissible, i.e.\ $\mu_h \in \mP_h$. In particular, $H(\mu_h) < \infty$ and $J_h(\mu_h)$ is well defined.
With this observation in mind, we
next give, for each $h \ge 1$, an upper bound for $\bchover_{\vdeg(\mu), \vvtype(\mu)}(\mu)\condmnun$, for 
$\mu \in \mP_u(\mTb_*)$
such that $0 < \deg(\mu) < \infty$ and 
$H(\mu_h) < \infty$.

%aa
\begin{prop}
  \label{prop:upper-bound}
  Let $h \ge 1$. Let $\mu \in
  \mP_u(\mTb_*)$ be a unimodular probability measure, with $0 < \deg(\mu) < \infty$ and $H(\mu_h) < \infty$. Then, if $\vmn$ and
  $\vun$ are sequences such that $(\vmn,\vun)$ is adapted to $(\vdeg(\mu), \vvtype(\mu))$, we have
  \begin{equation}
    \label{eq:upper-bound-Ph-statement}
    \bchover_{\vdeg(\mu), \vvtype(\mu)}(\mu)\condmnun \leq J_h(\mu_h).
  \end{equation}
%    Here, $\bchover(\mu)$ is obtained by using the sequences $\vmn$ and $\vun$. 
\end{prop}

% bb

The proof of Proposition~\ref{prop:upper-bound} is given in Section~\ref{sec:upperbound}.
Now, we consider the case $\evwrt{\mu}{\deg_T(o) \log \deg_T(o)} = \infty$
and show that the marked BC entropy is $-\infty$ in this case. 

%aa
\begin{prop}
  \label{prop:upper-bound-infty}
  %Assume that $h \in \nats$ and a unimodular probability measure 
  Let $\mu \in
  \mP_u(\mTb_*)$  be a unimodular probability measure such that $0 < \deg(\mu) < \infty$ and
  $\evwrt{\mu}{\deg_T(o) \log \deg_T(o)} = \infty$. Then, if $\vmn$ and
  $\vun$ are sequences such that $(\vmn,\vun)$ is adapted to $(\vdeg(\mu), \vvtype(\mu))$, we have
  \begin{equation}
    \label{eq:upper-bound-infty-statement}
    \bchover_{\vdeg(\mu), \vvtype(\mu)}(\mu)\condmnun = -\infty.
  \end{equation}
  % Here, $\bchover(\mu)$ is obtained by using the sequences $\vmn$ and $\vun$.
  %{\color{pedit} Consequently, we may unambiguously write the notation
   % $\bchover(\mu)$, and we have $\bchover(\mu) = -\infty$.}
\end{prop}

% bb

% Proposition~\ref{prop:upper-bound-infty} is proved in Section~\ref{sec:upperbound}.
% Now, we put the above statements together and show how they complete the proof
% of Theorems~\ref{thm:bch-properties} and \ref{thm:Jh}.
% We have already shown part \ref{thm:BC--infty} of
% Theorem~\ref{thm:bch-properties} in Propositions~\ref{prop:BC-not-unim_dQnotmatch--infty}
% and \ref{prop:BC-no-tree--infty}. In particular, if one of the conditions in part
% \ref{thm:BC--infty} hold, for any  sequences $\vmn$ and
% $\vun$, we have $\bchunder_{\vd, Q}(\mu) = \bchover_{\vd,Q}(\mu) = -\infty$.
% On the other hand, Proposition~\ref{prop:upper-bound-infty} implies that if
% $\evwrt{\mu}{\deg_T(o) \log \deg_T(o)} = \infty$, $\bchover(\mu) =
% \bchunder(\mu) = \bch(\mu) = -\infty$. Therefore, if $\mu$ satisfies one of the conditions of
% part \ref{thm:BC--infty} of Theorem~\ref{thm:bch-properties}, or 
% $\evwrt{\mu}{\deg_T(o) \log \deg_T(o)} = \infty$,  part~\ref{thm:Jh--infty} of
% Theorem~\ref{thm:Jh} and also parts \ref{thm:BC-invariant} and  \ref{thm:BC-well}
% of Theorem~\ref{thm:bch-properties} hold for such $\mu$. 

Proposition~\ref{prop:upper-bound-infty} is proved in
Section~\ref{sec:upperbound-new}. 

We now demonstrate how the propositions in this section can be used to prove 
Theorems~\ref{thm:bch-properties} and
\ref{thm:Jh}. We have already observed that 
Propositions~\ref{prop:BC-not-unim_dQnotmatch--infty}
and~\ref{prop:BC-no-tree--infty} imply that, in order to address
parts~\ref{thm:BC-invariant} and \ref{thm:BC-well} of
Theorem~\ref{thm:bch-properties}, we may assume that $\mu \in \mP_u(\mTb_*)$, $\vd = \vdeg(\mu)$, and $Q = \vvtype(\mu)$.
Proposition~\ref{prop:upper-bound-infty} then immediately implies parts \ref{thm:BC-invariant}
and  \ref{thm:BC-well} of Theorem~\ref{thm:bch-properties} and part~\ref{thm:Jh--infty} of
Theorem~\ref{thm:Jh}, for every $\mu \in \mP(\mGb_*)$ for which $0 < \deg(\mu) <
\infty$ and $\evwrt{\mu}{\deg_G(o) \log \deg_G(o)} = \infty$.

%{\color{pedit}
%Proposition~\ref{prop:upper-bound-infty} implies parts \ref{thm:BC-invariant}
%and  \ref{thm:BC-well} of Theorem~\ref{thm:bch-properties} and part~\ref{thm:Jh--infty} of
%Theorem~\ref{thm:Jh}, for every $\mu \in \mP(\mGb_*)$ for which $0 < \deg(\mu) <
%\infty$ and $\evwrt{\mu}{\deg_G(o) \log \deg_G(o)} = \infty$.

%To see this, note that part~\ref{thm:BC-invariant} of
%Theorem~\ref{thm:bch-properties} claims that for every $\mu \in \mP(\mGb_*)$ for
%which $0 < \deg(\mu) < \infty$,
%the quantities $\bchover_{\vd, Q}(\mu)\condmnun$ and $\bchunder_{\vd,
 % Q}(\mu)\condmnun$  do not depend on the choice of $(\vmn, \vun)$, 
%while part~\ref{thm:BC-well} of Theorem~\ref{thm:bch-properties} claims that these two quantities are the same and lie in the specified
%range. We have already proved part~\ref{thm:BC--infty} of
%Theorem~\ref{thm:bch-properties}, so it suffices to consider the case where
%$\mu$  is  
%unimodular, supported on $\mTb_*$, and $(\vd, Q) = (\vdeg(\mu), \vvtype(\mu))$. Under the additional
%assumption that  $\evwrt{\mu}{\deg_G(o) \log \deg_G(o)} = \infty$, Proposition~\ref{prop:upper-bound-infty} is seen to
%establish both claims. This then allows us to unambiguously use the notation
%$\bchover(\mu)$, $\bchunder(\mu)$ and $\bch(\mu)$ for all $\mu \in \mP(\mGb_*)$
%for which $0 < \deg(\mu) < \infty$ and  $\evwrt{\mu}{\deg_G(o) \log \deg_G(o)} =
%\infty$, and to conclude that $\bchover(\mu) = \bchunder(\mu) = \bch(\mu)$ for
%such $\mu$, which then establishes the claim of part~\ref{thm:BC-invariant} of Theorem~\ref{thm:Jh}.
%}

Thus it remains to consider the case of
unimodular $\mu \in \mP_u(\mTb_*)$ with 
$0 < \deg(\mu) < \infty$
and $\evwrt{\mu}{\deg_T(o)
  \log \deg_T(o)} < \infty$. 
  %Now, we validate part~\ref{thm:Jh-lim} of
%Theorem~\ref{thm:Jh} for such $\mu$. 
We have already observed that Corollary~\ref{cor:deg-log-deg-Ph} implies that  %for all $ h \in \nats$, 
for such $\mu$,
for all $ h \ge 1$, 
$\mu_h$ is strongly admissible, i.e.\ $\mu_h \in \mP_h$ and that this implies, in particular, that $H(\mu_h) < \infty$ and $J_h(\mu_h)$ is well
  defined.
  
We first show that, in this case, the sequence $J_h(\mu_h)$ is nonincreasing in $h$.
For $h \ge 1$, let $\nu^{(h)} := \ugwt_h(\mu_h)$. Observe that 
$\vdeg(\mu) = \vdeg(\nu^{(h)})$ and $\vvtype(\mu) =
  \vvtype(\nu^{(h)})$.
From Propositions~\ref{prop:lower-bound} and \ref{prop:upper-bound}, we
  have
  \begin{align*}
    J_{h+1}(\mu_{h+1}) &\leq \bchunder_{\vdeg(\mu), \vvtype(\mu)}(\nu^{(h+1)})\condmnun  \\
                       &\leq \bchover_{\vdeg(\mu), \vvtype(\mu)}(\nu^{(h+1)})\condmnun \\
                       &\leq J_h((\nu^{(h+1)})_h) \\
                       &= J_h(\mu_h),
  \end{align*}
  where the last equality uses the fact that $(\nu^{(h+1)})_h =
  (\ugwt_{h+1}(\mu_{h+1}))_h = \mu_h$,
  which is proved in Proposition~\ref{prop:ugwthP-markov}.
Hence, $J_\infty(\mu):= \lim_{h \rightarrow \infty} J_h(\mu_h)$ exists. 
Further, since Proposition~\ref{prop:upper-bound}  proves that $\bchover_{\vdeg(\mu), \vvtype(\mu)}(\mu)\condmnun \leq
    J_h(\mu_h)$ holds for all $h \ge 1$, we get $\bchover_{\vdeg(\mu), \vvtype(\mu)}(\mu)\condmnun
    \leq J_\infty(\mu)$.
Now, we show that $\bchunder_{\vdeg(\mu), \vvtype(\mu)}(\mu)\condmnun \geq J_\infty(\mu)$. 
Note that, since $\mu_h \in \mP_h$ is strongly admissible, Proposition~\ref{prop:lower-bound} implies
that $\bchunder_{\vdeg(\nu^{(h)}), \vvtype(\nu^{(h)})}(\nu^{(h)})\condmnun \geq J_h(\mu_h) \geq J_\infty(\mu)$,
where we have noted that, since 
$\vdeg(\mu) = \vdeg(\nu^{(h)})$ and $\vvtype(\mu) =
  \vvtype(\nu^{(h)})$, the pair of sequences 
  $(\vmn,\vun)$ is adapted to $(\vdeg(\nu^{(h)}), \vvtype(\nu^{(h)}))$.
On the other  hand,  $\nu^{(h)} \Rightarrow \mu$. 
  Therefore, using Lemma~\ref{lem:BC-ent-upper-semicontinouous}, we have
  \begin{align*}
    \bchunder_{\vdeg(\mu), \vvtype(\mu)}(\mu)|_{\vmn, \vun} &\geq \limsup_{h \rightarrow \infty} \bchunder_{\vdeg(\mu), \vvtype(\mu)}(\nu^{(h)})|_{\vmn, \vun} \\
                                                            &= \limsup_{h \rightarrow \infty} \bchunder_{\vdeg(\nu^{(h)}), \vvtype(\nu^{(h)})}(\nu^{(h)})|_{\vmn, \vun} \\
    &\geq J_\infty(\mu)
  \end{align*}
%
%{\color{pedit2}[comment by Payam: I realized that the following part of the
%  argument, which is colored in green, is a repetition of
%  Lemma~\ref{lem:BC-ent-upper-semicontinouous}. Therefore, I substituted it with
 % the above text in blue, and suggest removing the following green part.]}
%{\color{pedit}
%Therefore, given any $\epsilon > 0$,
%there exists $h \ge 1$ such that $\dlp(\mu, \nu^{(h)}) < \epsilon/2$.
%Using the triangle inequality for $\dlp$, this implies that for any sequences $\vmn$
%and $\vun$ such that $(\vmn,\vun)$ is adapted to $(\vdeg(\mu), \vvtype(\mu))$,  we have $\mGnmnun(\nu^{(h)},
%\epsilon/2) \subseteq \mGnmnun(\mu,\epsilon)$ for all $n$.
%Therefore, 
%\begin{align*}
%  \liminf_{n \rightarrow \infty} \frac{\log |\mGnmnun(\mu,\epsilon)| - \snorm{\vmn}_1 \log n}{n} &\geq  \liminf_{n \rightarrow \infty} \frac{\log |\mGnmnun(\nu^{(h)},\epsilon/2)| - \snorm{\vmn}_1 \log n}{n} \\
                                                                                %                 &\geq %\lim_{\delta \rightarrow 0} \liminf_{n \rightarrow \infty} \frac{\log |\mGnmnun(\nu^{(h)},\delta)| - \snorm{\vmn}_1 \log n}{n} \\
%  &= \bchunder_{\vdeg(\nu^{(h)}), \vvtype(\nu^{(h)})}(\nu^{(h)})\condmnun \\
                                                                                %                 &\geq J_\infty(\mu). 
%\end{align*}
%Sending $\epsilon$ to zero, we conclude that $\bchunder_{\vdeg(\mu), \vvtype(\mu)}(\mu)\condmnun \geq
%J_\infty(\mu)$.
%}

We have established that $J_\infty(\mu) \leq \bchunder_{\vdeg(\mu), \vvtype(\mu)}(\mu)\condmnun \leq \bchover_{\vdeg(\mu), \vvtype(\mu)}(\mu)\condmnun \leq J_\infty(\mu)$. 
This, in particular, implies that $\bchunder_{\vdeg(\mu), \vvtype(\mu)}(\mu)\condmnun = \bchover_{\vdeg(\mu), \vvtype(\mu)}(\mu)\condmnun$. 
To complete the proof of part~\ref{thm:BC-invariant} of
Theorem~\ref{thm:bch-properties} 
and the proof of part~\ref{thm:Jh-lim} of Theorem~\ref{thm:Jh} note that
$J_\infty(\mu)$ does not depend on the choice of the sequences $\vmn$ and
$\vun$. 
%Furthermore, we have already shown that if $(\vd, Q) \neq (\vdeg(\mu),
%\vvtype(\mu))$, we have $\bchunder_{\vd, Q}(\mu) = \bchover_{\vd, Q}(\mu) =
%\bch_{\vd, Q} = -\infty$. 

To complete the proof of part~\ref{thm:BC-well}
of
Theorem~\ref{thm:bch-properties},
note that the inequality
$\bch_{\vd, Q}(\mu) \leq s(\vd)+ H(Q)$ is a
direct consequence of \eqref{eq:log-mGnmnun-Stirling}.

%The proof of part~\ref{thm:BC-semic} of Theorem~\ref{thm:bch-properties} is
%evident from Lemma~\ref{lem:BC-ent-upper-semicontinouous} and the fact that
%$\bchover_{\vd, Q}(\mu_k) = \bchunder_{\vd, Q}(\mu_k)$ and also $\bchover_{\vd, Q}(\mu) = \bchunder_{\vd, Q}(\mu)$.

The proof of the propositions stated in this section rely on a generalization of the classical
graph configuration model called a {\emph{colored configuration model}}, which was 
 introduced in \cite{bordenave2015large}. In
Section~\ref{sec:color-conf-model} below, we review this framework and
generalize its properties to the marked regime. 
Using the tools developed in Section~\ref{sec:color-conf-model}, we give the
proof of these propositions in Section~\ref{sec:ent-properties}.

To close this section, assuming the truth of 
all the preceding propositions (which are proved in the subsequent sections),
we prove an upper semicontinuity result of marked BC entropy, which 
supersedes the result of Lemma~\ref{lem:BC-ent-upper-semicontinouous}.

% bb
\begin{thm}
  \label{thm:bch-semicont}
  Let an average degree vector $\vd = (d_{x,x'} : x,x' \in \edgemark)$ and a
  probability distribution $Q = (q_\theta: \theta \in \vermark)$ be given. For
  any  $\mu \in \mP(\mGb_*)$ with 
  $0 < \deg(\mu) < \infty$, the BC entropy
     $\bch_{\vd, Q}(.)$ is upper semicontinuous at $\mu$, i.e.\ if
    $\mu_k$ is a sequence in $\mP(\mGb_*)$ converging weakly to $\mu \in
    \mP(\mGb_*)$ such that $0 < \deg(\mu_k) < \infty$ for all $k$, then we have $\bch_{\vd, Q}(\mu) \geq \limsup_{k \rightarrow \infty} \bch_{\vd, Q}(\mu_k)$.
  
  \end{thm}

\begin{proof}
%\color{red}
Let
  $\vmn, \vun$ be sequences such that $(\vmn,\vun)$ is adapted to $(\vd, Q)$.
  Then, as established in
  part~\ref{thm:BC-invariant} of
Theorem~\ref{thm:bch-properties},
$\bch_{\vd, Q}(\mu)$ equals
  $\bchunder_{\vd, Q}(\mu)\condmnun$
  and $\bch_{\vd, Q}(\mu_k)$ 
  equals $\bchunder_{\vd,Q}(\mu_k)\condmnun$ for all $k$.
%This theorem needs to be proved.
%The proof is embedded in the preceding text. Move the relevant parts here.
The claim is therefore an immediate consequence of Lemma~\ref{lem:BC-ent-upper-semicontinouous}.
%\color{black}
\end{proof}

%%% Local Variables: 
%%% mode: latex
%%% TeX-master: "Note-41_BC-ent-arxiv.tex"
%%% End: 

%aa
\section{Colored Configuration Model}
\label{sec:color-conf-model}

% bb

In this section, we review and generalize results from \cite[Section~4]{bordenave2015large}. 
First, in Section~\ref{sec:directed-colored-multipgraph}, we review the notion
of \emph{directed colored multigraphs} from \cite[Section~4.1]{bordenave2015large}. 
Then, in Section~\ref{sec:color-conf-model-1}, we review the \emph{colored
  configuration model} from \cite[Section~4.2]{bordenave2015large}.
In Sections~\ref{sec:color-unim-galt} we review the notion of \emph{colored unimodular
Galton--Watson trees} and a local weak convergence result related to such trees, from
\cite[Sections~4.4, 4.5]{bordenave2015large}. In
Sections~\ref{sec:dct-to-mark-and-mark-to-dcg} and \ref{sec:graphs-with-given}, we draw a connection  between directed colored
multigraphs
%\color{blue}
%(multigraphs?)
%\color{black}
and marked graphs,
%by treating the type of an edge (defined in
%\eqref{eq:depth-h-type}) as its color. This is done by 
generalizing the results
in \cite[Sections~4.6]{bordenave2015large}. 
%In Section~\ref{sec:marked-to-colored-to-CM} 
We also discuss the colored configuration 
model arising from the colored degree sequences associated to the directed colored graphs arising from a marked graph.
This discussion is used in Section~\ref{sec:weak-convg-to-admissible}
to prove a weak convergence result for any admissible probability 
distribution $P \in \mP(\mTb_*^h)$ with finite support, for any $h \ge 1$.
Finally, in Section~\ref{sec:connection-color-ugwt}, we use the tools developed
in this section to prove a local weak convergence result for marked graphs
obtained from a colored configuration model, which will be useful in our
analysis in Section~\ref{sec:ent-properties}.
Note that the terms \emph{color} (defined in this section) and \emph{mark}
(defined in Section \ref{sec:marked-graphs}) refer to two different concepts and should not be confused with each other. 

%aa
\subsection{Directed Colored Multigraphs}
\label{sec:directed-colored-multipgraph}

%bb

Let $L \ge 1$ be a fixed integer, and define $\mC := \{ (i,j) : 1 \leq i , j \leq L \}$. \glsadd{not:colorset}
Each element $(i,j) \in \mC$ is interpreted as a \emph{color}.
Note that the terms \emph{color} and \emph{mark} refer to different
concepts and should not be confused with each other. 
Let $\mC_= := \{(i,i) \in \mC\}$, $\mC_< := \{(i,j) \in \mC: i<j\}$ and
$\mC_{\neq} := \{ (i,j) \in \mC: i \neq j\}$. We define $\mC_{\leq}$, $\mC_>$,
and $\mC_{\geq}$
similarly. For $c := (i,j) \in \mC$, we use the notation $\bar{c} := (j,i)$.
%A {\em directed multigraph} is a multigraph, as 
%defined in Section~\ref{sec:multigraph-lwc}, where in addition each edge
%is given a direction from one of its vertices to the 
%other vertex (this is also the case for self loops).

We now define a set $\hmG(\mC)$ of {\em directed colored multigraphs} with colors in $\mC$, comprised of multigraphs (as defined in Section~\ref{sec:multigraph-lwc}) where the edges are colored with
elements in $\mC$ in a \emph{directionally consistent way}.
%i.e. this set is determined by $L \ge 1$.
%\color{blue}
%Why are you calling them ``directed multigraphs"?
%Earlier the concept of a ``multigraph" was discussed at length. Here also, you appear to be imposing the condition 1 below on the union (over colors) of the underlying structure for each color, which makes the overall underlying structure (i.e. what you call the associated colorblind multigraph) a ``multigraph'', in the sense defined earlier. It is true that there appears to some kind of connection between the substructure of the multigraph associated with color $c$ and that associated with $\bar{c}$, which suggests some kind of directionality, 
%but a different terminology from ``directed multigraph'' would be preferable, since the term suggests something else and does not bring to mind the underlying lack of directionality.
%\color{black}
\glsadd{not:mhGmC}
More precisely, each $G \in \hmG(\mC)$ is of the
form  $G = (V, \omega)$ where $V$ is a finite or a countable vertex set, and $\omega = (\omega_c: c \in \mC)$ where for each $c \in \mC$, $\omega_c: V^2 \rightarrow \integers_+$ with the following properties:
\begin{enumerate}
\item For $c \in \mC_=$, $\omega_c(v,v)$ is even for all $v \in V$, and $\omega_c(u,v) = \omega_c(v,u)$ for all $u, v \in V$.
\item For $c \in \mC_{\neq}$, we have $\omega_c(u,v) = \omega_{\bar{c}}(v,u)$ for all $u,v \in V$.
\item For all $u \in V$ and $c \in \mC$, $\sum_{v \in V} \omega_c(u,v) < \infty$.
\end{enumerate}
%The interpretation is as follows:
%  \begin{enumerate}
%  \item for $v \neq u$, $\omega_c(v,u)$ is the number of directed edges with color $c$ going out of $v$ towards $u$.
%  \item For $c \in \mC_=$, $\frac{1}{2} \omega_c(v,v)$ is the number of loops
%    with color $c$ at vertex $v$.
%  \item For $v \in V$ and $c \in \mC_<$, $\omega_c(v,v) = \omega_{\bar{c}}(v,v)$
%    is the number of loops with color $c$ at vertex $v$.
    %By convention, we
    %assume that there is no loop at a vertex with colors in $\mC_>$.
    % I removed the sentence starting with by convention, since I found it
    % confusing. In fact, for such a loop, following the rule in $\omega$, we
    % also have a loop with color $\bar{c}$. Also, I found this sentence
    % possibly confusing when later on we discuss configuration model for \mC_<
    % and specially when the two vertices associated with half edges being the
    % same vertex. 
    % this is motivated by the configuration model, because half edges are matched to form edges, so a loop of color in \mC_< is formed by a c half edge and a \bar{c} half edge, forming one single loop not two but both appear in \omega_c and \omega_{\bar{c}}. The same reasoning motivates the definition for \mC_=
 % \end{enumerate}
See Figure~3 in \cite{bordenave2015large} for an example of an element of 
$\hmG(\mC)$.

  For a directed colored multigraph
  $G = (V, \omega) \in \hmG(\mC)$ 
  %on a finite or countable vertex set,  
  the associated {\em colorblind multigraph} 
 % associated to $G$, which is denoted by
 % $\CB(G)$, is defined to be the 
  is the multigraph $\CB(G) := (V, \bar{\omega})$ on the same vertex set $V$, where $\bar{\omega}: V^2 \rightarrow \integers_+$ is defined via
  \begin{equation*}
    \bar{\omega}(u,v) := \sum_{c \in \mC} \omega_c(u,v).
  \end{equation*}
\glsadd{not:colorblind}
It can be checked that $\CB(G)$ is a multigraph, as defined in Section~\ref{sec:multigraph-lwc}.
Distinct directed colored multigraphs can give rise to the same multigraph as their associated colorblind multigraph, and we can think of each of them as arising from this multigraph by coloring it in a directionally consistent way as expressed in properties 1 and 2.
%Note that, since for all $u,v
%\in V$ with $u \neq v$ we have $\bar{\omega}(u,v) = \bar{\omega}(v,u)$,  and since $\bar{\omega}(v,v)$ is even for all $v \in V$, we can treat $\bar{G}$ 
%as a multigraph, as defined in Section~\ref{sec:multigraph-lwc}. 
%Motivated by this, we may treat $\CB(G)$ as an undirected multigraph.

Given $G \in \hmG(\mC)$,
if $\CB(G)$ has no multiple edges and
no self--loops, i.e.
%\ for $u \neq v \in V$, we have $\bar{\omega}(u,v) \in \{0,1\}$,
%and we have $\bar{\omega}(v,v) = 0$ for all $v \in V$,
%then 
it is a graph,
%In this case, 
then we call 
$G$ a {\em directed colored graph}.
%(where we have dropped the term ``simple'' since all graphs we consider are assumed to be simple).
%With this, 
We let $\mG(\mC)$ denote the subset of 
$\hmG(\mC)$ comprised of 
directed colored graphs.
 \glsadd{not:mGmC}

We introduce the notation $\mM_L$ for the set of $L \times L $ matrices with nonnegative integer valued entries.  \glsadd{not:mML}

  Let $G = (V, \omega) \in \hmG(\mC)$, where $V$ is a finite set. For $u \in V$ and $c \in \mC$, define
  \begin{equation*}
    D^G_c(u) := \sum_{v \in V} \omega_c(u,v).
  \end{equation*}
  $D^G_c(u)$ is the number of color $c$ edges going out of the vertex $u$. %Note that, by definition, this quantity is finite.
  Let $D^G(v) := (D^G_c(v): c \in \mC)$. 
  Note that $D^G(v) \in \mM_L$. %can be identified with an 
  $D^G(v)$
  is called the \emph{colored degree matrix} of the vertex $v$. 
 %When $G \in \hmG(\mC)$, 
 Let $\vD^G := (D^G(v): v \in
  V)$.
  We call $\vD^G$ the \emph{colored degree sequence} corresponding to $G$.
  \glsadd{not:colored-degree}%
  \glsadd{not:DGc}%
  \glsadd{not:DG}%
  %\color{red}
  %%Note to Payam:
  %$D^G_c(u)$ and $D^G(v)$ should also be included in the glossary (in addition to 
  %$\vD^G$, which is already there).
  %\color{black}
  %\pres{added to glossary}

  % I think this is what Bordenave and Caputo meant for $\mG(\mC)$, because they say the set of "graphs", but they say it in the following form:
  % , where
  % \begin{itemize}
  % \item For all $c \in \mC$ and $u,v \in V$, $\omega_c(u,v) \in \{ 0,1 \}$.
  % \item For all $c \in \mC$ and $u \in V$, $\omega_c(u,u) = 0$.
  % \end{itemize}
  % Essentially, $\mG(\mC)$ consists of objects in $\hmG(\mC)$ where no loops and no  multiple edges of the same color between vertices
  % are  present.
% Further evidence to support this: 
% on their page 23 after Definition 4.1, they say that $\mG(\mathbb{D})$
% coincides with $\mG(\mathbb{D}, 2)$. 

%aa
\subsection{Colored Configuration Model}
%for given Colored Degree Sequence}
\label{sec:color-conf-model-1}

% bb

Fix an integer $L \ge 1$, and 
let $\mC := \{ (i,j) : 1 \leq i , j \leq L \}$
be the associated set of colors.
%$\mC$, as defined above. 
\glsadd{not:mDn}For $n \in \nats$, let $\mD_n$ be the set of vectors $(D(1),
\dots, D(n))$ where, for each $1 \le i \le n$, we have $D(i) = (D_c(i): c \in \mC) \in \mM_L$ and, further, $S :=
\sum_{i=1}^n D(i)$ is a symmetric matrix with even coefficients on the diagonal.
%i.e. $S = (S_c: c \in \mC)$ where $S_c = S_{\bar{c}}$ for all $c \in \mC$ and
%$S_c$ is even for all $c \in \mC_=$.
Note that  for  $G \in \hmG(\mC)$ we have  $\vD^G \in \mD_n$.
 \glsadd{not:mhGmD}
  Given $\vD = (D(1), \dots, D(n)) \in \mD_n$, define $\hmG(\vD)$ to be  the set
  of directed colored multigraphs $G \in \hmG(\mC)$ with the vertex set $V =
  [n]$ such that, for all $i \in [n]$, we have $D^G(i) = D(i)$.
  \glsadd{not:mGmDh}
  Further, given $n \in \nats$, $\vD \in \mD_n$, and $h \ge 1$, let $\mG(\vD, h)$ be
  the set of directed colored multigraphs 
  %$G \in \hmG(\mC)$ 
  $G \in \hmG(\vD)$
  such that $\CB(G)$
  has no cycles of length $l \leq h$.
  %\color{red}
  %Note to Payam: Replaced $G \in \hmG(\mC)$ by
  %$G \in \hmG(\vD)$ in the preceding sentence. Check.
  %\color{black}
  %\pres{correct}
  %\begin{marginpar}
    %{
  %\color{blue}
  %The notation $\mG(\vD, h)$ is not appropriate
  %for $h=1$. It is unclear to me at this point whether
  %you need to consider $h=1$. If not, it is best to stick with making this definition only for $h \ge 2$.
  %\color{pedit}
  %\textbf{Response:} A cycle of length $1$ is basically a self loop. Hence,
  %$\mG(\vD, 1)$ consists of $G$ such that $\CB(G)$ has no self loops. I think
  %one can restrict it to $h \geq 2$, but maybe it is better to remain
  %consistent with the notation in \cite{bordenave2014large}.
  %  }
  %\end{marginpar}
  %Recall that we treat $\CB(G)$ as an
  %undirected multigraph.
  Note that $\mG(\vD, h+1) \subseteq \mG(\vD, h)$
  for all $h \ge 1$, and that
  $\mG(\vD, 2) \subset \mG(\mC)$. 

 Now, given $\vD = (D(1), \dots, D(n)) \in \mD_n$, we give a recipe to generate
 a random directed colored multigraph $G \in \hmG(\mC)$ such that $\vD^G = \vD$, i.e. a random directed colored multigraph in $\hmG(\vD)$.
The procedure is similar to that in the classical configuration model.
%in the absence of colors.
%There, to each vertex we attach a number of ``half edges'' where the count of
%these objects are prescribed by a given degree sequence, and then these half
%edges are connected to each other randomly to construct a multigraph. In the
%presence of colors, this can be generalized so that colored  half edges are
%connected to vertices, and they get connected to each other following a certain
%rule. More precisely, 
For each $c \in \mC$, let $W_c := \cup_{i=1}^n W_c(i)$ be a
set of distinct 
{\em half edges} of color $c$ where $|W_c(i)| = D_c(i)$. 
%Here,
%$W_c(i)$ is the set of half edges of color $c$ attached to vertex $i$. 
We think of the half edges in $W_c(i)$ as attached to the vertex $i$.
%With this, 
We require a half edge with color $c$ to get connected to another half
edge with color $\bar{c}$. 
For this, for $c \in \mC_<$, let $\Sigma_c$ be the
set of bijections $\sigma_c : W_c \rightarrow W_{\bar{c}}$. 
Since $\vD
\in \mD_n$, $|W_c| = |W_{\bar{c}}|$ and such bijections exist. 
Likewise, for $c \in \mC_=$, let $\Sigma_c$ be the set of perfect matchings
on the set $W_c$. Since $\vD \in \mD_n$, $|W_c|$ is even and such
matchings exist.

Given a choice of $\sigma_c \in \Sigma_c$
for each $c \in \mC_\leq$,
we write $\sigma$ for $(\sigma_c: c \in \mC_\leq)$.
Let $\Sigma$ denote the product of $\Sigma_c$ for $ c\in \mC_{\leq}$. 
Given $\sigma \in \Sigma$, we construct a directed colored multigraph,
denoted $\Gamma(\sigma)$, as follows.
%and , let $\Gamma_c(\sigma_c)$ be a directed colored
%multigraph on the vertex set $[n]$ defined as follows: 
For $c \in \mC_<$, if $\sigma_c$ maps a half
edge of color $c$ at vertex $u$ to another half edge of color $\bar{c}$ at
vertex $v$,  then 
%in $\Gamma_c(\sigma_c)$, 
we place an edge 
directed from $u$ towards $v$ 
having color $c$ 
and an edge directed from 
 $v$ towards $u$, having
color $\bar{c}$.
Here it is allowed that $u = v$.
%Moreover, for $\sigma_c \in \Sigma_c$, let
%$\Gamma_c(\sigma_c)$ be a directed colored multigraph on the vertex set $[n]$
%defined as follows. 
For $c \in \mC_=$,
if $\sigma_c$ matches a half edge of color $c$ at vertex $u$
to another half edge of the same color at vertex $v$, then we place two directed edges,
one directed from $u$ towards $v$, and one directed from $v$ towards $u$, both
with color $c$. Here also it is allowed that $u = v$. 
  % to avoid confusion regarding loops, note that if we have $c \in \mC_=$, and $\omega_c(u,u) = 2$ for some vertex $u$, then there are two edges directed from $u$ to $u$ (we say that there is a single loop with color $c$ at vertex $u$). Then, $D_c(u) = 2$ and there are two half edges at vertex $u$. If we match these two half edges, by the above recipe, we put two directed edges from $u$ towards $u$ which is consistent with the original picture.
  %For
%$\sigma = (\sigma_c: c \in \mC_\leq)$, define $\Gamma(\sigma)$ to be the colored
%directed multigraph defined as the superposition of $(\Gamma_c(\sigma_c): c \in
%\mC_\leq)$. 

Note that, for $\sigma \in \Sigma$, the construction above gives $\Gamma(\sigma) \in \hmG(\vD)$.
For $\vD \in \mD_n$, let $\CM(\vD)$ be the law of $\Gamma(\sigma)$ where $\sigma$ is chosen uniformly at random in $\Sigma$.
\glsadd{not:CMD}

Theorem~\ref{thm:alpha-h} below is from \cite{bordenave2015large},
and states a key property of the configuration model defined above. To state that theorem, given a positive integer $\delta$,
let $\mM_L^{(\delta)}$ denote the set of $L \times L$ matrices with nonnegative
integer entries bounded by $\delta$.
\glsadd{not:mMLdelta}
Assume that $R \in \mP(\mM_L^{(\delta)})$ is given.
Let $\vDn = ( \Dn(1), \dots, \Dn(n)) \in \mD_n$ be a sequence satisfying the following two conditions:
%aa
\begin{subequations}
  \begin{gather}
    \Dn(i) \in \mM_L^{(\delta)} \qquad \forall i \in [n], \label{eq:H1}\\
    \frac{1}{n} \sum_{i=1}^n \delta_{\Dn(i)} \Rightarrow R. \label{eq:H2}
  \end{gather}
\end{subequations}
%bb
% using R instead of P or Q here, since P is usually used for distributions on
% \mG_*, and Q was used for the vertex mark distribution
Theorem~\ref{thm:alpha-h} states that, given the above conditions, for every $h \ge 1$, a positive fraction of
random directed colored multigraphs generated from the above configuration
model do not have and cycles of length $h$ or less. 

%aa
\begin{thm}[Theorem 4.5 in \cite{bordenave2015large}]
  \label{thm:alpha-h}
  Fix $\delta \in \nats$, $R \in \mP(\mM_L^{(\delta)})$, and a sequence $\vDn$
  satisfying~\eqref{eq:H1} and \eqref{eq:H2}. Let $G_n$ have distribution
  $\CM(\vDn)$ on $\hmG(\vDn)$. Then, for every $h \ge 1$, there exists $\alpha_h > 0$ such that
  \begin{equation*}
    \lim_{n \rightarrow \infty} \pr{G_n \in \mG(\vDn, h)} = \alpha_h.
  \end{equation*}
\end{thm}

% bb

To close this section, we give an asymptotic counting for the set $\mG(\vDn, h)$. This calculation is also from \cite{bordenave2015large}. For two
sequences $a_n$ and $b_n$ we write $a_n \sim b_n$ if $a_n / b_n \rightarrow 1$ as
$n \rightarrow \infty$. Moreover, for an even integer $N$, $(N-1)!!$ is defined
as $\frac{N!}{(N/2)! 2^{N/2}}$, or equivalently $(N-1)\times (N-3) \times \dots
3 \times 1$. Note that $(N-1)!!$ is the number of perfect matchings on a set of size $N$.

% aa
\begin{cor}[Corollary 4.6 in \cite{bordenave2015large}]
  \label{cor:GDnh-counting}
  In the setting of Theorem~\ref{thm:alpha-h}, write $\Sn_c := \sum_{i \in [n]} \Dn_c(i)$, which, we recall, form the entries of a symmetric matrix. For all $h \geq 2$ we have
  \begin{equation*}
    |\mG(\vDn, h)| \sim \alpha_h \frac{\prod_{c \in \mC<} \Sn_c! \prod_{c \in \mC_=} (\Sn_c - 1)!!}{\prod_{c \in \mC} \prod_{i = 1}^n \Dn_c(i)!}.
  \end{equation*}
\end{cor}

% bb

We give a brief sketch of how this counting statement results from Theorem~\ref{thm:alpha-h} and refer the reader to \cite{bordenave2015large}
for the proof. 
%Theorem~\ref{thm:alpha-h} implies that for each $h \in \nats$, with a nonzero
%asymptotic probability, the colored configuration model satisfying
%\eqref{eq:H1} and \eqref{eq:H2} generates a directed multigraph whose
%colorblind version has no cycles with length smaller than $h$. On the other
%hand, note that 
By construction, $|\Sigma|$, which is the total number of
configurations, is equal to $\prod_{c \in \mC_<} \Sn_c! \prod_{c \in \mC_=}
(\Sn_c-1)!!$. Each directed colored multigraph can be constructed via different
configurations. However, every $G \in \mG(\vD, h)$ for $h \geq 2$ is a colored
graph, i.e.\ is in $\mG(\mC)$. It is easy to see that, for such $G$, there are precisely
$\prod_{c \in \mC} \prod_{i=1}^n \Dn_c(i)!$ many configurations $\sigma \in
\Sigma$ for which $\Gamma(\sigma) = G$. Also, from Theorem~\ref{thm:alpha-h},
the asymptotic probability of $\Gamma(\sigma)$ being in $\mG(\vDn, h)$ is
$\alpha_h$. This provides a rough explanation of where  Corollary~\ref{cor:GDnh-counting} comes from.

%aa
\subsection{Colored Unimodular Galton--Watson trees}
\label{sec:color-unim-galt}
% bb

In this section we review the definition of colored unimodular Galton--Watson
trees from \cite[Section~4.4]{bordenave2015large}. This should not be confused with the
notion of marked unimodular Galton--Watson trees defined in
Section~\ref{sec:markovian-ugwt}. Later, in
Section~\ref{sec:connection-color-ugwt}, we explain the connection
between the two notions.
To reduce the chance of confusion, we employ the notation $\cugwt$ to denote the object
constructed here, which is slightly different from the notation used in
\cite{bordenave2015large}.

Given $L \in \nats$ and the set of colors
$\mC := \{ (i,j) : 1 \leq i , j \leq L \}$,
%a color set $\mC$ with $|\mC| = L^2$ as above, 
we first define a set of equivalence classes of %connected
rooted 
%colored directed locally finite multigraphs. 
directed colored multigraphs, denoted by
$\hmG_*(\mC)$.
\glsadd{not:hmG-star-mC}
%\color{red}
%Note to Payam:
%Need to put $\hmG_*(\mC)$ in the glossary.
%\color{black}
%\pres{added to glossary}
%\glsadd{not:mhGmCstar}
%More precisely, 
Each member 
of $\hmG_*(\mC)$ is of the form $[G, o]$ where $G \in \hmG(\mC)$ is
connected and $o$ is a distinguished vertex in $G$.
%Here, $[G,o]$ 
$[G, o]$ denotes the
equivalence class corresponding to $(G, o)$ where the equivalence relation is
defined through relabeling of the vertices,
while preserving the root and the edge structure together
with the directed colors.
As is discussed in \cite{bordenave2015large},
the framework of local weak convergence introduced in
Section~\ref{sec:multigraph-lwc} for multigraphs can be naturally extended to directed
colored multigraphs.
%\color{red}
%Note to Payam:
%Need to flesh this out by discussing why 
%$\hmG(\mC)$ is Polish, or at least pointing out where this is proved.
%\color{black}
%\pres{This is from \cite{bordenave2014large}, so added the reference.}
An element  $[G, o] \in
\hmG_*(\mC)$ is called a rooted directed colored tree if its associated colorblind multigraph $\CB(G)$ has no
cycles. 

Recall that $\mM_L$ denotes the set of $L \times L$ matrices with
nonnegative integer valued entries.
Let $P \in \mP(\mM_L)$ be a probability distribution such that for all $c \in \mC$, we have
%\begin{equation*}
  $\ev{D_c} = \ev{D_{\bar{c}}}$,
%\end{equation*}
where $D \in \mM_L$ has law $P$.
For $c \in \mC$ such that $\ev{D_c} > 0$, define $\hP^c \in \mP(\mM_L)$ as follows:
\begin{equation}
  \label{eq:cugwt-phat}
  \hP^c(M) := \frac{(M_{\bar{c}} + 1) P(M + E^{\bar{c}})}{\ev{D_c}},
\end{equation}
where 
%$D$ is distributed according to $P$ and  
$E^{\bar{c}} \in \mM_L$ denotes the
matrix with the entry at coordinate $\bar{c}$ being 1 and all the other
entries being zero. If $\ev{D_c} = 0$, we set $\hP^c(M) = 1$ if $M = 0$ and zero otherwise. It is straightforward to check that 
$\sum_{M \in \mM_L} \hP^c(M) =1$ for all $c \in \mC$.

With this setup, we define the \emph{colored unimodular Galton--Watson tree} $\cugwt(P) \in \mP(\hmG_*(\mC))$ to be the law of
$[T,o]$ where $(T,o)$ is a rooted directed colored multigraph defined as
follows. We start from the root $o$ and generate $D(o)$ with law $P$. Then,
for each $c \in \mC$, we attach $D_c(o)$ many vertex offspring of type $c$ to
the root. For each offspring $v$ of type $c$, we add a directed edge from $o$
to $v$ with color $c$ and another directed edge from $v$ to $o$ with color
$\bar{c}$. Subsequently, for an offspring $v$ of type $c$, we generate $D(v)$
with law $\hP^c$, independent from all other offspring. Then, we continue the
process. Namely, for each $c \in \mC$, we add $D_c(v)$ many vertex offspring of type
$c$ to $v$ where, for each offspring $w$ of type $c$, there is an edge directed
from $v$ towards $w$ with color $c$ and another edge directed from $w$ towards
$v$ with color $\bar{c}$. This process is continued inductively to define
$\cugwt(P)$. \glsadd{not:cugwt}

%The local topology of $\mG_*$ could be naturally extended to a local topology on
%$\hmG_*(\mC)$. This leads to a framework of local weak convergence for colored
%directed multigraphs. 
Let $\mG_*(\mC)$ denote the subset of $\hmG_*(\mC)$ consisting of
equivalence classes of
rooted directed colored graphs,
%which are simple, 
i.e.\ 
%$[G,o] \in \hmG_*(\mC)$ where
%$G$ is simple. Recall that this means 
for which the associated colorblind multigraph
$\CB(G)$ 
%has no multiple edges or loops.
is a graph,
see the end of Section \ref{sec:directed-colored-multipgraph}.
\glsadd{not:mG-star-mC}%
%\color{red}
%Note to Payam:
%Need to put $\mG_*(\mC)$ in the glossary.
%\color{black}
%\pres{added to glossary}
Note that 
%by definition, 
$\cugwt(P)$
is supported on $\mG_*(\mC)$.
The following result from \cite{bordenave2015large} will be useful for our future analysis.

%aa
\begin{thm}[Theorem 4.8 in \cite{bordenave2015large}]
    \label{thm:BC-CUGWT-LWC}
    Let $R \in \mP(\mM_L^{(\delta)})$ be given.
Let $\vDn \in \mD_n$ be a sequence 
     satisfying \eqref{eq:H1} and \eqref{eq:H2}.
     %\color{red} Note to Payam: added the condition that $\vDn \in \mD_n$, which
     %was not there in the earlier version. Check. \color{black}
     %\pres{correct}
     Moreover, assume that $G_n \in \hmG(\vDn)$ has law $\CM(\vDn)$, and that $G_n$
    are jointly defined to be independent on a single probability space. Then, with probability one, $U(G_n) \Rightarrow \cugwt(R)$. Also, the same result holds when $G_n$ is uniformly sampled from $\mG(\vDn, h)$, for any $h \geq 2$.
  \end{thm}
%bb

%\color{red} 
%Note to Payam:
%Replaced $G_n \in \hmG_n(\vDn)$ with $G_n \in \hmG(\vDn)$
%in the statement of the theorem. Check. \pres{correct}\\ 
%Note to Payam: What role does the $\delta$ appearing in \eqref{eq:H1} and
%\eqref{eq:H2} play in this result?
%\pres{this ensures that the entries in the matrices are bounded by $\delta$, and
%should be important for technical reasons.}
%Note to Payam: Also,
%I changed what was written as $\cugwt(P)$ for the weak 
%limit to $\cugwt(R)$. Check. \color{black}
%\pres{correct}

%aa
%\subsection{Graphs with given neighborhood structure}
\subsection{From a Marked Graph to a Directed Colored Graph and Back}
\label{sec:dct-to-mark-and-mark-to-dcg}
% bb

%In the previous sections, we studied colored multigraphs, where the notion of
%color was defined based on pairs of integers in a color set $\mC$ of size $L
%\times L$. Here, we try to connect this to marked graphs. In order to do so, we
%treat the types of edges, as defined in \eqref{eq:depth-h-type}, as colors. 

In this section we first associate, for any fixed $h \ge 1$, a specific
directed colored graph to a given marked graph, by treating 
the types of edges, as defined in \eqref{eq:depth-h-type}, as colors. 
We also discuss a procedure that, starting with 
a directed colored graph whose colors can be interpreted in terms of the types for a given $h \ge 1$, returns a marked graph.

%\glsadd{not:psihG}
%For a marked graph $G$ on the vertex set $[n]$ and  $h \ge 1$, we define
%\begin{equation*}
%  \psi_h(G) := ([G,1]_h, [G,2]_h, \dots, [G,n]_h),
%\end{equation*}
%to be the sequence of equivalence classes of depth $h$ neighborhoods of the vertices in $G$. 
%Moreover, we define the colored
%  version of $G$ denoted by $\colored(G)$ as follows.

For a marked graph $G$ on the vertex set $[n]$ and  an integer $h \ge 1$,
we define a directed colored graph 
denoted by $\colored(G)$.
 Let $\mF \subset \edgemark \times \mGb_*^{h-1}$ be the set of
 all distinct $G[u,v]_{h-1}$ for adjacent vertices $u$ and $v$ in
 $G$. Since $G$ is finite, $\mF$ is a finite subset of $\edgemark \times
 \mGb_*^{h-1}$. Therefore, with $L := |\mF|$, we can enumerate the elements in $\mF$ in some order,
 with integers $1, \dots, L$. 
 Recall from
 \eqref{eq:depth-h-type} that $\etype_G^h(u,v) = (G[v,u]_{h-1},
 G[u,v]_{h-1})$ is the depth $h$ type of the edge $(u,v)$. 
 Now, we define $\colored(G)$ to be a directed
 colored graph with colors in  $\mC = \mF \times \mF$ on the vertex set $[n]$ as
 follows.
 For two adjacent vertices $u$ and $v$ in $G$, in $\colored(G)$ we put an edge
 directed from $u$ towards $v$ with color $\etype_G^h(u,v)$ and another directed
 edge from $v$ towards $u$ with color $\etype_G^h(v,u)$.
 Since $G$
 is simple, 
 $\colored(G)$ is a directed colored graph, i.e. $\colored(G) \in \mG(\mC)$
\glsadd{not:coloredG}.
In fact, $\CB(\colored(G))$ is just the graph which results from $G$
by erasing its marks.

%{\bf Check with Payam.}

%So far, we discussed how to construct a directed colored graph from a simple
%marked graph. Now, we discuss the other direction. 
We can also go in the other direction.
Fix $h \ge 1$.
Let $\mF \subset \edgemark \times \mGb_*^{h-1}$ be a finite set with cardinality $L$,
%Therefore, elements in $\mF$ could be identified with integers $1, \dots, L$.
whose elements are 
identified with the integers $1, \dots, L$.
Let $\mC := \mF \times \mF$.
%let $H$ be a simple directed colored graph in
%$\mG(\mC)$ defined on a finite or countable vertex set $V$ 
Given a directed colored graph $H \in\mG(\mC)$,
defined on a finite or countable vertex set $V$,
%(recall that $H$ being simple means
%that $\CB(H)$ has no loops or multiple edges). 
and a sequence $\vbeta = (\beta(v): v \in V)$ with elements in $\vermark$, we define 
a marked graph on $V$,
called the \emph{marked color blind} version of $(\vbeta,H)$, denoted by $\MCB_{\vbeta}(H)$, as follows.
%to be
%a simple marked graph defined on the same vertex set $V$ constructed as follows.
For any pair of adjacent vertices $u$ and $v$ in $H$ where the color of the edge
directed from $u$ to $v$ is $(g,g')$ (and hence the color of the edge directed
from $v$ to $u$ is $(g',g)$), we put an edge in $\MCB_{\vbeta}(H)$ between $u$
and $v$ with the mark towards $u$ and $v$ being $g[m]$ and $g'[m]$, respectively.
%Recall that $g[m]$ and $g'[m]$ are the marked components of $g$ and $g'$, respectively. 
Moreover, the mark of a vertex $v \in V$ in $\MCB_{\vbeta}(H)$ is
defined to be $\beta(v)$. \glsadd{not:MCB}

Note that
%, although the colors in $H$ are of the form $(g,g')$ where $g,g' \in \edgemark \times \mGb_*^{h-1}$, 
it is not necessarily the case that 
%they 
the colors of $H$
are consistent with 
%the edge types of 
those in the directed colored graph $\colored(\MCB_{\vbeta}(H))$.
%$G:= \MCB_{\vbeta}(H)$. 
Namely,
$\etype_{\MCB_{\vbeta}}^h(u,v)$ for adjacent vertices $u, v$ can be different from the
color of the edge between $u$ and $v$ in $H$. See Figure~\ref{fig:triangle_MCB} for an example.  Proposition~\ref{prop:MCB-consistent-general} below
gives conditions under which this consistency holds. 
To be able to state this result, we
first need some definitions and tools, which are gathered in the next section.

%aa
%\subsection{Graphs with given neighborhood structure}
\subsection{Consistency in going from a directed colored graph to a marked graph and back}
\label{sec:graphs-with-given}
% bb

In this section we first give conditions under which the edge colors of a
directed colored graph are related to the edge colors 
of the directed colored graph derived from
its marked colorblind
version. This is done in Proposition~\ref{prop:MCB-consistent-general}.
Next, building on this result, we
study the configuration model given by the colored
degree sequence of the directed colored graph
associated to a given marked graph, and relate the marked color blind versions of the directed colored graphs arising as realizations from this configuration model to the original marked graph we started with.

\begin{figure}
  \centering
    \begin{tikzpicture}
    \begin{scope}[xshift=-5cm]
      \node[nodeB] (n0) at (0,0)  {};
      \node[nodeB] (n1) at (0,-1) {};
      \node[nodeB] (n2) at (0,-2) {};

      \drawedge{n0}{n1}{B}{B}
      \drawedge{n1}{n2}{B}{B}
      \draw[edgeO] (n0) -- (0,0.6);
      \node at (0,1) {$g$};
      \node at (0,-3) {$(i)$};
    \end{scope}
    \begin{scope}[xshift=0cm]
      \node at (0,1) {$H$};
      \node at (0,-3) {$(ii)$};
      \begin{scope}[yshift=-1cm]
      \node[fill,circle, inner sep=2pt] (n1) at (90:1.5) {};
      \node[fill,circle, inner sep=2pt] (n2) at (210:1.5) {};
      \node[fill,circle, inner sep=2pt] (n3) at (-30:1.5) {};

      \path[->,>=stealth']  (n1) edge[bend right=20] node[above,sloped] {$(g,g)$} (n2);
      \path[->,>=stealth']  (n2) edge[bend right=20] node[above,sloped] {$(g,g)$} (n1);

      \path[->,>=stealth']  (n1) edge[bend right=20] node[above,sloped] {$(g,g)$} (n3);
      \path[->,>=stealth']  (n3) edge[bend right=20] node[above,sloped] {$(g,g)$} (n1);

      \path[->,>=stealth']  (n2) edge[bend right=20] node[below,sloped] {$(g,g)$} (n3);
      \path[->,>=stealth']  (n3) edge[bend right=20] node[below,sloped] {$(g,g)$} (n2);

      \end{scope}
    \end{scope}

    \begin{scope}[xshift=5cm]
      \node at (0,1) {$\mathsf{MCB}_{\vec{\beta}}(H)$};
      \node at (0,-3) {$(iii)$};
      \begin{scope}[yshift=-1cm]
      \node[nodeB] (n1) at (90:1.2) {};
      \node[nodeB] (n2) at (210:1.2) {};
      \node[nodeB] (n3) at (-30:1.2) {};

      \draw[edgeO] (n1) -- (n2) -- (n3) -- (n1);

      \end{scope}
    \end{scope}
  \end{tikzpicture}
  \caption[MCB counterexample]{$(i)$: $g \in \edgemark \times \mGb_*^2$ where $\vermark = \{\tikz{\node[nodeB,scale=0.5] at (0,0) {};}, \tikz{\node[nodeR,scale=0.5] at (0,0)
  {};}\}$ and $\edgemark =
  \{\text{\color{blueedgecolor} Blue (solid)}, \text{\color{orangeedgecolor}
    Orange (wavy)} \}$.  a simple
  directed colored graph $H \in \mG(\mC)$ where $\mC = \{(g,g)\}$, $(iii)$:
  $G = \MCB_{\vbeta}(H)$ where $\vbeta = \{\tikz{\node[nodeB,scale=0.5] at (0,0)
    {};}, \tikz{\node[nodeB,scale=0.5] at (0,0) {};}, \tikz{\node[nodeB,scale=0.5] at (0,0) {};}\}$. Note
  that none of $\etype_G^2(1,2)$, $\etype_G^2(1,3)$ and $\etype_G^2(2,3)$ is
  equal to $(g,g)$. }
  \label{fig:triangle_MCB}
\end{figure}

%aa
\begin{definition}
  \label{def:theta-D-graphical}
  Fix $h \in \nats$ and assume $\mF \subset \edgemark \times \mTb_*^{h-1}$ is a
  finite set with cardinality  $L$. Define $\mC := \mF \times \mF$. Given a
  matrix $D = (D_{t,t'}: t,t' \in \mF) \in \mM_L$ and $\theta \in \vermark$, we
  say that the pair $(\theta, D)$ is ``graphical'' if there exists %a rooted marked tree 
  $[T, o]\in \mTb_*^h$ such that $\tau_T(o) = \theta$ and, for all
  $t, t' \in \mF$, we have $E_h(t, t')(T, o) = D_{t,t'}$. Moreover, for
  $\tilde{t}, \tilde{t}' \in \edgemark \times \mTb_*^{h-1}$ such that either
  $\tilde{t} \notin \mF$ or $\tilde{t}' \notin \mF$, we require $E_h(\tilde{t},
  \tilde{t}')(T, o)$ to be zero. \glsadd{trm:graphical}
\end{definition}

%bb
%\color{red}
%Note for Payam:
%Deleted ``finite rooted marked tree". Check.
%\color{black}
%\pres{correct}

%It is easy to see that the marked rooted tree isomorphism class
From Lemma~\ref{lem:r-r'-same-degree_are-the-same}
in Appendix~\ref{sec:marked-rooted-trees-some-props},
$[T, o]$  in the above definition,
if it exists, is unique. 
%In fact, this is a result of the following general fact:

Fix an integer $h \ge 1$. For $t \in \edgemark \times \mTb_*^{h-1}$, $x \in \edgemark$,
and $\theta \in \vermark$, define $(\theta, x) \otimes t$ to be the  element in
$\mTb_*^h$ where the root $o$ has mark $\theta$,  and attached to it is one offspring $v$, where the
subtree of $v$ is isomorphic to $t[s]$ and the edge connecting $o$ to $v$ has
mark $x$ towards $o$ and $t[m]$ towards $v$. See Figure~\ref{fig:otimes} for an
example. \glsadd{not:otimes}
For $s \in \mTb_*$
and $x \in \edgemark$, let $x \times s$ be  $t \in \edgemark
\times \mTb_*$ where $t[m] = x$ and $t[s] = s$. \glsadd{not:times}
For two rooted trees $s, s'
\in \mTb_*$ which have the same vertex mark at the root, define $s \odot s'$ to be the
element in $\mTb_*$ obtained by joining $s$ and $s'$ at a common root, see
Figure~\ref{fig:odot} for an example. \glsadd{not:odot}
Note that $\odot$ is commutative and associative. Therefore, we may write $\bigodot_{i=1}^k s_k$ for a collection $s_i, 1 \leq i \leq k$, of elements in $\mTb_*$, which all have the same mark at the root.

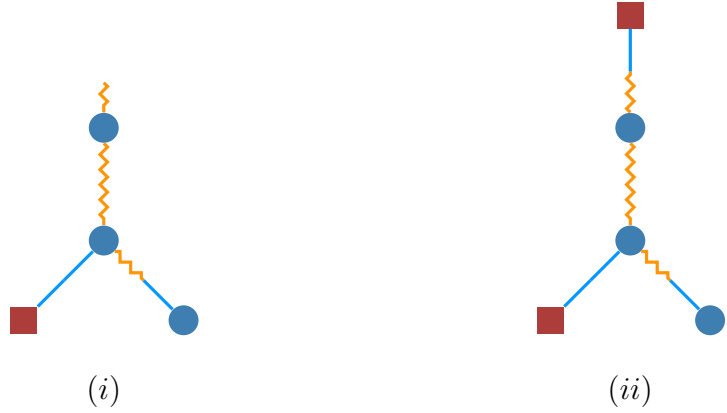
\begin{figure}
  \centering
    \begin{tikzpicture}
      \begin{scope}[xshift=-3.5cm]
        \node[nodeB] (n1) at (0,0) {};
        \node[nodeB] (n2) at (0,-1.5) {};
        \node[nodeR] (n3) at ($(0,-1.5) + (-135:1.5)$) {};
        \node[nodeB] (n4) at ($(0,-1.5) + (-45:1.5)$) {};
        
        % \drawedge{n1}{n2}{O}{O}
        \draw[edgeO] (n1) -- (n2);
        \drawedge{n2}{n3}{B}{B}
        \drawedge{n2}{n4}{O}{B}
        \draw[edgeO] ($(n1)+(0,0.6)$) -- (n1);
        \node at (0,-3.5) {$(i)$};
      \end{scope}
      
      \begin{scope}[xshift=3.5cm]
        \node[nodeR] (n0) at (0,1.5) {};
        \node[nodeB] (n1) at (0,0) {};
        \node[nodeB] (n2) at (0,-1.5) {};
        \node[nodeR] (n3) at ($(0,-1.5) + (-135:1.5)$) {};
        \node[nodeB] (n4) at ($(0,-1.5) + (-45:1.5)$) {};
        
        % \drawedge{n1}{n2}{O}{O}
        \draw[edgeO] (n1) -- (n2);
        \drawedge{n2}{n3}{B}{B}
        \drawedge{n2}{n4}{O}{B}
        % \draw[edgeO] ($(n1)+(0,0.6)$) -- (n1);
        \drawedge{n1}{n0}{O}{B}
        \node at (0,-3.5) {$(ii)$};
      \end{scope}

    \end{tikzpicture}
  \caption[$(\theta,x) \otimes t$ for $t \in \mTb_*^2$]{$(ii)$ depicts $(\theta, x) \otimes t$ for $t \in \edgemark \times
    \mTb_*^{2}$ as shown in $(i)$, $\theta = \tikz{\node[nodeR,scale=0.5] at (0,0)
  {};}$ and $x = \text{\color{blueedgecolor} Blue (solid)}$. We have used
    our convention of Figure~\ref{fig:Guv} for showing $t$, i.e.\ the half edge
    towards the root is the mark component.}
  \label{fig:otimes}
\end{figure}

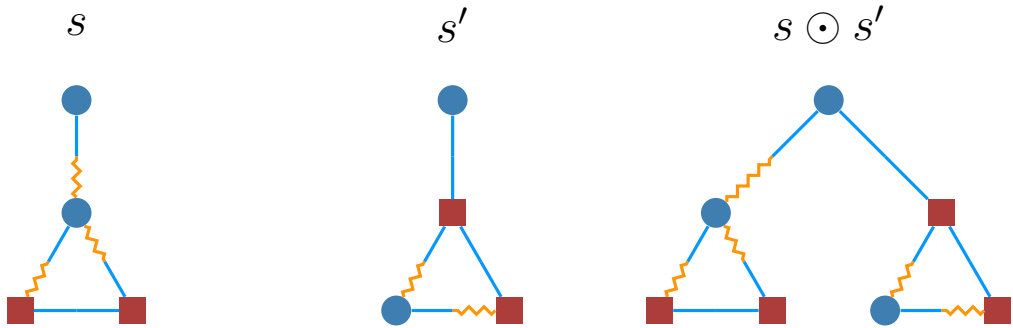
\begin{figure}
  \centering
  \begin{tikzpicture}
  \begin{scope}[xshift=-5cm]
    \node[nodeB] (n1) at (0,0) {};
    \node[nodeB] (n2) at (0,-1.5) {};
    \node[nodeR] (n3) at ($(0,-1.5) + (-120:1.5)$) {};
    \node[nodeR] (n4) at ($(0,-1.5) + (-60:1.5)$) {};

    \drawedge{n1}{n2}{B}{O}
    \drawedge{n2}{n3}{B}{O}
    \drawedge{n2}{n4}{O}{B}
    \drawedge{n3}{n4}{B}{B}
    \node[scale=1.5] at (0,1) {$s$};

  \end{scope}

  \begin{scope}[xshift=0cm]
    \node[nodeB] (n1) at (0,0) {};
    \node[nodeR] (n2) at (0,-1.5) {};
    \node[nodeB] (n3) at ($(0,-1.5) + (-120:1.5)$) {};
    \node[nodeR] (n4) at ($(0,-1.5) + (-60:1.5)$) {};

    \drawedge{n1}{n2}{B}{B}
    \drawedge{n2}{n3}{B}{O}
    \drawedge{n2}{n4}{B}{B}
    \drawedge{n3}{n4}{B}{O}

    \node[scale=1.5] at (0,1) {$s'$};
  \end{scope}

  \begin{scope}[xshift=5cm]
    \node[nodeB] (n1) at (0,0) {};
    \begin{scope}[xshift=-1.5cm, yshift=-1.5cm]
      \node[nodeB] (l2) at (0,0) {};
      \node[nodeR] (l3) at (-120:1.5) {};
      \node[nodeR] (l4) at (-60:1.5) {};
      
      \drawedge{n1}{l2}{B}{O}
      \drawedge{l2}{l3}{B}{O}
      \drawedge{l2}{l4}{O}{B}
      \drawedge{l3}{l4}{B}{B}

    \end{scope}

    \begin{scope}[xshift=1.5cm, yshift=-1.5cm]
      \node[nodeR] (r2) at (0,0) {};
      \node[nodeB] (r3) at (-120:1.5) {};
      \node[nodeR] (r4) at (-60:1.5) {};
      
      \drawedge{n1}{r2}{B}{B}
      \drawedge{r2}{r3}{B}{O}
      \drawedge{r2}{r4}{B}{B}
      \drawedge{r3}{r4}{B}{O}

    \end{scope}
    \node[scale=1.5] at (0,1) {$s \odot s'$};
  \end{scope}

  \end{tikzpicture}
  \caption{$s \odot s'$ for two rooted marked trees $s, s' \in \mTb_*$ that have the same vertex mark at the root.}
  \label{fig:odot}
\end{figure}

Let $G$ be a locally finite marked graph on a finite or countable vertex set $V$. Let $v$ and $w$ be adjacent vertices in $G$ such that $\deg_G(v) \geq 2$.
For $h \ge 1$, if $(G, v)_h$ is a rooted tree, then it is easy
  to see that we have
  \begin{equation}
    \label{eq:G-w-v-h--odot-otimes}
    G[w,v]_h = \xi_G(w, v) \times \left [  \bigodot_{\stackrel{w' \sim_G v}{w' \neq w}} \left( (\tau_G(v), \xi_G(w',v)) \otimes G[v, w']_{h-1} \right) \right ].
  \end{equation}
  Also, if $v$ is a vertex in $G$ with $\deg_G(v) \geq 1$, it is easy to see that if $(G, v)_h$ is a rooted tree, we have
  \begin{equation}
    \label{eq:G-v-h--odot-otimes}
    [G, v]_h = \bigodot_{w \sim_G v} \left( (\tau_G(v), \xi_G(w, v)) \otimes G[v,w]_{h-1} \right).
  \end{equation}

With this, we are ready to state conditions under which the edge colors of a
directed colored graph are related to the edge colors 
of the directed colored graph derived from
its marked colorblind
version. The following proposition can be considered to be a generalization of
Lemma 4.9 in \cite{bordenave2015large}.

%aa
\begin{prop}
  \label{prop:MCB-consistent-general}
  Fix an integer $h \ge 1$. Let $\mF \subset \edgemark \times \mTb_*^{h-1}$ be a finite set with cardinality $L$ and set $\mC = \mF \times \mF$. Let $H \in \mG(\mC)$ be a simple directed colored graph on a finite or countable vertex set $V$, and let $\vbeta = (\beta(v): v \in V)$
  have elements in $\vermark$.
  Define $A_h$ to be the set of vertices $v \in V$ such that the $h$--neighborhood of $v$ in $\CB(H)$ is a rooted tree and also, for all vertices $w$ with distance no more than $h$ from $v$ in $\CB(H)$, $(\beta(w), D^H(w))$ is graphical. Then, if $G = \MCB_{\vbeta}(H)$, it holds that
  \begin{enumerate}
  \item For each vertex $v \in A_h$, we have $(G, v)_h \equiv [T_v, o_v]_h$ where $[T_v, o_v]$ is the rooted tree corresponding to the graphical pair $(\beta(v), D^H(v))$.
  \item If $ v,w \in A_h$ are adjacent vertices in $H$ and the edge directed
    from $v$ towards $w$ has color $(t, t')$ in $H$, we have $\etype_G^h(v,w) =
    (t,t')$, i.e.\ $G(w, v)_{h-1} \equiv t$ and $G(v,w)_{h-1} \equiv t'$.
  \end{enumerate}
\end{prop}

% bb
% summary of proof can be found in the note 2018-08-24_PS_marked-BC_prop-MCB-consistent-general.pdf

  \begin{proof}
    For adjacent vertices $u$ and $v$ in $H$ (which are, by definition, also
    adjacent in $G$), let $c(u,v) \in \mF$ be the first component of the color
    of the edge directed from $u$ towards $v$.  Note that $H$ is simple, meaning
    that there is only one edge directed from $u$ towards $v$, so $c(u,v)$ is well-defined.
Also, recall from the definition of $\mG(\mC)$ that the color of the edge directed
from $v$ towards $u$ is $\bar{c}$, with $c$ being the color of the edge directed
from $u$ towards $v$. 
    Therefore, the color of the edge directed from $u$ towards $v$ is $(c(u,v), c(v,u))$.
    Define $A_0$ to be the set of vertices $v \in V$ such that $(\beta(v),
    D^H(v))$ is graphical. Moreover, for $1 \leq l \leq h$, define $A_l$ to be
    the set of vertices $v \in V$ such that $(G, v)_l$ is a rooted tree and, for
    all $w \in V$ with distance at most $l$ from $v$ in $G$, $(\beta(w),
    D^H(w))$ is graphical.
    Note that we have $A_0 \supseteq A_1 \supseteq
    \dots \supseteq A_h$.
    On the other hand, note that removing the marks in $G$ yields $\CB(H)$, hence $A_h$ defined above coincides with that in the statement of Proposition~\ref{prop:MCB-consistent-general}.
    For each $v \in A_0$, let $[T_v, o_v] \in \mTb_*^h$ be the rooted tree corresponding to the graphical pair $(\beta(v), D^H(v))$, and let $(T_v, o_v)$ be an arbitrary member of the isomorphism class $[T_v, o_v]$.
    Observe that, for each vertex $v \in A_0$ with $\deg_G(v) \geq 1$, there exists a  bijection $f_v$ that maps the set of vertices adjacent to $v$ in $G$ to the set of vertices adjacent to $o_v$ in $T_v$ such that for all $w \sim_G v$, we have
    \begin{equation}
      \label{eq:c-v-w--T-v-f-v}
      \begin{aligned}
        c(v,w) &= T_v[f_v(w), o_v]_{h-1}, \\
        c(w, v) &= T_v[o_v, f_v(w)]_{h-1}.
      \end{aligned}
    \end{equation}
    This is because applying to $(\beta(v),
    D^H(v))$ the  definition of what it means to be a graphical pair implies that, for each $t, t' \in \edgemark \times \mTb_*^{h-1}$,
    we have that $E_h(t, t')(T_v, o_v)$, which is the number of vertices $\tilde{w}
    \sim_{T_v} o_v$ such that $T_v(\tilde{w}, o_v) \equiv t$ and $T_v(o_v,
    \tilde{w}) \equiv t'$, is equal to the number of vertices $w \sim_G v$ such that $c(v,w) = t$ and $c(w, v) = t'$.

    Now, for each pair of adjacent vertices $(v,w)$ in $G$, and $0 \leq r \leq
    h-1$, we inductively define $M_r(v,w) \in \edgemark \times \mTb_*^r$ as
    follows, We first define $M_0(v,w) \in \edgemark \times \mTb_*^0$ to have its mark
    component equal to $\xi_G(w,v) = c(v,w)[m]$ and its subtree component
    a single vertex with mark $\beta(v)$. In fact, $M_0(v,w) = G[w,v]_0$. For $v \sim_G w$ and $1 \leq r \leq h-1$, if $\deg_G(v) = 1$, i.e. $w$ is the only vertex adjacent to $v$, we define $M_r(v,w)$ to be equal to $M_0(v,w)$. Otherwise, we define
    \begin{equation}
      \label{eq:Mt-def}
      M_r(v,w) := c(v,w)[m] \times \left [ \bigodot_{\stackrel{w' \sim_G v}{w' \neq w}} (\beta(v), c(v, w')[m]) \otimes M_{r-1}(w',v) \right ].
    \end{equation}
    % {\color{red} the message passing interpretation}
    % This relation is motivated by \eqref{eq:G-w-v-h--odot-otimes}.
    See Remark~\ref{rem:message-passing} below for a  message passing interpretation for $M_r(v, w)$ motivated by \eqref{eq:G-w-v-h--odot-otimes}.
    By induction on $r$, we show the following
    \begin{subequations}
      \begin{align}
        v \in A_r,  w \sim_G v \quad \Rightarrow \quad M_r(v, w) &= c(v, w)_r, \qquad \forall 0 \leq r \leq h-1, \label{eq:induction-t-M-c}\\
        v \in A_r,  w \sim_G v \quad \Rightarrow \quad M_r(v, w) &= G[w, v]_r, \qquad \forall 0 \leq r \leq h-1. \label{eq:induction-t-M-G}
      \end{align}
    \end{subequations}
    Recall that $c(v,w)_r = (x,t_r)$, where $x$ and $t$ are the mark and the
    subgraph components of $c(v,w) \in \mF$, respectively. 
    Then, we use \eqref{eq:induction-t-M-c}
    and \eqref{eq:induction-t-M-G}
    to show that
    \begin{equation}
      \label{eq:induction-t-G-v-t}
      v \in A_r \quad \Rightarrow \quad (G, v)_r \equiv (T_v, o_v)_r, \qquad \forall 0 \leq r \leq h.
    \end{equation}
    Combining \eqref{eq:induction-t-M-c} and \eqref{eq:induction-t-M-G}, we
    realize that for adjacent vertices $v, w \in A_h$, we have $G(v, w)_{h-1}
    \equiv c(w,v)_{h-1} = c(w,v)$ and $G(w, v)_{h-1} \equiv c(v, w)_{h-1} = c(v,w)$, which is the second part
    of the  statement in Proposition~\ref{prop:MCB-consistent-general}. The
    first part is a result of \eqref{eq:induction-t-G-v-t} for $r = h$.
    Therefore, it suffices to show \eqref{eq:induction-t-M-c},
    \eqref{eq:induction-t-M-G} and \eqref{eq:induction-t-G-v-t} to complete the
    proof.

    To start the proof, note that, for $r = 0$,  $v \in A_0$, and $w \sim_G v$, the mark component of
    $M_0(v, w)$ is $\xi_G(w, v)$ and its subtree component is a single root with
    mark $\beta(v)$. On the other hand, the mark component of $c(v, w)_0$ is
    $c(v, w)[m] = \xi_G(w, v)$ and its subtree component, using
    \eqref{eq:c-v-w--T-v-f-v}, is the subtree component of $T_v[f_v(w), o_v]_0$.
    But since the pair $(\beta(v), D^H(v))$ is graphical, $T_v[f_v(w), o_v]_0$
    is a single root with mark $\beta(v)$. This establishes
    \eqref{eq:induction-t-M-c} for $r=0$. Moreover, \eqref{eq:induction-t-M-G}
    follows from the facts that, by the definition of $G = \MCB_{\vbeta}(H)$,
    the mark component of $G[w, v]_0$ is $\xi_G(w, v)=c(v,w)[m]$ and its subtree component is a single root with mark $\beta(v)$.

    Now, we use induction to show \eqref{eq:induction-t-M-c} and \eqref{eq:induction-t-M-G}.
    First, we directly show \eqref{eq:induction-t-M-c} and
    \eqref{eq:induction-t-M-G} for a vertex $v$ with $\deg_G(v) = 1$. If $w$ is the
    only vertex adjacent to such $v$, we have $M_r(v, w) = M_0(v, w) \in
    \edgemark \times \mTb_*^0$, by definition. Recall that the mark component of $M_0(v, w)$  is
    $\xi_G(w, v) = c(v, w)[m]$ and its subtree component is a single root with mark
    $\beta(v)$. But this is precisely $G[w, v]_0$, which shows
    \eqref{eq:induction-t-M-G}. To show \eqref{eq:induction-t-M-c}, from
    \eqref{eq:c-v-w--T-v-f-v}, we have $c(v,w) = T_v[f_v(w), o_v]_{h-1}$. But since
    $f_v$ is a bijection, and the pair $(\beta(v), D^H(v))$ is graphical, we have $\deg_{T_v}(o_v) = 1$ and hence the subtree
    component of $T_v[f_v(w), o_v]_{h-1}$ is a single root with mark $\beta(v)$,
    which is precisely the subtree component of $M_0(v, w)$. The mark components of
    $M_r(v, w) = M_0(v,w)$ and $c(v, w)$ are both equal to $\xi_G(w, v)$. This establishes
    \eqref{eq:induction-t-M-c} in case $\deg_G(v) = 1$.

    Now, we  show \eqref{eq:induction-t-M-c} and
    \eqref{eq:induction-t-M-G} for $v \in A_r$ such that $\deg_G(v) \geq 2$. If $v \in A_r$ then all the vertices adjacent to $v$ are in $A_{r-1}$. Therefore,
    using the induction hypothesis \eqref{eq:induction-t-M-c}
    for $r-1$ on the right hand side of \eqref{eq:Mt-def}, we
    realize that for such $v$ and $w \sim_G v$ we have
    \begin{equation*}
      M_r(v, w) = c(v, w)[m] \times \left [ \bigodot_{\stackrel{w' \sim_G v}{w' \neq w}} (\beta(v), c(v, w')[m]) \otimes c(w', v)_{r-1} \right ].
    \end{equation*}
    Using \eqref{eq:c-v-w--T-v-f-v} and the fact that $\beta(v) = \tau_{T_v}(o_v)$, we get
    \begin{equation*}
      M_r(v, w) = \xi_{T_v}(f_v(w), o_v) \times \left [ \bigodot_{\stackrel{w' \sim_G v}{w' \neq w}} (\tau_{T_v}(o_v), \xi_{T_v}(f_v(w'), o_v)) \otimes T_v[o_v, f_v(w')]_{r-1} \right ].
    \end{equation*}
    Observe that $f_v$ is a bijection, hence  the set of vertices $w'$ in $G$ such with $w' \sim_G
    v$ and  $w' \neq w$ is mapped by $f_v$ to the set of vertices $\tilde{w}$ in $T_v$
    such that  $\tilde{w} \sim_{T_v}o_v$ and $\tilde{w} \neq f_v(w)$. With this, we can rewrite the above relation
    as
    \begin{equation*}
      M_r(v, w) = \xi_{T_v}(f_v(w), o_v) \times \left [ \bigodot_{\stackrel{\tilde{w} \sim_{T_v} v}{\tilde{w} \neq f_v(w)}} (\tau_{T_v}(o_v), \xi_{T_v}(\tilde{w}, o_v)) \otimes T_v[o_v, \tilde{w}]_{r-1} \right ].
    \end{equation*}
    Using \eqref{eq:G-w-v-h--odot-otimes}, since $(T_v,o_v)$ is a rooted tree, the right hand side is precisely $T_v[f_v(w), v]_r$. Another usage of \eqref{eq:c-v-w--T-v-f-v} implies \eqref{eq:induction-t-M-c}.

    To show \eqref{eq:induction-t-M-G} for $v \in A_r$ with $\deg_G(v) \geq 2$
    and $w \sim_G v$, again using the fact that $w' \in A_{r-1}$ for all $w'
    \sim_G v$, we realize that by first using \eqref{eq:induction-t-M-G} for $r-1$ and substituting in the right hand side of \eqref{eq:Mt-def}, then using $c(v,w')[m] = \xi_G(w', v)$ for all $w' \sim_G v$,  and finally using $\beta(v) = \tau_G(v)$, we get
    \begin{equation*}
      M_r(v,w) = \xi_G(w, v)\times \left [ \bigodot_{\stackrel{w' \sim_G v}{w' \neq w}} (\tau_G(v), \xi_G(w', v)) \otimes G[v, w']_{r-1} \right ].
    \end{equation*}
    Since $v \in A_r$, $(G, v)_r$ is a rooted tree. Thereby, \eqref{eq:G-w-v-h--odot-otimes} implies that the right hand side of the preceding equation is precisely $G[w, v]_r$ which completes the proof of  \eqref{eq:induction-t-M-G}.

    Now, it remains to show \eqref{eq:induction-t-G-v-t}. We first do this for $v
    \in A_r$ such that $\deg_G(v) \geq 1$. Observe that, since  $v \in A_r$, $(G,
    v)_r$ is a rooted tree. Consequently, using \eqref{eq:G-v-h--odot-otimes}, we
    have
    \begin{equation*}
      [G, v]_r = \bigodot_{w \sim_G v} (\tau_G(v), \xi_G(w, v)) \otimes G[v,w]_{r-1}.
    \end{equation*}
    Since $w \in A_{r-1}$ for all $w \sim_G v$, using \eqref{eq:induction-t-M-c} and \eqref{eq:induction-t-M-G} for $r-1$, we
    realize that, for each $w \sim_G v$ on the right hand side, we have $G[v,
    w]_{r-1} = c(w, v)_{r-1}$. Moreover, we have $\tau_G(v) = \beta(v) =
    \tau_{T_v}(o_v)$ for $w \sim_G v$. Furthermore, 
    by \eqref{eq:c-v-w--T-v-f-v},  $\xi_G(w, v) = c(v, w)[m] = \xi_{T_v}(f_v(w), o_v)$ for all $w \sim_G v$. Substituting these into the
    above relation and using  \eqref{eq:c-v-w--T-v-f-v}, we get
    \begin{equation*}
      [G, v]_r = \bigodot_{w \sim_G v} (\tau_{T_v}(o_v), \xi_{T_v}(f_v(w), o_v)) \otimes T_v[o_v, f_v(w)]_{r-1}.
    \end{equation*}
    % in the above, we have used $\xi_G(w, v) = c(v, w)_m = (T_v[f_v(w),
    % o_v]_{h-1})_m = \xi_{T_v}(f_v(w), o_v)$$.
    %
    Since $f_v$ induces a one to one correspondence between the neighbors $w$ of
    $v$ in $G$ and the neighbors $\tilde{w}$ of $o_v$ in $T_v$,  we may rewrite
    the above as
    \begin{equation*}
      [G, v]_r = \bigodot_{\tilde{w} \sim_{T_v} o_v} (\tau_{T_v}(o_v), \xi_{T_v}(\tilde{w}, o_v)) \otimes T_v[o_v, \tilde{w}]_{r-1}.
    \end{equation*}
    Since $(T_v, o_v)$ is a rooted tree,  \eqref{eq:G-v-h--odot-otimes} implies that
    the right hand side is precisely $[T_v, o_v]_r$. This means that $[G, v]_r =
    [T_v, o_v]_r$ or equivalently $(G, v)_r \equiv (T_v, o_v)_r$, which is precisely
    \eqref{eq:induction-t-G-v-t}. 

To show \eqref{eq:induction-t-G-v-t} for $v \in
    A_r$ such that $\deg_G(v) = 0$, note that, for such $v$, $(G, v)_r$ is a single
    root with mark $\beta(v)$. Moreover, since $\deg_G(v) = \sum_{t, t' \in \mF}
    D^H_{t, t'}(v)$, we must have $D^H_{t,t'}(v) = 0$ for all $t, t' \in \mF$.
    Therefore, it must be the case that, for all $t, t' \in \edgemark \times
    \mTb_{h-1}^*$, $E_h(t, t')(T_v, o_v) = 0$. This means that $\deg_{T_v}(o_v) = 0$,
    and hence $(T_v, o_v)$ is a single root with mark $\beta(v)$. Therefore, $(G,
    v)_r \equiv (T_v, o_v)_r$ and the proof is complete.
  \end{proof}

%aa
\begin{rem}
  \label{rem:message-passing}
  Motivated by the definition of $M_r(v,w)$ in \eqref{eq:Mt-def},
  we can interpret $M_r(v,w)$ as the message the vertex $v$ sends to the vertex
  $w$ at time $r$, which is obtained by aggregating the messages sent by the
  neighbors of $v$, except for $w$, at time $r-1$. The proof of
  Proposition~\ref{prop:MCB-consistent-general} above implies that, if $v \in
  A_r$, $M_r(v,w)$ is in fact the local $r$--neighborhood of $v$ in $G$ after removing
  the edge between $v$ and $w$, i.e.\ $G[w,v]_r$. In fact, motivated by
  \eqref{eq:G-w-v-h--odot-otimes}, the message $M_r(v,w)$ is inductively constructed in a
  way so that this holds. 
\end{rem}

% bb

%aa
%\subsection{Graphs with given neighborhood structure}
%\subsection{Configuration model from a directed colored graph of a marked graph}
%\label{sec:marked-to-colored-to-CM}
% bb

In the second part of Section \ref{sec:dct-to-mark-and-mark-to-dcg},
we started with a directed colored graph
$H \in\mG(\mC)$,
defined on a finite or countable vertex set $V$,
%(recall that $H$ being simple means
%that $\CB(H)$ has no loops or multiple edges). 
and a sequence $\vbeta = (\beta(v): v \in V)$ with elements in $\vermark$,
and studied the corresponding
marked color blind version,
denoted by $\MCB_{\vbeta}(H)$. 
We will now start with a marked graph, consider the associated directed colored graph,
for a given $h \ge 1$,
and study the configuration model given by the colored
degree sequence of this graph. The purpose is to relate the marked color blind versions of the directed colored graphs arising as realizations from this configuration model to the original marked graph we started with.
%This is done below using the above tools.
The results we prove next are corollaries of 
Proposition \ref{prop:MCB-consistent-general}.

%\color{red}

\begin{definition}
  \label{def:h-tree-like}
  
A marked or unmarked graph $G$ is said to be $h$ tree--like if, for all vertices $v$ in
  $G$, the depth $h$ local neighborhood of $v$ in $G$, i.e.\ $(G, v)_h$, is a
  rooted tree. This condition is equivalent to requiring that there is no cycle
  of length $2h+1$ or less in $G$. \glsadd{trm:h-treelike}
  %\glsadd{not:Gbarn}
  
  \end{definition}
  
  %\color{black}

%aa
\begin{cor}
  \label{cor:h-tree-like-conf-same}
  Let $n \in \nats$. 
  Recall that $\mGb_n$ denotes the set of marked graphs on the vertex set $[n]$. 
  For $h \ge 1$, assume that a marked $h$ tree--like graph $G \in \mGb_n$ is
  given.  Let $\vD = \vD^{\colored(G)}$ be the colored degree sequence
  associated to the directed colored version, $\colored(G)$, of $G$. 
  Let $\vbeta :=
  (\beta(v): 1 \leq v \leq n)$ denote the vertex mark vector of $G$.
  %i.e. $\beta(v)$ is the mark of the vertex $v$ in $G$. 
  Then, for any directed colored
  graph $H \in \mG(\vD, 2h+1)$, 
  %we have 
  %$\psi_h(\MCB_{\vbeta}(H)) = \psi_h(G)$.
  we have $(\MCB_{\vbeta}(H), v)_h \equiv (G, v)_h$
  for all $v \in [n]$.
\end{cor}

% bb

  \begin{proof}
    By definition, since $H \in \mG(\vD, 2h+1) \subset \hmG(\vD)$, we have $D^H(v) = D(v) = D^{\colored(G)}(v)$
    for every vertex $v \in
    [n]$.
    Moreover, since $(G,
    v)_h$ is a rooted tree, using the rooted tree $(G, v)_h$ in
    Definition~\ref{def:theta-D-graphical}, 
    we realize that  the pair $(\beta(v), D^{\colored(G)}(v))$ is
    graphical. On
    the other hand, since $H \in \mG(\vD, 2h+1)$, the colorblind
    graph $\CB(H)$ is $h$ tree--like. Consequently, the set $A_h$ in
    Proposition~\ref{prop:MCB-consistent-general} coincides with $[n]$. Thus,
    the first part of Proposition~\ref{prop:MCB-consistent-general} implies that
    for all $v \in [n]$, we have $(\MCB_{\vbeta}(H), v)_h \equiv (G, v)_h$ which
    completes the proof.
  \end{proof}

%aa
  \begin{cor}
    \label{cor:NhG-count-nDb}
    Let $n \in \nats$ and $h \ge 1$. Let $G \in \mGb_n$ be an $h$ tree--like graph. Define
    \begin{equation}
      \label{eq:NhG}
      N_h(G) := |\{ G' \in \mGb_n: U(G')_h = U(G)_h\}|.
    \end{equation}
    Then, we have
    \begin{equation*}
      N_h(G) = n(\vD, \vbeta) | \mG(\vD, 2h+1) |,
    \end{equation*}
    where $\vD := \vD^{\colored(G)}$ and $\vbeta = (\beta(i): 1 \leq i \leq n)$
    with $\beta(i) := \tau_G(i)$. Here $n(\vD, \vbeta)$ denotes the number
    of distinct pairs $(\vD^\pi, \vbeta^\pi)$ where $\pi$ ranges over the set of
    permutations $\pi: [n] \rightarrow [n]$ and where, for $1\leq i \leq n$,
    $D^\pi(i) := D(\pi(i))$ and $\beta^\pi(i) := \beta(\pi(i))$. \glsadd{not:nDbeta}
  \end{cor}

% bb

    \begin{proof}
    For a permutation $\pi: [n] \rightarrow [n]$, define $G^\pi \in \mGb_n$ to be
    the marked graph obtained from $G$ by relabeling vertices using $\pi$. More
    precisely, for $v \in [n]$, we have  $\tau_{G^\pi}(v) := \tau_G(\pi(v))$.
    Also, we place an edge between the vertices $v$ and $w$ in $G^\pi$ if $\pi(v)$
    and $\pi(w)$ are adjacent in $G$. In this case, we set $\xi_{G^\pi}(v, w) =
    \xi_G(\pi(v), \pi(w))$. With this, for any permutation $\pi: [n] \rightarrow
    [n]$ and $H \in \mG(\vD^\pi, 2h+1)$,
    Corollary~\ref{cor:h-tree-like-conf-same} implies that $U(\MCB_{\vbeta^\pi}(H))_h =
    U(G^\pi)_h = U(G)_h$. On the other hand, if the permutations $\pi$ and $\pi'$
    are such that $(\vD^\pi, \vbeta^\pi)$ and $(\vD^{\pi'}, \vbeta^{\pi'})$ are
    distinct, the sets $\{ \MCB_{\vbeta^\pi}(H): H \in \mG(\vD^\pi, 2h+1)\}$ and
    $\{ \MCB_{\vbeta^{\pi'}}(H'): H' \in \mG(\vD^{\pi'}, 2h+1)\}$ are disjoint.
    Moreover,   part 2 of Proposition~\ref{prop:MCB-consistent-general}
    implies that for any permutation $\pi$ and any $H \in \mG(\vD^\pi, 2h+1)$,
    $\colored(\MCB_{\vbeta^\pi}(H)) = H$.
    % this is because the edges carry correct types which in turn will transfer
    % into colors in the colorization process.
    Thereby, distinct elements $H_1, H_2
    \in \mG(\vD^\pi, 2h+1)$ yield distinct marked colorblind graphs
    $\MCB_{\vbeta^\pi}(H_1)$ and $\MCB_{\vbeta^\pi}(H_2)$.
    This establishes the inequality $N_h(G) \geq n(\vD, \vbeta) | \mG(\vD, 2h+1)
    |$.
    The other direction can be seen by observing that if $G' \in \mGb_n$ is such that
    $U(G')_h = U(G)_h$, then there exists a permutation $\pi: [n] \rightarrow [n]$
    such that, for each vertex $v \in [n]$, we have $(G', v)_h \equiv
    (G,\pi(v))_h$. Consequently, for all vertices $v \in [n]$, we have $D^{\colored(G')}(v) = D^{\colored(G)}(\pi(v))
    = D^\pi(v)$ and $\tau_{G'}(v) = \tau_G(\pi(v)) = \beta^\pi(v)$. Also, since
    $G$ is $h$ tree--like, so is $G'$. Hence, with $H := \colored(G')$, we have
    $H \in \mG(\vD^\pi, 2h+1)$ and $G' = \MCB_{\vbeta^\pi}(H)$. This shows that
    $N_h(G) \leq n(\vD, \vbeta)|\mG(\vD, 2h+1)|$ and completes the proof.
  \end{proof}
  
%aa
%\subsection{Graphs with given neighborhood structure}
\subsection{Realizing Admissible Probability Distributions with Finite Support}
\label{sec:weak-convg-to-admissible}
% bb

Next, using the tools developed above, we 
show that, for all $h \ge 1$ and any
admissible  probability distribution $P \in \mP(\mTb_*^h)$ having finite support, there exists a sequence of marked graphs which converges to $P$ in the sense of local weak convergence. This result can
be considered a generalization of Lemma~4.11 in \cite{bordenave2015large}.

%aa  
\begin{lem}
    \label{lem:P-finite-seq-converging}
    Let $h \ge 1$ and $P \in \mP(\mTb_*^h)$. %, i.e. $P$ is admissible.
    Assume that
    $P$ is admissible and has finite support. For $x, x'
    \in \edgemark$, let $d_{x,x'} := \evwrt{P}{\deg_T^{x,x'}(o)}$ and $\vd :=
    (d_{x,x'}: x, x' \in \edgemark)$. Moreover, for $\theta \in \vermark$, let
    $q_\theta$ be the probability of the mark at the root in $P$ being $\theta$
    and define $Q = (q_{\theta}: \theta \in \vermark)$. If $\vmn$ and $\vun$ are
    sequences of edge and vertex mark count vectors such that $(\vmn,\vun)$ is adapted to $(\vd, Q)$, then there exists a finite set
    $\Delta \subset \mTb_*^h$ and a sequence of marked graphs $G_n \in
    \mGn_{\vmn, \vun}$ such that the support of $U(G_n)_h$ is
    contained in $\Delta$ for each $n$, and $U(G_n)_h \Rightarrow P$.
  \end{lem}

  % bb

% summary of proof can be found in the note
%  2018-08-29_PS_lem-P-finite-seq-converging.pdf

  %%%
  % \begin{proof}[Proof Highlights]
  % \begin{itemize}
  % \item Take a sequence $\gn = (\gn(i): i \in [n])$ with empirical converging to $P$.
  % \item  Obtain $(\vbetan, \vDn)$ from this sequence.
  % \item Convert $\vDn$ to $\vDpn \in \mD_n$.
  % \item Get $\Hn \in \mG(\vDpn, 2h+1)$ and $\tGn = \MCB_{\vbetan}(\Hn)$.
  % \item Modify $\tGn$ to obtain $\Gn \in \mGn_{\vmn, \vun}$.
  % \end{itemize}
  % \end{proof}
  %%%

\begin{proof}
  Let $S = \{r_1, \dots, r_k\} \subset \mTb_*^h$ be the finite support
  of $P$. Since $S$ is finite, we can construct,
  for each $n \in \nats$, a
  sequence $(\gn(i): 1 \leq i \leq n)$ where $\gn(i) \in S$, $1 \leq i \leq
  n $, and
  \begin{equation}
    \label{eq:deltagn--P}
    \frac{1}{n} \sum_{i=1}^n \delta_{\gn(i)} \Rightarrow P.
  \end{equation}
  Let $\delta$ be the maximum degree at the root over all the elements of $S$.
  Moreover, let $\mF \subset \edgemark \times \mTb_*^{h-1}$ be 
  the set
  %{\color{pedit2} the} set
  %\color{red}
  %Note to Payam: replace ``the set" by ``a set".
  %This may make a difference. Check.
  %\color{black}
  %containing
  comprised of 
  %{\color{pedit2}all the objects} 
  $T[o,v]_{h-1}$ and $T[v,o]_{h-1}$ for each  $[T,o]
  \in S$ and $v \sim_T o$.
 % \pres{It should be ``the'' set, since I want $\mF$ to conclude all the
  %possible such objects. Added the term ``all the objects'' to emphasize this
  %and hopefully avoid confusion.}
  Since $S$ is finite, 
  $\mF$ is finite, 
  %$\mF$ can be taken to be finite, 
  %\color{red}
  %Note to Payam: replaced ``$\mF$ is finite" by
  %``$\mF$ can be taken to be finite" in view of the
  %preceding change in pinning down $\mF$. If the preceding change is reversed,
  %this change will also need to be reversed. Check.
  %color{black}
   % {\color{pedit2} reversed the change}
  hence can be identified with $\{1,
  \ldots, L\}$ with $L := |\mF|$. With this, define the color set $\mC := \mF
  \times \mF$.
  
  For each $n \in \nats$, define the sequences $\vbetan = (\betan(i): 1
  \leq i \leq n)$ and $\vDn = (\Dn(i): 1 \leq i \leq n)$ as follows. For $1
  \leq i \leq n$, let $\betan(i) \in \vermark$ be the mark at the root in
  $\gn(i)$. Further, let $\Dn(i) \in \mM_L^{(\delta)}$ be such that, for $c
  \in \mC$, $\Dn_c(i) = E_h(c)(\gn(i))$. Here, $E_h(c)(\gn(i)) =
  E_h(t,t')(\gn(i))$ with $c = (t,t')$, as was defined in \eqref{eq:Eh-g-g'}.

Now, we try to construct directed colored graphs given $\vDn$. However, it might
be the case that $\vDn \notin \mD_n$. Therefore, we modify $\vDn$ slightly to get
a sequence in $\mD_n$. In order to do this, for $c \in \mC$,
let $\Sn_c := \sum_{i=1}^n \Dn_c(i)$. Moreover, for $c \in \mC_=$, let $\tSn_c := 2 \lfloor \Sn_c/2 \rfloor$, and for
$c \in \mC_{\neq}$, let $\tSn_c := \Sn_c \wedge \Sn_{\bar{c}}$. Note that, because of \eqref{eq:deltagn--P}, for all $c \in \mC$, $\Sn_c/n \rightarrow e_P(c)$
as $n \rightarrow \infty$. 
%\color{red}
%Note to Payam:
%The notation $e_P(c)$ needs to be included in the glossary.
%This should be done in addition to notation like
%$e_P(g,g')$ and $e_P(t,t')$.
%\color{black}
%\pres{added to glossary}
\glsadd{not:eP-c}%
On the other hand, as  $P$ is admissible, we have  $e_P(c) =
e_P(\bar{c})$. Hence, $|\tSn_c - \Sn_c| = o(n)$ for all $c \in \mC$.  Therefore, we can find a sequence $\vDpn = (\Dpn(i): 1 \leq i \leq n)$
such that for all $1 \leq i \leq n$ we have $\Dpn(i) \in \mM_L^{(\delta)}$, and for all
$c \in \mC$ we have $\Dpn_c(i) \leq \Dn_c(i)$, and we have $\sum_{i=1}^n \Dpn_c(i) =
\tSn_c$. Moreover, since $\sum_{c \in \mC} |\tSn_c - \Sn_c| = o(n)$, we
may construct $\vDpn$ such that, for all but $o(n)$ vertices, we have $\Dpn(i) =
\Dn(i)$.
In particular, if $\tP \in \mP(\mM_L^{(\delta)})$ is defined to be
the law of $D = (D_c: c \in \mC)$, where $D_c = E_h(c)(r)$ with $r$
having law $P$, we have
\begin{equation}
  \label{eq:delta-tDn--tP}
  \frac{1}{n} \sum_{i=1}^n \delta_{\Dpn(i)} \Rightarrow \tP.
\end{equation}
Indeed, due to \eqref{eq:deltagn--P}, we have $(\sum_{i=1}^n \delta_{\Dn(i)}) /
n \Rightarrow \tilde{P}$, which implies  \eqref{eq:delta-tDn--tP} since $\Dn(i)
= \Dpn(i)$  for all but $o(n)$ many $1 \leq i \leq n$.
Note that, by definition, $\tSn_c$ is even for $c \in \mC_=$ and, for
$c \in \mC_{\neq}$, $\Sn_c = \Sn_{\bar{c}}$. Therefore, $\vDpn \in
\mD_n$.

Furthermore, since conditions \eqref{eq:H1} and \eqref{eq:H2} are
both satisfied for $\vDpn$ and $\tP$, Theorem \ref{thm:alpha-h} then implies that $\mG(\vDpn,2h+1)$ is non empty for $n$ large enough. For such $n$, let $\Hn$ be a
member of $\mG(\vDpn, 2h+1)$ and let $\tGn = \MCB_{\vbetan}(\Hn)$.
Since for each $1 \leq i \leq n$, $\betan(i)$ and $\Dn(i)$
are defined based on $\gn(i) \in \mTb_*^h$, they form a graphical pair
in the sense of Definition~\ref{def:theta-D-graphical}. Also, $\Dpn(i) =
\Dn(i)$ for all but $o(n)$ vertices. On the other hand, all the degrees
in $\tGn$ are bounded by $\delta$. Therefore, the number of vertices $v$
in $\tGn$ such that $(\betan(w), \Dpn(w))$ is graphical for all vertices
$w$ in the $h$--neighborhood of $v$ is $n - o(n)$. Moreover, since $\Hn \in
\mG(\vDpn, 2h+1)$, $\tGn$ has no cycle of length $2h+1$ or less, which means
that $\tGn$ is $h$ tree--like. Thereby,
Proposition~\ref{prop:MCB-consistent-general} implies that the number of
vertices $v$ in $\tGn$ such that $(\tGn, v)_h \equiv \gn(v)$ is $n - o(n)$. This means that $U(\tGn)_h \Rightarrow P$.

Now, the only remaining step is to modify $\tGn$ to obtain a simple marked graph
in $\mGn_{\vmn, \vun}$. To do this, note that if $(\tmn(x, x'): x, x' \in \edgemark)$ is the edge mark count
vector of $\tGn$, we have
\begin{equation*}
  \tmn(x, x') =
  \begin{cases}
    \sum_{v=1}^n \vDpn_{x,x'}(v) & x \neq x', \\
    \frac{1}{2} \sum_{v=1}^n \vDpn_{x,x}(v) & x = x',
  \end{cases}
\end{equation*}
where
\begin{equation*}
  \Dpn_{x,x'}(v) := \sum_{\stackrel{t,t' \in \mF}{t[m] = x,t'[m] = x'}} \Dpn_{t,t'}(v).
\end{equation*}
This together with  condition \eqref{eq:delta-tDn--tP}, implies that for $x \neq x'
\in \edgemark$ we have $\tmn(x,x') / n \rightarrow d_{x,x'}$,  and for $x \in
\edgemark$ we have $\tmn(x,x) /
n \rightarrow d_{x,x}/2$. Consequently, $|\tmn(x,x') - \mn(x,x')| =
o(n)$. On the other hand, if $(\tun(\theta): \theta \in \vermark)$ is the
vertex mark count vector of $\tGn$, since $\vbetan$ is the vertex mark
vector of $\tGn$,~\eqref{eq:deltagn--P} implies that for $\theta \in
\vermark$, $\tun(\theta) / n \rightarrow q_\theta$ and hence
$\sum_{\theta \in \vermark} | \tun(\theta) - \un(\theta)| = o(n)$.
Now, we modify $\tGn$ to obtain $\Gn$. In order to do this, for each $x \leq x' \in \edgemark$ such that
$\tmn(x,x') < \mn(x,x')$, we add $\mn(x,x') - \tmn(x,x')$ many
edges with mark $x,x'$. We can do this for all such $x,x'$ so that all vertices 
in the graph are connected to at most one of the newly added edges. This is possible
for $n$ large enough since $\edgemark$ is finite, $\sum_{x \leq x' \in
  \edgemark} |\mn(x,x') - \tmn(x,x')| = o(n)$, and the total number of
edges in $\tGn$ is $O(n)$.
Next, for $x \leq x' \in \edgemark$ such that $\tmn(x,x') >
\mn(x,x')$, we arbitrarily remove $\tmn(x,x') - \mn(x,x')$ many edges with mark $x,x'$.
Moreover, since $\sum_{\theta \in \vermark} |\un(\theta) -
\tun(\theta)| = o(n)$, we may change the vertex mark of all but $o(n)$
many vertices so that for all $\theta \in \Theta$, the number of vertices
with mark $\theta$ becomes precisely equal to $\un(\theta)$. 
Let $\Gn$ be the resulting simple marked graph, which is indeed a member of 
$\mGnmnun$.

Note that, by construction, all the degrees in $\Gn$ are
bounded by $\delta+1$. Hence, the support of $U(\Gn)_h$ is contained
in the set $\Delta$, defined as the set of $[T, o] \in \mTb_*^h$ such that the
degrees of all vertices in $T$ are bounded by $\delta+1$. Note that $\Delta$ is
finite. Also, adding or removing each edge affects the $h$--neighborhood
of at most $2(\delta+1)^{h+1}$ many vertices. Likewise, changing the
mark of a vertex can affect the $h$--neighborhood of at most
$(\delta+1)^{h+1}$ many vertices. 
Hence, $(\Gn,v)_h = (\tGn,v)_h$ for all but $o(n)$ vertices $v \in
[n]$.
Consequently. $U(\Gn)_h \Rightarrow P$ and the proof is complete. 
  \end{proof}

%aa
  % \subsection{Connection between colored UGWT and UGWT }
  \subsection{Local Weak Convergence of a Sequence of Graphs obtained from a
    Colored Configuration Model}
\label{sec:connection-color-ugwt}

% bb

In Lemma~\ref{lem:P-finite-seq-converging} in the previous section,
given $h \ge 1$ and an admissible $P \in \mP(\mTb_*^h)$ with finite support, we
constructed a sequence of marked graphs $\Gn$ such that $U(\Gn)_h
\Rightarrow P$. In this section, we show how to use
a colored
configuration model based on this sequence to generate marked graphs which converge to
$\ugwt_h(P)$ in the local weak sense. In the process of doing this, we also draw a
connection between the 
%$\ugwt_h(P)$ 
marked unimodular Galton--Watson trees
introduced in
Section~\ref{sec:markovian-ugwt} and the colored unimodular Galton--Watson trees
introduced in Section~\ref{sec:color-unim-galt}.

Fix $h \ge 1$. Let $\Delta  \subset \mTb_*^h$  be a fixed finite set. Let $P \in \mP(\mTb_*^h)$ be
admissible with support contained in $\Delta$.
We write $\mF$ for the set of $T[o, v]_{h-1}$ and $T[v,o]_{h-1}$ arising from $[T,o] \in \Delta$ and vertices $v \sim_T o$. Since $\Delta$ is finite, $\mF$
is also finite. We use the notation $L := |\mF|$.
Define the color set $\mC := \mF
\times \mF$. 
Let $\delta$ be an upper bound for the degree of each vertex of
each $[T,o] \in \Delta$. For $r \in \Delta$, define $D(r) \in \mM_L^{(\delta)}$ to be the
matrix such that, for $t,
t' \in \mF$, $D_{t,t'}(r)= E_h(t, t')(r)$. Furthermore, define $\theta(r)
\in \Theta$ to be the mark at the root in $r$.

%aa
\begin{prop}
  \label{prop:MCB-colored-unif-D-2h+1--converge}
  With the above setup, let $(\Gamma_n: n \in \nats)$
  be a sequence of marked graphs, with $\Gamma_n$ having the vertex set $[n]$ and the support of $U(\Gamma_n)_h$
  contained in $\Delta$ for each $n$, and such that
  $U(\Gamma_n)_h \Rightarrow P$.
     Define $\vDn = (\Dn(v): v \in [n])$ where  for $v \in
  [n]$,  $\Dn(v) \in
  \mM_L^{(\delta)}$ is defined such that for $t, t' \in \mF$, $\Dn_{t,t'}(v) :=
  E_h(t, t')(\Gamma_n, v)$. Moreover, define  $\vbetan =
  (\betan(v): v \in [n])$  such that $\betan(v) :=
  \tau_{\Gamma_n}(v)$ for $v \in [n]$. 
  For $n \ge 1$, let $H_n$ be a random directed colored graph uniformly
  distributed in $\mG(\vDn, 2h+1)$, and assume that
  $(H_n: n \in \nats)$ are independent on a joint
  probability space. Let $G_n := \MCB_{\vbetan}(H_n)$. Then,
  with probability one, we have $U(G_n) \Rightarrow \ugwt_h(P)$.
\end{prop}

% bb

We prove this proposition in two steps. First, in Lemma~\ref{lem:cugwt--ugwt-P}
below, we draw a connection between $\ugwt_h(P)$ and a colored unimodular
Galton--Watson tree. Then, we use this to state
Lemma~\ref{lem:MCB-colored-converge}, which will then complete the proof of the
above statement. Before this, we need to set up some notation. 
%\color{red}
%Note to self:
%Need to rephrase the preceding paragraph to indicate what Lemma~\ref{lem:MCB-colored-converge} is proving.
%\color{black}

%With the distribution $P$ described above, we define a  probability measure on
%$\vermark \times \mG_*(\mC)$ as follows. 

Let $\tP \in \mP(\mM_L^{(\delta)})$ be the law of $D(r)$ where $r
\sim P$. Since $P$ is admissible, we have $\evwrt{\tP}{D_c} =
\evwrt{\tP}{D_{\bar{c}}}$  for all $c \in \mC$. Now, we generate a random rooted directed colored tree $(F,o)$ using the procedure described in
Section~\ref{sec:color-unim-galt} by starting with $D^{(0)} = D(r^{(0)})$ with
$r^{(0)} \sim P$ at the root, and then adding further layers as in the colored
unimodular Galton--Watson tree. Let $Q \in \mP(\vermark
  \times \mG_*(\mC))$ be the law of the pair $(\theta(r^{(0)}), [F, o])$. Furthermore, let $Q_1$
  and $Q_2$ be the law of $\theta(r)$ and $[F, o]$, respectively. Note that $Q_2
  = \cugwt(\tP)$. 
For vertices $v, w$ in $F$, let $c(v, w) \in \mF$ be the first component of the
color of the edge going from $v$ towards $w$. For a vertex $v$ in $F$ other than
the root, let
$p(v)$ be the parent of $v$, and let $c(v)$ be the shorthand for $(c(v,
  p(v)), c(p(v), v))$. Moreover, for a vertex $v$, let $M(v) \in \mM_L^{(\delta)}$ be such
that for $c \in \mC$,
\begin{equation*}
  M_c(v) := |\{w: p(w) = v, c(v,w) = c\}|.
\end{equation*}
%\color{red}
%Note to Payam:
%Introduced $|\cdot|$ which was missing earlier, since
%$M_c(v)$ is supposed to be an integer and not a set. Check.
%\color{black}
%\pres{correct, thanks}
In fact, $M(v)$ is the part of the colored degree matrix of $v$ corresponding to
its offspring, so that if $v \neq o$, $D^F(v) = M(v) + E^{c(v)}$ and $D^F(o) = M(o)$. Recall that $E^{c(v)} \in
\mM_L$ is the matrix with value 1 in entry $c(v)$ and zero elsewhere. 

A matrix $D \in \mM_L^{(\delta)}$ is said to be $\Delta$--graphical if there
exists $r \in \Delta$ such that $D = D(r)$.
\glsadd{trm:Delta-graphical-matrix}%
%\color{red}
%Note to Payam:
%The terminology ``$\Delta$--graphical" needs to be put in the list of terminology at the end.
%\color{black}
%\pres{both ``$\Delta$''--graphical matrix'' and ``$\Delta$--graphical rooted
%  directed colored graph'' defined in the next paragraph added to the list of terminology}
If $D \in \mM_L^{(\delta)}$ is $\Delta$--graphical and nonzero, define
$\theta(D)$ to be the mark at the root for some $r \in \mTb_*^h$ for which
we have $D = D(r)$. To see why $\theta(D)$ is well-defined for $D \neq 0$, take $r, r'
\in \Delta$ so that $D = D(r) = D(r')$. Since $D$ is nonzero, there
exist $t, t' \in \edgemark \times \mF$ such that $D_{t,t'} = D_{t,t'}(r) =
D_{t,t'}(r') > 0$. Hence, the marks at the root in both $r$ and $r'$ are the same
as the mark at the root in the subgraph part of $t$, i.e.\ $t[s]$. This shows that $\theta(D)$ is well
defined. In fact this together with
Lemma~\ref{lem:r-r'-same-degree_are-the-same} 
in Appendix~\ref{sec:marked-rooted-trees-some-props} implies that if $D \neq 0$ is
$\Delta$--graphical there is only one $r \in \Delta$ such that $D = D(r)$. 

We say that  a rooted directed colored graph $[F, o] \in \mG_*(\mC)$ is 
$\Delta$--graphical if for each vertex $v$ in $F$, $D^F(v)$ is
$\Delta$--graphical.
\glsadd{trm:Delta-graphical-rooted-graph}
Let $\mH$ be the subset of $\vermark \times \mG_*(\mC)$ which consists
of the pairs $(\theta, [F, o])$ such that $[F, o]$ is a $\Delta$--graphical
rooted directed colored graph, and
if $o$ is not isolated in $F$ we have $\theta = \theta(D^F(o))$.
For $(\theta, [F, o]) \in \mH$, by an abuse of notation, we define
  $\MCB_\theta(F)$ to be the simple marked graph defined as follows. Let $\beta(o) :=
  \theta$, and for 
  $v \neq o$ in $F$, define $\beta(v) := \theta(D^F(v))$. 
  Note that if $v$ is a vertex other than the root,
  since $F$ is connected by definition, $v$ is not isolated and hence
  $D^F(v)$ is not the zero matrix. Thereby, $\theta(D^F(v))$ is well-defined.  With
  this, let $\vbeta $ be the vector
  consisting of  $\beta(v)$ for vertices $v$ in $F$, and define $\MCB_\theta(F)
  := \MCB_{\vbeta}(F)$.
  Note that if $(\theta, [F, o]) \sim Q$ then, with probability one,
 we have $(\theta, [F, o]) \in \mH$. The reason is that $D^F(o) = D(r^{(0)})$ and
 $\theta = \theta(r^{(0)})$, where $r^{(0)}$ is in the support of $P$ and hence
 in $\Delta$. Moreover, by the construction of $\cugwt(\tP)$ and
 \eqref{eq:cugwt-phat}, with probability one, for all vertices $v \neq o$ in $F$, $D^F(v) = M(v) +
 E^{c(v)}$ is in the support of $\tP$, and hence $D^F(v)$ is
 $\Delta$--graphical. 

 Now we are ready to state two lemmas. 
 Lemma~\ref{lem:MCB-colored-converge} will prove
 Proposition~\ref{prop:MCB-colored-unif-D-2h+1--converge},
 and itself depends on Lemma~\ref{lem:cugwt--ugwt-P}.
 The proposition will be proved assuming the truth of 
 the lemmas, and then the lemmas will be proved.

 %\color{red}
 %Note to self:
 %In the following lemma, need to set up what $Q$ is, starting with $P$.
 %\color{black}

  %aa
 \begin{lem}
   \label{lem:cugwt--ugwt-P}
   If $(\theta, [F,o])$ has the law $Q$ described above, then $[\MCB_\theta(F), o]$ has the law $\ugwt_h(P)$.
 \end{lem}

 % bb

 %aa
\begin{lem}
  \label{lem:MCB-colored-converge}
  With the above setup, assume that a sequence $\vDn
  \in \mD_n$ together with a sequence $\vbetan = (\betan(v): v \in [n])$
  are given such that $\betan(v) \in \vermark$ for all $v \in V$. Moreover, assume that for
  each $n \in \nats$ and $v \in [n]$, we have $\Dn(v) \in \mM_L^{(\delta)}$ and $(\betan(v),
  \Dn(v)) = (\theta(r), D(r))$ for some $r \in \Delta$.  Also, with $\tQ$ being the law
  of $(\theta(r), D(r))$ when $r \sim P$, assume that
  \begin{equation}
    \label{eq:betan-v-Dn-v---tilde-Q}
    \frac{1}{n} \sum_{v=1}^n \delta_{(\betan(v), \Dn(v))} \Rightarrow \tQ.
  \end{equation}
  With $H_n$ uniformly distributed in $\mG(\vDn, 2h+1)$ and
  independently for each $n$, define $G_n := \MCB_{\vbetan}(H_n)$. Then, with
  probability one, we have $U(G_n) \Rightarrow \ugwt_h(P)$.
\end{lem}

% bb

   \begin{proof}[Proof of Proposition~\ref{prop:MCB-colored-unif-D-2h+1--converge}]
    Note that since $U(\Gamma_n)_h \Rightarrow P$, the sequences $\vbetan$ and
    $\vDn$ obtained from $\Gamma_n$ as in the statement of the proposition
    satisfy \eqref{eq:betan-v-Dn-v---tilde-Q}. Therefore,
    Lemma~\ref{lem:MCB-colored-converge} completes the proof.
  \end{proof}

   \begin{proof}[Proof of Lemma~\ref{lem:cugwt--ugwt-P}]
    Note that, with probability one, $(\theta, [F,o]) \in \mH$. Let $T
    := \MCB_\theta(F)$. Since $[F,o]$ is almost surely $\Delta$--graphical and $T$
    is a simple marked tree, Proposition~\ref{prop:MCB-consistent-general}
    implies that, for all vertices $v$ in $T$, we have $[T,v]_h \in \Delta$
    almost surely.
    Therefore, given $r \in \Delta$, using
    Proposition~\ref{prop:MCB-consistent-general}, we have $(T, o)_h \equiv r$
    iff $M(o) = D(r)$ and $\theta$ is the mark at
      the root in $r$, i.e.\ $(\theta, M(o)) = (\theta(r),D(r))$. By the definition
      of $Q$, this has probability $P(r)$. To sum up, we have $\pr{(T,o)_h \equiv r} = P(r)$.

Now, assume that $v \sim_T o$ is an offspring
of the root in $T$ such that $T(o,v)_{h-1} \equiv t$ and
$T(v,o)_{h-1} \equiv t'$. If $\tilde{t} \in \edgemark \times \mTb_*^h$ is such
that $\tilde{t}_{h-1} = t$, Lemma~\ref{lem:t1-oplus-t'--t2-oplus-t'} in
Appendix~\ref{sec:marked-rooted-trees-some-props} implies that $T(o,v)_h \equiv \tilde{t}$
iff $(T,v)_h \equiv \tilde{t} \oplus t'$. Since $[T,v]_h \in \Delta$ almost
surely, $T(o,v)_h \equiv \tilde{t}$
has probability zero unless $\tilde{t} \oplus t' \in \Delta$. Assuming
that $\tilde{t} \oplus t' \in \Delta$ is satisfied, by the construction of
$\MCB_\theta(F)$ and Proposition~\ref{prop:MCB-consistent-general}, we
know that $(T,v)_h \equiv \tilde{t} \oplus t'$ iff $D^F(v) = D(\tilde{t}
\oplus t')$, or equivalently,  $M(v) = D(\tilde{t} \oplus t')
- E^{(t,t')}$. From \eqref{eq:cugwt-phat}, the probability of this is precisely 
\begin{equation*}
  \widehat{\tP}^{\overline{c(v)}}(M(v)) = \frac{(M_{c(v)}(v)+1) \tP(M(v) + E^{c(v)})}{e_P(\overline{c(v)})}.
\end{equation*}
Since $c(v) = (t,t')$, we have $M_{c(v)}(v)+1 = D_{(t,t')}(\tilde{t} \oplus t') =
E_h(t,t')(\tilde{t} \oplus t')$. On the other hand, $\tP(M(v) + E^{c(v)}) = \tP(D(\tilde{t} \oplus t')) =
P(\tilde{t} \oplus t')$. Comparing this with \eqref{eq:size-biased-def},
we realize that 
\begin{equation*}
  \pr{T(o,v)_h \equiv \tilde{t} | T(o,v)_{h-1} \equiv t, T(v,o)_{h-1} \equiv t'} = \one{\tilde{t} \oplus t' \in \Delta} \hP_{t,t'}(\tilde{t}) = \hP_{t,t'}(\tilde{t}),
\end{equation*}
where the last equality uses the fact that the support of $P$ is contained in
$\Delta$. Comparing these with the definition of $\ugwt_h(P)$, the proof is complete by repeating this argument inductively for
further depths in $T$ and noting that the choice of $M(v)$ in $F$ is done
conditionally independently for vertices with the same parent.
%\color{red}
%Note for Payam:
%Added the phrase ``with the same parent" at the end of the last line. Check.
%\pres{correct}
%\color{black}
  \end{proof}

    \begin{proof}[Proof of Lemma~\ref{lem:MCB-colored-converge}]
      From Theorem~\ref{thm:BC-CUGWT-LWC} we know that, with probability
      one, we have  $U(H_n)
      \Rightarrow Q_2 = \cugwt(\tP)$. Moreover, we claim that, with probability one,
      \begin{equation}
        \label{eq:betan-v-Hn-v-converge-Q}
        \frac{1}{n} \sum_{v=1}^n \delta_{(\betan(v), [H_n, v])} \Rightarrow Q.
      \end{equation}
      Recall that $[H_n,v] \in \mG_*(\mC)$ is the isomorphism class of the
      connected component of $v$ in $H_n$ rooted at $v$.
      Since $H_n \in \mG(\vDn, 2h+1)$, for each $v \in [n]$,
      it holds that $[H_n(v), v] \in \mG_*(\mC)$ is a simple colored directed rooted
      graph.
      Here,
      to make sense of the weak convergence, we turn $\Theta \times
      \mG_*(\mC)$ into a metric space with the metric
      \begin{equation*}
        d((\theta, [H, o]), (\theta', [H', o'])) = d_{\vermark}(\theta, \theta') + d_{\mG_*(\mC)}([H, o], [H', o']),
      \end{equation*}
      where $d_{\vermark}$ in the first term on the right hand side is an arbitrary metric
      on the finite set $\vermark$, e.g.\ the discrete metric,  and
      $d_{\mG_*(\mC)}$ denotes the the local metric of $\mG_*(\mC)$ from
      Section~\ref{sec:color-unim-galt}.
      To show \eqref{eq:betan-v-Hn-v-converge-Q}, we take a
      bounded continuous function $f : \Theta \times \mG_*(\mC)
      \rightarrow \reals$ and show that
      \begin{equation}
        \label{eq:f-betan-v-Hn-v-converge-int-Q}
        \frac{1}{n} \sum_{v=1}^n f(\betan(v), [H_n, v]) \rightarrow \int f d Q \qquad \text{a.s.}.
      \end{equation}
      With such a function $f$, define $f_1 : \mG_*(\mC) \rightarrow \reals$ as follows: for
      $[F, o] \in \mG_*(\mC)$, if $o$ is not isolated in $F$ and $D^F(o)$ is $\Delta$--graphical, define $f_1([F,o]) :=
      f(\theta(D^F(o)), [F,o])$. Recall that, since $D^F(o)$ is nonzero and $\Delta$--graphical,
      $\theta(D^F(o))$ is well-defined. Otherwise, if $o$ is isolated in $F$ or
      if $D^F(o)$ is not $\Delta$--graphical,
      define $f_1([F,o])$ to be zero.
      Moreover, define $f_2: \vermark \times \mM_L \rightarrow \reals$ as
      follows: for $\theta \in \vermark$ and $D \in \mM_L$, if $D$ is
      not the zero matrix, define $f_2(\theta, D)$ to be zero.
      Otherwise, define $f_2(\theta, D) := f(\theta, [F,o])$ where $[F,
      o] \in \mG_*(\mC)$ is an isolated root.
      Now, take $(\theta, [F, o]) \in \vermark \times \mG_*(\mC)$
      such that $(\theta, D^F(o)) = (\theta(r) , D(r))$ for some $r \in \Delta$. If $o$ is isolated in $F$, $f_1([F, o]) = 0$ and
      $f_2(\theta, D^F(o))  = f(\theta, [F, o])$. Otherwise, $f_1([F,
      o]) = f(\theta, [F, o]) $ and $f_2(\theta, D^F(o))=0$. In both cases,
      we have
      \begin{equation}
        \label{eq:f-theta-Fo--f1+f2}
        f(\theta, [F, o]) = f_1([F,o]) + f_2(\theta, D^F(o)).
      \end{equation}
      On the other hand, if $(\theta, [F, o]) \sim Q$, with probability one,
      we have 
      $(\theta, D^F(o)) = (\theta(r) , D(r))$ for some $r \in \Delta$. Thereby,
      \begin{equation}
        \label{eq:f-f1+f2--Q-as}
        f(\theta, [F, o]) = f_1([F, o]) + f_2(\theta, D^F(o)) \qquad Q\text{--a.s.}.
      \end{equation}

 Note that, by assumption, for all $n \in \nats$ and $v \in [n]$, we have $(\betan(v),
 \Dn(v)) = (\theta(r), D(r))$ for some $r \in \Delta$. Also, we have $D^{H_n}(v)
 = \Dn(v)$. Consequently, from \eqref{eq:f-theta-Fo--f1+f2}, for all $n \in \nats$ and $v
 \in [n]$, we have
 \begin{equation}
   \label{eq:f-f1+f2--Hn}
   f(\betan(v), [H_n, v]) = f_1([H_n, v]) + f_2(\betan(v), \Dn(v)) \qquad \text{a.s.}.
 \end{equation}
 Moreover, if $(\theta, [F, o]) \sim Q$, $[F, o]$ is distributed
      according to $Q_2$ and $(\theta, D^F(o))$ is distributed
      according to $\tQ$. Thereby, using
      \eqref{eq:f-f1+f2--Q-as}, we have
      \begin{equation}
        \label{eq:int-f-Q-f1+f2}
        \int f dQ = \int f_1 d Q_2 + \int f_2 d \tQ.
      \end{equation}
 It is easy to see that if $f$ is continuous, both $f_1$
      and $f_2$ are continuous. Therefore, using the fact that, with
      probability one, $U(H_n)\Rightarrow Q_2$, we realize that,
      \begin{equation}
        \label{eq:int-f1-conv-Q2}
        \frac{1}{n} \sum_{v=1}^n f_1([H_n, v]) = \int f_1 dU(H_n) \rightarrow \int f_1 dQ_2 \qquad \text{a.s.}.
      \end{equation}
 Also, due to \eqref{eq:betan-v-Dn-v---tilde-Q}, we have
      \begin{equation}
        \label{eq:sum-f2-int-tilde-Q}
        \frac{1}{n} \sum_{v=1}^n f_2(\betan(v), \Dn(v)) \rightarrow \int f_2 d \tQ.
      \end{equation}
Substituting \eqref{eq:int-f1-conv-Q2} and
\eqref{eq:sum-f2-int-tilde-Q} into \eqref{eq:f-f1+f2--Hn}
and comparing with \eqref{eq:int-f-Q-f1+f2}, we arrive at
\eqref{eq:f-betan-v-Hn-v-converge-int-Q} which shows
\eqref{eq:betan-v-Hn-v-converge-Q}.

Now, define the function $J$ that maps $(\theta, [F,o])
\in \mH$  to $[\MCB_\theta(F), o] \in \mGb_*$.
It is easy to see that $J$ is continuous.
Moreover, Lemma~\ref{lem:cugwt--ugwt-P} asserts that the pushforward
of $Q$ under the mapping $J$ is precisely $\ugwt_h(P)$.
On the other hand, since for all $n \in \nats$ and $v \in [n]$ we have $(\betan(v), \Dn(v)) =
(\theta(r), D(r))$ for some $r \in \Delta$, we realize that, with probability
one, $(\betan(v),
[H_n,v]) \in \mH$ and $J(\betan(v), [H_n, v]) =
[\MCB_{\vbetan}(H_n), v] = [G_n, v]$.
Consequently, the pushforward of the
left hand side in \eqref{eq:betan-v-Hn-v-converge-Q} under the map $J$ is
precisely $U(G_n)$, while the pushforward of its right hand side
under the map $J$ is $\ugwt(P)$. This means that, with probability one,
$U(G_n) \Rightarrow \ugwt(P)$ and this completes the proof.
  \end{proof}

%%% Local Variables: 
%%% mode: latex
%%% TeX-master: "Note-41_BC-ent-arxiv.tex"
%%% End: 

\section{Properties of the Entropy}
\label{sec:ent-properties}

In this section, we give the proof of steps taken in
Section~\ref{sec:main-results} in order to prove 
Theorems~\ref{thm:badcases}, \ref{thm:bch-properties} and \ref{thm:Jh}.
First, in Section~\ref{sec:BC-conditions--infty}, we prove Propositions~\ref{prop:BC-not-unim_dQnotmatch--infty} and
\ref{prop:BC-no-tree--infty}, which specify conditions under which the entropy
is $-\infty$. 
%Then, in Section~\ref{sec:proof-lem-Ptilde-Ph+1}, we prove
%Lemma~\ref{lem:PPh--PtildePh+1}, which will be useful for the rest of the section. 
Afterwards, in Section~\ref{sec:lowerbound}, we prove the lower bound
result of Proposition~\ref{prop:lower-bound}. In
Section~\ref{sec:upperbound}, we prove the upper bound result of 
Propositions~\ref{prop:upper-bound}.
Finally, in Section~\ref{sec:upperbound-new}, 
we prove the upper bound result of Proposition~\ref{prop:upper-bound-infty}.

\subsection{Conditions under which the entropy is $-\infty$}
\label{sec:BC-conditions--infty}

In this section, we prove Propositions~\ref{prop:BC-not-unim_dQnotmatch--infty} and
\ref{prop:BC-no-tree--infty}. Before that, we state and prove the following
useful lemma:

\begin{lem}
  \label{lem:upper-bound-on-G_n-m_mlogn-n}
If $\mG_{n, m}$ denotes the set of simple unmarked graphs on the vertex set $[n]$
having exactly $m$  edges, we have, 
  \begin{equation*}
    \log |\mG_{n,m}| = \log \left | \binom{\binom{n}{2}}{m} \right | \leq m \log n + ns\left ( \frac{2m}{n} \right ),
  \end{equation*}
where $s(x) := \frac{x}{2} - \frac{x}{2} \log x$ for $x > 0$ and $s(0) := 0$. Moreover, since $s(x) \leq 1/2$ for all $x \geq 0$, we have in particular
\begin{equation*}
  \log |\mG_{n,m}| \leq m \log n + \frac{n}{2}.
\end{equation*}
\end{lem}

\begin{proof}%[Proof of Lemma~\ref{lem:upper-bound-on-G_n-m_mlogn-n}]
  Using the classical upper bound $\binom{r}{s} \leq (re/s)^s$, we have 
  \begin{equation*}
    \log \left | \binom{\binom{n}{2}}{m} \right | \leq m \log \frac{n^2 e}{2m} = m \log n + m \log \frac{ne}{2m} = m \log n + n s(2m/n),
  \end{equation*}
which completes the first part. Also, it is easy to see that $s(x)$ is
increasing for $x\leq 1$, decreasing for $x> 1$ and attains its maximum value $1/2$ at $x=1$. Therefore, $s(x) \leq 1/2$. This completes the proof of the second statement. 
\end{proof}

\begin{proof}[Proof of Proposition~\ref{prop:BC-not-unim_dQnotmatch--infty}]
Suppose $\bchover_{\vd, Q}(\mu)\condmnun > -\infty$. Then, for all $\epsilon>0$,
$\mGnmnun(\mu, \epsilon)$ is non empty for infinitely many $n$. Therefore, there exists a sequence of integers $n_i$ going to infinity
together with simple marked graphs $\Gni \in \mG^{(n_i)}_{\vec{m}^{(n_i)}, \vec{u}^{(n_i)}}$ such that $U(\Gni)
\Rightarrow \mu$.
This already implies that if $\bchover_{\vd, Q}(\mu)\condmnun > -\infty$, $\mu$
must be sofic and hence unimodular. In other words, if $\mu$ is not
unimodular, $\bchover_{\vd, Q}(\mu)\condmnun = -\infty$.

Consequently, it remains to show
that if either $\vd \neq \vdeg(\mu)$ or $Q \neq \vvtype(\mu)$ we have $\bchover_{\vd, Q}(\mu)\condmnun = -\infty$.
Similar to the above, assume  $\bchover_{\vd, Q}(\mu)\condmnun > -\infty$ and take the above
sequence of simple marked graphs $\Gni$. 
First note that, for any $\alpha > 0$, and $x, x' \in
\edgemark$, the function $[G,o] \mapsto \deg_G^{x,x'}(o) \wedge \alpha$ is
continuous and bounded on $\mGb_*$. Thereby,
\begin{equation*}
  \int \deg_G^{x,x'}(o) dU(\Gni)([G,o]) \geq  \int (\deg_G^{x,x'}(o) \wedge \alpha) dU(\Gni)([G,o]) \rightarrow   \int (\deg_G^{x,x'}(o) \wedge \alpha) d \mu([G,o]).
\end{equation*}
Sending $\alpha$ to infinity on the right hand side and using the monotone
convergence theorem, we realize that
\begin{equation}
  \label{eq:evwrtUgn-deg-liminf--bigger}
  \liminf_{i \rightarrow \infty} \int \deg_G^{x,x'}(o) dU(\Gni)([G,o]) \geq \deg_{x,x'}(\mu).
\end{equation}
On the other hand, we have
\begin{equation*}
  \int \deg_G^{x,x'}(o) dU(\Gni)([G,o]) =
  \begin{cases}
    m^{(n_i)}(x,x') / n & x \neq x', \\
    2m^{(n_i)}(x,x') / n & x = x'.
  \end{cases}
\end{equation*}
We know that if $x \neq x'$ we have $\mn(x,x') / n \rightarrow d_{x,x'}$, and 
%if $x = x'$, 
we also have
$\mn(x,x) / n \rightarrow d_{x,x}/2$
for all $x$.
Comparing this with~\eqref{eq:evwrtUgn-deg-liminf--bigger}, we realize
that if $\bchover_{\vd, Q}(\mu)\condmnun > -\infty$ then, for all $x,x' \in
\edgemark$, we have $d_{x,x'} \geq \deg_{x,x'}(\mu)$.
Similarly, using the fact that for all $\theta \in \vermark$ the mapping $[G,
o] \mapsto \one{\tau_G(o) = \theta}$ is continuous and bounded on $\mGb_*$, we
realize that, if $\bchover_{\vd, Q}(\mu)\condmnun> -\infty$, with the sequence $\Gni$ as
above we have $u^{(n_i)}(\theta) \rightarrow \vtype_\theta(\mu)$. But
$u^{(n)}(\theta) / n \rightarrow q_\theta$. This means that $Q = \vvtype(\mu)$.
As a result, to complete the proof, we assume that  for some
$\tx,  \txp \in \edgemark$, we have $d_{\tx,
  \txp} > \deg_{\tx, \txp}(\mu)$ and then we show that $\bchover_{\vd, Q}(\mu)\condmnun =
-\infty$.
In order to do this it suffices to prove that for any sequence $\epsilon_n \rightarrow
0$ we have
\begin{equation}
  \label{eq:dneq-limsup-epsilonn--infty-claim}
  \limsup_{n \rightarrow \infty} \frac{1}{n} \left ( \log |\mGn_{\vmn, \vun}(\mu, \epsilon_n)| - \snorm{\mn}_1 \log n \right) = - \infty.
\end{equation}
% note that if $\bchover > -\infty$, as it is $\lim_{\epsilon \rightarrow 0}
% \bchover(\mu, \epsilon)$, a diagonalization argument would imply that there
% exists a sequence $\epsilon_n \rightarrow 0$ such that the LHS above
% converge to $\bchover$
For an integer $\Delta > 0$ define $A_\Delta:= \{ [G,o] \in \mGb_*:
\deg_G^{\tx, \txp}(o) > \Delta\}$.
Recall that, by definition of the \LP distance, if $\Gn \in \mGnmnun(\mu, \epsilon_n)$ then
\begin{equation}
  \label{eq:UGn-ADelta}
  U(\Gn)(A_\Delta) \leq \mu(A_\Delta^{\epsilon_n}) + \epsilon_n,
\end{equation}
where $A_\Delta^{\epsilon_n}$ is the $\epsilon_n$--extension of
the set $A_\Delta$.
Note that if we have $d_*([G,o], [G',o']) < 1/2$ for $[G,o]$ and $[G',o']$ in
$\mGb_*$ then we have $[G,o]_1 \equiv [G',o']_1$ and hence $\deg_G^{\tx, \tx'}(o) =
\deg_{G'}^{\tx, \tx'}(o')$. This implies that if $\epsilon_n < 1/2$, which
indeed holds for $n$ large
enough, then $A_\Delta^{\epsilon_n} = A_\Delta$.
Therefore, using~\eqref{eq:UGn-ADelta}, we realize that if  $n$ is large enough
so that $\epsilon_n < 1/2$,  for any $\Delta > 0$ we have 
\begin{equation}
  \label{eq:sizeof-v-deg-tx-txp-Delta--nmuADelta}
  | \{ v \in [n]: \deg_{\Gn}^{\tx, \txp}(v) > \Delta\}|  \leq n (\mu(A_\Delta) + \epsilon_n). 
\end{equation}
A similar argument shows that, for $n$ large enough
such that $\epsilon_n < 1/2$, for any integer $k$ and any $\Gn \in
\mGnmnun(\mu, \epsilon_n)$, we have
\begin{equation}
  \label{eq:v-deg-tx-txp-=k-Pmu}
  | \{ v \in [n]: \deg_{\Gn}^{\tx, \txp}(v) = k\}|  \leq n \left (\mu\left( \left\{ [G,o]: \deg_G^{\tx, \txp}(o) = k\right\} \right) + \epsilon_n \right ).
\end{equation}
Now, fix a sequence of integers $\Delta_n$ such that as
$n\rightarrow \infty$, $\Delta_n \rightarrow \infty$, but $\Delta_n^2
\epsilon_n \rightarrow 0$. For instance, one could make the choice $\Delta_n = \lceil  \epsilon_n^{-1/3}
\rceil$.
Using~\eqref{eq:v-deg-tx-txp-=k-Pmu} for $k = 0, \dots, \Delta_n$,
we realize that, for $n$ large enough and for any $\Gn \in
\mGnmnun(\mu, \epsilon_n)$, we have
\begin{equation}
  \label{eq:sum-c-Deltan-deg--deg-mu}
  \begin{aligned}
    \sum_{v \in [n]: \deg_{\Gn}^{\tx, \txp}(v) \leq \Delta_n} \deg_{\Gn}^{\tx, \txp}(v) &\leq \sum_{k = 0}^{\Delta_n} k n \left(  \mu\left(\left\{[G,o]:\deg_G^{\tx, \txp}(o) = k\right\}\right) + \epsilon_n \right) \\
    &\leq n \left(  \evwrt{\mu}{\deg_G^{\tx, \txp}(o) \one{\deg_G^{\tx, \txp}(o) \leq \Delta_n}}  + \Delta_n^2 \epsilon_n \right) \\
    &\leq n \deg_{\tx, \txp}(\mu) + n \Delta_n^2 \epsilon_n.
  \end{aligned}
\end{equation}
On the other hand, for $\Gn \in \mGnmnun(\mu, \epsilon_n)$ we have 
\begin{equation*}
  \sum_{v \in [n]} \deg_G^{\tx, \txp}(v) =
  \begin{cases}
    \mn(\tx, \txp) & \tx \neq \txp, \\
    2 \mn(\tx, \txp) & \tx = \txp.
  \end{cases}
\end{equation*}
Moreover, as $n \rightarrow \infty$, 
for $\tx \neq \txp$ we have $\mn(\tx,
\txp) / n  \rightarrow d_{\tx, \txp}$ and we also have
$\mn(\tx, \tx) /
n \rightarrow d_{\tx,\tx} / 2$ for all $\tilde{x}$.
Consequently, we have 
\begin{equation*}
  \sum_{v \in [n]} \deg_{\Gn}^{\tx,\txp}(v) = n (d_{\tx, \txp} + \alpha_n),
\end{equation*}
where $\alpha_n$ is a sequence such that $\alpha_n \rightarrow 0$
as $n\rightarrow \infty$. Comparing this with
\eqref{eq:sum-c-Deltan-deg--deg-mu}, we realize that for $n$ large
enough and $\Gn \in \mGnmnun(\mu, \epsilon_n)$ we have
\begin{equation*}
  \sum_{v \in [n]: \deg_{\Gn}^{\tx, \txp}(v) > \Delta_n} \deg_{\Gn}^{\tx, \txp}(v) \geq n ( d_{\tx,\txp} - \deg_{\tx, \txp}(\mu) + \alpha_n - \Delta_n^2 \epsilon_n).
\end{equation*}
Recall that, by assumption,  $d_{\tx,\txp} > \deg_{\tx, \txp}(\mu)$,  $\alpha_n
\rightarrow 0$ and $\Delta_n^2 \epsilon_n \rightarrow 0$. Hence,
there exists $\delta> 0$ such that for $n$ large enough and $\Gn
\in \mGnmnun(\mu, \epsilon_n)$ we have
\begin{equation}
  \label{eq:sum-v-n-Deltan-positive }
  \sum_{v \in [n]: \deg_{\Gn}^{\tx, \txp}(v) > \Delta_n} \deg_{\Gn}^{\tx, \txp}(v) \geq n \delta.
\end{equation}
Comparing this to
\eqref{eq:sizeof-v-deg-tx-txp-Delta--nmuADelta}, we realize that,
for $n$ large enough, 
$\Gn \in \mGnmnun(\mu, \epsilon_n)$ implies that for the subset $S_n \subset [n]$
defined as
$S_n:= \{v \in [n]: \deg_{\Gn}^{\tx, \tx'}(v) > \Delta_n\}$ we have 
$\sum_{v \in S} \deg_{\Gn}^{\tx,
  \txp}(v) \geq n \delta$ and $|S_n| \leq n\beta_n$, where $\beta_n
:= \mu(A_{\Delta_n}) + \epsilon_n$. Note that, since $\Delta_n \rightarrow
\infty$ and $\epsilon_n \rightarrow 0$, we have $\beta_n \rightarrow 0$.
Observe that  $\sum_{v \in S} \deg_{\Gn}^{\tx,\txp}(v) \geq n \delta$
implies that there are at least $n \delta/2$ many edges in $\Gn$ with mark
$\tx, \txp$ with at least one endpoint in the set $S_n$.
Let
 $\mS_n$ denote the family of subsets $S_n \subset [n]$
with $|S_n| \leq n \beta_n$. For $S_n \in \mS_n$, let $B_n(S_n)$
denote the set of  simple marked graphs $\Gn \in
\mGnmnun$ such that  there are at least $n \delta/2$ many edges with mark
$\tx, \txp$ with at least one endpoint in $S_n$. The above
discussion implies that, for $n$ large enough, we have 
\begin{equation}
  \label{eq:mGnmnun--BnS}
  \mGnmnun(\mu, \epsilon_n) \subset \bigcup_{S \in \mS_n} B_n(S_n).
\end{equation}
Now, in order to find an upper bound for the size of the set on the right hand
side, note that for $S_n \in \mS_n$ there are $\binom{|S_n|}{2} + |S_n|(n - |S_n|)$
many slots to choose for the edges with at least one endpoint in $S_n$ and with
marks $\tx, \tx'$. Since
there are at least $n \delta / 2$ many such edges, the number of ways to pick
the $\tx, \tx'$ edges of a graph in $B_n(S_n)$ is at most 
\begin{equation*}
  \binom{\binom{|S_n|}{2} + |S_n|(n-|S_n|)}{n \delta/2} \binom{\binom{n}{2}}{\mn(\tx, \txp) - n \delta/2} 2^{\mn(\tx, \txp)} =: C_n(S_n). 
\end{equation*}
Here, the term $2^{\mn(\tx, \txp)}$ is an upper bound for the number of ways we
can apply the marks $\tx$ and $\tx'$ for the chosen edges (if $\tx = \tx'$, this
number is in fact 1).
Now, since $|S_n| \leq n \beta_n$ for $S_n \in \mS_n$, using the standard bound $\binom{r}{s} \leq (re / s)^s$ and
Lemma~\ref{lem:upper-bound-on-G_n-m_mlogn-n}, if $n$ is large enough so that
$\beta_n \leq 1/2$, we get
\begin{align*}
  \max_{S_n \in \mS_n} \log C_n(S_n) &\leq \frac{n \delta}{2} \log \left ( n e \frac{\frac{\beta_n^2}{2} +\beta_n (1-\beta_n)}{\delta/2} \right ) + (\mn(\tx, \txp) - n \delta/2) \log n + \frac{n}{2} + \mn(\tx, \txp) \log 2 \\
                                     &= \mn(\tx, \txp) \log n + n \Bigg( \frac{1}{2} + \frac{\delta}{2} - \frac{\delta}{2} \log \frac{\delta}{2} + \frac{\mn(\tx, \txp)}{n} \log 2 \\
  &\quad \qquad + \frac{\delta}{2} \log \left(\frac{\beta_n^2}{2} + \beta_n (1-\beta_n)\right) \Bigg).
\end{align*}
Note that $\delta> 0$ is fixed. On the other hand, as $n
\rightarrow \infty$, $\mn(\tx,
\txp) / n $ either converges to $d_{\tx, \txp}$ or $d_{\tx,
  \txp} / 2$, depending on whether $\tx \neq \txp$
  or $\tx = \txp$ respectively.  
  %\color{red}
  %Note for Payam:
  %Changed the order of $\tx \neq \txp$ and $\tx = \txp$. Check.
  %\color{black}
  %\pres{correct}
  But, in any case, 
it remains bounded. However, $\beta_n \rightarrow
0$, hence $\delta \log (\beta_n^2/2 + \beta_n(1-\beta_n))
\rightarrow -\infty $. Consequently, we have 
\begin{equation}
  \label{eq:CS--infty}
  \lim_{n \rightarrow \infty} \frac{1}{n} \left( \max_{S \in \mS_n} \log C_n(S_n) - \mn(\tx, \txp) \log n \right) = - \infty.
\end{equation}
Now, in order to find an upper bound for $|B_n(S_n)|$ given $S_n \in \mS_n$, we
multiply the term $C_n(S_n)$ defined above by the number of ways we can
add vertex marks to the graph and also add edges with marks different from $\tx,
\tx'$, to get
\begin{equation*}
  |B_n(S_n)| \leq C_n(S_n) |\vermark|^n \prod_{\stackrel{x \leq x' \in \edgemark}{(x, x') \neq (\tx, \txp)}} \binom{\binom{n}{2}}{\mn(x, x')} 2^{\mn(x, x')}.
\end{equation*}
Using \eqref{eq:CS--infty} and Lemma~\ref{lem:upper-bound-on-G_n-m_mlogn-n} for
each term, we
realize that
\begin{equation}
  \label{eq:log-BnS--infty}
  \lim_{n \rightarrow \infty} \frac{1}{n} \left( \max_{S \in \mS_n} \log |B_n(S_n)| - \snorm{\vmn}_1\log n \right) = - \infty.
\end{equation}
Moreover, if $n$ is large enough so that $\beta_n < 1/2$, we have 
\begin{equation*}
  |\mS_n| \leq \sum_{k=0}^{n \beta_n} \binom{n}{k} \leq (1 + n \beta_n) \binom{n}{n \beta_n}.
\end{equation*}
Observe that, since $\beta_n \rightarrow 0$, we have
$\frac{1}{n} \log |\mS_n| \rightarrow 0$ as $n \rightarrow
\infty$.
Putting this together with \eqref{eq:log-BnS--infty}
and comparing with \eqref{eq:mGnmnun--BnS}, we arrive at
\eqref{eq:dneq-limsup-epsilonn--infty-claim}, which completes the
proof.
\end{proof}

\begin{proof}[Proof of Proposition~\ref{prop:BC-no-tree--infty}]
Let $\mG_*$ denote the space of isomorphism classes of rooted simple unmarked
connected graphs, which is  defined in a similar way as $\mGb_*$, with the
difference that  vertices and
edges do not carry marks.
We can equip $\mG_*$ with a local metric similar to that of $\mGb_*$.
With this, let  $F: \mGb_* \rightarrow \mG_*$ be such that
$[G, o]$ is mapped to $[\tG, o]$ under $F$, where $\tG$ is the unmarked graph
obtained from $G$ by removing all vertex and edge marks.
For $[G,o]$ and $[G',o']$ in $\mGb_*$, let $\tG$ and
$\tG'$ be obtained from $G$ and $G'$ by removing vertex and edge marks,
respectively. Observe that if $[G, o]_h \equiv [G', o']_h$ for $h \geq 0$, then $[\tG,o]_h
    \equiv [\tG',o']_h$.
This means that $F$ is $1$--Lipschitz, and in particular continuous.

Now, let $\vmn, \vun$ be any sequences such that $(\vmn,\vun)$ is adapted to $(\vdeg(\mu),
    \vvtype(\mu))$ and define $m_n = \snorm{\vmn}_1$.
Moreover, for integer $n$, let $\mG_{n, m_n}$ be the set of simple unmarked graphs on the vertex set $[n]$
    having $m_n$ edges.
Observe that if $\Gn \in \mGnmnun(\mu, \epsilon)$ for some $\epsilon>0$ and $n \in
\nats$, and $\tGn \in \mG_{n, m_n}$ is the unmarked graph obtained from
$\Gn$ by removing all vertex and edge marks, then $U(\tGn)$ is the pushforward of
$U(\Gn)$ under the mapping $F$.
Let $\rho \in \mP(\mG_*)$ be the pushforward of $\mu$ under $F$.
Since $F$ is $1$--Lipschitz, it is easy to see that for $\Gn \in \mGnmnun$, we
have  $\dlp(U(\tGn), \rho)
\leq \dlp(U(\Gn), \mu) < \epsilon$.
% see proof note 2018-05-09_Levy-Prokhorov-metric_1-Lipschitz-map.pdf
Therefore, if $\mG_{n, m_n}(\rho, \epsilon)$ denotes the set of unmarked
graphs $H \in \mG_{n, m_n}$ such that $\dlp(U(H), \rho) < \epsilon$, the
above discussion implies that for $\Gn \in \mGnmnun$, we have $\tGn \in
\mG_{n, m_n}(\rho, \epsilon)$. Moreover, for a simple unmarked  graph $H \in
\mG_{n, m_n}$, there are at most $(|\edgemark|^2)^{m_n} |\vermark|^n$ many
ways of adding marks to vertices and edges. Thereby, 
\begin{equation*}
  |\mGnmnun(\mu, \epsilon)| \leq |\mG_{n, m_n}(\rho, \epsilon)| (|\edgemark|^2)^{m_n} |\vermark|^n.
\end{equation*}
Note that as $n \rightarrow \infty$, $m_n / n \rightarrow d/2$ where $d = \deg(\mu) = \deg(\rho)$.
Consequently, 
    \begin{equation}
      \label{eq:mu-T*-rho}
      \limsup_{n \rightarrow \infty} \frac{\log |\mGnmnun(\mu, \epsilon)| - \snorm{\vmn}_1 \log n}{n} \leq \limsup_{n \rightarrow \infty} \frac{\log |\mG_{n, m_n}(\rho, \epsilon)| - m_n \log n }{n} + d \log |\edgemark| + \log |\vermark|.
    \end{equation}
Now, the assumption   $\mu(\mTb_*) < 1$ implies that $\rho(\mT_*) < 1$. Hence, Theorem~1.2 in
\cite{bordenave2015large} implies that the unmarked BC entropy of $\rho$ is
$-\infty$, i.e.
\begin{equation*}
  \lim_{\epsilon\rightarrow 0} \limsup_{n \rightarrow \infty} \frac{\log |\mG_{n, m_n}(\rho, \epsilon)| - m_n \log n }{n} = - \infty.
\end{equation*}
  Comparing this with \eqref{eq:mu-T*-rho}, we realize that
  $\bchover_{\vdeg(\mu), \vvtype(\mu)}(\mu)\condmnun = - \infty$ which completes the proof.
\end{proof}

\subsection{Lower bound}
\label{sec:lowerbound}

In this section, we prove the lower bound result of
Proposition~\ref{prop:lower-bound}.

\begin{proof}[Proof of Proposition~\ref{prop:lower-bound}]
For $x, x' \in \edgemark$, let $d_{x,x'} := \deg_{x,x'}(\mu)$ and $\vd :=
(d_{x,x'}: x,x' \in \edgemark)$. Furthermore, for $\theta \in \vermark$, let
$q_\theta := \vtype_{\theta}(\mu)$ and $Q := (q_\theta: \theta \in \vermark)$. 
We prove the result in two steps: first we assume that $P$ has a finite
support, and then relax this assumption.

\underline{Case 1: $P$ has a finite support:} Using
Lemma~\ref{lem:P-finite-seq-converging} from %Section~\ref{sec:graphs-with-given},
Section~\ref{sec:weak-convg-to-admissible}
we realize that there exists a finite set
    $\Delta \subset \mTb_*^h$ containing the support of $P$, and a sequence of
    simple marked graphs $\Gamma_n \in \mGnmnun$ such that $U(\Gamma_n)_h
    \Rightarrow P$ and, for all $n$, the
    support of $U(\Gamma_n)_h$ is contained in $\Delta$.
     To find a lower bound for $\mGn_{\vmn, \vun}(\mu, \epsilon)$, we may
    restrict ourselves to the graphs $G \in \mGn_{\vmn, \vun}$ such that
    $U(G)_h = U(\Gamma_n)_h$, since
    \begin{equation}
      \label{eq:P-finite-Gn-mu-epsilon-Gamma-h}
      |\mGn_{\vmn, \vun}(\mu, \epsilon)| \geq |\{ G \in \mGn_{\vmn, \vun}: U(G)_h = U(\Gamma_n)_h, \dlp(U(G), \mu) < \epsilon\}|.
    \end{equation}
    In order to find a lower bound for the right hand side of \eqref{eq:P-finite-Gn-mu-epsilon-Gamma-h}, we employ the tools from 
    Section~\ref{sec:color-conf-model}.

More precisely, define $\mF \subset \edgemark \times \mTb_*^{h-1}$ to be the set comprised of
$T[o,v]_{h-1}$ and $T[v,o]_{h-1}$ for all $[T, o] \in
    \Delta$ and $v \sim_T o$. Since $\Delta$ is finite, $\mF$ is also finite and
    hence can be identified with the set of integers $\{1, \ldots, L\}$ where $L
    = |\mF|$. Moreover, define the color set $\mC := \mF \times \mF$. Also, let
    $\delta$ be the maximum degree at the root among the members of $\Delta$. 
Since the support of $U(\Gamma_n)_h$ lies in $\Delta$, the colored
version of $\Gamma_n$, $\colored(\Gamma_n)$,  is a member of $\mG(\mC)$. 
Let $\vDn := \vD^{\colored(\Gamma_n)}$ be the colored degree sequence of
$\colored(\Gamma_n)$. Recall that, for $t, t' \in \mF$ and $v \in [n]$, we have $\Dn_{t,t'}(v) = E_h(t,
t')(\Gamma_n, v)$.
Moreover, since the support of $U(\Gamma_n)_h$ lies in $\Delta$, we have $\vDn(v) \in
\mM_L^{(\delta)}$ for all $n$ and $v \in [n]$. Furthermore, define $\vbetan =
(\betan(v): v \in [n])$ such that for $v \in [n]$, $\betan(v) := \tau_{\Gamma_n}(v)$.

From Corollary~\ref{cor:NhG-count-nDb}, we know that $N_h(\Gamma_n)$,
which is the number of simple marked graphs $G$ in $\mGb_n$ such that
$U(G)_h = U(\Gamma_n)_h$, is precisely $n(\vDn, \vbetan) |\mG(\vDn, 2h+1)|$.
Note that if $U(G)_h = U(\Gamma_n)_h$, then $\vm_{G} = \vm_{\Gamma_n} = \vmn$
and $\vu_{G} = \vu_{\Gamma_n} = \vun$, thus $G \in \mGnmnun$.
Moreover, from the proof of Corollary~\ref{cor:NhG-count-nDb}, we know that for
two permutations $\pi$ and $\pi'$, if
$((\vDn)^\pi, (\vbetan)^\pi) \neq ((\vDn)^{\pi'}, (\vbetan)^{\pi'})$, the sets
$\{\MCB_{(\vbetan)^\pi}(H): H \in \mG((\vDn)^\pi, 2h+1)\}$ and
$\{\MCB_{(\vbetan)^{\pi'}}(H): H \in \mG((\vDn)^{\pi'}, 2h+1)\}$ are disjoint.
On the other hand, for $H_n \neq H'_n \in \mG(\vDn, 2h+1)$, we have 
%$\MCB_{\vbeta}(H_n) \neq \MCB_{\vbetan}(H'_n)$. 
$\MCB_{\vbetan}(H_n) \neq \MCB_{\vbetan}(H'_n)$. 
%\color{red}
%Note for Payam: Changed $\MCB_{\vbeta}(H_n)$
%to $\MCB_{\vbetan}(H_n)$ in the preceding equation. Check.
%\color{black}
%\pres{correct}
    These observations, together with \eqref{eq:P-finite-Gn-mu-epsilon-Gamma-h},
imply that with $\tilde{H}_n$ being uniformly
    distributed in $\mG(\vDn, 2h+1)$ and $\tilde{G}_n :=
    \MCB_{\vbetan}(\tilde{H}_n)$,
    %[Payam's comment: just changed the wording], 
    we have 
    \begin{align*}
      |\mGn_{\vmn, \vun}(\mu, \epsilon)| & \geq n(\vDn, \vbetan) | \{ H_n \in \mG(\vDn, 2h+1): \dlp(U(\MCB_{\vbetan}(H_n)), \mu) < \epsilon\}| \\
                                         &= n(\vDn, \vbetan) \left|\mGn(\vDn, 2h+1)\right| \pr{\dlp(U(\tilde{G}_n), \mu) < \epsilon} \\
                                         &= N_h(\Gamma_n) \pr{\dlp(U(\tilde{G}_n), \mu) < \epsilon}.
    \end{align*}
From Proposition~\ref{prop:MCB-colored-unif-D-2h+1--converge},
    we know that, for any $\epsilon> 0$, $\pr{\dlp(U(\tilde{G}_n), \mu) < \epsilon}
    \rightarrow 1$ as $n\rightarrow \infty$. Therefore, we have 
    \begin{equation*}
      \liminf_{n \rightarrow \infty} \frac{\log |\mGnmnun(\mu, \epsilon)| - \snorm{\vmn}_1 \log n}{n} \geq \liminf_{n \rightarrow \infty} \frac{\log N_h(\Gamma_n) - \snorm{\vmn}_1 \log n}{n}.
    \end{equation*}
Consequently, if we show that 
\begin{equation}
  \label{eq:LB-f-log-Nh--Jh}
  \lim_{n \rightarrow \infty}\frac{1}{n} \left ( \log N_h(\Gamma_n) - \snorm{\vmn}_1 \log n \right ) = J_h(P),
\end{equation}
then we can conclude that for $\epsilon>0$, $\bchunder_{\vd, Q}(\mu, \epsilon)\condmnun \geq J_h(P)$ and
    hence $\bchunder_{\vd, Q}(\mu)\condmnun \geq J_h(P)$, which completes
    the proof for this case. Thereby, it suffices to show~\eqref{eq:LB-f-log-Nh--Jh}.

In order to do this, first note that, as a result of
Lemma~\ref{lem:r-r'-same-degree_are-the-same}
in Appendix~\ref{sec:marked-rooted-trees-some-props},
%from Section~\ref{sec:graphs-with-given}, 
for $v, w \in [n]$ we have $(\Gamma_n,v)_h
\equiv (\Gamma_n, w)_h$ iff $(\betan(v), \Dn(v)) = (\betan(w), \Dn(w))$.
% see the proof note
% 2018-04-04_Proof-note_BC-lower-bound-finite-support_betaD.pdf for a
% proof. In fact, I realized that this proof note contains the lemma which was
% referenced above!
Thereby, since $U(\Gamma_n)_h \Rightarrow P$ and $\Delta$ is finite, we have
\begin{equation}
  \label{eq:LB-f-nDb}
  \lim_{n \rightarrow \infty} \frac{1}{n} \log n(\vDn, \vbetan) = H(P).
\end{equation}
Moreover, from Corollary~\ref{cor:GDnh-counting} and Stirling's approximation,
if, for $c \in \mC$, $\Sn_c$ denotes $\sum_{v=1}^n \Dn_c(v)$, we have
\begin{equation}
  \label{eq:LB-f-log-GDn-2h+1}
  \begin{aligned}
    \log |\mG(\vDn, 2h+1)| &= \sum_{c \in \mC_<} \left( \Sn_c \log \Sn_c - \Sn_c \right) + \sum_{c \in \mC_=} \left( \frac{\Sn_c}{2} \log \Sn_c - \frac{\Sn_c}{2} \right) \\
    &\qquad - \sum_{c \in \mC} \sum_{v=1}^n \log \Dn_c(v)! + o(n)\\
    &= \frac{1}{2} \sum_{c \in \mC} \left( \Sn_c \log \Sn_c - \Sn_c  \right) - \sum_{c \in \mC} \sum_{v=1}^n \log \Dn_c(v)! + o(n).
  \end{aligned}
\end{equation}
Here, we have used the following facts: $(i)$ $\log k! = k \log k - k + o(k)$,
$(ii)$ $\log (k-1)!! =
\frac{k}{2} \log k - \frac{k}{2} + o(k)$, $(iii)$ for all $c \in \mC$, $\limsup_{n \rightarrow \infty} \Sn_c / n < \infty$
or equivalently $\Sn_c = O(n)$, and $(iv)$ for $c \in \mC$, $\Sn_c =
\Sn_{\bar{c}}$.
% for calculation details, see the note 
% 2018-09-11_PN_LB-Stirling-Sc-log-Sc.pdf
Note that, since $U(\Gamma_n)_h \Rightarrow P$ and $\Delta$ is finite, for each
$c = (t, t') \in \mC$ we have 
\begin{equation*}
  \frac{1}{n} \Sn_c = \frac{1}{n} \sum_{v=1}^n \Dn_c(v) \underset{n \rightarrow \infty}{\longrightarrow} \evwrt{P}{E_h(t,t')(T,o)} = e_P(t,t').
\end{equation*}
Likewise, for $c = (t, t') \in \mC$, we have
\begin{equation*}
  \frac{1}{n} \sum_{v=1}^n \log \Dn_c(v)! \underset{n \rightarrow \infty}{\longrightarrow} \evwrt{P}{\log E_h(t,t')(T,o)!}.
\end{equation*}
Using these in \eqref{eq:LB-f-log-GDn-2h+1} and simplifying, we get
\begin{equation}
  \begin{aligned}
  \label{eq:LB-f-log-GDn-2h+1--simplified}
  \log |\mG(\vDn, 2h+1)| &= \frac{n}{2} \sum_{c \in \mC} \left( \frac{\Sn_c}{n} \log \frac{\Sn_c}{n} + \frac{\Sn_c}{n} \log n - \frac{\Sn_c}{n} \right)- n \sum_{c \in \mC} \frac{1}{n} \sum_{v=1}^n \log \Dn_c(v)! + o(n) \\
  &=\snorm{\vmn}_1 \log n - \snorm{\vmn}_1 + \frac{n}{2} \sum_{t, t' \in \mF} e_P(t, t') \log e_P(t, t') \\
  &\quad - n \sum_{t,t' \in \mF} \evwrt{P}{\log E_h(t, t')(T,o)!} + o(n),
  \end{aligned}
\end{equation}
where in the second line we have used $\sum_{c \in \mC} \Sn_c= 2\snorm{\vmn}_1$.
Note that, since $\vmn$ and $\vun$ are such that $(\vmn,\vun)$ is adapted to $(\vdeg(\mu), \vvtype(\mu))$, as $n \rightarrow \infty$ we have $\snorm{\vmn}_1 / n \rightarrow \deg(\mu) /
2$.
From \eqref{eq:LB-f-nDb} and \eqref{eq:LB-f-log-GDn-2h+1--simplified}, with $d
:= \deg(\mu)$, we get 
%and
%simplifying,  we arrive at \eqref{eq:LB-f-log-Nh--Jh} which completes the
%proof for this case.
% for algebraic manipulations in the above simplification, see the note
% 2018-09-12_PN_LB-finite-Jh-simplification.pdf
\begin{align*}
  \frac{\log N_h(\Gamma_n) - \snorm{\vmn}_1 \log n}{n} &= H(P) - \frac{d}{2} + \frac{1}{2} \sum_{t,t' \in \mF} e_P(t,t') \log e_P(t,t') \\
                                                       &\qquad \qquad - \sum_{t,t' \in \mF} \evwrt{P}{\log E_h(t,t')(T,o)!} + o(1) \\
                                                       &= H(P) - \frac{d}{2} + \frac{d}{2} \sum_{t,t' \in \mF} \frac{e_P(t,t')}{d} \left( \log d + \log \frac{e_P(t,t')}{d} \right) \\
                                                       &\qquad \qquad - \sum_{t,t' \in \mF} \evwrt{P}{\log E_h(t,t')(T,o)!} + o(1) \\
                                                       &\stackrel{(a)}{=} - s(d) + H(P) + \frac{d}{2} \sum_{t,t' \in \edgemark \times \mTb_*^{h-1}} \pi_P(t,t') \log \pi_P(t,t') \\
                                                       &\qquad \qquad - \sum_{t,t' \in \edgemark \times \mTb_*^{h-1}} \evwrt{P}{\log E_h(t,t')(T,o)!} + o(1) \\
  &= J_h(P) + o(1),
\end{align*}
where in $(a)$ we have used the facts that the support of $P$ is contained in
$\Delta$ and  $\sum_{t,t' \in \edgemark \times
  \mTb_*^{h-1}} e_P(t,t') = d$. This shows 
 \eqref{eq:LB-f-log-Nh--Jh} and thus completes the
proof for the finite support case.

\underline{Case 2: For general $P$:} We use a truncation procedure together with the proof in 
the above finite support case. More precisely, for an integer $k > 1$, we start
from a random rooted marked tree $(T,o)$ with law $\mu$ and, for all vertices $v$ in
$T$ with degree more than $k$, we remove all the edges
connected to $v$. Let $T^{(k)}$ denote the connected
component of the root in the resulting forest. With this, define $\muk$ to be
the law of $[T^{(k)}, o]$. It is easy to see that $\muk$ is unimodular.
% for a proof, see the note
% 2018-10-10_PN_truncated-muk-is-unimodular.pdf
Furthermore, let $P_k
:= (\muk)_h \in \mP(\mTb_*^h)$  be the law of the depth $h$ neighborhood of the root
in $\muk$. Since $\muk$ is
unimodular, $P_k$ is admissible. On the other hand, $P_k$ has a finite support,
and hence $P_k$ is strongly admissible, i.e.\ $P_k \in \mP_h$. 

With the above construction, we have $\deg_{x,x'}(\muk) \leq \deg_{x,x'}(\mu)$
for all $x, x' \in \edgemark$ and $\vvtype(\muk) =
  \vvtype(\mu)$. We do not directly apply the result of the previous case to
  $P_k$, since the sequences $\vmn$ and $\vun$ are such that $(\vmn,\vun)$ is adapted to
  $(\vdeg(\mu), \vvtype(\mu))$ which might be different from $(\vdeg(\muk),
  \vvtype(\muk))$.  Instead, we modify $\muk$ to obtain a measure $\widetilde{\mu}^{(k)}$
  such that $(\vdeg(\widetilde{\mu}^{(k)}), \vvtype(\widetilde{\mu}^{(k)})) =
  (\vdeg(\mu), \vvtype(\mu))$. In order to do this, for each pair of edge marks $x \leq  x' \in
  \edgemark$, we choose an integer $\tilde{d}_{x,x'} > 2(|\edgemark|^2 d_{x,x'} \vee 1 )$.
Moreover, define $\nu_{x,x'}$ to be the law of $[T,o]
  \in \mTb_*$ where $(T, o) $ is the random rooted marked 
  $\tilde{d}_{x,x'}$--regular tree defined as follows. With
  probability $1/2$, we have 
  \begin{equation*}
    \xi_T(v,w) =
    \begin{cases}
      x & \dist_T(o,w) \text{ is even}, \\
      x' & \dist_T(o,w) \text{ is odd},
    \end{cases}
    \qquad \forall v,w \in V(T),
  \end{equation*}
  and with probability $1/2$, we have
  \begin{equation*}
    \xi_T(v,w) =
    \begin{cases}
      x' & \dist_T(o,w) \text{ is even}, \\
      x & \dist_T(o,w) \text{ is odd},
    \end{cases}
    \qquad \forall v,w \in V(T).
  \end{equation*}
  Additionally,  each vertex in $T$ is independently given a mark with distribution
  $\vvtype(\mu)$.
It is easy to check that $\nu_{x,x'}$ is unimodular, $\vvtype(\nu_{x,x'}) = \vvtype(\mu)$, and
$\deg_{x,x'}(\nu_{x,x'}) = \deg_{x',x}(\nu_{x,x'}) =
\tilde{d}_{x,x'}/2$. Let $U_{x,x'} := (\nu_{x,x'})_h \in
\mP(\mTb_*^h) $  be the law of the depth $h$ neighborhood of the root in
$\nu_{x,x'}$.
Due to the way we chose $\tilde{d}_{x,x'}$ for $x \leq x' \in \edgemark$, 
%Since for all $x \leq x' \in \edgemark$, we have $\tilde{d}_{x,x'} / 2 >
%d_{x,x'}$,
%for
%$k$ large enough,
we can choose $p_k \in [0,1]$ together with nonnegative numbers $(\alpha^k_{x,x'}:
x \leq x' \in \edgemark)$
 so that $p_k + \sum_{x \leq x' \in \edgemark} \alpha^k_{x,x'} =
1$ and such that with
\begin{equation}
  \label{eq:Ptilde-k-def}
  \tP_k := p_k P_k + \sum_{x\leq x' \in \edgemark} \alpha^k_{x,x'}
  U_{x,x'},
\end{equation}
 we have $\evwrt{\tP_k}{\deg_T^{x,x'}(o)} = d_{x,x'}$ for all $x, x'
 \in \edgemark$.
 More precisely, with $d^k_{x,x'} := \deg_{x,x'} (\muk)$, we may set 
 \begin{equation*}
   p_k := \frac{1 - \sum_{x\leq x' \in \edgemark} 2d_{x,x'} / \tilde{d}_{x,x'}}{1 - \sum_{x\leq x' \in \edgemark} 2d^k_{x,x'} / \tilde{d}_{x,x'}},
 \end{equation*}
 and, for $x \leq x' \in \edgemark$,
 \begin{equation*}
   \alpha^k_{x,x'} := \frac{2(d_{x,x'} - p_k d^k_{x,x'})}{\tilde{d}_{x,x'}}.
 \end{equation*}
Then, using $\tilde{d}_{x,x'} > 2(|\edgemark|^2 d_{x,x'} \vee 1)$ and $d^k_{x,x'} <
d_{x,x'}$, all the desired properties mentioned above would follow.
% for more details, see the note
% 2018-10-10_PN_pk-alphak.pdf
On the other hand, since $\deg_{x,x'}(\muk) \uparrow \deg_{x,x'}(\mu)$ as $k \rightarrow \infty$,
we have  $p_k \rightarrow 1$ as $k \rightarrow \infty$.
Furthermore, since $P_k$ is admissible and $\nu_{x,x'}$ is unimodular, $\tP_k$ is
admissible, and in addition has a finite support. This implies that
$\tP_k$ is strongly admissible, i.e.\ $\tP_k \in \mP_h$.
Thus, with $\widetilde{\mu}^{(k)}:= \ugwt_h(\tP_k)$, we have
$\vvtype(\widetilde{\mu}^{(k)}) = \vvtype(\mu)$ and
$\vdeg(\widetilde{\mu}^{(k)}) = \vdeg(\mu)$. Now, we  claim that 
\begin{equation}
  \label{eq:e-tPk->e-P}
  \lim_{k \rightarrow \infty} e_{\tP_k}(t,t') = e_P(t,t') \qquad  \forall t,t' \in \edgemark \times \mTb_*^{h-1}.
\end{equation}
In order to show this, note that from \eqref{eq:Ptilde-k-def} we have
\begin{equation}
  \label{eq:etPk-ePk-eUxx}
  e_{\tP_k}(t,t') = p_k e_{P_k}(t,t') + \sum_{x \leq x' \in \edgemark} \alpha^k_{x,x'} e_{U_{x,x'}}(t,t'), \qquad   \forall t,t' \in \edgemark \times \mTb_{*}^{h-1}.
\end{equation}
But $U_{x,x'}$ are fixed, $\edgemark$ is finite, and $\alpha^k_{x,x'} \rightarrow
0$. Hence, to show~\eqref{eq:e-tPk->e-P}, it suffices to show that 
\begin{equation}
  \label{eq:e-Pk->e-P}
  \lim_{k \rightarrow \infty} e_{P_k}(t,t') = e_P(t,t') \qquad  \forall t,t' \in \edgemark \times \mTb_*^{h-1}.
\end{equation}
Observe that for $t, t' \in \edgemark \times \mTb_*^{h-1}$ we have $e_{P_k}(t,t') =
  \evwrt{\mu}{E_h(t,t')([T^{(k)}, o])}$. But, for $[T, o] \in \mTb_*$, if $k$
  is large enough, $[T^{(k)}, o]_h = [T, o]_h$. Thereby, $E_h(t,t')([T^{(k)},
  o]) \rightarrow E_h(t,t')([T, o])$ as $k \rightarrow \infty$. On the
  other hand, $E_h(t,t')([T^{(k)},o]) \leq \deg_{T^{(k)}}(o) \leq
  \deg_T(o)$. Hence,
  \begin{equation*}
\evwrt{\mu}{E_h(t,t')([T^{(k)}, o])} \leq \evwrt{\mu}{\deg_T(o)} < \infty.
  \end{equation*}
This together with the dominated convergence theorem implies
\eqref{eq:e-Pk->e-P}. Thus, we arrive at~\eqref{eq:e-tPk->e-P}.
On the other hand, we have $P_k \Rightarrow P$, and from
\eqref{eq:Ptilde-k-def} we have $\tP_k \Rightarrow P$. Therefore 
Lemma~\ref{lem:Pn-conv-P--ugwt-Pn-conv-ugwt-P} in
Appendix~\ref{sec:UGWT-convg} implies that 
$\widetilde{\mu}^{(k)}
\Rightarrow \mu$ as $k \rightarrow \infty$. 
Therefore, from Lemma~\ref{lem:BC-ent-upper-semicontinouous} and the lower
  bound for the finite support case, we have
  \begin{equation}
    \label{eq:LB-bchunder--liminf-Jh-tPk}
    \bchunder_{\vd, Q}(\mu)\condmnun \geq \limsup_{k \rightarrow \infty} \bchunder_{\vd, Q}(\ugwt_h(\tP_k))\condmnun \geq \limsup_{k \rightarrow \infty} J_h(\tP_k) \geq \liminf_{k \rightarrow \infty} J_h(\tP_k).
  \end{equation}
  Here, all the entropy terms are obtained via the same sequences $\vmn$ and $\vun$. 
  Therefore, it suffices to show that $\liminf_{k \rightarrow \infty}
  J_h(\tP_k) \geq J_h(P)$.
  Note that, by definition, we have
  \begin{equation*}
    J_h(\tP_k) = -s(d) + H(\tP_k) - \frac{d}{2} H(\pi_{\tP_k}) - \sum_{t,t' \in \edgemark \times \mTb_*^{h-1}} \evwrt{\tP_k}{\log E_h(t,t')!},
  \end{equation*}
  where $d = \deg(\mu)$.
  We claim that
  \begin{equation}
    \label{eq:claim-liminf-Jh--tPk--Jh--Pk}
    \liminf_{k \rightarrow \infty} J_h(\tP_k) \geq \liminf_{k \rightarrow \infty} J_h(P_k).
  \end{equation}
Note that $P_k$ is
    admissible and is finitely supported, and hence $H(P_k) < \infty$. Furthermore,  since $d>0$, for $k$
    large enough $P_k$ has positive expected degree at the root. Hence
    $J_h(P_k)$ is well defined for $k$ large enough.  
In order to show~\eqref{eq:claim-liminf-Jh--tPk--Jh--Pk}, first note that if $d_k$ is the average degree at the
root in $P_k$ then we have $d_k \rightarrow d$ as $k \rightarrow \infty$. Hence we have
\begin{equation}
  \label{eq:s-dk--sd}
  \lim_{k \rightarrow \infty} s(d_k) = s(d).
\end{equation}
On the other hand, using~\eqref{eq:Ptilde-k-def}, we
have
\begin{equation*}
  H(\tP_k) = p_k \log \frac{1}{p_k} + \sum_{x \leq x' \in \edgemark} \alpha^k_{x,x'} \log \frac{1}{\alpha^k_{x,x'}} + p_k H(P_k) + \sum_{x \leq x' \in \edgemark} \alpha^k_{x,x'} H(U_{x,x'}).
\end{equation*}
Here, $U_{x,x'}$ are fixed distributions and have no dependence on $k$. Also,
$p_k \rightarrow 1$ and $\alpha^k_{x,x'} \rightarrow 0$ for all $x \leq x' \in
\edgemark$. Hence, if we show that the sequence $H(P_k)$ is bounded, we can
conclude that $\liminf_{k \rightarrow \infty} H(\tP_k) \geq \liminf_{k
  \rightarrow \infty} H(P_k)$. In order to show that the sequence $H(P_k)$ is
bounded, recall that $P_k$ is the distribution of $[T^{(k)},o]_h$. Observe that
$[T^{(k)}, o]_h$ is a function of $[T, o]_{h+1}$. The reason is that, by
definition, $T^{(k)}$ is obtained from $T$ by removing all the edges connected
to vertices with degree more than  $k$, and the degree of a vertex with distance
at most $h$ from the root is completely determined by $[T,o]_{h+1}$. This means
that $H(P_k) \leq H(R)$ where $R := \mu_{h+1} \in \mP(\mTb_*^{h+1})$ is the law
of the $h+1$ neighborhood of the root in $\mu$. From
Lemma~\ref{lem:PPh--PtildePh+1}, we have $R \in \mP_{h+1}$ and hence $H(R) <
\infty$. This shows that $H(P_k)$ is a bounded sequence and
\begin{equation}
  \label{eq:liminf-H-tPk--liminf-H-Pk}
  \liminf_{k \rightarrow \infty} H(\tP_k) \geq \liminf_{k
  \rightarrow \infty} H(P_k).
\end{equation}
%Furthermore, from \eqref{eq:Ptilde-k-def}, we have
%\begin{equation*}
%  e_{\tP_k}(t,t') = p_k e_{P_k}(t,t') + \sum_{x \leq x' \in \edgemark} \alpha^k_{x,x'} e_{U_{x,x'}}(t,t') \qquad \forall t,t' \in \edgemark \times \mTb_{*}^{h-1}.
%\end{equation*}
%In particular
On the other hand, from \eqref{eq:etPk-ePk-eUxx}, we have 
\begin{equation*}
\pi_{\tP_k} = \frac{d_k}{d} p_k \pi_{P_k} + \sum_{x \leq x' \in \edgemark} \frac{\tilde{d}_{x,x'}}{d} \alpha^k_{x,x'} \pi_{U_{x,x}}.
\end{equation*}
But, as $k \rightarrow \infty$, we have $d_k \rightarrow d$, $p_k \rightarrow
1$, and $\alpha^k_{x,x'} \rightarrow 0$ for all $x \leq x' \in \edgemark$. Also,
$U_{x,x'}$ for $x \leq x' \in \edgemark$ are fixed and do not  depend on $k$.
Thereby, we
conclude that
\begin{equation}
  \label{eq:limsup-H-pi_tPk--limsup-H-pi_Pk}
 \limsup_{k \rightarrow \infty} H(\pi_{\tP_k}) \leq \limsup_{k
  \rightarrow \infty} H(\pi_{P_k}). 
\end{equation}
Moreover, from \eqref{eq:Ptilde-k-def}, we have
\begin{align*}
  \evwrt{\tP_k}{\sum_{t,t' \in \edgemark \times \mTb_*^{h-1}}\log E_h(t,t')!} &= p_k  \evwrt{P_k}{\sum_{t,t' \in \edgemark \times \mTb_*^{h-1}}\log E_h(t,t')!} \\
  &\quad + \sum_{x \leq x' \in \edgemark}  \alpha^k_{x,x'} \evwrt{U_{x,x'}}{\sum_{t,t' \in \edgemark \times \mTb_*^{h-1}}\log E_h(t,t')!}.
\end{align*}
Again, as $k \rightarrow \infty$, we have $p_k \rightarrow 1$ and $\alpha^k_{x,x'}
\rightarrow 0$ for all $x \leq x' \in \edgemark$. Hence
\begin{equation}
  \label{eq:ilmsup-sum-t-t'--tPk-Pk}
  \limsup_{k \rightarrow \infty} \evwrt{\tP_k}{\sum_{t,t' \in \edgemark \times \mTb_*^{h-1}}\log E_h(t,t')!} \leq \limsup_{k \rightarrow \infty}  \evwrt{P_k}{\sum_{t,t' \in \edgemark \times \mTb_*^{h-1}}\log E_h(t,t')!}.
\end{equation}
Putting together \eqref{eq:s-dk--sd}, \eqref{eq:liminf-H-tPk--liminf-H-Pk},
\eqref{eq:limsup-H-pi_tPk--limsup-H-pi_Pk} and
\eqref{eq:ilmsup-sum-t-t'--tPk-Pk}, we arrive at
\eqref{eq:claim-liminf-Jh--tPk--Jh--Pk}. Comparing this
with~\eqref{eq:LB-bchunder--liminf-Jh-tPk}, in order to complete the proof, it suffices to show that 
  \begin{equation}
    \label{eq:liminf-Jh-Pk--Jh-P}
    \liminf_{k \rightarrow \infty} J_h(P_k) \geq J_h(P).
  \end{equation}
Without loss of generality, for the rest of the proof, we may assume that $J_h(P) > - \infty$, otherwise nothing remains to be proved. 
In order to show~\eqref{eq:liminf-Jh-Pk--Jh-P}, it suffices to show the following
  \begin{subequations}
    \begin{align}
      \liminf_{k \rightarrow \infty} H(P_k) &\geq H(P), \label{eq:LB-g-HPk-HP}\\
      \lim_{k \rightarrow \infty} \sum_{t, t' \in \edgemark \times \mTb_*^{h-1}} \evwrt{P_k}{\log E_h(t,t')!} &= \sum_{t, t' \in \edgemark \times \mTb_*^{h-1}} \evwrt{P}{\log E_h(t,t')!}, \label{eq:LB-g-Eh} \\
      \limsup_{k \rightarrow \infty} H(\pi_{P_k}) &\leq H(\pi_P). \label{eq:LB-g-H-pi-k-pi}
    \end{align}
  \end{subequations}

First, to show~\eqref{eq:LB-g-HPk-HP}, note that for all $[\tilde{T},
\tilde{o}] \in \mTb_*^h$, we have $P_k([\tilde{T},\tilde{o}]) =
\int \mathbbm{1}[[T^{(k)}, o]_h = [\tilde{T}, \tilde{o}]] d \mu([T,o])$.
Therefore, the dominated convergence theorem implies that $P_k([\tilde{T},
\tilde{o}]) \rightarrow P([\tilde{T}, \tilde{o}])$ as $k \rightarrow \infty$.
Hence, \eqref{eq:LB-g-HPk-HP} follows from this and lower semi--continuity
of the Shannon entropy (see, for instance \cite{ho2010interplay}).
% this basically follows by Fatou's lemma:
% in general, if p_n is a sequence of discrete distributions so that p_n(x) ->
% p(x), then we have using Fatou's lemma:
% \liminf H(p_n) = \liminf \sum p_n(x) \log 1 / p_n(x) \geq \sum p(x) \log 1 /
% p(x) = H(P). 

Now we turn to showing \eqref{eq:LB-g-Eh}. Define $\mC := (\edgemark \times
  \mTb_*^{h-1}) \times (\edgemark \times \mTb_*^{h-1})$. Moreover, for $r \in
  \mTb_*$, let $F(r) := \sum_{c \in \mC} \log E_h(c)(r)!$. With this,
  we have $\sum_{t, t' \in \edgemark \times \mTb_*^{h-1}}
  \evwrt{P_k}{\log E_h(t,t')!} = \evwrt{\mu}{F([T^{(k)}, o])}$. Recall that 
  $[T^{(k)}, o] \in \mTb_*$, as defined above, is the rooted tree
  obtained from $[T, o]$ by removing all edges connected to vertices
  with degree larger than $k$ followed by taking the connected component
  of the root. Likewise, the right hand side of \eqref{eq:LB-g-Eh} is precisely
  $\evwrt{\mu}{F([T,o])}$. Observe that for each $[T,o] \in \mTb_*$, if $k$ is large enough,
  $[T^{(k)},o]_h = [T, o]_h$. Thereby, $F([T^{(k)}, o]) \rightarrow
  F([T,o])$ pointwise. Now, for $[T, o] \in \mTb_*$, using the inequality $\log a! \leq a \log
  a$ that holds for any nonnegative integer $a$ by interpreting $0 \log
  0 = 0$, we get
  \begin{align*}
    F([T, o]) &\leq \sum_{c \in \mC} E_h(c)(T, o) \log E_h(c)(T, o) \\
                 &\leq \sum_{c \in \mC} E_h(c)(T, o) \log \deg_{T}(o) \\
                 &= \deg_T(o) \log \deg_T(o).
  \end{align*}
Consequently,
\begin{equation*}
  \evwrt{P}{|F([T^{(k)}, o])|} = \evwrt{P}{F([T^{(k)}, o])} \leq \evwrt{P}{\deg_{T^{(k)}}(o) \log \deg_{T^{(k)}}(o)} \leq \evwrt{P}{\deg_{T}(o) \log \deg_{T}(o)}.
\end{equation*}
The fact that $P \in \mP_h$ implies that the right hand side is finite.
Thereby, we arrive at  \eqref{eq:LB-g-Eh} using the dominated convergence
theorem.

Next, we show \eqref{eq:LB-g-H-pi-k-pi}. Recall that, without
loss of generality, we have assumed  that $J_h(P) > -\infty$, which means
$H(\pi_P)< \infty$. Consequently, we have
  \begin{equation}
    \label{eq:sum-epc-log-epc-infty}
    \begin{aligned}
      \sum_{c \in \mC} |e_P(c) \log e_P(c) | &\leq \sum_{c \in \mC} |e_P(c) \log e_P(c) - e_P(c) \log d| + |e_P(c)\log d| \\
      &= \sum_{c \in \mC} e_P(c) \log \frac{d}{e_P(c)} + \sum_{c \in \mC} e_P(c) \log d \\
      &= d H(\pi_P) + d \log d < \infty,
    \end{aligned}
  \end{equation}
where, in the second line, we have used the fact that $e_P(c) \leq d$ for all $c
\in \mC$. Therefore, the sequence $e_P(c) \log e_P(c)$ is absolutely summable.
Hence, we may write 
\begin{equation*}
  H(\pi_P) = \sum_{c \in \mC} \frac{e_P(c)}{d} \log \frac{d}{e_P(c)} = \log d - \frac{1}{d} \sum_{c \in \mC} e_P(c) \log e_P(c). 
\end{equation*}
% I just wanted to make sure that the infinite sum $\sum_{c \in \mC} e_P(c)
% \log e_P(c)$ is absolutely convergent so that I can manipulate the sum in
% $H(\pi_P)$.
On the other hand, $P_k$ has finite support. Hence, with  $d_k = \sum_{c \in \mC}
e_{P_k}(c)$ being the  expected degree at the root in $P_k$, we have
\begin{equation*}
  H(\pi_{P_k}) = \log d_k - \frac{1}{d_k} \sum_{c \in \mC} e_{P_k}(c) \log e_{P_k}(c). 
\end{equation*}
Therefore, as $d_k \uparrow d$, in order to show \eqref{eq:LB-g-H-pi-k-pi}, 
it suffices to show that
  \begin{equation}
    \label{eq:LB-g-epk-ep}
    \liminf_{k \rightarrow \infty} \sum_{c \in \mC} e_{P_k}(c) \log e_{P_k}(c) \geq \sum_{c \in \mC} e_P(c) \log e_P(c).
  \end{equation}
  % to see why it is sufficient to show this, see the note
  % 2018-10-11_PN_ePk-ep-dk-liminf.pdf
  %
  %
  %
% Observe that for $c \in \mC$, we have $e_{P_k}(c) =
%   \evwrt{\mu}{E_h(c)([T^{(k)}, o])}$. But for $[T, o] \in \mTb_*$, if $k$
%   is large enough, $[T^{(k)}, o]_h = [T, o]_h$. Thereby, $E_h(c)([T^{(k)},
%   o]) \rightarrow E_h(c)([T, o])$ as $k \rightarrow \infty$. On the
%   other hand, $E_h(c)([T^{(k)},o]) \leq \deg_{T^{(k)}}(o) \leq
%   \deg_T(o)$. Hence,
%   \begin{equation*}
%     \evwrt{\mu}{|E_h(c)([T^{(k)}, o])|} = \evwrt{\mu}{E_h(c)([T^{(k)}, o])} \leq \evwrt{\mu}{\deg_T(o)} < \infty.
%   \end{equation*}
% This together with the dominated convergence theorem implies that 
%   \begin{equation}
%     \label{eq:ePkc-->ePc}
%     \lim_{k \rightarrow \infty}e_{P_k}(c) =  e_P(c) \qquad \forall c \in \mC.
%   \end{equation}
  Recall from \eqref{eq:e-Pk->e-P} that for all $c \in \mC$, we have $e_{P_k}(c)
  \rightarrow e_P(c)$ as $k \rightarrow \infty$.
Now, for a nonnegative integer $\delta$, define $\mCdelta \subset \mC$ to be the set
  of $(t,t') \in \mC$ such that all vertices in the subgraph components of $t$
  and $t'$, i.e.\ $t[s]$ and $t'[s]$, have
  degrees bounded by $\delta$.
Therefore, due to   \eqref{eq:sum-epc-log-epc-infty} and the fact that 
$\mC = \cup_{\delta = 1}^\infty \mCdelta$, we have 
\begin{equation}
  \label{eq:sum-ep-log-ep-Cdelta}
  \sum_{c \notin \mCdelta} |e_P(c) \log e_P(c)| < \epsilon_1(\delta),
\end{equation}
where $\epsilon_1(\delta) \rightarrow 0$ as $\delta \rightarrow
\infty$.
Note that $\mCdelta$ is finite. This together with \eqref{eq:e-Pk->e-P} and
  \eqref{eq:sum-ep-log-ep-Cdelta} implies that for $\delta > 0$ we have 
\begin{equation}
    \label{eq:lim-CdeltaePk--eP-epsilon1}
    \begin{aligned}
      \lim_{k \rightarrow \infty} \sum_{c \in \mCdelta} e_{P_k}(c) \log e_{P_k} (c) &= \sum_{c \in \mCdelta} e_P(c) \log e_P(c)  \\
      &\geq \sum_{c \in \mC} e_P(c) \log e_P(c) - \epsilon_1(\delta).
    \end{aligned}
  \end{equation}
Hence, we may write 
  \begin{align*}
    \liminf_{k \rightarrow \infty} \sum_{c \in \mC} e_{P_k}(c) \log e_{P_k}(c) &\geq \liminf_{k \rightarrow \infty} \sum_{c \in \mCdelta} e_{P_k}(c) \log e_{P_k}(c)  + \liminf_{k \rightarrow \infty} \sum_{c \notin \mCdelta} e_{P_k}(c) \log e_{P_k}(c) \\
&\geq \sum_{c \in \mC} e_P(c) \log e_P(c) - \epsilon_1(\delta) +  \liminf_{k \rightarrow \infty} \sum_{c \notin \mCdelta} e_{P_k}(c) \log e_{P_k}(c).
  \end{align*}
As this holds for all $\delta > 0$ and since $\epsilon_1(\delta) \rightarrow 0$ when
$\delta \rightarrow 0$, in order to show \eqref{eq:LB-g-epk-ep} it suffices to prove that
for all positive integers $k$ and $\delta$ such that $k > \delta$, we have
\begin{equation}
  \label{eq:sum-ePk--log-ePk--epsilon2}
  \sum_{c \notin \mCdelta} e_{P_k}(c) \log e_{P_k}(c) \geq -\epsilon_2(\delta),
\end{equation}
where $\epsilon_2(\delta) \rightarrow 0$ as $\delta \rightarrow
\infty$.
Now, we fix positive integers $k > \delta$ and show that
\eqref{eq:sum-ePk--log-ePk--epsilon2} holds for an appropriate choice of
$\epsilon_2(\delta)$. For an integer $r > 0$, let 
$\mAr_h \subset \mTb_*^h$ be the set of marked rooted trees of depth at most $h$
where all degrees
are 
bounded by $r$.
We  define 
$\mAr_{h+1} \subset \mTb_*^{h+1}$ similarly.
Note that $P_k$ has a finite support and so the left hand side of \eqref{eq:sum-ePk--log-ePk--epsilon2} is a
finite sum.
Indeed, $P_k$ is supported on  the finite set $\mAk_h \subset \mTb_*^h$
and $e_{P_k}(c) = 0$ for $c \notin \mCk$. Consequently, we have 
\begin{equation}
  \label{eq:ePkc-log-ePkc---simp-1}
  \begin{aligned}
    \sum_{c \notin \mCdelta} e_{P_k}(c) \log e_{P_k}(c) &= \sum_{c \in \mCk \setminus \mCdelta} \left (\sum_{s \in \mAk_h} P_k(s) E_h(c)(s) \right ) \log \left (\sum_{s' \in \mAk_h} P_k(s') E_h(c)(s') \right ) \\
    &\geq \sum_{c \in \mCk \setminus \mCdelta} \sum_{s \in \mAk_h} P_k(s) E_h(c)(s) \log (P_k(s) E_h(c)(s)).
  \end{aligned}
\end{equation}
Note that if $c \notin \mCdelta$ and $s \in \mAdelta_h$, we have   $E_h(c)(s) =
0$. Thereby,
  \begin{equation}
    \label{eq:ePkc-log-ePkc--simp-2}
    \begin{aligned}
      \sum_{c \notin \mCdelta} e_{P_k}(c) \log e_{P_k}(c) &\geq \sum_{s \in \mAk_h \setminus \mAdelta_h} \sum_{c \in \mCk \setminus \mCdelta}  P_k(s) E_h(c)(s) \log (P_k(s) E_h(c)(s)) \\
      &\geq \sum_{s \in \mAk_h \setminus \mAdelta_h} \sum_{c \in \mCk \setminus \mCdelta}  P_k(s) E_h(c)(s) \log P_k(s) \\
      &\geq \sum_{[T,o] \in \mAk_h \setminus \mAdelta_h}  \deg_T(o) P_k([T,o])  \log P_k([T,o])
    \end{aligned}
  \end{equation}
where the last inequality uses the fact that for $[T,o] \in \mAk_h \setminus
\mAdelta_h$, we have
\begin{equation*}
 \sum_{c \in \mCk \setminus \mCdelta}
E_h(c)([T,o]) \leq \sum_{c \in \mC} E_h(c)([T,o]) = \deg_T(o), 
\end{equation*}
 and $P_k([T,o]) \log P_k([T,o])
\leq 0$. %Recall that $\deg_s(0)$ denotes the degree at the root in $s$.
As we discussed above, for $[T,o] \in \mTb_*$, $[T^{(k)}, o]_h$ is determined by
$[T,o]_{h+1}$, since the degree of all vertices up to depth $h$ is determined by
the structure of the tree up to depth $h+1$.  Moreover, define $F_k :
  \mTb_*^{h+1} \rightarrow \mAk_{h}$ such that  for $[T, o] \in
  \mTb_*^{h+1}$, we have  $F_k([T,o]) := [T^{(k)},o]_h$.
With this,  $P_k$ is the pushforward of
  $R := \mu_{h+1}$ under the mapping $F_k$, i.e.\ for $[T,o] \in \mAk_h$, $P_k([T,o]) =
  R(F_k^{-1}([T,o]))$.
On the other hand, for $[T,o] \in \mAk_h$, if $[T',o'] \in F_k^{-1}([T,o])$, then $R([T',o']) \leq
  R(F_k^{-1}([T,o])) = P_k([T,o])$. Using these in 
  \eqref{eq:ePkc-log-ePkc--simp-2}, we have 
  \begin{align*}
    \sum_{c \notin \mCdelta} e_{P_k}(c) \log e_{P_k}(c) &\geq \sum_{[T,o] \in \mAk_h \setminus \mAdelta_h} \sum_{[T',o'] \in F_k^{-1}([T,o])}  \deg_T(o) R([T',o']) \log P_k([T,o]) \\
                                                        &\geq \sum_{[T,o] \in \mAk_h \setminus \mAdelta_h} \sum_{[T',o'] \in F_k^{-1}([T,o])} \deg_T(o) R([T',o']) \log R([T',o'])\\
                                                        &= \sum_{[T',o'] \in \mTb_*^{h+1}} \one{F_k([T',o']) \notin \mAdelta_h} \deg_{F_k([T',o'])}(o') R([T',o']) \log R([T',o']).
  \end{align*}
Here, in the last equality, we were allowed to change the order of
  summations since all the terms are nonpositive. Also note that, by definition,
  for  $[T',o'] \in \mTb_*^{h+1}$
  we have $F_k([T',o']) \in \mAk_h$. 
  Since the mapping $F_k$ decreases the degree of all the vertices, for all $[T',o']
  \in \mTb_*^{h+1}$ we have  $\one{F_k([T',o'])
    \notin \mAdelta_h} \leq \one{[T',o']\notin \mAdelta_{h+1}}$ and
  $\deg_{F_k([T',o'])}(o') \leq \deg_{[T',o']} (o')$.
Using these observations in the above chain of inequalities, since all the terms in the
summation are  nonpositive, we get 
  \begin{equation}
    \label{eq:ePkc-log-ePkc--simp-3}
    \sum_{c \notin \mCdelta} e_{P_k}(c) \log e_{P_k}(c) \geq  \sum_{[T',o'] \in \mTb_*^{h+1} \setminus \mAdelta_{h+1}} \deg_{T'}(o') R([T',o']) \log R([T',o']).
  \end{equation}
Since $P$ is strongly admissible, i.e.\ $P \in \mP_h$, Lemma~\ref{lem:PPh--PtildePh+1} implies that  $R
\in \mP_{h+1}$, which means $H(R) < \infty$. Also,
  $\evwrt{R}{\deg_T(o) \log \deg_T(o)} = \evwrt{P}{\deg_T(o) \log
    \deg_T(o)} < \infty$. Therefore, from
  Lemma~\ref{lem:pllogl--pllogp}, we have
  \begin{equation*}
   \sum_{[T',o'] \in \mTb_*^{h+1}}
  \deg_{T'}(o') R([T',o']) \log R([T',o']) > -\infty. 
  \end{equation*}
 Since $\cup_\delta
  \mAdelta_{h+1} = \mTb_*^{h+1}$, we have
  \begin{equation*}
    \sum_{[T',o'] \in \mTb_*^{h+1} \setminus \mAdelta_{h+1}} \deg_{T'}(o') R([T',o']) \log R([T',o']) \geq -\epsilon_2(\delta),
  \end{equation*}
  where $\epsilon_2(\delta) \rightarrow 0$ as $\delta \rightarrow
  \infty$.
  Putting this into \eqref{eq:ePkc-log-ePkc--simp-3}, we arrive at 
 \eqref{eq:sum-ePk--log-ePk--epsilon2}, which completes the proof.
\end{proof}

\subsection{Upper bound}
\label{sec:upperbound}

In this section, we prove the upper bound result of
Proposition~\ref{prop:upper-bound}.

\begin{proof}[Proof of Proposition~\ref{prop:upper-bound}]
  Let $P:= \mu_h$. Since $\mu$ is unimodular and $d < \infty$, 
  from Lemma~\ref{lem:unimodular-is-admissible} we see that $P$ is
  admissible. Further, since $P \in \mP(\mTb_*)$ with $\evwrt{P}{\deg_T(o)} = d >0 $ and $H(P) <
  \infty$, we see that $J_h(P)$, as introduced in \eqref{eq:Jh-def}, is 
  well-defined.
  From the local topology on $\mGb_*$ 
  %for all $h \in \nats$ 
  one sees that,
  for all $h \ge 1$ 
  and
  $\epsilon>0$, there exists $\eta_1(\epsilon)$ such that, for all $\rho_1,
  \rho_2 \in \mP(\mGb_*)$,  $\dlp(\rho_1, \rho_2) < \epsilon$ implies
  $\dtv((\rho_1)_h, (\rho_2)_h) < \eta_1(\epsilon)$. Here the function $\eta_1(.)$
  depends only on $h$ and has the property that $\eta_1(\epsilon) \rightarrow 0$ as
  $\epsilon \rightarrow 0$.
  % see proof note 2018-04-13_Proof-note_Levy-Prokhorov_Total-variation.pdf
  % for details. The note proves the claim for $\mP(\mTb_*)$, but by
  % substituting $\mTb_*$ with $\mGb_*$, the claim holds also for $\mP(\mGb_*)$.
  Therefore, if for $\delta > 0$ we define
  \begin{equation*}
    \Anmnun(P, \delta) := \{ G \in \mGnmnun: \dtv(U(G)_h, P) < \delta\},
  \end{equation*}
  we have
  \begin{equation*}
    \mGnmnun(\mu, \epsilon) \subseteq \Anmnun(P, \eta_1(\epsilon)).
  \end{equation*}
  Hence, to show \eqref{eq:upper-bound-Ph-statement}, it suffices to
    show that
    \begin{equation}
      \label{eq:upperbound-A-claim}
      \lim_{\epsilon\rightarrow 0} \limsup_{n \rightarrow \infty} \frac{1}{n} \left ( \log |\Anmnun(P, \epsilon)| - \snorm{\vmn}_1 \log n \right ) \leq J_h(P).
    \end{equation}

    In order to do this, fix a finite subset $\Delta \subset \mTb_*^h$ and define
    $\mF(\Delta) \subset \edgemark \times \mTb_*^{h-1}$ to be the set of  $T[o,v]_{h-1}$ and $T[v,o]_{h-1}$ for $[T, o] \in \Delta$
    and $v \sim_T o$. Since $\Delta$ is finite, $\mF(\Delta)$ is also
    finite and can be identified with the set of integers $\{1,
    \dots, L\}$ where $L =L(\Delta) := |\mF(\Delta)|$. With this, define the
    color set $\mC(\Delta) := \mF(\Delta) \times \mF(\Delta)$. Furthermore, 
    %let $\bar{\mF}(\Delta) := \mF \cup \{\star_x: x \in \edgemark\}$, 
    let $\bar{\mF}(\Delta) := \mF(\Delta) \cup \{\star_x: x \in \edgemark\}$, 
    where
    $\star_x$ for $ x \in \edgemark$ are additional distinct elements  not present in
    $\mF(\Delta)$. 
    %\color{red}
    %Note for Payam:
    %Changed $\mF$ to $\mF(\Delta)$ in the preceding equation. Check.
    %\color{black}
    %\pres{correct}
    Note that $\bar{\mF}(\Delta)$ is finite, thus can be
    identified with the set of integers $\{1,
    \dots, \bar{L}\}$ where $\bar{L} = \bar{L}(\Delta) =  L(\Delta) +
    |\edgemark|$, where the first $L$ elements represent $\mF(\Delta)$.
    Finally, extend the color set $\mC(\Delta)$ to
    $\bar{\mC}(\Delta) := \bar{\mF}(\Delta) \times \bar{\mF}(\Delta)$.

  Now, for a fixed $\Delta$ as above, given a simple marked graph $G$ on the
  vertex set $[n]$, we construct a simple
  directed colored
    graph $\tG \in \mG(\bar{\mC}(\Delta))$ on the same vertex set $[n]$, with color set $\bar{\mC}(\Delta)$, as follows. For each edge between vertices $u$ and $v$ in $G$, if $\etype_G^h(u,v) \in
      \mC(\Delta)$, we place an edge in $\tG$ directed from $u$ towards $v$ with color
      $\etype_G^h(u,v)$, and another edge directed from $v$ towards $u$ with color $\etype_G^h(v,u)$.
      Otherwise, if $\etype_G^h(u,v) \notin \mC(\Delta)$, we place an edge in $\tG$
      directed from $u$ towards $v$ with
      color $(\star_x, \star_{x'})$ and an edge directed from $v$ towards $u$ with color
      $(\star_{x'}, \star_x)$, where $x = \xi_G(v,u)$ and $x' = \xi_G(u,v)$.
      Let $\vD^{\tG}$ be the colored degree sequence of $\tG$. More precisely,
    for
    a vertex $v$, $D^{\tG}(v) \in \mM_{\bar{L}(\Delta)}$ has the following form.
    For $c \in \mC(\Delta)$, $D^{\tG}_c(v)$ is the number of
    vertices $w \sim_G v$ such that $\etype_G^h(v,w) = c$. Moreover, for $x, x'
    \in \edgemark$, $D^{\tG}_{(\star_x, \star_{x'})}(v)$ is the number
    of vertices $w \sim_G v$ such that $\etype_G^h(v,w) \notin
    \mC(\Delta)$, $\xi_G(w,v) = x$ and $\xi_G(v,w) = x'$. Note that, for $x \in
    \edgemark$ and $t \in \mF(\Delta)$, we have 
    $D^{\tG}_{(\star_x, t)}(v) = D^{\tG}_{(t, \star_x)}(v) = 0$.
    
With an abuse of notation, for a marked rooted graph $(G,o)$ on a finite or
    countable vertex set, and $c \in \bar{\mC}(\Delta)$ of the form $c =
    (\star_x, \star_{x'})$, $x, x' \in \edgemark$, we define
    \begin{equation}
      \label{eq:E_h-star-a}
      %\begin{gathered}
        E_h(\star_x, \star_{x'})(G, o) = |\{ v \sim_G o: \etype_G^h(o,v) \notin \mC(\Delta), \xi_G(v,o) = x,\xi_G(o,v)= x'\}|,
        \end{equation}
        and, for $x \in \edgemark$, we define
        \begin{equation}
        \label{eq:E_h-star-b}
        E_h(\star_x, t)(G, o) = E_h(t, \star_x)(G,o) = 0, \qquad \forall t \in \mF(\Delta).
      %\end{gathered}
    \end{equation}
With this convention, define the map $F_\Delta : \mGb_*^h \rightarrow \vermark \times
    \mM_{\bar{L}(\Delta)}$, such that for $[G,o] \in \mGb_*^h$, $F([G,o]) := (\theta, D)$ where $\theta = \tau_G(o)$ is the
    mark at the root in $G$, and for  $c \in \bar{\mC}(\Delta)$, $D_c =
    E_h(c)(G,o)$.
Moreover, let $\bar{P}^\Delta \in \mP(\vermark \times \mM_{\bar{L}(\Delta)})$ be
    the law of $F_\Delta([T,o])$ when $[T, o] \sim P$.
Note that if $G$ is a marked graph on the vertex set $[n]$, with the
directed colored graph $\tG$ defined above then for all vertices $v$ in $G$,
we have $(\tau_G(v),
D^{\tG}(v)) = F_\Delta([G,v]_h)$.
Therefore, if $G \in \Anmnun(P, \epsilon)$, since $\dtv(U(G)_h, P) < \epsilon$, we
    have
    \begin{equation}
      \label{eq:dtv-tau-G-D-G-PbarDelta-epsilon}
      \dtv \left ( \frac{1}{n} \sum_{i=1}^n \delta_{(\tau_G(v), D^{\tG}(v))}, \bar{P}^\Delta \right ) < \epsilon.
    \end{equation}
Let $\Bnmnun(P, \Delta, \epsilon)$ be the set of pairs of sequences
$(\vbeta, \vD)$, $\vbeta =
(\beta(i): 1 \leq i \leq n)$ and $\vD = (D(i): 1 \leq i \leq n) \in \mD_n$ where, for
$1 \leq i \leq n$, $\beta(i) \in \vermark$, $D(i) \in
\mM_{\bar{L}(\Delta)}$ are such that with
\begin{equation}
  \label{eq:R-beta-D}
  R(\vbeta, \vD) := \frac{1}{n} \sum_{i=1}^n \delta_{(\beta(i), D(i))},
\end{equation}
we have
\begin{subequations}
  \begin{gather}
    \dtv(R(\vbeta, \vD), \bar{P}^\Delta) < \epsilon, \label{eq:Bn-prop-TV}
    \\
    \sum_{i=1}^n \sum_{c \in \bar{\mC}(\Delta)} D_c(i) = 2
    \snorm{\vmn}_1, \label{eq:Bn-prop-mn} \\
    \sum_{i=1}^n \one{\beta(i) = \theta}
    = \un(\theta), \quad \forall \theta \in \vermark, \label{eq:Bn-prop-theta}
    \\
    D_{(\star_x, t)}(v) = D_{(t, \star_x)}(v) = 0, \quad \forall  x \in
    \edgemark, t \in \mF(\Delta), i \in [n]. \label{eq:Bn-prop-x-t-0}
  \end{gather}
\end{subequations}
Let $\vtau_G$
%, i.e.\ $\vtau_G := (\tau_G(v): v \in V(G))$. 
denote $(\tau_G(v): v \in V(G))$.
\glsadd{not:vtauG}
Note that, from \eqref{eq:dtv-tau-G-D-G-PbarDelta-epsilon}, for $G
    \in \Anmnun(P,\epsilon)$ we have $(\vtau_G, \vD^{\tG}) \in \Bnmnun(P,
    \Delta, \epsilon)$.
Now, we claim that for $(\vbeta, \vD) \in \Bnmnun(P, \Delta, \epsilon)$ and a colored
    directed graph $H \in \mG(\vD, 2)$, there is at most one marked graph
    $G \in \Anmnun(P, \epsilon)$ such that $\vtau_G = \vbeta$ and $\tG = H$.
 The reason is that, to start with, the condition $\vtau_G = \vbeta$ uniquely determines vertex
 marks for $G$. Moreover, since $\tG = H$, vertices $v$ and $w$ are adjacent
    in $G$ iff they are adjacent in $H$. Let $v$ and $w$ be adjacent vertices in
    $H$ with $c \in \bar{\mC}(\Delta)$ being the color of the edge directed from
    $v$ towards $w$. Note that, by the definition of the set $\Bnmnun(P, \Delta, \epsilon)$,
    $c$ is either of the form $(t, t')$ where $t, t' \in \mF(\Delta)$, or
    $c = (\star_{x_1}, \star_{x_2})$ for some $x_1, x_2 \in \edgemark$.
    Since $\tG = H$, in the former case, we have $\xi_G(w,v) = t[m]$ and
    $\xi_G(v,w) = t'[m]$, while, in the latter case, we have  $\xi_G(w,v) = x_1$ and
    $\xi_G(v,w) = x_2$. This implies that there is at most one marked
    graph satisfying $\vtau_G = \vbeta$ and $\tG = H$.
 Consequently, we have 
 \begin{equation}
   \label{eq:Amnun-Bmn-max-G-vD}
   |\Anmnun(P, \epsilon)| \leq |\Bnmnun(P, \Delta, \epsilon)| \max_{(\vbeta, \vD) \in \Bnmnun(P, \Delta, \epsilon)} |\mG(\vD, 2)|.
 \end{equation}
 Now, we find bounds for the two terms on the right hand side
 of \eqref{eq:Amnun-Bmn-max-G-vD}.

 \underline{Bounding $|\Bnmnun(P, \Delta, \epsilon)|$}: we claim that
     \begin{equation}
      \label{eq:UB-Bn-HP-claim}
      \lim_{\epsilon \rightarrow 0} \limsup_{n \rightarrow \infty} \frac{1}{n} \log |\Bnmnun(P, \Delta, \epsilon)| \leq H(P).
    \end{equation}
Note that $P$ and $\bar{P}^\Delta$ are not necessarily finitely supported,
so this is not a direct consequence of \eqref{eq:Bn-prop-TV} and requires
some work. In order establish the claim, fix a finite subset $X \subset \vermark \times
\mM_{\bar{L}(\Delta)}$ of the form  $X = \{ (\beta^1, D^1), \dots,
(\beta^{|X|}, D^{|X|})\}$. With this, for $1 \leq j \leq |X|$, let $\bar{p}_j :=
    \bar{P}^\Delta(\beta^j, D^j)$. Furthermore, define $\bar{p}_0 := 1 - \sum_{j=1}^{|X|}
    \bar{p}_j = 1 - \bar{P}^\Delta(X)$.
 Now, fix $(\vbeta, \vD ) \in \Bnmnun(P, \Delta, \epsilon)$ and let $I_j:= \{ i
 \in [n]: (\beta(i), D(i)) = (\beta^j, D^j)\}$ for $1 \leq j \leq |X|$.
 Additionally, let $I_0 := \{ i \in [n]: (\beta(i), D(i)) \notin X\}$. Moreover,
 define $a_j := |I_j|/ n$ for $0 \leq j \leq |X|$.
 Then, because of \eqref{eq:Bn-prop-TV}, we have $|a_j - \bar{p}_j| <
 \epsilon$ for $0 \leq j \leq |X|$.
 Also, due to \eqref{eq:Bn-prop-mn}, we have
    \begin{equation}
      \label{eq:UB-sum-Dc-I0-ub}
      \begin{aligned}
        \sum_{i \in I_0} \sum_{c \in \bar{\mC}(\Delta)} D_c(i) &= 2 \snorm{\vmn}_1 - \sum_{j=1}^{|X|} \sum_{i \in I_j} \sum_{c \in \bar{\mC}(\Delta)} D_c(i) \\
        &= 2 \snorm{\vmn}_1  - \sum_{j=1}^{|X|} n a_j \sum_{c \in \bar{\mC}(\Delta)} D^j_c \\
        &\leq 2 \snorm{\vmn}_1 - n \sum_{j=1}^{|X|} \bar{p}_j  \sum_{c \in \bar{\mC}(\Delta)} D^j_c + n \epsilon \sum_{j=1}^{|X|} \sum_{c \in \bar{\mC}(\Delta)} D^j_c \\
        &= 2 \snorm{\vmn}_1 - n \bar{d}(X) + n \epsilon \alpha(X),
      \end{aligned}
    \end{equation}
    where
    \begin{equation*}
      \bar{d}(X) := \sum_{j=1}^{|X|} \bar{p}_j \sum_{c \in
      \bar{\mC}(\Delta)} D^j_c = \evwrt{\bar{P}^\Delta}{\one{(\beta,
        D) \in X}\sum_{c \in \bar{\mC}(\Delta)} D_c},
    \end{equation*}
    and
    \begin{equation*}
     \alpha(X) := \sum_{j=1}^{|X|} \sum_{c \in \bar{\mC}(\Delta)} D^j_c.
    \end{equation*}

Motivated by this, in order to  find an upper bound for $\Bnmnun(P, \Delta, \epsilon)$,
 we may first count the number of choices  for $I_0, \dots, I_{|X|}$, and then the number of
pairs $(\vbeta, \vD)$ consistent with each of them. For this, let $\Yn_\epsilon$ be the set of $(a_0, \dots, a_{|X|})$
    such that $\sum_{j=0}^n a_j = 1$ and such that for $0 \leq j \leq |X|$ we have $n a_j \in
    \integers_+$ and  $|a_j - \bar{p}_j| < \epsilon$. We can see that
    $|\Yn_\epsilon| \leq (2n\epsilon)^{1+|X|}$. Moreover,  given $(a_0, \dots, a_{|X|})$, there are $\binom{n}{na_0
      \dots na_{|X|}}$ many ways to chose a partition $I_0, \dots, I_{|X|}$ of
    $[n]$ such that $|I_j| = n a_j$, $0 \leq j \leq |X|$. 
Fixing such a partition, for  $i \in I_j$, $j \neq 0$, we must have $(\beta(i), D(i)) =
    (\beta^j, D^j)$. Hence, we only need to count the
    number of choices for $\{ (\beta(i), D(i)) :i \in I_0\}$.
Note that, there are at most $|\vermark|^{|I_0|} \leq
    |\vermark|^{n(\bar{p}_0 + \epsilon)}$ many ways to choose $\beta(i)$ for $i \in I_0$.
On the other hand, for $(D_c(i): i \in I_0, c \in
    \bar{\mC}(\Delta))$ there are $|I_0|{\bar{\mC}(\Delta)}= na_0 |\bar{\mC}(\Delta)|$ many
    nonnegative integers satisfying \eqref{eq:UB-sum-Dc-I0-ub}.
    Hence, there are at most
    \begin{equation*}
      %\max_{(a_0, \dots, a_{|X|}) \in \Yn_\epsilon}
      \binom{2\snorm{\vmn}_1 - n \bar{d}(X) + n \epsilon \alpha(X)  + n a_0 |\bar{\mC}(\Delta)|}{n a_0 |\bar{\mC}(\Delta)|}
    \end{equation*}
    many ways to choose $D(i)$ for $i \in I_0$.
    % this is just counting
    % the number of solutions for x_1 + \dots x_k \leq
    % N among nonnegative integers which is \binom{N+k}{k}.
 Putting all these together, we have
    \begin{equation}
      \label{eq:UB-log-Bn-ub-1}
      \begin{aligned}
        \log |\Bnmnun(P, \Delta, \epsilon)| &\leq (1+ |X|) \log (2n \epsilon) + \max_{(a_0, \dots, a_{|X|}) \in \Yn_\epsilon} \log \binom{n}{na_0  \dots n a_{|X|}} + n (\bar{p}_0 + \epsilon) \log |\vermark| \\
        &\quad + \max_{(a_0, \dots, a_{|X|}) \in \Yn_\epsilon} \log \binom{2\snorm{\vmn}_1 - n \bar{d}(X) + n \epsilon \alpha(X)  + n a_0 |\bar{\mC}(\Delta)|}{n a_0 |\bar{\mC}(\Delta)|}.
      \end{aligned}
    \end{equation}
 Furthermore, using Stirling's approximation, one can show that for
    $(a_0, \dots, a_{|X|}) \in \Yn_\epsilon$, we have
    \begin{equation}
      \label{eq:UB-log-n-c-aj-1}
      \begin{aligned}
        \log \binom{n}{na_0  \dots n a_{|X|}} &\leq 1 + \frac{1}{2} \log n - n \sum_{j=0}^{|X|} a_j \log a_j \\
        &\leq 1 + \frac{1}{2} \log n - n \sum_{j=0}^{|X|} \bar{p}_j \log \bar{p}_j + n \eta_2(\epsilon),
      \end{aligned}
    \end{equation}
    % for calculation details, see the note
    % 2018-09-04_PN_log-n-naj-upper-bound.pdf
    % also, notes around the time 2018-04-19 and 2018-04-16 are relevant 
    where, in the second inequality, $\eta_2(\epsilon)
    \rightarrow 0$ as $\epsilon \rightarrow 0$, and we have
    used the fact that $x \mapsto x \log x$ is uniformly continuous on $(0,1]$ 
    and also the assumption that $|a_j - \bar{p}_j| < \epsilon$
    for $0 \leq j \leq |X|$.
    Note that
    \begin{equation}
      \label{eq:UB-pbarj-H-Pbardelta}
      \begin{aligned}
        - \sum_{j=0}^{|X|} \bar{p}_j \log \bar{p}_j &= - \left ( \sum_{(\beta, D) \in X} \bar{P}^\Delta(\beta, D) \log \bar{P}^\Delta(\beta, D) \right )- (1 - \bar{P}^\Delta(X)) \log (1 - \bar{P}^\Delta(X)) \\
        &\leq H(\bar{P}^\Delta) \leq H(P),
      \end{aligned}
    \end{equation}
    where the last inequality follows from the fact that
    $\bar{P}^\Delta$ is the pushforward of $P$
    under $F_\Delta$.
Putting \eqref{eq:UB-pbarj-H-Pbardelta} back in
    \eqref{eq:UB-log-n-c-aj-1}, we get
    \begin{equation}
      \label{eq:UB-log-n-c-aj--HP}
      \lim_{\epsilon \rightarrow 0} \limsup_{n \rightarrow \infty} \frac{1}{n} \max_{(a_0, \dots, a_{|X|}) \in \Yn_\epsilon} \log \binom{n}{na_0 \dots na_{|X|}} \leq H(P).
    \end{equation}

On the other hand, 
%for $(a_0, \dots, a_{|X|}) \in Y^{(n)}_\epsilon$, 
using the general  inequality $\log \binom{r}{s} \leq \log
    2^r \leq r$, which holds for integers $r \geq s \geq
    0$, together with $a_0 \leq \bar{p}_0 + \epsilon$, we realize that for all $(a_0, \dots, a_{|X|}) \in
    \Yn_\epsilon$, we have
    \begin{equation*}
      \frac{1}{n} \log \binom{2\snorm{\vmn}_1 - n \bar{d}(X) + n \epsilon \alpha(X) + n a_0 |\bar{\mC}(\Delta)|}{n a_0 |\bar{\mC}(\Delta)|} \leq 2 \frac{\snorm{\vmn}_1}{n} - \bar{d}(X) +\epsilon \alpha(X) + (\bar{p}_0 + \epsilon) |\bar{\mC}(\Delta)|.
    \end{equation*}
 Note that, since the sequences $\vmn$ and $\vun$ are such that $(\vmn,\vun)$ is adapted to $(\vdeg(\mu), \vvtype(\mu))$, as $n \rightarrow \infty$ we have $\snorm{\vmn}_1 / n
    \rightarrow d/2$, where $d = \deg(\mu)$ is the
    average degree at the root in $\mu$. Therefore,
    \begin{eqnarray}
      \label{eq:UB-log-mn-bar-dbar-ub}
      && \lim_{\epsilon \rightarrow 0} \limsup_{n \rightarrow \infty} \frac{1}{n} \max_{(a_0, \dots, a_{|X|}) \in \Yn_\epsilon} \log \binom{2\snorm{\vmn}_1 - n \bar{d}(X) + n \epsilon \alpha(X)  + n a_0 |\bar{\mC}(\Delta)|}{n a_0 |\bar{\mC}(\Delta)|} \nonumber \\
      && ~~~~ \leq d - \bar{d}(X) + (1 - \bar{P}^\Delta(X)) |\bar{\mC}(\Delta)|.
    \end{eqnarray}   
 Using \eqref{eq:UB-log-n-c-aj--HP} and
    \eqref{eq:UB-log-mn-bar-dbar-ub} in
    \eqref{eq:UB-log-Bn-ub-1}, we get
    \begin{equation}
      \label{eq:UB-log-Bn-ub-2}
      \lim_{\epsilon \rightarrow 0} \limsup_{n \rightarrow \infty} \frac{1}{n} \log |\Bnmnun(P, \Delta, \epsilon)| \leq H(P) + d - \bar{d}(X) + ( 1 - \bar{P}^\Delta(X))( \log |\vermark| + |\bar{\mC}(\Delta)|).
    \end{equation}
Since this holds for any finite $X \subset
    \Theta \times \mM_{\bar{L}(\Delta)}$, we may take a nested sequence $X_k$
    converging to $\Theta \times \mM_{\bar{L}(\Delta) + 1}$ so that
    $\bar{P}^\Delta(X_k) \rightarrow 1$ and
    \begin{equation*}
      \bar{d}(X_k) = \evwrt{\bar{P}^\Delta}{\one{(\beta, D) \in X_k} \sum_{c \in \bar{\mC}(\Delta)} D_c} \rightarrow \evwrt{\bar{P}^\Delta}{ \sum_{c \in \bar{\mC}(\Delta)} D_c} = \evwrt{P}{\deg_T(o)} = \deg(\mu) = d.
    \end{equation*}
 % the convergence is basically a result of Monotone convergence theorem, since
 % $X_k$ are nested.
 Using this in \eqref{eq:UB-log-Bn-ub-2}, we arrive at 
 \eqref{eq:UB-Bn-HP-claim}.
 
\underline{Bounding $|\mG(\vD, 2)|$:} Now, we find an upper bound for the second
term in the right hand side of \eqref{eq:Amnun-Bmn-max-G-vD}. We claim that for
$(\vbeta, \vD) \in \Bnmnun(P, \Delta, \epsilon)$, we have
\begin{equation}
  \label{eq:log-mD-vD--factorial-bound}
  |\mG(\vD, 2)| \leq \frac{\prod_{c \in \bar{\mC}(\Delta)_<} \Sn_c(\vD)! \prod_{c \in \bar{\mC}(\Delta)_=} (\Sn_c(\vD) - 1)!!}{\prod_{c \in \bar{\mC}(\Delta)} \prod_{v=1}^n D_c(v)!},
\end{equation}
where $\Sn_c(\vD) = \sum_{v=1}^n D_c(v)$.
In order to show this, we take a simple directed colored graph $H \in \mG(\vD, 2)$
and construct $N:= \prod_{c \in \bar{\mC}(\Delta)} \prod_{v=1}^n D_c(v)!$ many
configurations in the space $\Sigma$. Recall from
Section~\ref{sec:color-conf-model-1} that $\Sigma$ is the set of all possible
matchings of half edges. Note that by the definition of the set $\Bnmnun(P, \Delta,
\epsilon)$, we have $\vD \in \mD_n$ and $\Sigma$ is well-defined.
For $v \in [n]$ and $c \in
\mC(\Delta)$, we consider all the possible
numberings of $D_c(v)$ many edges going out of the vertex $v$, which is $D_c(v)!$. It is
easy to see that since $H$ does not have loops and there is at most one directed
edge between each two pair of vertices, the $N$ many objects constructed in this way
are all distinct members of $\Sigma$. Also, for distinct simple directed colored
graphs $H \neq H' \in
\mG(\vD, 2)$, the $N$ many objects corresponding to $H$ are indeed distinct from the
$N$ many objects corresponding to $H'$.  Hence, $|\mG(\vD, 2)| N \leq |\Sigma|$.  But $|\Sigma|$ is
precisely $\prod_{c \in \bar{\mC}(\Delta)_<} \Sn_c(\vD)! \prod_{c \in
  \bar{\mC}(\Delta)_=} (\Sn_c(\vD) - 1)!!$. This establishes \eqref{eq:log-mD-vD--factorial-bound}.
Applying Stirling's approximation to \eqref{eq:log-mD-vD--factorial-bound}, in a manner
similar to what we did in \eqref{eq:LB-f-log-GDn-2h+1}, we get 
\begin{equation}
  \label{eq:UB-log-G-vD-2-expand}
  \begin{aligned}
    \log | \mG(\vD, 2) | &\leq \frac{1}{2} \sum_{c \in \bar{\mC}(\Delta)} \left (\Sn_c(\vD) \log \Sn_c(\vD) - \Sn_c(\vD)\right) - \sum_{v=1}^n \sum_{c \in \bar{\mC}(\Delta)} \log D_c(v)! + o(n) \\
    &= \frac{n}{2} \sum_{c \in \bar{\mC}(\Delta)} \left (\frac{\Sn_c(\vD)}{n} \log \frac{\Sn_c(\vD)}{n} - \frac{\Sn_c(\vD)}{n}\right) - \sum_{v=1}^n \sum_{c \in \bar{\mC}(\Delta)} \log D_c(v)! \\
    &\qquad + \frac{1}{2} \sum_{c \in \bar{\mC}(\Delta)} \Sn_c(\vD) \log n + o(n)  \\
    &= \frac{n}{2} \sum_{c \in \bar{\mC}(\Delta)} \left (\frac{\Sn_c(\vD)}{n} \log \frac{\Sn_c(\vD)}{n} - \frac{\Sn_c(\vD)}{n}\right) - \sum_{v=1}^n \sum_{c \in \bar{\mC}(\Delta)} \log D_c(v)! \\
    &\qquad + \snorm{\vmn}_1 \log n + o(n), 
    \end{aligned}
\end{equation}
% for calculation details, see the note
% 2018-09-05_PN_UB-Stirling-Sc-log-Sc.pdf
where the $o(n)$ term does not depend on $\vD$.
% this is because it is related to $\log \snorm{\vmn}_1 = O(\log n)$. 
Now, we claim that, for $c \in \bar{\mC}(\Delta)$,
    \begin{equation}
      \label{eq:UB-limsup-Sc-Ep-0}
      \limsup_{\epsilon\rightarrow 0} \limsup_{n \rightarrow \infty} \max_{(\vbeta, \vD) \in \Bnmnun(P, \Delta, \epsilon)} | \Sn_c(\vD) / n - \evwrt{\bar{P}^\Delta}{D_c} | = 0.
    \end{equation}
Note that for $c \in \bar{\mC}(\Delta)$, we have $\evwrt{\bar{P}^\Delta}{D_c} \leq d = \deg(\mu)$, which is finite. On the
other hand, we have $\Sn_c(\vD) / n = \evwrt{R(\vbeta, \vD)}{D_c}$. Therefore, condition \eqref{eq:Bn-prop-TV} implies that for any integer $k > 0$ and
    any $(\vbeta, \vD) \in \Bnmnun(P, \Delta, \epsilon)$, we have 
    \begin{equation*}
      \Sn_c(\vD) / n \geq \evwrt{R(\vbeta, \vD)}{D_c \wedge k} \geq \evwrt{\bar{P}^\Delta}{D_c \wedge k} - 2k \epsilon.
    \end{equation*}
% this is basically the 2 factor that appears in bounding d_{TV}(\mu, \nu) = \sup
% |\mu(A) - \nu(A)| and \sum |\mu(\{z\}) - \nu(\{z\})| for probability distributions on
% countable samples spaces. In general, we have |\int f d \mu - \int f d \nu|
% \leq 2 \norm{f}_\infty d_{TV}(\mu, \nu). The above trick shows this for simple
% functions, and then for integrable ones. 
Taking the $\liminf$ as $n\rightarrow \infty$ and then sending $\epsilon$ to
zero, we get 
\begin{equation*}
  \liminf_{\epsilon\rightarrow 0} \liminf_{n \rightarrow \infty} \min_{(\vbeta, \vD) \in \Bnmnun(P, \Delta, \epsilon)} \frac{\Sn_c(\vD)}{n} \geq \evwrt{\bar{P}^\Delta}{D_c \wedge k}.
\end{equation*}
Furthermore, sending $k \rightarrow \infty$, we get
    \begin{equation}
      \label{eq:UB-liminf-Sc-Ep-k}
      \liminf_{\epsilon\rightarrow 0} \liminf_{n \rightarrow \infty} \min_{(\vbeta, \vD) \in \Bnmnun(P, \Delta, \epsilon)} \frac{\Sn_c(\vD)}{n} \geq \evwrt{\bar{P}^\Delta}{D_c}.
    \end{equation}
Now, we show that a matching upper bound exists. To do this, note that due to
\eqref{eq:Bn-prop-mn}, for $c \in \bar{\mC}(\Delta)$, we have $\Sn_c(\vD)
    = 2\snorm{\vmn}_1 - \sum_{\stackrel{c' \in \bar{\mC}(\Delta)}{c' \neq c}}
    \Sn_{c'}(\vD)$. 
Using $2\snorm{\vmn}_1 / n \rightarrow d = \deg(\mu) \in (0,\infty)$ and
\eqref{eq:UB-liminf-Sc-Ep-k}, since $\bar{\mC}(\Delta)$ is finite, we get
\begin{equation*}
  \begin{aligned}
    \limsup_{\epsilon\rightarrow 0} \limsup_{n \rightarrow \infty} \max_{(\vbeta, \vD) \in \Bnmnun(P, \Delta, \epsilon)} \frac{\Sn_c(\vD)}{n} &\leq d - \sum_{\stackrel{c' \in \bar{\mC}(\Delta)}{c' \neq c}} \liminf_{\epsilon\rightarrow 0} \liminf_{n \rightarrow \infty}\min_{(\vbeta, \vD) \in \Bnmnun(P, \Delta, \epsilon)} \frac{\Sn_{c'}(\vD)}{ n} \\
    &\leq d - \sum_{\stackrel{c' \in \bar{\mC}(\Delta)}{c' \neq c}} \evwrt{\bar{P}^\Delta}{D_{c'}} \\
    &= \sum_{c'' \in \bar{\mC}(\Delta)} \evwrt{\bar{P}^\Delta}{D_{c''}} - \sum_{\stackrel{c' \in \bar{\mC}(\Delta)}{c' \neq c}} \evwrt{\bar{P}^\Delta}{D_{c'}} \\
    &= \evwrt{\bar{P}^\Delta}{D_c}.
  \end{aligned}
\end{equation*}
This together with \eqref{eq:UB-liminf-Sc-Ep-k} completes the proof of
\eqref{eq:UB-limsup-Sc-Ep-0}.

On the other hand, observe that for $(\vbeta, \vD) \in \Bnmnun(P, \Delta, \epsilon)$ and
      $c \in \bar{\mC}(\Delta)$, $\frac{1}{n} \sum_{v=1}^n \log D_c(v)! =
      \evwrt{R(\vbeta, \vD)}{\log D_c!}$. Therefore, a similar truncation
      argument as in \eqref{eq:UB-liminf-Sc-Ep-k} implies that
      \begin{equation}
        \label{eq:UB-log-D!-lowerbound}
        \liminf_{\epsilon \rightarrow 0} \liminf_{n \rightarrow \infty} \min_{(\vbeta, \vD) \in \Bnmnun(P, \Delta, \epsilon)} \frac{1}{n} \sum_{v=1}^n \log D_c(v)! \geq \evwrt{\bar{P}^\Delta}{\log D_c!}. 
      \end{equation}
Note that $\log D_c! \geq 0$, hence %$\ev{\log D_c!}$ 
$\evwrt{\bar{P}^\Delta}{\log D_c!}$
is well-defined,
although it can be $\infty$. 
Also,   $\bar{\mC}(\Delta)$ is finite.
Therefore, using
\eqref{eq:UB-log-D!-lowerbound} together with \eqref{eq:UB-limsup-Sc-Ep-0}
in \eqref{eq:UB-log-G-vD-2-expand} and simplifying, we get
\begin{equation*}
  \begin{aligned}
    &\limsup_{\epsilon \rightarrow 0} \limsup_{n \rightarrow \infty} \max_{(\vbeta, \vD) \in \Bnmnun(P, \Delta, \epsilon)} \frac{1}{n} \left( \log |\mG(\vD, 2) - \snorm{\vmn}_1 \log n \right) \\
    &\quad \leq \frac{1}{2} \sum_{c \in \bar{\mC}(\Delta)} \left( \evwrt{\bar{P}^\Delta}{D_c} \log \evwrt{\bar{P}^\Delta}{D_c}  - \evwrt{\bar{P}^\Delta}{D_c} \right)  - \sum_{c \in \bar{\mC}(\Delta)} \evwrt{\bar{P}^\Delta}{\log D_c!} \\
    &\quad= -s(d) + \frac{d}{2} \sum_{c \in \bar{\mC}(\Delta)} \frac{\evwrt{\bar{P}^\Delta}{D_c}}{d} \log  \frac{\evwrt{\bar{P}^\Delta}{D_c}}{d} - \sum_{c \in \bar{\mC}(\Delta)} \evwrt{\bar{P}^\Delta}{\log D_c!}.
  \end{aligned}
\end{equation*}
% we have used continuity of x \mapsto x \log x
Note that, for each $c \in \bar{\mC}(\Delta)$,
$0 \leq \evwrt{\bar{P}^\Delta}{D_c} \leq d < \infty$, hence each term in the first
summation is nonpositive and finite. Also, $\evwrt{\bar{P}^\Delta}{\log D_c!} \geq
0$ for $c \in \bar{\mC}(\Delta)$.
As a result, the bound on the right hand side is well-defined,
although it can be $-\infty$. Also, since each term in the first summation is
nonpositive while each term in the second summation is nonnegative, we may
restrict both the summations to $\mC(\Delta) \subset \bar{\mC}(\Delta)$ to find
an upper bound for the right hand side. But for $c \in \mC(\Delta)$, we have 
$\evwrt{\bar{P}^\Delta}{D_c} = e_P(c)$ and $\evwrt{\bar{P}^\Delta}{\log
  D_c!}  = \evwrt{P}{\log E_h(c)!}$, which yields
\begin{equation}
\label{eq:log-mG-vD-leq-sum-C-Delta}
  \begin{aligned}
    \limsup_{\epsilon \rightarrow 0} \limsup_{n \rightarrow \infty} \max_{(\vbeta, \vD) \in \Bnmnun(P, \Delta, \epsilon)} \frac{1}{n} \left( \log |\mG(\vD, 2)| - \snorm{\vmn}_1 \log n \right) \\ \leq -s(d) + \frac{d}{2} \sum_{c \in \mC(\Delta)} \pi_P(c) \log \pi_P(c) - \sum_{c \in \mC(\Delta)} \evwrt{P}{\log E_h(c)!}.
  \end{aligned}
\end{equation}
Again, note that the terms in the first summation are finite and nonpositive,
while the terms in the second summation are nonnegative, but possibly $+\infty$.
Thereby, the above bound is well-defined, although it can be $-\infty$. 
% -------
% Note that for $c \in \mC(\Delta)$, we have
% \begin{equation*}
%   \evwrt{P}{\log E_h(c)!} \leq \evwrt{P}{\log \deg_T(o)!} \leq \evwrt{P}{\deg_T(o) \log \deg_T(o)} < \infty,
% \end{equation*}
% % k! \leq k^k, so \log k! \leq k \log k
% where the last inequality follows from the fact that $P \in \mP_h$. Therefore,
% all the terms in \eqref{eq:log-mG-vD-leq-sum-C-Delta} are finite.
% ------
% the above block is not correct, I want to prove this Proposition without the
% assumption P \in \mP_h so that later in Proposition 6 I can prove when
% \ev{\deg \log \deg} = \infty, the entropy is -\infty. 

By assumption, we have $H(P) < \infty$. Therefore, we can put the
bounds in \eqref{eq:log-mG-vD-leq-sum-C-Delta} and \eqref{eq:UB-Bn-HP-claim} back in
\eqref{eq:Amnun-Bmn-max-G-vD} to  get
\begin{equation}
\label{eq:boundforupper}
  \begin{aligned}
    &\lim_{\epsilon\rightarrow 0} \limsup_{n \rightarrow \infty} \frac{1}{n} \left(  \log |\Anmnun(P, \epsilon)| - \snorm{\vmn}_1 \log n \right) \\
    &\quad \leq -s(d) + H(P) -   \frac{d}{2} \sum_{c \in \mC(\Delta)} \pi_P(c) \log \frac{1}{\pi_P(c)} - \sum_{c \in \mC(\Delta)} \evwrt{P}{\log E_h(c)!} .
  \end{aligned}
\end{equation}
Note that this holds for any finite $\Delta \subset \mTb_*^h$, and that
$\pi_P(c) \log \frac{1}{\pi_P(c)} \geq 0$ and $\evwrt{P}{\log E_h(c)!}
\geq 0$ for all $c \in (\edgemark \times \mTb_*^{h-1}) \times (\edgemark
\times \mTb_*^{h-1})$.
%\color{red}
Interpreting the summations on the 
right hand side of \eqref{eq:boundforupper} 
as integrals, 
restricted to $\mC(\Delta)$, with respect to the 
uniform measure on 
$(\edgemark \times \mTb_*^{h-1}) \times (\edgemark \times \mTb_*^{h-1})$, by sending $\Delta$ to $\mTb_*^h$ and using
the monotone convergence theorem, we arrive at \eqref{eq:upperbound-A-claim}
which completes the proof.
%\color{black}
\end{proof}
%\end{mybox}

\subsection{Proof of Proposition~\ref{prop:upper-bound-infty}}
\label{sec:upperbound-new}

In this section, we prove the upper bound result of
%Proposition~\ref{prop:upper-bound}.
%and 
Proposition~\ref{prop:upper-bound-infty}. 

\begin{proof}[Proof of Proposition~\ref{prop:upper-bound-infty}]
  Let $P:= \mu_1 \in \mP(\mTb_*^1)$ be the distribution of the depth--1
  neighborhood of the root in $\mu$.
  Borrowing the idea in the proof of  Corollary~\ref{cor:deg-log-deg-Ph}, note
  that 
  each rooted tree equivalence class $[T,o] \in \mTb_*^1$
  is uniquely determined by knowing the integers 
  \begin{equation}
    \label{eq:Nxx'--theta-theta'--notation}
    N_{x,x'}^{\theta, \theta'}(T,o) := |\{v \sim_T o: \xi_T(v,o) = x, \tau_T(o) = \theta, \xi_T(o,v) = x', \tau_T(v) = \theta'\}|,
  \end{equation}
  for each $x, x' \in \edgemark$ and $\theta, \theta' \in \vermark$. Now, for $x,
  x' \in \edgemark$ and $\theta, \theta' \in \vermark$, we have
  $\evwrt{P}{N_{x,x'}^{\theta, \theta'}(T,o)} \leq \evwrt{P}{\deg_T(o)} <
  \infty$. Consequently, when  $[T, o] \sim P$, the entropy of the random
  variable $N_{x,x'}^{\theta, \theta'}(T,o)$ is
  finite for all $x, x'\in \edgemark$ and
  $\theta, \theta' \in \vermark$. Therefore, since $\edgemark$ and $\vermark$ are finite sets, we conclude that
  $H(P) < \infty$. 
Hence, using Proposition~\ref{prop:upper-bound} for $h =1$, we have
    \begin{equation}
      \label{eq:ub-inty-J}
      \bchover_{\vdeg(\mu), \vvtype(\mu)}(\mu)\condmnun \leq -s(d) + H(P) - \frac{d}{2} H(\pi_P) - \sum_{t, t' \in \edgemark \times \mTb_*^0} \evwrt{P}{\log E_1(t, t')!}.
    \end{equation}
Now we show that there exist $t$ and $ t'$ in $ \edgemark \times \mTb_*^0$   such that
$\evwrt{P}{\log E_1(t,t')!} = \infty$. Since every element of $\mTb_*^0$ is a marked
isolated vertex, $\mTb_*^0$ can be identified with $\vermark$. With an abuse of notation, we
may therefore write $t, t' \in \edgemark
    \times \mTb_*^0$ as  $ t= (x, \theta)$ and $t' = (x', \theta')$ respectively, where $x,x'
    \in \edgemark$ and $\theta, \theta' \in \vermark$. 
With this, for $[T,o] \in \mTb_*$, 
%using the notation in
%\eqref{eq:Nxx'--theta-theta'--notation}, 
we have $E_h(t, t')(T, o) =
N_{x,x'}^{\theta, \theta'}(T, o)$.
Therefore, from \eqref{eq:ub-inty-J}, it suffices to prove that
$\evwrt{P}{\log N_{x,x'}^{\theta, \theta'}(T, o)!} = \infty$ for some
$x,x' \in \edgemark$, $\theta, \theta' \in \vermark$.

We prove this by contradiction. Assume that $\evwrt{P}{\log
  N_{x,x'}^{\theta, \theta'}(T, o)!} < \infty$ for all $x,x' \in \edgemark$,
$\theta, \theta' \in \vermark$.
Using Stirling's approximation, for $k \geq 0$, we have $\log k! \geq k \log k -
k$, where $0 \log 0$ is interpreted as $0$.
Therefore, for $x,x' \in \edgemark$ and $\theta, \theta' \in \vermark$, we have
\begin{equation}
  \label{eq:log-Nxx'--theta-theta'--NlogN-N}
  \infty > \evwrt{P}{\log N_{x,x'}^{\theta, \theta'}(T,o)!} \geq \evwrt{P}{N_{x,x'}^{\theta, \theta'}(T,o) \log N_{x,x'}^{\theta, \theta'}(T,o)} - \evwrt{P}{N_{x,x'}^{\theta, \theta'}(T,o)}.
\end{equation}
On the other hand, $\deg_T(o) = \sum_{\stackrel{x,x' \in
    \edgemark}{\theta, \theta' \in \vermark}} N_{x,x'}^{\theta,
  \theta'}(T, o)$ for all $[T, o] \in \mTb_*$. Also, we have
$\evwrt{P}{\deg_T(o)} = \deg(\mu) < \infty$. Hence
$\evwrt{P}{N_{x,x'}^{\theta, \theta'}(T,o)} < \infty$ for all $x,x' \in
\edgemark$ and $ \theta, \theta' \in \vermark$. 
Using this in \eqref{eq:log-Nxx'--theta-theta'--NlogN-N}, we realize that 
\begin{equation}
  \label{eq:contradiction--NlogN<infty}
 \evwrt{P}{N_{x,x'}^{\theta, \theta'}(T, o) \log N_{x,x'}^{\theta,
     \theta'}(T, o)} < \infty \qquad \forall  x, x' \in \edgemark \text{ and } \theta,
    \theta' \in \vermark
\end{equation}
Moreover, for $[T,o] \in \mTb_*$, using $\deg_T(o) = \sum_{\stackrel{x,x' \in
    \edgemark}{\theta, \theta' \in \vermark}} N_{x,x'}^{\theta,
  \theta'}(T, o)$ and the convexity of $x \mapsto x \log x$, we have 
\begin{equation*}
  \frac{\deg_T(o)}{|\edgemark|^2 |\vermark|^2} \log \frac{\deg_T(o)}{|\edgemark|^2 |\vermark|^2} \leq \frac{1}{|\edgemark|^2 |\vermark|^2} \sum_{\stackrel{x,x' \in
          \edgemark}{\theta, \theta' \in \vermark}} N_{x,x'}^{\theta, \theta'}(T, o) \log N_{x,x'}^{\theta, \theta'}(T, o),
    \end{equation*}
    where as usual, we interpret $0 \log 0$ as $0$. 
   % \color{red}
   % Note for Payam: 
   % I changed the normalization (by squaring terms) in the preceding calculation. Check.
    %\color{black}
   % \pres{correct, thanks.}
Taking the expectation with respect to $P$ on both sides 
%followed by  using the
%fact that $\edgemark$ and $\vermark$ are finite and 
%$\evwrt{P}{\deg_T(o)} = \deg(\mu) < \infty$ together with \eqref{eq:contradiction--NlogN<infty}, 
we realize that $\evwrt{P}{\deg_T(o) \log \deg_T(o)} < \infty$, which is a
    contradiction.
Hence, there must exist $x, x' \in \edgemark$ and $\theta, \theta'
      \in \vermark$ such that $\evwrt{P}{\log N_{x,x'}^{\theta, \theta'}!} =
      \infty$. 
Finally, using this in \eqref{eq:ub-inty-J} implies $\bchover_{\vdeg(\mu),
  \vvtype(\mu)}(\mu)\condmnun = -\infty$ and completes the proof.
\end{proof}

\subsection*{Acknowledgments}
%\label{sec:upperbound}

Research supported by the NSF grants
CNS-1527846 and CCF-1618145, the NSF Science \& Technology
Center grant CCF-0939370 (Science of Information), and the
William and Flora Hewlett Foundation supported Center for
Long Term Cybersecurity at Berkeley.

%%% Local Variables: 
%%% mode: latex
%%% TeX-master: "Note-41_BC-ent-arxiv.tex"
%%% End: 

\appendix

\section{Some Properties of Marked Rooted Trees of Finite Depth}
\label{sec:marked-rooted-trees-some-props}

In this section we gather some useful properties of marked rooted trees of finite depth, which are used at various points during the discussion.

%Before that, we define some notation and state some lemmas. 
Given a marked rooted tree $(T,
o)$, integers $ k,l \geq 1$,  $t \in \edgemark \times  \mTb_*^{k-1}$, and $t' \in \edgemark \times  \mTb_*^{l-1}$, define
\begin{equation*}
  E_{k,l}(t,t')(T,o) := |\{v \sim_T o : T(v,o)_{k-1} \equiv t, T(o,v)_{l-1} \equiv t' \}|.
\end{equation*}
%\color{red}
%Note to Payam:
%The notation $E_{k,l}(t,t')(T,o)$ needs to be introduced in the glossary.
%\color{black}
%\pres{added to the glossary}
When $k=l$ this reduces to the notation we defined in Section~\ref{sec:markovian-ugwt}, i.e.\ $E_{k,k}(t,t')(T,o)$
is the same as $E_k(t,t')(T,o)$.
\glsadd{not:Ekl-tt'-To}

\begin{lem}
  \label{lem:ov=ov'--vo-v'o}
  Assume $(T,o)$ is a rooted marked tree with finite depth, and $v$ and $v'$ are
  offspring of the root. Then, if $T(o,v) \equiv T(o,v')$ and $\xi_T(v,o) =
  \xi_T(v',o)$, we have $T(v,o) \equiv T(v',o)$.
\end{lem}

\begin{proof}
  We construct the rooted automorphism $f: V(T) \rightarrow V(T)$ as follows. We
  set $f(o) = o$, $f(v) = v'$ and $f(v') = v$. Moreover, we use the isomorphism
  $T(o,v) \equiv T(o,v')$ to map the nodes in the subtree of  $v$ to the
  nodes in the subtree of  $v'$ and vice versa. Finally, we set $f$ to be
  the identity map on the rest of the tree. Indeed, $f$ is an adjacency preserving
  bijection. On the other hand, the assumptions $\xi_T(v,o) \equiv \xi_T(v',o)$
  and $T(o,v) \equiv T(o,v')$ imply that $f$ preserves the marks. Therefore, $f$
  is an automorphism which maps $T(v,o)$ to $T(v',o)$. This completes the proof.
\end{proof}

\begin{lem}
  \label{lem:t1-oplus-t'--t2-oplus-t'}
  Assume $t^{(1)}, t^{(2)} \in \edgemark \times  \mTb_*^h$ and $t' \in \edgemark
 \times \mTb_*^k$ are given such
  that $t^{(1)} \oplus t' = t^{(2)} \oplus t'$. Further, assume that
  %$t^{(1)}$ and $t^{(2)}$ have the same mark components, i.e.
  $t^{(1)}[m] =
  t^{(2)}[m]$. Then, we have $t^{(1)} = t^{(2)}$. 
\end{lem}
% k, h \geq 0

\begin{proof}
  Let $(T,o)$ be an arbitrary member of the equivalence class $t^{(1)} \oplus t'$. Therefore, we have
  $T(v,o) \equiv t^{(1)}$ and $T(o,v) \equiv t'$ for some $v \sim_T o$. On the
  other hand, by assumption, $(T,o)$ is also a member of the equivalence class
  $t^{(2)} \oplus t'$. This means that $T(v',o) \equiv t^{(2)}$ and $T(o,v')
  \equiv t'$ for some $v' \sim_T o$. Moreover, by assumption, $\xi_T(v,o) =
  t^{(1)}[m] = t^{(2)}[m] = \xi_T(v',o)$. Since $T(o,v) \equiv T(o,v') \equiv t'$,
  Lemma~\ref{lem:ov=ov'--vo-v'o} above implies that $T(v,o) \equiv T(v',o)$, or
  equivalently $t^{(1)} = t^{(2)}$. 
\end{proof}

\begin{lem}
  \label{lem:Eh-1+Nh}
  Assume $(T, o)$ is a rooted marked tree with depth at most 
  $k \ge 1$. Moreover, assume that, for some $l \ge 1$, $t \in \edgemark \times \mTb_*^l$, and $t' \in \edgemark \times \mTb_*^{k-1}$, we have $E_{l+1, k}(t,t')(T,o) > 0$. Then,
  \begin{equation*}
    E_{l+1, k}(t, t')(T,o) = |\{ v \sim_T o: T(o,v)_{k-1} \equiv t', \xi_T(v,o) = t[m] \}|.
  \end{equation*}
\end{lem}

\begin{proof}
  From the definition of $E_{l+1, k}(t, t')(T,o)$, we have
  \begin{align*}
    E_{l+1, k}(t, t')(T,o) &= |\{ v \sim_T o: T(o,v)_{k-1} \equiv t', T(v,o)_l \equiv t \}| \\
                           &\leq | \{ v \sim_T o : T(o,v)_{k-1} \equiv t', \xi_T(v,o) = t[m] \}|.
  \end{align*}
  Now, we show the inequality in the opposite direction. The assumption $E_{l+1, k}(t, t')(T, o) > 0$ implies that there exists $v \sim_T o$ such that $T(o,v)_{k-1} \equiv t'$ and $T(v,o)_l \equiv t$. This in particular means that $\xi_T(v, o) = t[m]$. On the other hand, if $v' \sim_T o$ is such that $T(o,v')_{k-1} \equiv t'$ and $\xi_T(v',o) = t[m]$, Lemma~\ref{lem:ov=ov'--vo-v'o} above implies that $T(v',o) \equiv T(v,o)$, which means  $T(v',o)_l \equiv t$. This establishes the other direction of the inequality and completes the proof.
\end{proof}

%aa
\begin{lem}
  \label{lem:r-r'-same-degree_are-the-same}
  Given  $h \ge 1$ and two marked rooted trees $(T,o)$ and
  $(T',o')$ with depth at most $h$,  assume that the
  mark at the root in $T$ and $T'$ are the same and also, for all $t, t' \in
  \edgemark \times \mTb_*^{h-1}$, we have $E_h(t, t')(T,o) = E_h(t, t')(T',o')$.
  Then, $(T,o) \equiv (T',o')$.
  %\color{red} 
  %Note for Payam:
  %It is unclear why an assumption about the mark at the root is being singled
  %out. It seems to follow from the other assumption. \color{black} \pres{it does
    %not necessarily follow. When $(T,o)$ and $(T',o')$ are isolated roots, then we automatically have $E_h(t, t')(T,o) = E_h(t, t')(T',o') = 0$ for
    %all $t, t' \in \edgemark \times \mTb_*^{h-1}$, but $(T,o) \equiv (T',o')$
    %only when the mark at the root in the two rooted trees are the same. I also
    %added this below as a comment in the text before the proof starts.}
\end{lem}

% bb

\begin{proof}

Note that the assumption regarding the mark at the root in $T$ and that in $T'$ being
equal is necessary in this statement. To see this, consider the example where $(T,o)$ and $(T',o')$ are isolated roots, then we automatically have $E_h(t, t')(T,o) = E_h(t, t')(T',o') = 0$ for
    all $t, t' \in \edgemark \times \mTb_*^{h-1}$, but $(T,o) \equiv (T',o')$
    only when the marks at the root in the two rooted trees are the same.

  Since the root marks in $[T,o]$ and $[T',o']$ are the same, %(\color{red} It is unclear why an assumption about the mark at the root is being singled out. It seems to follow from the other assumption. \color{black}) 
  it suffices to
  show that for all $x \in \edgemark$ and $t \in \edgemark \times
  \mTb_*^{h-1}$ we have
  %$N_{x,t}(T,o) = N_{x,t}(T',o')$ where 
  \begin{equation}
  \label{eq:trees-hanging}
    %N_{x,t}(T,o) := 
    | \{ v \sim_T o: \xi_T(v,o) = x, T(o,v)_{h-1} \equiv t\}|
    =
    | \{ v \sim_{T'} o': \xi_{T'}(v,o') = x, {T'}(o',v)_{h-1} \equiv t\}|
    .
  \end{equation}
  %and $N_{x,t}(T',o')$ is defined similarly.
  
  Note that if 
  $| \{ v \sim_T o: \xi_T(v,o) = x, T(o,v)_{h-1} \equiv t\}| > 0$
  %$N_{x,t}(T,o) > 0$ 
  %for some $x \in \edgemark$ and $t \in \edgemark \times \mTb_*^{h-1}$, 
  then there exists $v \sim_T o $ such that
  $\xi_T(v,o) = x$ and $T(o,v)_{h-1} \equiv t$. This means that with $t' :=
  T[v,o]_{h-1}$, we have $E_h(t',t)(T,o) > 0$. Moreover,
  Lemma~\ref{lem:Eh-1+Nh} 
  %in Appendix~\ref{sec:marked-rooted-trees-some-props} 
  implies that
  %$E_h(t',t)(T,o) = N_{x,t}(T,o)$. 
  $E_h(t',t)(T,o) = | \{ v \sim_T o: \xi_T(v,o) = x, T(o,v)_{h-1} \equiv t\}|$. 
  On the other hand, from the hypothesis of this lemma, we also know that
  $E_h(t',t)(T',o') = E_h(t',t)(T,o)>0$. Another usage of
  Lemma~\ref{lem:Eh-1+Nh} shows
  \eqref{eq:trees-hanging}.
  %that
  %\begin{equation*}
   % N_{x,t}(T',o') = E_h(t',t)(T',o') = E_h(t',t)(T,o) = N_{x,t}(T,o). 
  %\end{equation*}
  So far, we have shown that 
  $| \{ v \sim_T o: \xi_T(v,o) = x, T(o,v)_{h-1} \equiv t\}| > 0$
  %$N_{x,t}(T,o) > 0$ for some $x \in \edgemark$ and
  %$t \in \edgemark \times \mTb_*^{h-1}$ 
  implies \eqref{eq:trees-hanging}.
  %that $N_{x,t}(T,o) = N_{x,t}(T',o')$. 
  Similarly, 
  $| \{ v \sim_{T'} o: \xi_{T'}(v,o') = x, {T'}(o',v)_{h-1} \equiv t\}| > 0$
  %$N_{x,t}(T',o') > 0$ 
  implies \eqref{eq:trees-hanging}.
  %$N_{x,t}(T',o') = N_{x,t}(T,o)$. 
  %Thus, $N_{x,t}(T,o) = N_{x,t}(T',o')$ for all $x \in
  %\edgemark$ and $t \in \edgemark \times \mTb_*^{h-1}$, and the proof is complete. 
  This completes the proof.
\end{proof}

\section{Some Properties of Unimodular Galton--Watson Trees with given
  Neighborhood Distribution}
\label{sec:UGWT-some-props}

\newcommand{\parent}{p}
\newcommand{\hT}{\hat{T}}
\newcommand{\ho}{\hat{o}}
\newcommand{\mun}{\mu^{(n)}}
\newcommand{\gamman}{\gamma^{(n)}}

Fix $h \ge 1$ and $P \in \mP(\mTb_*^h)$ admissible. 
In this section, we prove some properties of $\ugwt_h(P)$.
%First, we introduce some notation. 

Given
$[T,o] \in \mTb_*$, and $v \in V(T)$ with $v \neq o$, let
$\parent(v)$ denote the parent node of $v$.
For such $v$, we denote $(T[\parent(v), v]_{h-1}, T[v, \parent(v)]_{h-1})$
by $c(v)$ and $(T[v,\parent(v)]_{h-1}, T[\parent(v),
v]_{h-1})$ by $\bar{c}(v)$.
%With these, given a marked rooted tree $[T,o] \in \mTb_*$ and a vertex $v \in V(T)$, we define
Let
\begin{equation}
  \label{eq:gamma-def}
  \gamma_{[T,o]}(v) =
  \begin{cases}
    P([T,o]_h) & v = o, \\
    \hP_{c(v)}(T[\parent(v), v]_h) & v \neq o.
  \end{cases}
\end{equation}
When the marked rooted tree $[T,o]$ is clear from the context
we will simply write $\gamma(v)$ for $\gamma_{[T,o]}(v)$.

\begin{lem}
  \label{lem:UGWT-gamma-pos-ep-pos}
  Given $[T,o] \in \mTb_*$ and $v \in V(T)$ 
 with $\dist_T(v,o) = k$ where $k \geq 1$,
 let $o = v_0, v_1, \dots, v_k =
  v$ denote the path connecting $v$ to the root. 
  
  Then, if $\gamma(v_i) > 0$
  for all $0 \leq i \leq k-1$,  we have $P([T,v_i]_h) > 0$ for $0 \leq i
  \leq k-1$, and $e_P(c(v_i)) > 0$ for $1 \leq i \leq k$. 
\end{lem}

\begin{proof}
  We prove this by induction on $k$. First, consider $k=1$. In this case, we have $\gamma(v_0) = \gamma(o) = P([T,o]_h)$,
  and so the hypothesis that $\gamma(v_0) > 0$ implies that $P([T,v_0]_h) = P([T,o]_h) > 0$, which establishes the first claim. Using this, we get
  \begin{align*}
    e_P(\bar{c}(v)) &= e_P(T[v,o]_{h-1}, T[o,v]_{h-1}) \geq P([T,o]_h) E_h(T[v,o]_{h-1}, T[o,v]_{h-1})([T,o]_h) \\
    &\geq P([T,o]_h) > 0,
  \end{align*}
  where the last equality uses the fact that, by definition, $E_h(T[v,o]_{h-1},
  T[o,v]_{h-1})([T,o]_h) \geq 1$. But, since $P$ is admissible, we have
  $e_P(c(v)) = e_P(\bar{c}(v)) > 0$ which completes the proof for $k=1$.

  Now, for $k > 1$, we have $e_P(c(v_{k-1})) > 0$
  from the  induction hypothesis. This
  implies that, with $t:= T[v_{k-2}, v_{k-1}]_{h-1}$, $t' := T[v_{k-1},
  v_{k-2}]_{h-1}$ and $\tilde{t} := T[v_{k-2}, v_{k-1}]_{h}$, we have
  \begin{align*}
    0 &< \gamma(v_{k-1}) = \hP_{t,t'}(\tilde{t}) \\
    & = \frac{P(\tilde{t} \oplus t')E_h(t,t')(\tilde{t} \oplus t')}{e_P(t,t')}.
  \end{align*}
  In particular, we have $P(\tilde{t} \oplus t') > 0$. But $\tilde{t} \oplus t'
  = [T,v_{k-1}]_h$. This together with the induction hypothesis implies that
  $P([T,v_i]_h) > 0$ for $0 \leq i \leq k-1$. Moreover, we have 
  \begin{equation*}
    e_P(\bar{c}(v)) \geq P([T,v_{k-1}]_h) E_h(\bar{c}(v))([T,v_{k-1}]_h) \geq P([T,v_{k-1}]_h)> 0,
  \end{equation*}
  where the last equality follows from the fact that, by definition, we have
  \begin{equation*}
   E_h(\bar{c}(v))([T,v_{k-1}]_h) \geq 1. 
  \end{equation*}
  The proof is complete by noting that,
  since $P$ is admissible, we have $e_P(\bar{c}(v)) = e_P(c(v))$. 
\end{proof}

  \begin{cor}
    \label{cor:ugwt-gamma>0-ep>0}
  Let $\mu = \ugwt_h(P)$ with $h \geq 1$ and let  $P \in \mP(\mTb_*^h)$, i.e. $P$ admissible.
  Then, for $\mu$--almost all $[T,o] \in \mTb_*$, we have $\gamma_{[T,o]}(v) > 0$ for
  all $v \in V(T)$ and $e_P(c(w)) > 0$ for
  all $w \in V(T) \setminus \{o\}$. 
\end{cor}

\begin{proof}
  First recall that $\gamma_{[T,o]}(o) = P([T,o]_h)$ is the probability of
  sampling $[T,o]_h$ in the process of generating $[T,o]$ with law $\mu =
  \ugwt_h(P)$. Hence, $\mu$--almost surely, we have $\gamma_{[T,o]}(o) > 0$.
  Moreover, for a vertex $v \in V(T) \setminus \{o\}$, $\gamma_{[T,o]}(v) =
  \hP_{c(v)}(T[\parent(v), v]_h)$ is the probability of sampling $T[\parent(v),
  v]_h$ given $T[\parent(v), v]_{h-1}$ and $T[v, \parent(v)]_{h-1}$ in the process
  of generating $[T,o]$ with law $\mu$. Since there are countably many vertices
  in $[T,o] \in \mTb_*$, $\mu$--almost surely we have
  $\gamma_{[T,o]}(v) > 0$ for all $v \in V(T)$.

  Motivated by the above discussion, if $[T,o] \in \mTb_*$ is outside a measure zero set with respect to
  $\mu$, we have $\gamma_{[T,o]}(v) > 0$ for all $v \in V(T)$. Thus,
  Lemma~\ref{lem:UGWT-gamma-pos-ep-pos} above implies that $e_P(c(w)) > 0$ for
  all $w \in V(T) \setminus \{o\}$ and completes the proof.
\end{proof}

%%% Local Variables:
%%% mode: latex
%%% TeX-master: "Note-41_BC-ent-arxiv"
%%% End:

\section{A Convergence Property of Unimodular Galton--Watson Trees with respect to the Neighborhood Distribution}
\label{sec:UGWT-convg}

In this section we give the proof of Lemma~\ref{lem:Pn-conv-P--ugwt-Pn-conv-ugwt-P}.

\begin{proof}[Proof of Lemma~\ref{lem:Pn-conv-P--ugwt-Pn-conv-ugwt-P}]
Let $\mun := \ugwt_h(\Pn)$ and $\mu:= \ugwt_h(P)$.
  We claim  that for any integer $l \in \nats$ and $[\hT, \ho] \in
  \mTb_*^l$, we have
  \begin{equation}
    \label{eq:mun-ATo--mu-Ato}
    \lim_{n \rightarrow \infty} \mun(A_{[\hT, \ho]}) = \mu(A_{[\hT, \ho]}),
  \end{equation}
  where
  \begin{equation*}
    A_{[\hT, \ho]} := \{[T,o] \in \mTb_*: [T,o]_l = [\hT, \ho]\}. 
  \end{equation*}
  Before proving our claim, we show why this implies $\mun \Rightarrow \mu$.
  To do this, we take a bounded and uniformly continuous function $f: \mTb_*
  \rightarrow \reals$ and show that $\int f d \mun \rightarrow \int f d \mu$.
  Fix $\epsilon > 0$. Due to the local topology on $\mTb_*$, there is $l \in
  \nats$ such that for all $[\hT, \ho] \in \mTb_*^l$ and $[T,o] \in A_{[\hT,
    \ho]}$, we have $d_*([T,o], [\hT, \ho]) < \epsilon$. Recall that $d_*$
  denotes the local metric on $\mTb_*$. Since $f$ is uniformly continuous, this
  implies that $|f([T,o]) - f([\hT, \ho])| < \eta(\epsilon)$ where
  $\eta(\epsilon) \rightarrow 0$ as $\epsilon \rightarrow 0$. Now, fix a finite
  collection $\mS$ of marked rooted trees $[\hT, \ho] \in \mTb_*^l$ such that
  \begin{equation*}
    \sum_{[\hT, \ho] \in \mS} \mu(A_{[\hT, \ho]}) > 1 - \epsilon.
  \end{equation*}
  Then, \eqref{eq:mun-ATo--mu-Ato} implies that, for $n$ large enough, we have
  \begin{equation*}
    \sum_{[\hT, \ho] \in \mS} \mun(A_{[\hT, \ho]}) > 1 - 2\epsilon.
  \end{equation*}
  This implies that
  \begin{equation*}
    \left|\int f d \mun - \int f d \mu\right| \leq 2 \eta(\epsilon) + 3\epsilon \snorm{f}_\infty + \sum_{[\hT, \ho] \in \mS} |f([\hT, \ho])| | \mun(A_{[\hT, \ho]}) - \mu(A_{[\hT, \ho]}) |.
  \end{equation*}
  % to see calculations see the note
  % 2018-10-11_PN_lwc-ATo-integral-diff.pdf
  By first sending $n$ to infinity and then $\epsilon$ to zero, we get our desired
  result.

  We now get back to proving our claim in \eqref{eq:mun-ATo--mu-Ato}.
First, observe that, by the definition of
$\ugwt_h(P)$, we have
\begin{equation}
  \label{eq:mu-A-C-prod-gamma}
  \mu(A_{[\hT, \ho]}) = C \prod_{v \in B} \gamma(v),
\end{equation}
where $C$ is a constant which only depends on $[\hT, \ho]$,
and $B := \{ v \in
V(\hT): \dist_{\hT}(v,\ho) \leq (l-h)_+\}$ where $(l-h)_+ := \max\{l-h, 0\}$. 
Here, we have employed the notation $\gamma(v) =
  \gamma_{[\hat{T}, \hat{o}]}(v)$ from \eqref{eq:gamma-def} in Appendix~\ref{sec:UGWT-some-props}.
On
the other hand, if we define $\gamman$ by replacing $P$ with $\Pn$ and $\hP_{c(v)}$
with $\hP^{(n)}_{c(v)}$ in the definition of $\gamma$, we have
\begin{equation}
  \label{eq:mun-A-C-prod-gamma}
  \mun(A_{[\hT, \ho]}) = C \prod_{v \in B} \gamman(v).
\end{equation}
Note that, as $C$ only depends on  $[\hT, \ho]$, the constants on
\eqref{eq:mu-A-C-prod-gamma} and \eqref{eq:mun-A-C-prod-gamma} are the same. 
With this, we show  \eqref{eq:mun-ATo--mu-Ato} by considering two cases.

\underline{Case 1, $\mu(A_{[\hT, \ho]}) > 0$:} Using
\eqref{eq:mu-A-C-prod-gamma}, this means that for all $v \in
B$, we have $\gamma(v) > 0$. In particular, $P([\hT, \ho]_h)>  0$ and 
Lemma~\ref{lem:UGWT-gamma-pos-ep-pos} above
% (note that $\gamma(v)>0$ for all $v \in B$, hence the conditions of this lemma
% hold)
implies that 
for all $v \in B$, $v \neq \ho$,  we have $e_P(c(v)) > 0$. Hence, for $v \in
B$, $v\neq \ho$, we have
\begin{align*}
  0 < \gamma(v) = \hP_{c(v)} (\hT[\parent(v), v]_h) = \frac{P([\hT,v]_h) E_h(c(v))([\hT, v]_h)}{e_P(c(v))}.
\end{align*}
Consequently, we have $P([\hT,v]_h) > 0$. As a result, for $n$ large enough, we
have $\Pn([\hT, v]_h) > 0$. On the other hand, since $e_{\Pn}(c(v)) \rightarrow
e_P(c(v))$, for $n$ large enough, we have $e_{\Pn}(c(v)) > 0$ for all $v \in B$,
$v \neq \ho$. Therefore, for $v \in B$, $v \neq \ho$ and $n$ large enough, we
have
\begin{equation}
  \label{eq:gamman-v-->gamma-v-B}
  \gamman(v) = \frac{\Pn([\hT, v]_h) E_h(c(v)) ([\hT, v]_h)}{e_{\Pn}(c(v))} \rightarrow  \frac{P([\hT, v]_h) E_h(c(v)) ([\hT, v]_h)}{e_{P}(c(v))} = \gamma(v).
\end{equation}
Moreover,
\begin{equation*}
  \gamman(\ho) = \Pn([\hT, \ho]_h) \rightarrow P([\hT, \ho]_h) = \gamma(\ho).
\end{equation*}
Thus, together with~\eqref{eq:gamman-v-->gamma-v-B}, and comparing with
\eqref{eq:mu-A-C-prod-gamma} and~\eqref{eq:mun-A-C-prod-gamma}, we realize that
$\mun(A_{[\hT, \ho]}) \rightarrow \mu(A_{[\hT, \ho]})$.

\underline{Case 2, $\mu(A_{[\hT, \ho]}) = 0$:} Using
\eqref{eq:mu-A-C-prod-gamma}, there is at least one node $v \in B$ such that
$\gamma(v) = 0$. If $\gamma(\ho) = P([\hT, \ho]_h)= 0$, we have $\Pn([\hT, \ho]_h)
\rightarrow P([\hT, \ho]_h) = 0$. Hence, $\mun(A_{[\hT, \ho]}) \rightarrow 0$
and we are done. Otherwise, 
let $v \in B$, $v \neq \ho$, be  a node with minimal depth such that $\gamma(v) = 0$,
i.e.\ if $1 \leq k =
\dist_{\hT}(v, \ho)$ and $\ho = v_0,
v_1, \dots, v_k = v$ is the path connecting $v$ to the root, we have
$\gamma(v_i) > 0$ for $0 \leq i \leq k-1$ and $\gamma(v_k) = 0$. Using
Lemma~\ref{lem:UGWT-gamma-pos-ep-pos}, we conclude that $e_P(c(v_k)) > 0$ and
thus
\begin{align*}
  0 = \gamma(v_k) &= \frac{P([\hT, v_k]_h) E_h(c(v_k))([\hT, v_k]_h)}{e_P(c(v_k))} \\
  & \geq \frac{P([\hT, v_k]_h)}{e_P(c(v_k))},
\end{align*}
where the last line uses the fact that $E_h(c(v_k))([\hT, v_k]_h)
\geq 1$. This implies that $P([\hT, v_k]_h) = 0$. Furthermore, since $\Pn([\hT,
v_k]_h) \rightarrow P([\hT, v_k]_h) = 0$ and $e_{\Pn}(c(v_k)) \rightarrow
e_P(c(v)) > 0$, we realize that for $n$ large enough,
\begin{align*}
  \gamman(v_k) &= \frac{\Pn([\hT, v_k]_h) E_h(c(v_k))([\hT, v_k]_h)}{e_{\Pn}(c(v_k))} \rightarrow 0.
\end{align*}
Consequently, using~\eqref{eq:mun-A-C-prod-gamma}, we have $\mun(A_{[\hT, \ho]}) \rightarrow 0= \mu(A_{[\hT, \ho]})$
which completes the proof.

\section{Unimodularity of $\ugwt_h(P)$}
\label{sec:UGWT-unimod}

%{\color{blue}
%The following is the proof of the first part of Proposition~1 in the document of
%2018-11-25, i.e.\ a standalone proof of the fact that for $h \in \nats$ and $P
%\in \mP(\mTb_*^h)$ admissible, $\ugwt_h(P)$ is unimodular.

%The proof is along the lines of the proof Lemma~3.1 in the BC paper, but with
%some differences. Specifically, I found that in part of their proof, they claim
%without enough justification, therefore I tried to clarify. Also, calculation
%details are different. For example, instead of labeling the children of the root
%as $1, \dots, \deg_T(o)$ in the BC paper, here we pick a node chosen uniformly at random among
%the children of the root. This in particular leads to a viewpoint connected to the ``edge perspective'' distribution
%$\pi_P$ and 
%Lemma 2 in the document of 2019-03-11 (this document was sent a few days ago,
%and its goal was to prove that for $P \in \mP(\mTb_*^h)$ with
%$\evwrt{P}{\deg_T(o)} > 0$, when $\evwrt{P}{\deg_T(o) \log \deg_T(o)} <
%\infty$, then $J_h(P) > -\infty$).
%Therefore, I believe that the  proof here is more intuitive and easy to
%understand compared to that of Lemma~3.1 in the BC paper.

%The proof uses two lemmas: Lemma 2 in the document of 2019-03-11, and Lemma  13
%in the 2018-11-25 version. Both of these lemmas have short proofs, and do not
%depend on the earlier results in the document.
%(Lemma 13 in the 2018-11-25 version only depends on Lemma 11 therein, which itself has a
%short proof and does not depend on any other result).

%}
%\vspace{1cm}
%\hrule
%\vspace{1cm}

We give a proof of Lemma \ref{lem:ugwt-P-is-unimodular}.
Let $h \ge 1$ and $P \in \mP(\mTb_*^h)$ be an
admissible probability distribution.
Let $\mu = \ugwt_h(P)$. 
In order to show that $\mu$ is unimodular, we need to show that for any Borel
function $f: \mTb_{**} \rightarrow \reals_+$, we have 
\begin{equation*}
    \evwrt{\mu}{\sum_{v \sim_T o} f(T, o, v)} =   \evwrt{\mu}{\sum_{v \sim_T o} f(T, v, o)}.
\end{equation*}
Without loss of generality, we may assume that $\deg(\mu) > 0$, since
otherwise nothing remains to be proved. 
We have
\begin{equation}
  \label{eq:ugwt-unim_1}
  \begin{aligned}
  \evwrt{\mu}{\sum_{v \sim_T o} f(T, o, v)} &= \sum_{g \in \mTb_*^h} P(g) \evwrt{\mu}{\sum_{v \sim_T o} f(T, o, v)\bigg|(T,o)_h \equiv g} \\
  &= \sum_{g \in \mTb_*^h: \deg(g) > 0} \deg(g) P(g) \evwrt{\mu}{\frac{1}{\deg(g)}\sum_{v \sim_T o} f(T, o, v)\bigg|(T,o)_h \equiv g}, \\
  \end{aligned}
\end{equation}
where $\deg(g)$ denotes the degree at the root in $g$. Define the
probability distribution $\tP \in \mP(\mTb_*^h)$ such that
\begin{equation*}
  \tP([T,o]) := \frac{P([T,o]) \deg_T(o)}{d},
\end{equation*}
where $d:= \evwrt{P}{\deg_T(o)}$ is the expected degree at the root in $P$.
Moreover, define the probability measure $\tilde{\mu} \in \mP(\mTb_*)$ in a way
identical to $\ugwt_h(P)$, with the exception that $(T,o)_h$ in $\tilde{\mu}$ is
sampled from $\tP$ instead of $P$, and we use the distributions $\hP_{t,t'}$ to
extend $(T,o)_h$ exactly as in $\ugwt_h(P)$.
Since, by definition, conditioned on $(T,o)_h$, the distribution of $(T,o)$ is
the same in $\mu$ and $\tilde{\mu}$, we may write~\eqref{eq:ugwt-unim_1} as
follows
\begin{equation*}
  \evwrt{\mu}{\sum_{v \sim_T o} f(T, o, v)}  = d \sum_{g \in \mTb_*^h} \tP(g) \evwrt{\tilde{\mu}}{\frac{1}{\deg(g)}\sum_{v \sim_T o} f(T, o, v)\bigg|(T,o)_h \equiv g}.
\end{equation*}
With $\hat{v}$ being a node chosen uniformly at random among the nodes $v\sim_T
o$ adjacent to the root in $[T,o] \sim \tilde{\mu}$, we may rewrite the above expression as follows,
\begin{equation}
  \label{eq:ugwt-unim_2}
  \evwrt{\mu}{\sum_{v \sim_T o} f(T, o, v)}  = d \evwrt{\tilde{\mu}}{f(T,o,\hat{v})}.
\end{equation}
Note that, $\tilde{\mu}$--almost surely, $\deg_T(o) > 0$ and $\hat{v}$ is well
defined. 
Now, we find the distribution of $[T,o,\hat{v}] \in \mTb_{**}$ when $[T,o] \sim
\tilde{\mu}$ and $\hat{v}$ is chosen uniformly at random among the neighbors of
the root, as was defined above. 

%For $t, t' \in \edgemark \times \mTb_*^{h-1}$, define $H_{t,t'} \in \mTb_{**}$
%such that $H_{t,t'}[o', o] = t$ and $H_{t,t'}[o,o'] = t'$.

In order to do so,  we define the
probability measure $\nu \in \mP(\mTb_{**})$ to be the law of $[H, o, o']$ where
$H$ is a connected random marked tree with two distinguished adjacent vertices $o$ and
$o'$, defined as follows. We first sample $t,t'$
from the distribution $\pi_P(t,t') = e_P(t,t') / d$, and construct $H$ such that
$H(o',o) = H(o',o)_{h-1} \equiv t$ and $H(o,o') = H(o,o')_{h-1} \equiv t'$.
Then, similar to the construction of 
$\ugwt_h(P)$, we extend $H(o',o)$ and $H(o,o')$ inductively to construct $H$.
More precisely, first we sample $\tilde{t}$ from $\hP_{t,t'}(.)$ and use it to add
at most one layer to $H(o',o)_{h-1}$ so that $H(o',o)_h \equiv \tilde{t}$.
Similarly, we sample $\tilde{t}'$ from $\hP_{t',t}(.)$ and use it to add at most
one layer to $H(o,o')_{h-1}$ so that $H(o,o')_h \equiv \tilde{t}'$. Next,
independently for $v \sim_H o, v \neq o'$, we sample $\tilde{t}$ from
$\hP_{H[o,v]_{h-1}, H[v,o]_{h-1}}(.)$ and use it to add at most one layer to
$H(o,v)_{h-1}$ such that $H(o,v)_h \equiv \tilde{t}$. We apply the same
procedure to $w \sim_H o', w \neq o$. We continue this procedure inductively and
define $\nu$ to be the law of $[H, o, o']$.

We now claim that if $[T,o]$ has distribution $\tilde{\mu}$ and $\hat{\nu}$ is
chosen uniformly at random among the neighbors of the root in $T$ as above, then $[T,o,
\hat{v}]$ has distribution $\nu$. Before proving this, we show how it
completes the proof of the unimodularity of $\mu$. Note that, with this claim proved, 
%motivated from~\eqref{eq:ugwt-unim_2}, 
%we have
\eqref{eq:ugwt-unim_2} becomes
\begin{equation*}
    \evwrt{\mu}{\sum_{v \sim_T o} f(T, o, v)} = d \evwrt{\nu}{f(H, o, o')}.
  \end{equation*}
  Similarly, we have
\begin{equation*}
  \evwrt{\mu}{\sum_{v \sim_T o} f(T, v, o)} = d \evwrt{\nu}{f(H, o', o)}.
\end{equation*}
However, the admissibility of $P$ implies that $\pi_P(t,t') = \pi_P(t', t)$ for
all $t, t' \in \edgemark \times \mTb_*^{h-1}$.
Therefore, $\nu$ is symmetric in the sense that $[H, o, o']$ and $[H, o', o]$
have the same distribution. Therefore, we have $    \evwrt{\mu}{\sum_{v \sim_T
    o} f(T, o, v)} =     \evwrt{\mu}{\sum_{v \sim_T o} f(T, v, o)}$, which
is precisely what we needed to show. 

Therefore, it remains to prove that with $[T,o] \sim \tilde{\mu}$ and $\hat{v}$
defined as above, $[T,o,\hat{v}] \sim \nu$. First, we claim that since $[T,o]_h\sim \tP$,
%from {\color{blue} [Lemma 2 in the document of 2019-03-11, with
%  label \verb+lem:pi-p_tP+]},
we have $(T[\hat{v},o]_{h-1}, T[o,\hat{v}]_{h-1})$
has distribution $\pi_P$.
In order to show this, note that due to the definition of $\tP$ above, for $[T,o] \sim \tP$, we have
$\deg_T(o) \geq 1$ almost surely.  Let $Q$ be the distribution of   $(T[\hat{v} ,o]_{h-1}, T[o,\hat{v}]_{h-1})$ with $[T,o]$ and
$\hat{v}$ as stated. Then, for $t, t' \in \edgemark \times \mTb_*^{h-1}$, we have
\begin{align*}
  Q(t, t') &= \sum_{[T,o] \in \mTb_*^h: \deg_T(o) \geq 1} \tP([T,o]) \frac{E_h(t,t') (T,o)}{\deg_T(o)} \\
           &= \sum_{[T,o] \in \mTb_*^h: \deg_T(o) \geq 1} \frac{P([T,o]) \deg_T(o)}{d} \frac{E_h(t,t') (T,o)}{\deg_T(o)} \\
  &= \frac{e_P(t,t')}{d} = \pi_P(t,t'),
\end{align*}
which completes the proof of our claim.
%\color{red}
%Note to Payam: Please include
%in this document a proof of the lemma mentioned in 
%the blue text.
%\color{black}
%\pres{brought the proof of the mentioned lemma inside the text.}
This, in particular, implies that, $\tilde{\mu}$--almost
surely, we have $\pi_P(T[\hat{v},o]_{h-1}, T[o,\hat{v}]_{h-1}) > 0$.  Moreover, we
claim that for $t, t' \in \mTb_*^{h-1}$ such that $\pi_P(t,t') > 0$ and $\tilde{t}
\in \mTb_*^h$ such that $\tilde{t}_{h-1} = t$, we have
\begin{equation}
  \label{eq:ugwt-unim_tildemu-Tvo}
  \prwrt{\tilde{\mu}}{T[\hat{v}, o]_h = \tilde{t} \, \big| \,  T[o,\hat{v}]_{h-1} = t', T[\hat{v}, o]_{h-1} = t} = \hP_{t,t'}(\tilde{t}). 
\end{equation}
In order to show this, first note that, as was mentioned above, we have
\begin{equation}
  \label{eq:ugwt-unim_3}
  \prwrt{\tilde{\mu}}{T[o,\hat{v}]_{h-1} = t', T[\hat{v}, o]_{h-1} = t} = \pi_P(t,t'). 
\end{equation}
On the other hand, we have
\begin{align*}
  \prwrt{\tilde{\mu}}{T[\hat{v}, o]_h = \tilde{t}, T[o,\hat{v}]_{h-1} = t'} &\stackrel{(a)}{=} \prwrt{\tilde{\mu}}{[T,o]_h = \tilde{t} \oplus t'} \frac{1}{\deg(\tilde{t} \oplus t')} E_{h+1, h}(\tilde{t}, t')(\tilde{t} \oplus t') \\
                                                                            &= \frac{\tP(\tilde{t} \oplus t')}{\deg(\tilde{t} \oplus t')} E_{h+1, h}(\tilde{t}, t')(\tilde{t} \oplus t') \\
  &= \frac{1}{d} P(\tilde{t} \oplus t') E_{h+1, h}(\tilde{t}, t')(\tilde{t} \oplus t'),
\end{align*}
where  $(a)$ is obtained by employing the assumption that $\hat{v}$ is chosen
uniformly at random among the neighbors of the root,  and the fact that  conditioned on $[T,o]_h = \tilde{t}
\oplus t'$, there are precisely $E_{h+1, h}(\tilde{t}, t')(\tilde{t} \oplus t')$
many $v \sim_T o $ such that $T[v, o]_h = \tilde{t}$ and $T[o,v]_{h-1} = t'$.
%{\color{blue} [comment: the notation $E_{h+1, h}(\tilde{t}, t')(\tilde{t} \oplus
%  t')$ was defined at the beginning of Appendix B in the 2018-11-25 version.
%  Recall that with $[\tilde{T}, \tilde{o}] = \tilde{t} \oplus t'$, $E_{h+1,
 %   h}(\tilde{t}, t')(\tilde{t} \oplus t')$ denotes the number of $v
%  \sim_{\tilde{T}} \tilde{o}$ such that $\tilde{T}(v, \tilde{o})_h \equiv
%  \tilde{t} $ and $\tilde{T}(\tilde{o}, v)_{h-1} \equiv t'$.]}
Here, $\deg(\tilde{t} \oplus t')$ denotes the degree at the root in $\tilde{t}
\oplus t'$. Using this together with \eqref{eq:ugwt-unim_3}, we get
\begin{equation}
  \label{eq:ugwt-unim_4}
  \begin{aligned}
    \prwrt{\tilde{\mu}}{T[\hat{v}, o]_h = \tilde{t} \, \big| \,  T[o,\hat{v}]_{h-1} = t', T[\hat{v}, o]_{h-1} = t} 
    &= \frac{P(\tilde{t} \oplus t') E_{h+1, h}(\tilde{t}, t')(\tilde{t} \oplus t')}{d \pi_P(t,t')} \\
    &= \frac{P(\tilde{t} \oplus t') E_{h+1, h}(\tilde{t}, t')(\tilde{t} \oplus t')}{e_P(t,t')}.
    \end{aligned}
  \end{equation}
Now, if $\tilde{o}$ denotes the root in $\tilde{t} \oplus t'$, we have 
\begin{align*}
  E_{h+1, h}(\tilde{t}, t')(\tilde{t} \oplus t') & \stackrel{(a)}{=} |\{ v \sim_{\tilde{t} \oplus t'} \tilde{o}: (\tilde{t} \oplus t')(\tilde{o},v)_{h-1} \equiv t', \xi_{\tilde{t} \oplus t'}(v,\tilde{o}) = \tilde{t}[m]\}| \\
&\stackrel{(b)}{=} |\{ v \sim_{\tilde{t} \oplus t'} \tilde{o}: (\tilde{t} \oplus t')(\tilde{o},v)_{h-1} \equiv t', \xi_{\tilde{t} \oplus t'}(v,\tilde{o}) =t[m]\}| \\
&\stackrel{(c)}{=} E_h(t,t')(\tilde{t} \oplus t'),
\end{align*}
where in $(a)$ we have used Lemma~\ref{lem:Eh-1+Nh}, 
  %\color{red}
  %Note to Payam:
   %Please identify which lemma you are referring to here and (I think it is in the first appendix here) and put in the appropriate reference. If the lemma you had in mind is not already in this document, please provide it in the document.
   %\color{black}
   %\pres{added appropriate reference}
  $(b)$ is implied by the fact that $t[m] =  \tilde{t}[m]$,
and in $(c)$ we have again used Lemma~\ref{lem:Eh-1+Nh}.
Substituting this into~\eqref{eq:ugwt-unim_4} and comparing with the definition
of $\hP_{t,t'}$, we arrive at~\eqref{eq:ugwt-unim_tildemu-Tvo}.

So far, we have shown that the distribution of $(T[o,\hat{v}]_{h-1}, T[\hat{v},
o]_h)$ is the same as that of $(H[o,o']_{h-1}, H[o',o]_h)$ when $[H,o,o'] \sim
\nu$. Observe that $T[o,\hat{v}]_{h-1}$ and $T[\hat{v},o]_h$ together form
$[T,o]_h$. Moreover, by definition,  conditioned on $(T,o)_h$, $(T,o)$ is
constructed using the $\hP_{t,t'}(.)$ distributions,  in a
way similar to the process of defining $(H, o, o')$ above. Consequently, the
distribution of $[T,o, \hat{v}]$ is identical to $\nu$. As was discussed above,
this completes the proof of the unimodularity of $\ugwt_h(P)$. 
\end{proof}

\section{Proof of Proposition~\ref{prop:ugwthP-markov}}
\label{sec:prop-markov-galt}

In this section, we prove Proposition~\ref{prop:ugwthP-markov}.

\begin{proof}[Proof of Proposition~\ref{prop:ugwthP-markov}]

Let $\mu := \ugwt_h(P)$.
 Using induction, it suffices to show \eqref{eq:prop-consistency} only for $k = h+1$, i.e.\ with $Q := \mu_{h+1}$, we claim that
  \begin{equation}
    \label{eq:consistency-h-h+1}
    \ugwt_{h+1}(Q) = \mu.
  \end{equation}
  Recall from Section~\ref{sec:markovian-ugwt} that $\mu$ is the law of $[T,o]$
  where  $(T, o)_h$ is sampled from $P$ and
  the distributions $(\hP_{t,t'}: t,t' \in \edgemark \times \mTb_*^{h-1})$ are
  used to extend depth $h-1$ rooted trees to depth $h$ rooted trees in a
  recursive fashion. However, this process is equivalent to the following:
  first, we sample $(T, o)_{h+1}$ using $Q$ and then recursively use
  $(\hP_{t,t'}: t,t' \in \edgemark \times \mTb_*^{h-1})$ to extend depth $h-1$
  trees to depth $h$ trees, starting from nodes at depth $2$. More precisely,
  for $v$ being an offspring of the root and $w$ being an offspring of $v$, we
  extend $T(v,w)_{h-1}$ to $T(v,w)_h$ using $\hP_{T[v,w]_{h-1}, T[w,v]_{h-1}}$.
  This is done independently for all nodes $w$ with depth $2$.
  Equivalently,  for each offspring $v$ of the root, $T(o,v)_h$ is extended to
  $T(o,v)_{h+1}$, independent from all other offspring $v'$ of the root.
  However, motivated by the above discussion, in order to extend $T(o,v)_h$, we
  need to know $T[v,w]_{h-1}$ and $T[w,v]_{h-1}$ for the offspring $w$ of $v$. But
  this is known if we are given $T(o,v)_h$ and $T(v,o)_h$ (in fact, it is easy
  to see that even knowing $T(o,v)_h$ and $T(v,o)_{h-2}$ is sufficient). In
  other words, the distribution of $T[o,v]_{h+1}$ is uniquely determined by
  knowing $T[o,v]_h$ and $T[v,o]_h$. Motivated by this, for $s, s' \in \edgemark
  \times  \mTb_*^h$ such that $e_Q(s,s') = \evwrt{\mu}{E_{h+1}(s,s')(T,o)} > 0$
  and $\tilde{s} \in \mTb_*^{h+1}$ such that $\tilde{s}_h = s$, define
  $\tP_{s,s'}(\tilde{s})$ to be the probability of $T(o,v)_{h+1} \equiv
  \tilde{s}$ given $T(o,v)_h \equiv s$ and $T(v,o)_h \equiv s'$. The
  unimodularity of $\mu$ implies that if $e_Q(s,s') = 0$ for some $s,s' \in
  \edgemark \times \mTb_*^h$, the probability under $\mu$ of observing a node
  $w$ with parent $v$ such that $T(v,w)_h \equiv s$ and $T(w,v)_h \equiv s'$ is
  zero; therefore, we may define $\tP_{s,s'}$ arbitrarily for such $s,s'$.
  Continuing this argument recursively for nodes at higher depths, we realize
  that $\mu$ is the law of $(T,o)$ where $(T,o)_{h+1}$ is sampled from $Q$ and
  then $(\tP_{s,s'}(.): s,s' \in \edgemark \times \mTb_*^h)$ is used to extend
  subtrees of depth $h$ to subtrees of depth $h+1$.
  %, similar to the way
  %$(\hP_{t,t'}(.): t, t' \in \edgemark \times \mTb_*^{h-1})$  are used to extend
  %subtrees of depth $h-1$ to subtrees of depth $h$.
  Comparing this with the
  construction of $\ugwt_{h+1}(Q)$, we realize that in order to show
  \eqref{eq:consistency-h-h+1}, it suffices to show that for every  $s, s' \in
  \edgemark \times \mTb_*^h$ with $e_Q(s,s') > 0$,  and for all $\tilde{s} \in \edgemark \times \mTb_*^{h+1}$, we have
  \begin{equation}
    \label{eq:consistency-hQ-tP}
    \hQ_{s,s'}(\tilde{s}) = \tP_{s,s'}(\tilde{s}),
  \end{equation}
  where $\hQ_{s,s'}$ is defined using \eqref{eq:size-biased-def} based on the distribution $Q$. More precisely, for $s,s' \in \edgemark \times \mTb_*^h$ such that $e_Q(s,s') > 0$, and $\tilde{s} \in \edgemark \times \mTb_*^{h+1}$, we have
  \begin{equation}
    \label{eq:hQ-expanded}
    \hQ_{s,s'}(\tilde{s}) = \one{\tilde{s}_h = s} \frac{Q(\tilde{s} \oplus s') E_{h+1}(s,s')(\tilde{s} \oplus s')}{e_Q(s,s')}.
  \end{equation}

  We now fix $s,s' \in \edgemark \times \mTb_*^h$ and show \eqref{eq:consistency-hQ-tP}. Without loss of generality, we may assume that $\tilde{s}_h = s$, since otherwise both  sides of \eqref{eq:consistency-hQ-tP} are zero. We claim that
  \begin{equation}
    \label{eq:tP-frac-ev}
    \tP_{s,s'}(\tilde{s}) = \frac{\evwrt{\mu}{E_{h+1, h+2}(s', \tilde{s})(T, o)}}{\evwrt{\mu}{E_{h+1}(s', s)(T, o)}}.
  \end{equation}
  To see this, note that
  \begin{align*}
    &\evwrt{\mu}{E_{h+1, h+2}(s', \tilde{s})(T, o)} = \sum_{r \in \mTb_*^{h+1}} Q(r) \evwrt{\mu}{\sum_{v \sim_T o} \one{T(v,o)_h \equiv s', T(o,v)_{h+1} \equiv \tilde{s}} \Bigg | (T, o)_{h+1} \equiv r} \\
    &\qquad= \sum_{r \in \mTb_*^{h+1}} Q(r) \evwrt{\mu}{\sum_{v \sim_T o} \one{T(v,o)_h \equiv s', T(o,v)_h \equiv s} \one{T(o,v)_{h+1} \equiv \tilde{s}} \Bigg | (T, o)_{h+1} \equiv r}.
  \end{align*}
  %Note that for $v \sim_T o$ such that $T(v,o)_h \equiv s'$ and $T(o,v)_h \equiv s$, we have $T(o,v)_{h+1} \equiv \tilde{s}$ with probability $\tP_{s,s'}(\tilde{s})$, independent from $(T,o)_{h+1}$. 
  Note that the event $\{ T(o,v)_{h+1} \equiv \tilde{s} \}$
  is conditionally independent of the event
  $\{(T, o)_{h+1} \equiv r\}$, given the event
  $\{ T(v,o)_h \equiv s', T(o,v)_h \equiv s \}$.
  %\color{red}
  %The preceding sentence is a rephrasing of what I think you intended to say. Check.
  %\color{black}
  %\pres{correct.}
  Therefore,
  \begin{equation}
    \label{eq:tP-hQ-proof-1}
    \evwrt{\mu}{E_{h+1, h+2}(s', \tilde{s})(T, o)} = \sum_{r \in \mTb_*^{h+1}} Q(r) E_{h+1}(s',s)(r) \tP_{s,s'}(\tilde{s}) =  \evwrt{\mu}{E_{h+1}(s', s)(T, o)} \tP_{s,s'}(\tilde{s}).
  \end{equation}
  Using the unimodularity of $\mu$, we have
  \begin{equation}
    \label{eq:E-s'-s-eQ-s-s'}
    \evwrt{\mu}{E_{h+1}(s', s)(T, o)} = \evwrt{\mu}{E_{h+1}(s, s')(T, o)} = e_Q(s,s').
  \end{equation}
  Thereby,  $e_Q(s,s') > 0$ implies  that $\evwrt{\mu}{E_{h+1}(s', s)(T, o)} > 0$. Therefore, dividing both sides of \eqref{eq:tP-hQ-proof-1} by $\evwrt{\mu}{E_{h+1}(s', s)(T, o)}$, we arrive at \eqref{eq:tP-frac-ev}. Now, we simplify the right hand side of \eqref{eq:tP-frac-ev} to establish \eqref{eq:consistency-hQ-tP}. Using the  unimodularity of $\mu$ for the numerator, we have
  \begin{equation*}
    \evwrt{\mu}{E_{h+1,h+2}(s',\tilde{s})(T,o)} = \evwrt{\mu}{E_{h+2,h+1}(\tilde{s}, s')(T, o)}.
  \end{equation*}
  Observe that $E_{h+2,h+1}(\tilde{s}, s')(T, o) > 0$ iff $(T, o)_{h+1} \equiv
  \tilde{s} \oplus s'$.
  % reason: 
  %Indeed, if $(T, o)_{h+1} \equiv \tilde{s} \oplus s'$, we have
  %$E_{h+2,h+1}(\tilde{s}, s')(T, o) \geq 1$. To see the other direction, note
  %that if $E_{h+2,h+1}(\tilde{s}, s')(T, o) >0$, there exists $v \sim_T o$ such
  %that $T(v,o)_{h+1} \equiv \tilde{s}$ and $T(o,v)_h \equiv s'$, which implies
  %that $(T, o)_{h+1} \equiv \tilde{s} \oplus s'$.
  On the other hand, if $(T,o)_{h+1} \equiv \tilde{s} \oplus s'$, we have $E_{h+2,h+1}(\tilde{s}, s')(T, o) = E_{h+2,h+1}(\tilde{s}, s')(\tilde{s} \oplus s')$. Consequently,
  \begin{equation}
    \label{eq:ev-mu-s'-tildes-T-o}
    \begin{aligned}
      \evwrt{\mu}{E_{h+1,h+2}(s',\tilde{s})(T,o)} &= \prwrt{\mu}{(T,o)_{h+1} \equiv \tilde{s} \oplus s'} E_{h+2,h+1}(\tilde{s}, s')(\tilde{s} \oplus s') \\
      &= Q(\tilde{s} \oplus s') E_{h+2,h+1}(\tilde{s}, s')(\tilde{s} \oplus s').
    \end{aligned}
  \end{equation}
  Note that $\tilde{s} \oplus s'$ by construction has the property that
  $E_{h+2, h+1}(\tilde{s}, s')(\tilde{s} \oplus s') \geq 1$. Thereby,
  Lemma~\ref{lem:Eh-1+Nh} implies that $E_{h+2, h+1}(\tilde{s}, s')(\tilde{s}
  \oplus s') = |\{ v \sim_{\tilde{s} \oplus s'} o: (\tilde{s} \oplus s')(o,v)_h
  \equiv s', \xi_{\tilde{s} \oplus s'}(v, o) = \tilde{s}[m] \}|$. Here, $o$
  denotes the root in $\tilde{s} \oplus s'$. Likewise, since $\tilde{s}_h = s$,
  $E_{h+1}(s, s')(\tilde{s} \oplus s') \geq 1$, and another usage of
  Lemma~\ref{lem:Eh-1+Nh} implies that  $E_{h+1}(s, s')(\tilde{s} \oplus s') = |\{ v
  \sim_{\tilde{s} \oplus s'} o: (\tilde{s} \oplus s')(o,v)_h \equiv s',
  \xi_{\tilde{s} \oplus s'}(v, o) = s[m] \}|$. Also, $\tilde{s}_h = s$ in
  particular means $s[m] = \tilde{s}[m]$. Therefore, $E_{h+2, h+1}(\tilde{s},
  s')(\tilde{s} \oplus s') = E_{h+1}(s, s')(\tilde{s} \oplus s')$.
  Substituting into \eqref{eq:ev-mu-s'-tildes-T-o}, we get
  \begin{equation}
    \label{eq:ev-mu-s',tilde-s-simplified}
    \evwrt{\mu}{E_{h+1,h+2}(s',\tilde{s})(T,o)} = Q(\tilde{s} \oplus s') E_{h+1}(s, s')(\tilde{s} \oplus s').
  \end{equation}
  %On the other hand, the unimodularity of $\mu$ implies that
  %$\evwrt{\mu}{E_{h+1}(s', s)(T,o)} = \evwrt{\mu}{E_{h+1}(s, s')(T,o)} =
  %e_Q(s,s')$.
  Putting \eqref{eq:E-s'-s-eQ-s-s'} and \eqref{eq:ev-mu-s',tilde-s-simplified} back into
  \eqref{eq:tP-frac-ev} and comparing with \eqref{eq:hQ-expanded}, we arrive at
  \eqref{eq:consistency-hQ-tP}, which completes the
  proof.
\end{proof}

\section{Proof of Lemma~\ref{lem:PPh--PtildePh+1}}
\label{sec:proof-lem-Ptilde-Ph+1}

In this section we prove Lemma~\ref{lem:PPh--PtildePh+1}.
First, we state the following lemma from
\cite{bordenave2015large} which will be useful in the proof. 

\begin{lem}[Lemma~5.4 in \cite{bordenave2015large}]
  \label{lem:pllogl--pllogp}
  Let $P=\{p_x,\,x\in\mX\}$ be a probability measure on a discrete space $\mX$
  such that $H(P) < \infty$. Let $(\ell_x)_{x \in \mX}$ be  a sequence with
  $\ell_x \in \integers_+$, $x \in \mX$, such that $ \sum_{x} p_x \ell_x \log \ell_x < \infty$. Then
  $ - \sum_{x} p_x \ell_x \log  p_x < \infty$.
\end{lem}

\begin{proof}[Proof of Lemma~\ref{lem:PPh--PtildePh+1}]
By Lemma~\ref{lem:unimodular-is-admissible},
since $\mu$ is unimodular, $\tP$ is admissible. 
Also,
\begin{equation*}
 \evwrt{\tP}{\deg_T(o) \log \deg_T(o)} = \evwrt{P}{\deg_T(o) \log
  \deg_T(o)} < \infty.
\end{equation*}
Therefore, we only need to verify that $H(\tP) < \infty$.
Define $\nu := \ugwt_h(P)$ and let $P' := \nu_{h+1}\in \mP(\mTb_*^{h+1})$ be the
distribution of the $h+1$--neighborhood of the root in $\nu$.
Here we have again used Lemma~\ref{lem:unimodular-is-admissible} to note that the unimodularity of $\mu$ implies that $P$ is admissible, and hence
$\ugwt_h(P)$ is well-defined. 
Now, we claim that 
\begin{equation}
  \label{eq:tP-log-P'<infty-claim}
  \sum_{s \in \mTb_*^{h+1}} \tP(s) \log \frac{1}{P'(s)} < \infty.
\end{equation}
Using Gibbs' inequality, this  implies that $H(\tP) <
\infty$ and completes the proof. Hence, it suffices to show \eqref{eq:tP-log-P'<infty-claim}.
% for some details on Gibbs' inequality, see the note
% 2018-09-28_PN_Gibbs-for-lemma-Ph-Ph+1.pdf

  % \item Let  $P' := \nu_{h+1}$. 
  % \item First, we show that 
  % \item If we show this, using the identity
  %   \begin{equation}
  %     \label{eq:tP-P'-identity}
  %     \tP(s) \log \frac{1}{\tP(s)} + \tP(s) \log \frac{\tP(s)}{P'(s)} = \tP(s) \log \frac{1}{P'(s)} \qquad \forall s \in \mTb_*^{h+1},
  %   \end{equation}
  %   \begin{equation*}
  %     \begin{aligned}
  %       \sum_{s \in \mTb_*^{h+1}} \left | \tP(s) \log \frac{\tP(s)}{P'(s)} \right |  &=  \sum_{s \in \mTb_*^{h+1}} P'(s) \left |\frac{\tP(s)}{P'(s)} \log \frac{\tP(s)}{P'(s)} \right | \\
  %       &\leq \sum_{s \in \mTb_*^{h+1}} P'(s) \left( \frac{1}{e} + \tP(s) \log \frac{1}{P'(s)} \right) \\
  %       &< \infty
  %     \end{aligned}
  %   \end{equation*}
  %   where the first inequality follows from \eqref{eq:tP-log-P'<infty-claim} and
  %   the fact that $x \log x \geq -1/e$ for
  %   $x > 0$.
  % \item Consequently, using 
  %   \begin{equation*}
  %     \sum_{s \in \mTb_*^{h+1}} \tP(s) \log \frac{\tP(s)}{P'(s)}= H(\tP \Vert P')   \geq 0,
  %   \end{equation*}
  %   together with \eqref{eq:tP-P'-identity}, we conclude that $H(\tP) < \infty$.

Recall that, by the definition of $\ugwt_h(P)$, for $[T, o] \in \mTb_*^{h+1}$ we have
    \begin{equation}
      \label{eq:tP-T-o}
      P'([T, o]) = C P([T, o]_h) \prod_{v \sim_T o} \hP_{T[ o, v]_{h-1}, T[v,o]_{h-1}}(T[o,v]_h),
    \end{equation}
    where $C \geq 1$ is a constant that only depends on $[T,o]$ and counts the
    number of extensions of $[T,o]_h$ that result in $[T,o]$.
      Now, take $[T, o] \in \mTb_*^{h+1}$ such that 
$P'([T,o]) > 0$ and note that, using \eqref{eq:tP-T-o}, we have $P([T,o]_h) > 0$. This,
together with the fact that $P$ is admissible, implies that for all $v \sim_T o$
we have 
\begin{equation}
  \label{eq:P-T-o>0--ep>0}
  \begin{aligned}
    e_P(T[o,v]_{h-1}, T[v,o]_{h-1}) &=   e_P(T[v,o]_{h-1}, T[o,v]_{h-1}) \\
    & \geq P([T,o]_h) E_h(T[v,o]_{h-1}, T[o,v]_{h-1})(T,o) \\
    &\geq P([T,o]_h) > 0.
    \end{aligned}
  \end{equation}
  Therefore, 
for $v \sim_T o $,
with $t := T[o,v]_{h-1}$ and $t' := T[v,o]_{h-1}$, we have  $e_P(t,t') > 0$ and,
using~\eqref{eq:size-biased-def}, 
\begin{equation}
  \label{eq:hP-t-t'-lowerbound}
  \begin{aligned}
    \hP_{t,t'}(T[o,v]_h) &= \frac{P([T, v]_h) E_h(t, t')([T, v]_h)}{e_P(t,t')} \\
    &\geq \frac{P([T,v]_h)}{e_P(t,t')},
  \end{aligned}
\end{equation}
where the last line follows from the fact that $E_h(t,t')([T,v]_h) \geq 1$.
% this is simply because $o$ counts for it, i.e. $\phi_G^h(v,o) = (t,t')$ by
% definition.
Note that, as we have assumed $P'([T,o]) > 0$, from \eqref{eq:tP-T-o} we
have $\hP_{t,t'}(T[o,v]_h) > 0$. Thereby, the first line
in~\eqref{eq:hP-t-t'-lowerbound} implies that $P([T,v]_h) > 0$. So far, we have
shown that for $[T,o] \in \mTb_*^{h+1}$ such that $P'([T,o]) > 0$, for all $v
\sim_T o$ we have $P([T,v]_h) > 0$ and
\begin{equation*}
  \hP_{T[ o, v]_{h-1},
  T[v,o]_{h-1}}(T[o,v]_h) \geq \frac{P([T,v]_h)}{ e_P(T[ o, v]_{h-1},  T[v,o]_{h-1})}.
\end{equation*}

% Moreover, using~\eqref{eq:size-biased-def}, 
% for $v \sim_T o $,
% with $t := T[o,v]_{h-1}$ and $t' := T[v,o]_{h-1}$, if $e_P(t,t') > 0$, we have
% \begin{equation}
%   \label{eq:hP-t-t'-lowerbound}
%   \begin{aligned}
%     \hP_{t,t'}(T[o,v]_h) &= \frac{P([T, v]_h) E_h(t, t')([T, v]_h)}{e_P(t,t')} \\
%     &\geq \frac{P([T,v]_h)}{e_P(t,t')},
%   \end{aligned}
% \end{equation}
% where the last line follows from the fact that $E_h(t,t')([T,v]_h) \geq 1$.
% % this is simply because $o$ counts for it, i.e. $\phi_G^h(v,o) = (t,t')$ by
% % definition.
% Now, if $[T, o] \in \mTb_*^{h+1}$ is such that
% $P'([T,o]) > 0$ then using \eqref{eq:tP-T-o} we have $P([T,o]_h) > 0$. This,
% together with the fact that $P$ is admissible, implies that for all $v \sim_T o$
% we have 
% \begin{equation}
%   \label{eq:P-T-o>0--ep>0}
%   \begin{aligned}
%     e_P(T[o,v]_{h-1}, T[v,o]_{h-1}) &=   e_P(T[v,o]_{h-1}, T[o,v]_{h-1}) \\
%     & \geq P([T,o]_h) E_h(T[v,o]_{h-1}, T[o,v]_{h-1})(T,o) \\
%     &\geq P([T,o]_h) > 0.
%     \end{aligned}
% \end{equation}
% Observe that 
% \eqref{eq:P-T-o>0--ep>0} implies that \eqref{eq:hP-t-t'-lowerbound} holds.
% Also, the first line in~\eqref{eq:hP-t-t'-lowerbound} suggests that $P([T,v]_h) > 0$
%\color{red}
%Note to Payam:
%I have no idea what the preceding sentence means and 
%cannot figure out what you are intending to say.
%Therefore I will check this appendix in the next phase.
%\color{black}
%\pres{modified this part to clarify}
Substituting this
  in \eqref{eq:tP-T-o}, we realize that for
$[T,o] \in \mTb_*^{h+1}$ with $P'([T,o]) > 0$, we have 
\begin{equation}
  \label{eq:log-1-tP-T-o--expansion}
  \log \frac{1}{P'([T,o])} \leq \log \frac{1}{P([T,o]_h)} + \sum_{v \sim_T o} \log \frac{1}{P([T,v]_h)} + \sum_{v \sim_T o} \log e_P(T[o,v]_{h-1}, T[v,o]_{h-1}).
\end{equation}

Next, we claim that $\tP \ll P'$. Observe that from \eqref{eq:tP-T-o}, for
$[T,o] \in \mTb_*^{h+1}$, $P'([T,o])
    = 0$  implies that either $P([T,o]_h) = 0$ or $P([T,o]_h) > 0$ and $\hP_{T[o,v]_{h-1},
      T[v,o]_{h-1}}(T[o,v]_h) = 0$ for some $v \sim_T o$.
But if $P([T,o]_h) > 0$, \eqref{eq:P-T-o>0--ep>0} implies that for all $v \sim_T
o$, we have  $e_P(T[o,v]_{h-1}, T[v,o]_{h-1}) > 0$. Therefore,
using \eqref{eq:hP-t-t'-lowerbound}, if $ \hP_{T[o,v]_{h-1},
      T[v,o]_{h-1}}(T[o,v]_h) = 0$ for some $v \sim_T o$, it must be
the case that $P([T,v]_h) = 0$. Consequently, $P'([T,o]) = 0$ implies that
either $P([T,o]_h) = 0 $ or $P([T,o]_h) > 0$ and for some $v \sim_T o$, we have
$P([T,v]_h) = 0$.
Note that since $\tP_h = P$, if $P([T,o]_h) = 0$, we have $\tP([T,o]) = 0$. Now,
we claim that if $P([T,v]_h) = 0$ for some $v \sim_T o$, then $\tP([T,o]) = 0$.
In order to establish this claim, using $\tP = \mu_{h+1}$, we have
\begin{equation*}
  \begin{aligned}
    \evwrt{\tP}{\sum_{v \sim_T o} \one{P([T,v]_h) = 0}} &= \int \sum_{v \sim_T o } \one{P([T,v]_h) = 0} d \mu([T,o]) \\
    &\stackrel{(a)}{=} \int \sum_{v \sim_T o} \one{P([T,o]_h) = 0} d \mu([T,o]) \\
    &= \int \deg_T(o) \one{P([T,o]_h = 0)} d\mu([T,o]) \\
    &\stackrel{(b)}{=} 0,
  \end{aligned}
\end{equation*}
where $(a)$ uses the unimodularity of $\mu$ and $(b)$ uses the fact that
$\mu_h = P$.
This means that for $\tP$--almost all $[T,o] \in \mTb_*^{h+1}$,
$P([T,v]_h) > 0$ for all $v \sim_T o$.
Equivalently, if $P([T,v]_h) = 0$ for $v \sim_T o$, we have $\tP([T,o]) = 0$. 
To sum up, we showed that for $[T,o] \in \mTb_*^{h+1}$, $P'([T,o]) = 0$ implies
$\tP([T,o]) = 0$ and hence $\tP \ll P'$. 
As a result, using this and \eqref{eq:log-1-tP-T-o--expansion}, we may write the
LHS of \eqref{eq:tP-log-P'<infty-claim} as
\begin{equation}
  \label{eq:sum-tP-log-P'--3-components}
\begin{aligned}
  \sum_{s \in \mTb_*^{h+1}} \tP(s) \log \frac{1}{P'(s)} &= \sum_{s \in \mTb_*^{h+1}: P'(s) > 0} \tP(s) \log \frac{1}{P'(s)} \\
                                                        &\leq \sum_{[T,o]\in \mTb_*^{h+1}: P'([T,o]) > 0} \tP([T,o]) \Bigg(  \log \frac{1}{P([T,o]_h)} + \sum_{v \sim_T o } \log \frac{1}{P([T,v]_h)} \\
  &\qquad \qquad + \sum_{v \sim_T o} \log e_P(T[o,v]_{h-1}, T[v,o]_{h-1}) \Bigg).
\end{aligned}
\end{equation}
%Since all the terms are nonnegative,
We may bound each component separately as follows.
First, note that the facts $\tP \ll P'$, $\tP = \mu_{h+1}$ and $P = \mu_h$
imply that 
\begin{equation}
  \label{eq:int-log-1-P-mu}
  \begin{aligned}
    \sum_{[T,o]\in \mTb_*^{h+1}: P'([T,o]) > 0} \tP([T,o]) \log \frac{1}{P([T,o]_h)} &= \int \log \frac{1}{P([T,o]_h)} d \mu([T,o]) \\
    &= - \sum_{s \in \mTb_*^h} P(s) \log P(s) = H(P) < \infty.
    \end{aligned}
\end{equation}
We also have
    \begin{equation}
      \label{eq:int-sum-1-P-T-v}
      \begin{aligned}
        \sum_{[T,o]\in \mTb_*^{h+1}: P'([T,o]) > 0} \tP([T,o]) \sum_{v \sim_T o} \log \frac{1}{P([T,v]_h)} &= \int \sum_{v \sim_T o} \log \frac{1}{P([T,v]_h)} d \mu([T,o]) \\
        &\stackrel{(a)}{=} \int \sum_{v \sim_T o} \log \frac{1}{P([T,o]_h)} d \mu([T,o])\\
        &= \int \deg_T(o) \log \frac{1}{P([T,o]_h)} d \mu([T,o]) \\
        &= -\sum_{[T,o] \in \mTb_*^h} \deg_T(o) P([T,o]) \log P([T,o]) \\
        &\stackrel{(b)}{<}  \infty
      \end{aligned}
    \end{equation}
    where $(a)$ follows from unimodularity of $\mu$ and $(b)$ follows from
    Lemma~\ref{lem:pllogl--pllogp} and the
    fact that since {$P$ is strongly admissible, i.e.} $P \in \mP_h$, we have  $\evwrt{P}{\deg_T(o) \log \deg_T(o)} <
    \infty$.
    %Recall that $\deg_s(0)$ denotes the degree at the root in $s$.
Finally, for the third component, note that since
    \begin{equation*}
      d = \evwrt{P}{\deg_T(o)} = \evwrt{P}{\sum_{t,t' \in \edgemark \times \mTb_*^{h-1}} E_h(t,t')(T,o)} = \sum_{t,t' \in \edgemark \times \mTb_*^{h-1}} e_P(t,t'),
    \end{equation*}
    we have $e_P(T[o,v]_{h-1}, T[v,o]_{h-1}) \leq d$ for all $[T,o] \in \mTb_*^{h+1}$
    and $v \sim_T o$. Consequently,
    \begin{equation}
      \label{eq:sum-tP-sum-ep--dlogd--bound}
    \begin{aligned}
      \sum_{[T,o]\in \mTb_*^{h+1}: P'([T,o]) > 0} \tP( [T,o]) &\sum_{v \sim_T o} \log e_P(T[o,v]_{h-1}, T[v,o]_{h-1}) \\
      &=  \int \sum_{v \sim_T o} \log e_P(T[o,v]_{h-1}, T[v,o]_{h-1}) d \mu([T,o])\\
      &\leq \int \deg_T(o) (\log d) d \mu([T,o]) = d \log d < \infty.
    \end{aligned}
  \end{equation}
  Putting~\eqref{eq:int-log-1-P-mu}, \eqref{eq:int-sum-1-P-T-v} and
  \eqref{eq:sum-tP-sum-ep--dlogd--bound} back in
  \eqref{eq:sum-tP-log-P'--3-components} we arrive at
  \eqref{eq:tP-log-P'<infty-claim}, which completes the proof.
    % \item For the third term, let $\mC := (\vermark \times \mTb_*^{h-1}) \times
  %   (\edgemark \times \mTb_*^{h-1})$ and note that
  %   \begin{align*}
  %     \sum_{c \in \mC} e_P(c) \log e_P(c) &= \sum_{c \in \mC} \sum_{[T, o] \in \mTb_*^h} P([T,o]) E_h(c)(T,o) \log e_P(c) \\
  %                                         &\stackrel{(a)}{=} \sum_{[T,o] \in \mTb_*^h} \sum_{c \in \mC} P([T,o]) E_h(c)(T,o) \log e_P(c) \\
  %                                         &=\sum_{[T,o] \in \mTb_*^h}  P([T,o]) \sum_{c \in \mC} E_h(c)(T,o) \log e_P(c) \\
  %                                         &=\sum_{[T,o] \in \mTb_*^h} P([T,o]) \sum_{v \sim_T o} \log e_P(T[v,o]_{h-1}, T[o,v]_{h-1}) \\
  %                                         &= \int \sum_{v \sim_T o} \log e_P(T[v,o]_{h-1}, T[o,v]_{h-1}) d\mu([T,o]) \\
  %                                         &\stackrel{(b)}{=} \int \sum_{v \sim_T o} \log e_P(T[o,v]_{h-1}, T[v,o]_{h-1}) d\mu([T,o])
  %                                         % &= \sum_{[T,o] \in \mTb_*^h} P([T,o]) \sum_{v \sim_T o} \log e_P(T[o,v]_{h-1}, T[v,o]_{h-1})
  %   \end{align*}
  %   where $(a)$ follows the assumption $\sum_{c \in \mC}|e_P(c) \log e_P(c)| <
  %   \infty$ and Fubini's theorem, and $(b)$ follows from the unimodularity of
  %   $\mu$.
  % \item Since $\sum_{c \in \mC} e_P(c) \log e_P(c)<\infty $, we realize that
  %   the integral of the last term is also finite.
\end{proof}

\section{Calculations for Deriving \eqref{eq:log-mGnmnun-Stirling}}
\label{sec:app-Stirling}

First note that since $u^{(n)}(\theta) / n \rightarrow q_\theta$ for all $\theta \in
\vermark$, we have
\begin{equation}
  \label{eq:n!-un!-HQ-on}
  \log \frac{n!}{\prod_{\theta \in \vermark} u^{(n)}(\theta)!} = n H(Q) + o(n).
\end{equation}
Furthermore, since $\mn(x,x') / n \rightarrow d_{x,x'}$ for $x \neq x'$,
we have 
\begin{equation}
  \label{eq:Stirling-2-power-term-simplification}
  \begin{aligned}
    \log 2^{\sum_{x < x'} \mn(x,x')} &= n \sum_{x < x'} \frac{\mn(x,x')}{n} \log 2 \\
    &= n \left( \sum_{x<x'} d_{x,x'} \log 2 + o(1) \right) \\
    &= n \sum_{x < x'} d_{x,x'} \log 2 + o(n).
  \end{aligned}
\end{equation}
Moreover, 
from Stirling's approximation for the factorial, for a positive integer $k$ we have $\log k! = k \log
k - k + O(\log k)$. Moreover, since for all $x \neq  x' \in \edgemark$ we have
$\mn(x,x') / n \rightarrow d_{x,x'} < \infty$ and for $x \in
  \edgemark$ we have $\mn(x,x) / n \rightarrow d_{x,x} / 2 < \infty$, we conclude that we have $\mn(x,x') =
O(n)$ for all $x, x' \in \edgemark$, and so
\begin{align*}
  &\log \frac{\frac{n(n-1)}{2}!}{\prod_{x\leq x' \in \edgemark} \mn(x,x')! \times \left ( \frac{n(n-1)}{2} - \snorm{\vmn}_1 \right)!} = \frac{n(n-1)}{2} \log \frac{n(n-1)}{2} - \frac{n(n-1)}{2} \\
  &\qquad \qquad  - \sum_{x \leq x'} \left(  \mn(x,x')  \log \mn(x,x') - \mn(x,x')  \right) \\
  &\qquad \qquad - \left[ \left( \frac{n(n-1)}{2} - \snorm{\vmn}_1 \right)  \log \left( \frac{n(n-1)}{2} - \snorm{\vmn}_1 \right) - \left( \frac{n(n-1)}{2} - \snorm{\vmn}_1 \right) \right] \\
  &\qquad \qquad + O(\log n) \\
  &= \frac{n(n-1)}{2} \log \frac{n(n-1)}{2} 
  %- \frac{n(n-1)}{2} 
  - \sum_{x \leq x'} \mn(x,x') \log \mn(x,x') 
  %+ \underbrace{\sum_{x \leq x'} \mn(x,x')}_{= \snorm{\vmn}_1} 
  \\
  &\qquad \qquad - \frac{n(n-1)}{2} \log \left( \frac{n(n-1)}{2} - \snorm{\vmn}_1 \right) + \snorm{\vmn}_1 \log \left( \frac{n(n-1)}{2} - \snorm{\vmn}_1 \right) + o(n)\\
  %&\qquad \qquad + \frac{n(n-1)}{2} - \snorm{\vmn}_1 + o(n) \\
  &= - \frac{n(n-1)}{2} \log\left( 1 - \frac{2\snorm{\vmn}_1}{n(n-1)} \right) - n \sum_{x \leq x'} \frac{\mn(x,x')}{n} \log \frac{\mn(x,x')}{n} - n \sum_{x \leq x'} \frac{\mn(x,x')}{n} \log n \\
  &\qquad \qquad + \snorm{\vmn}_{1} \log \left[ n^2 \left( \frac{n-1}{2n} - \frac{\snorm{\vmn}_1}{n^2} \right) \right] + o(n).\\
  \intertext{Using the facts that $2\snorm{\vmn}_1 / n(n-1) \rightarrow 0$ and $\log(1-x) = -x + O(x^2)$, this simplifies to} 
  &= -\frac{n(n-1)}{2} \left[ -\frac{2\snorm{\vmn}_1}{n(n-1)} + O\left( \frac{4\snorm{\vmn}_1^2}{(n(n-1))^2} \right) \right] - n \left( \sum_{x < x'} d_{x,x'} \log d_{x,x'} + \sum_{x} \frac{d_{x,x}}{2} \log \frac{d_{x,x}}{2} \right) \\
  &\qquad \qquad - \snorm{\vmn}_1 \log n + 2 \snorm{\vmn}_1 \log n + n \frac{\snorm{\vmn}_1}{n} \log \left( \frac{n-1}{2n} - \frac{\snorm{\vmn}_1}{n^2} \right) + o(n).\\
  \intertext{Since, by assumption, $\snorm{\vmn}_1/n \rightarrow \sum_{x<x'} d_{x,x'} + \sum_x d_{x,x}/2 = \sum_{x, x'} d_{x,x'} / 2$, this simplifies to }
  &= n \frac{\snorm{\vmn}_1}{n} + \underbrace{O\left( \frac{\snorm{\vmn}_1^2}{n(n-1)} \right)}_{O(1)} -n \left( \sum_{x < x'} d_{x,x'} \log d_{x,x'} + \sum_{x} \frac{d_{x,x}}{2} \log \frac{d_{x,x}}{2} \right) \\
  &\qquad \qquad + \snorm{\vmn}_1 \log n + n \left( \sum_{x,x'} \frac{d_{x,x'}}{2} \log \frac{1}{2} + o(1) \right) + o(n) \\
 % &= n \sum_{x,x'} \frac{d_{x,x'}}{2} - n \left( \sum_{x \neq x'} \frac{d_{x,x'}}{2} \log d_{x,x'} + \sum_x \frac{d_{x,x}}{2} \log d_{x,x} + \sum_x \frac{d_{x,x}}{2} \log \frac{1}{2} \right) \\
  %&\qquad \qquad + \snorm{\vmn}_1 \log n + n \left( \sum_{x,x'} \frac{d_{x,x'}}{2} \log \frac{1}{2} \right) + o(n) \\
  %&= n \left( \sum_{x,x'} \frac{d_{x,x'}}{2} - \frac{d_{x,x'}}{2} \log d_{x,x'} \right) + \snorm{\vmn}_1 \log n + n \sum_{x \neq x'} \frac{d_{x,x'}}{2} \log \frac{1}{2} + o(n)\\
  &= \snorm{\vmn}_1 \log n + n \sum_{x,x'} s(d_{x,x'}) - n \sum_{x < x'} d_{x,x'} \log 2 + o(n).
\end{align*}
Using this together with \eqref{eq:n!-un!-HQ-on} and
\eqref{eq:Stirling-2-power-term-simplification}, we get 
\begin{align*}
  \log |\mGnmnun| &= n H(Q) + \snorm{\vmn}_1 \log n + n \sum_{x,x'} s(d_{x,x'}) - n \sum_{x < x'} d_{x,x'} \log 2 \\
                  &\qquad \qquad + n \sum_{x < x'} d_{x,x'} \log 2 + o(n) \\
  &= \snorm{\vmn}_1 \log n + n H(Q) + n \sum_{x,x'} s(d_{x,x'}) + o(n),
\end{align*}
which is precisely what was stated in~\eqref{eq:log-mGnmnun-Stirling}.

%%% Local Variables:
%%% mode: latex
%%% TeX-master: "Note-41_BC-ent-arxiv"
%%% End:

%\section{Properties of Unimodular Galton-Watson Trees with given neighborhood}
\section{Proof of Proposition~\ref{prop:ugwthP-properties}}
\label{sec:prop-unim-galt}

In this section, we prove Proposition~\ref{prop:ugwthP-properties}.

\begin{proof}[Proof of Proposition~\ref{prop:ugwthP-properties}]
  Let $\mu := \ugwt_h(P)$. If $P$ has a finite support then,
  %unimodularity of $\mu$
  %follows from the fact that,
  as implied by
   Lemma~\ref{lem:P-finite-seq-converging} in
    Section~\ref{sec:weak-convg-to-admissible} and
  Proposition~\ref{prop:MCB-colored-unif-D-2h+1--converge} in
  Section~\ref{sec:connection-color-ugwt}, $\mu$ is sofic. If $P$ does not have
  a finite support then, along the lines of the proof of Proposition~\ref{prop:lower-bound}
  in Section~\ref{sec:lowerbound}, for $k > 1$, let $\muk$ be the law of
  $[T^{(k)}, o]$ obtained from $[T, o] \sim \mu$ as follows. For each vertex $v
  \in V(T)$, we remove all the edges connected to $v$ if $\deg_T(v) \geq k$.
  Then, we let $T^{(k)}$ denote the connected component of the root in the resulting forest. As was shown
  in the proof of Proposition~\ref{prop:lower-bound}, with
  $P_k := (\muk)_h$, as $k \rightarrow \infty$,  we have $P_k \Rightarrow P$ and $e_{P_k}(t,t') \rightarrow
  e_P(t,t')$ for all $t, t' \in \edgemark \times \mTb_*^{h-1}$  (see
  \eqref{eq:e-Pk->e-P}).
  Note that, from Lemma~\ref{lem:ugwt-P-is-unimodular} in
    Appendix~\ref{sec:UGWT-unimod}, $\mu$ is unimodular. Thereby, it is easy to
    see that $\mu^{(k)}$ is also unimodular. Furthermore, from Lemma~\ref{lem:unimodular-is-admissible}, $P_k$
  is admissible.
   The above discussion together with
  Lemma~\ref{lem:Pn-conv-P--ugwt-Pn-conv-ugwt-P} in
  Appendix~\ref{sec:UGWT-convg} implies that 
$\ugwt_h(P_k) \Rightarrow \mu$. On the other hand, as we have discussed above,
since $P_k$ has a finite support,  $\ugwt_h(P_k)$ is sofic. Therefore, a
diagonal argument implies that $\mu$ is also  sofic and completes the proof. 
%\color{red}
%Note to Payam: Why is it only being claimed that 
%$\mu$ is unimodular? The claim that is supposed to be proved is that $\mu$ is
%sofic.
%\color{black}
%\pres{sorry, the last part of the argument was missing which is now added.}
It is worth recalling that we have earlier directly shown the unimodularity of $\mu$ in Appendix~\ref{sec:UGWT-unimod}.
However, in general, being sofic might be a stronger property than being unimodular for all one knows at the moment.
%\color{red}
%Note to Payam: Is the preceding sentence refering to the proof that was already given in this document?
%If so, the phrasing of that sentence should be modified.
%\color{black}
%\pres{modified the sentence and added reference to Appendix~\ref{sec:UGWT-unimod}.}
\end{proof}

%%% Local Variables: 
%%% mode: latex
%%% TeX-master: "Note-41_BC-ent-arxiv.tex"
%%% End: 

%aa
%\bibliographystyle{alpha}
%\bibliography{refs.bib}
% \input{Note-41_BC-ent-arxiv.bbl}

% bb

\pagebreak
\printglossaries
%\printglossary[type=notation]
%\printglossary[type=term]

%aa
\end{document}